\newif\iflncs
\lncstrue
\iflncs
\documentclass{llncs}
\else
\documentclass{article}[11]
\usepackage{amsthm}
\usepackage{fullpage}
\fi

\usepackage{appendix}
\usepackage{amsmath}
\usepackage{amssymb}
\usepackage{ifpdf}
\usepackage{algorithmic}
\usepackage{subfig}
\usepackage{wrapfig}
\usepackage{cite}
\usepackage{colortbl}
\usepackage{tikz}
\usepackage{morefloats}
\usepackage[]{algorithm2e}
\newcommand\numberthis{\addtocounter{equation}{1}\tag{\theequation}}
\renewcommand{\Re}{\operatorname{Re}}
\renewcommand{\Im}{\operatorname{Im}}

\makeatletter
\newcommand{\@chapapp}{\relax}%
\makeatother

\usepackage{commands-tam}

\ifpdf

  \usepackage[pdftex]{epsfig}
  \usepackage[pdftex]{hyperref}

\else

    \usepackage[dvips]{epsfig}
    \newcommand{\href}[2]{#2}

\fi

\newif\ifabstract
\newif\iffull

\abstracttrue

\ifabstract
	\fullfalse
\else
	\fulltrue
\fi

\newtoks\magicAppendix
\magicAppendix={}
\newtoks\magictoks
\newif\iflater
\laterfalse
\ifabstract
\long\def\later#1{\magictoks={#1}%
  \edef\magictodo{\noexpand\magicAppendix={\the\magicAppendix \par
    \the\magictoks%
  }}
  \magictodo}
\long\def\both#1{\magictoks={#1}%
  \edef\magictodo{\noexpand\magicAppendix={\the\magicAppendix \par
    \noexpand\setcounter{theorem-preserve}{\noexpand\arabic{theorem}}%
    \noexpand\setcounter{theorem}{\arabic{theorem}}%
    \noexpand\setcounter{section-preserve}{\noexpand\arabic{section}}%
    \noexpand\setcounter{section}{\arabic{section}}%
	\noexpand\let\noexpand\oldsection=\noexpand\thesection
	\noexpand\def\noexpand\thesection{\thesection}
	\noexpand\let\noexpand\oldlabel=\noexpand\label
	\noexpand\let\noexpand\label=\noexpand\blank
    \the\magictoks%
    \noexpand\setcounter{theorem}{\noexpand\arabic{theorem-preserve}}%
    \noexpand\setcounter{section}{\noexpand\arabic{section-preserve}}%
	\noexpand\let\noexpand\thesection=\noexpand\oldsection
	\noexpand\let\noexpand\label=\noexpand\oldlabel
  }}
  \magictodo
  \the\magictoks}
\else
\long\def\later#1{#1}
\long\def\both#1{#1}
\fi
\long\def\magicappendix{
	\latertrue%
	\the\magicAppendix%
}

\iflncs
\else
\theoremstyle{definition}
\newtheorem{theorem}{Theorem}[section]
\newtheorem{lemma}[theorem]{Lemma}
\newtheorem{definition}[theorem]{Definition}

\fi
\newtheorem{observation}[theorem]{Observation}

\title{Computing in continuous space with self-assembling polygonal tiles (extended abstract)}

\author{
    Oscar Gilbert
        \thanks{Department of Mathematical Sciences, University of Arkansas, Fayetteville, AR, USA.
        \protect\url{oogilber@email.uark.edu}.
        This author's research was supported in part by National Science Foundation Grants CCF-1117672 and CCF-1422152.}
\and
    Jacob Hendricks
        \thanks{Department of Computer Science and Computer Engineering, University of Arkansas, Fayetteville, AR, USA.
        \protect\url{jhendric@uark.edu}.
        This author's research was supported in part by National Science Foundation Grants CCF-1117672 and CCF-1422152.}
\and
	Matthew J. Patitz
        \thanks{Department of Computer Science and Computer Engineering, University of Arkansas, Fayetteville, AR, USA.
        \protect\url{mpatitz@self-assembly.net}.
        This author's research was supported in part by National Science Foundation Grants CCF-1117672 and CCF-1422152.}
\and
    Trent A. Rogers
        \thanks{Department of Computer Science and Computer Engineering, University of Arkansas, Fayetteville, AR, USA.
        \protect\url{tar003@uark.edu}.
        This author's research was supported by the National Science Foundation Graduate Research Fellowship Program under Grant No. DGE-1450079, and National Science Foundation grants CCF-1117672 and CCF-1422152.}
}

\date{}
\institute{}

\begin{document}

\maketitle

\vspace{-10pt}
\begin{abstract}
In this paper we investigate the computational power of the polygonal tile assembly model (polygonal TAM) at temperature $1$, i.e. in non-cooperative systems.  The polygonal TAM is an extension of Winfree's abstract tile assembly model (aTAM) which not only allows for square tiles (as in the aTAM) but also allows for tile shapes which are arbitrary polygons. Although a number of self-assembly results have shown computational universality at temperature $1$, these are the first results to do so by fundamentally relying on tile placements in continuous, rather than discrete, space. With the square tiles of the aTAM, it is conjectured that the class of temperature $1$ systems is not computationally universal.  Here we show that for each $n > 6$, the class of systems whose tiles are the shape of the regular polygon $P$ with $n$ sides is computationally universal.  On the other hand, we show that the class of systems whose tiles consist of a regular polygon $P$ with $n \leq 6$ sides cannot compute using any known techniques.  In addition, we show a number of classes of systems whose tiles consist of a non-regular polygon with $n \geq 3$ sides are computationally universal.
\end{abstract}
\vspace{-10pt}

\pagebreak

\vspace{-20pt}
\section{Introduction}
\vspace{-10pt}

Self-assembly is a process by which systems that evolve based only on simple local interactions form.  Studying self-assembling systems can lead to insights into everything from the origin of life~\cite{SchulmanWinfreeEvolution} to new and novel ways to guide atomically precise manufacturing.
Theoretical modeling of self-assembling systems has uncovered important mathematical properties \cite{RotWin00,AdChGoHu01,ACGHKMR02,IUSA,SFTSAFT,AKKR02,2PATS}, and physical realizations have been experimentally verified in the laboratory and used to create intricate nanostructures~\cite{RothTriangles,SchWin07,MaoLabReiSee00,BarSchRotWin09,LaWiRe99,evans2014crystals}.  In order to facilitate the design of these systems, a number of mathematical models have been introduced.  The work presented in this paper examines a model of self-assembly which is an extension of Erik Winfree's abstract Tile Assembly Model (aTAM)~\cite{Winf98}.  In the aTAM, the fundamental components are square ``tiles'' with ``glues'' on their edges.  These tiles can then combine depending on their glues to form surprisingly complex and mathematically interesting structures~\cite{SolWin07,RotWin00,jCCSA,IUSA,jSADSSF}.

A long standing open conjecture in regards to the aTAM is that systems in which tile attachments depend only on one exposed glue (we call such systems \emph{temperature-1} systems) are not computationally universal~\cite{jLSAT1,IUNeedsCoop,ManuchTemp1}.  It may appear clear that this conjecture is certainly true, but the ability of tile assembly systems to place a tile which prevents the attachment of a later tile gives these systems a surprising amount of power~\cite{GeoTiles,Polyominoes,Duples,OneTile} and has made proving such a result elusive.  In fact, the exploitation of this ability has been used to show that temperature-$1$ systems in other models are computationally universal~\cite{SingleNegative,CookFuSch11,GeoTiles,Duples,Polyominoes}.

This paper examines the computational power of a model which is similar to the aTAM with the exception that the shape of the tiles in the systems is relaxed to include any shape which is a polygon.  Unlike all previous work, our model makes no assumption about an underlying lattice and discrete space.  Instead, we must work in the real plane, and fundamentally exploit continuous space to precisely position polygonal tiles.  We call this model the polygonal TAM and show that certain classes of temperature-$1$ systems in the polygonal TAM are computationally universal.  In order to show our results about computational universality, we explicitly construct ``lattices'' for polygons and create geometric ``bit-readers''.  In the case of regular polygons with $n>6$ sides, we exploit the inability of these polygons to tile the plane to read bits. In fact, we show that for regular polygons which do tile the plane, bit-reading gadgets are impossible to construct. Interestingly, our exploits do not work for pentagons. In particular, we show that even though pentagons cannot tile the plane, bit-reading gadgets are impossible to construct with them.

The layout of the paper is as follows.  We first introduce the polygonal TAM.  Next, we introduce our main results which concentrate on the computational power of polygonal TAM systems at temperature $1$.  Our first main result states that for any regular polygon $P$ with $n > 6$ sides, there exists a polygonal TAM system consisting of tiles of shape $P$ which simulates any Turing machine on any input.  We then provide evidence that this computational boundary is tight by showing that the class of polygonal TAM systems composed only of tiles of a single shape which is any regular polygon $P$ with $n \leq 6$ sides cannot compute using any currently known techniques. On the other hand, we show that the class of polygonal TAM systems whose tiles are composed of any two regular polygons is capable of simulating any Turing machine on arbitrary input.  We then show two positive results about computing with systems whose tiles have the shape of non-regular polygons with less than seven sides.  In order to show these results we have two supporting sections.  One shows how we can create a ``lattice'' in the plane out of any regular polygon.  The other uses these ``lattices'' to connect together several components which ``read bits''.

\vspace{-20pt}
\section{Preliminaries}
\vspace{-15pt}
In this section we sketch definitions of the Polygonal Tile Assembly Model (Polygonal TAM) and relevant terminology.\footnote{The Polygonal TAM is simply a case of the polygonal free-body TAM defined in \cite{OneTile} with no rotational restriction and no tile flipping.  We define it here for completeness.}  Please see the Appendix for more detailed definitions.

\vspace{-5pt}
\paragraph{Polygonal Tiles}
A \emph{simple polygon} is a plane geometric figure consisting of straight, non-intersecting line segments or ``sides'' that are joined pair-wise to form a closed path. As is commonly the case, we omit the qualifier ``simple'' and refer to simple polygons as polygons. A polygon encloses a region called its \emph{interior}. The line segments that make-up a polygon meet only at their endpoints. Exactly two edges meet at each vertex.
We define the set of \emph{edges} of a polygon to be the line segments that make-up a polygon.
In our definition we find it useful to give a polygon a default position and rotation. First, we assume that the centroid, $c$ say, of any polygon is at the origin in $\R^2$. Then, for a polygon $P_n$ with $n$ edges, let $v=(x,y)\in \R^2$ be some vertex of $P_n$ such that $v\neq c$. By possibly rotating $P_n$ about $c$, we can ensure that $y = 0$ and $x>0$. For a given polygon $P$ and some vertex $v$ of $P$ that is not equal to the centroid of $P$, we call this position and rotation the \emph{standard position} for $P$ given $v$.

A \emph{polygonal tile} is a polygon with a subset of its edges labeled from some {\em glue} alphabet $\Sigma$, with each glue having an integer \emph{strength} value.  Two tiles are said to be \emph{adjacent} if they are placed so that two edges, one on each tile, intersect completely. Two tiles are said to \emph{bind} when they are placed so that they have non-overlapping interiors and have adjacent edges with complementary glues and matching lengths; each complementary glue pair binds with force equal to its strength value. An \emph{assembly} is any connected set of polygons whose interiors do not overlap such that every tile is adjacent to some other tile.~\footnote{As with the aTAM, the edges of two tiles of an assembly may intersect, but we do not allow for the interiors of two tiles of an assembly to have non-empty intersection.}  Given a positive integer $\tau \in \mathbb{N}$, an assembly is said to be \emph{$\tau$-stable} or (just \emph{stable} if $\tau$ is clear from context), if any partition of the assembly into two non-empty groups (without cutting individual polygons) must separate bound glues whose strengths sum to $\ge \tau$. We say a tile is in \emph{standard position} if the underlying polygon defining its shape is in standard position. We also refer to the centroid of a polygonal tile as the centroid of the underlying polygon defining the shape of the tile.

\vspace{-5pt}
\paragraph{Tile System}
A \emph{tile assembly system} (TAS) is an ordered triple $\calT = (T,\sigma,\tau)$ where $T$ is a set of polygonal tiles, and $\sigma$ is a $\tau$-stable assembly called the \emph{seed}. $\tau$ is the \emph{temperature} of the system, specifying the minimum binding strength necessary for a tile to attach to an assembly.  Throughout this paper, the temperature of all systems is assumed to be $1$, and we therefore frequently omit the temperature from the definition of a system (i.e. $\calT = (T,\sigma)$). If the tiles in $T$ all have the same polygonal shape, $\calT$ is said to be a \emph{single-shape} system; more generally $\calT$ is said to be a $c$-shape system if there are $c$ distinct shapes in $T$.  If not stated otherwise, systems described in this paper should by default be assumed to be single-shape systems.

We define a \emph{configuration} of $\calT$ to be a (possibly empty) arrangement of tiles in $\R^2$ where tiles of this arrangement are translations and/or rotations of copies of tiles in $T$. Formally, we define a configuration of $\calT$ as follows. For a $c$-shaped system $\calT = (T,\sigma,\tau)$, let $P_1$, $P_2$, $\dots$, $P_c$ denote the polygons that make up the shapes of $\calT$. For each $i$ such that $1\leq i \leq c$, assume that each $P_i$ is in standard position given some vertex $v_i$ of $P_i$. Then, a configuration of $\calT$ is a partial function $\pfunc{\alpha}{\R^2\times \left[ 0, 2\pi \right)}{T}$. One should think of this function as mapping centroid locations and an angle of rotation, $(r,\theta)$ say, to a tile in $T$ as follows. Starting from $t$ in standard position, $t$ is rotated counter-clockwise by $\theta$ and translated so that the centroid of $t$ is at $r$. Note that the definition of configuration makes no claim as to whether or not two tiles of a configuration have overlapping interiors or have matching glues. Similarly, we can define an assembly to be a configuration such that every tile is adjacent to some other tile and the intersection of the interiors of any two distinct tiles is empty. Then an assembly $\alpha'$ is a \emph{subassembly} of $\alpha$ if $\dom(\alpha')\subseteq \dom(\alpha)$ and if $(r,\theta)\in \dom(\alpha')$ then $\alpha((r,\theta)) = \alpha'((r,\theta))$.  We define subconfiguration analogously to the way we defined subassembly.

\ifabstract
\later{
\section{Full Description of the Polygonal TAM}

We now give a full description of the Polygonal TAM.

\paragraph{Polygonal Tiles}
A \emph{simple polygon} is a plane geometric figure consisting of straight, non-intersecting line segments or ``sides'' that are joined pair-wise to form a closed path. As is commonly the case, we omit the qualifier ``simple'' and refer to simple polygons as polygons. A polygon encloses a region called its \emph{interior}. The line segments that make-up a polygon meet only at their endpoints. Exactly two edges meet at each vertex.
We define the set of \emph{edges} of a polygon to be the line segments that make-up a polygon.
In our definition we find it useful to give a polygon a default position and rotation. First, we assume that the centroid, $c$ say, of any polygon is at the origin in $\R^2$. Then, for a polygon $P_n$ with $n$ edges, let $v=(x,y)\in \R^2$ be some vertex of $P_n$ such that $v\neq c$. By possibly rotating $P_n$ about $c$, we can ensure that $y = 0$ and $x>0$. For a given polygon $P$ and some vertex $v$ of $P$ that is not equal to the centroid of $P$, we call this position and rotation the \emph{standard position} for $P$ given $v$.

A \emph{polygonal tile} is a polygon with a subset of its edges labeled from some {\em glue} alphabet $\Sigma$, with each glue having an integer \emph{strength} value.  Two tiles are said to be \emph{adjacent} if they are placed so that two edges, one on each tile, intersect completely. Two tiles are said to \emph{bind} when they are placed so that they have non-overlapping interiors and have adjacent edges with matching glues and matching lengths; each matching glue binds with force equal to its strength value. An \emph{assembly} is any connected set of polygons whose interiors do not overlap such that every tile is adjacent to some other tile.~\footnote{As with the aTAM, the edges of two tiles of an assembly may intersect, but we do not allow for the interiors of two tiles of an assembly to have non-empty intersection.}  Given a positive integer $\tau \in \mathbb{N}$, an assembly is said to be \emph{$\tau$-stable} or (just \emph{stable} if $\tau$ is clear from context), if any partition of the assembly into two non-empty groups (without cutting individual polygon) must separate bound glues whose strengths sum to $\ge \tau$. We say that a tile is in \emph{standard position}, if the underlying polygon defining the shape of the tile is in standard position. We also refer to the centroid of a polygonal tile as the centroid of the underlying polygon defining the shape of the tile.

\paragraph{Assembly Process}
Given a tile-assembly system $\calT = (T,\sigma,\tau)$, we now define the set of \emph{producible} assemblies $\prodasm{T}$ that can be derived from $\calT$, as well as the \emph{terminal} assemblies, $\termasm{T}$, which are the producible assemblies to which no additional tiles can attach. The assembly process begins from $\sigma$ and proceeds by single steps in which any single copy of some tile $t \in T$ may be attached to the current assembly $A$, provided that it can be translated and/or rotated so that its placement does not overlap any previously placed tiles and it binds with strength $\ge \tau$. For a system $\calT$ and assembly $A$, if such a $t\in T$ exists, we say $A \rightarrow^\calT_1 A'$ (i.e. $A$ grows to $A'$ via a single tile attachment).  We use the notation $A \rightarrow^\calT A''$, when $A$ grows into $A''$ via $0$ or more steps. Assembly proceeds asynchronously and nondeterministically, attaching one tile at a time, until no further tiles can attach.  An \emph{assembly sequence} in a TAS $\mathcal{T}$ is a (finite or infinite) sequence $\vec{\alpha} = (\alpha_0 = \sigma,\alpha_1,\alpha_2,\ldots)$ of assemblies in which each $\alpha_{i+1}$ is obtained from $\alpha_i$ by the addition of a single tile.
The set of producible assemblies $\prodasm{T}$ is defined to be the set of all assemblies $A$ such that there exists an assembly sequence for $\calT$ ending with $A$ (possibly in the limit).  The set of \emph{terminal} assemblies $\termasm{T} \subseteq \prodasm{T}$ is the set of producible assemblies such that for all $A \in \termasm{T}$ there exists no assembly $B \in \prodasm{T}$ in which $A\rightarrow^\calT_1 B$.  A system $\calT$ is said to be directed if $|\termasm{T}|=1$, i.e., if it has exactly one terminal assembly.
} %
\fi

\vspace{-15pt}
\section{Geometric Bit-reading, Grids, and Turing Machine Simulation}
\vspace{-10pt}
In this section we state our main results and then give a high-level description of the machinery used to prove these results. In particular, we describe bit-reading gadget assemblies and grid assemblies, and briefly show how to simulate a Turing machine using these assemblies. The general strategy that motivates the work in this paper is similar to the the techniques used in~\cite{CookFuSch11, Duples, Polyominoes}. Unlike the techniques used in~\cite{CookFuSch11, Duples, Polyominoes}, we do not have an underlying integer lattice that is being tiled, and therefore, must rely on analysis of polygonal tile assemblies in $\R^2$.

\vspace{-15pt}
\subsection{Main results}
\vspace{-5pt}
We now state our main results.  The first set of results are positive and state that there are a variety of systems with polygons which can simulate any Turing machine.  The last result is a negative result which states that the class of systems whose tiles are composed of regular polygons with less than $7$ sides cannot compute using known techniques in self-assembly.

Informally, our first theorem states that if $P$ is a regular polygon with $\ge 7$ sides, then the class of systems with tiles of shape $P$ is computationally universal.
\vspace{-10pt}
\begin{theorem}\label{thm:regular-CU}
Let $P_n$ be a regular polygon with $n$ sides such that $n \geq 7$. Then for every standard Turing machine $M$ and input $w$, there exists a directed TAS with $\tau=1$ consisting only of tiles of shape $P_n$ that simulates $M$ on $w$.
\end{theorem}

The following theorem states that if we are allowed two different regular polygons as tile shapes, then the class of systems consisting only of these two shapes is computationally universal.
\begin{theorem}\label{thm:pairs-CU}
Let $P_n$ and $Q_m$ be regular polygons with $n$ and $m$ sides of equal length. Then for every $n \geq 3$ and $m \geq 3$ such that $n\neq m$, and every standard Turing machine $M$ with input $w$, there exists a directed 2-shaped system $\mathcal{T}_{n,m} = (T_{n,m}, \sigma_{n,m})$ consisting only of tiles of shape $P_n$ or $Q_m$ that simulates $M$ on $w$.
\end{theorem}

The next theorem differs from the previous two theorems in that it discusses the computational power of polygons which are not regular.  Roughly, it states that if we relax the condition that the polygon is regular (but still equilateral), then there exist polygons with only four sides which are capable of composing a class of computationally universal single shape systems.  It also implies this for shapes with five and six sides as well.

\begin{theorem}\label{thm:equilateral-CU}
Let $M$ be a standard Turing machine with input $w$. Then for all $n \geq 4$, there exists an equilateral polygon $P_n$ with $n$ sides and a directed single-shaped system $\mathcal{T}_n = (T_n, \sigma_n)$ consisting only of tiles of shape $P_n$ that simulates $M$ on $w$.
\end{theorem}

Our final positive result shows that there exists a class of single-shaped systems of obtuse isosceles triangle which is computationally universal.

\begin{theorem}\label{thm:triangular-CU}
Let $M$ be a standard Turing machine with input $w$. Then, there exists an obtuse isosceles triangle $P$ and a directed single-shaped system $\mathcal{T} = (T, \sigma)$ consisting only of tiles of shape $P$ that simulates $M$ on $w$.
\end{theorem}

We now state the negative result, which is based on the fact that regular polygonal tiles with $\le 6$ sides cannot form paths capable of blocking each other in specific ways allowing important geometric information encoding and decoding.

\begin{theorem}\label{thm:cant-bit-read}
Let $n\in \N$ be such that $3\leq n \leq 6$. Then, there exists no temperature 1 single-shaped polygonal tile assembly system $\mathcal{T} = (T,\sigma,1)$ where for all $t \in T$, $t$ is a regular polygon with $n$ sides, and a bit-reading gadget exists for $\mathcal{T}$.
\end{theorem}

Due to space constraints in this extended abstract, the proofs of most results are relegated to the Appendix.  However, in the main body we now sketch an overview of how the positive results work and a portion of the proof of Theorem~\ref{thm:regular-CU} for $n \ge 15$, which gives the general overall scheme for all of the positive results.

\vspace{-10pt}
\subsection{Bit-Reading Gadgets Overview}\label{sec:bit-gadgets}
\vspace{-10pt}
First, we discuss a primitive tile-assembly component that enables computation by self-assembling systems. This component is called the \emph{bit-reading gadget}, and essentially consists of pre-existing assemblies, \emph{bit writers}, that appropriately encode bit values (i.e., $0$ or $1$) and paths that grow past them and are able to ``read'' the values of the encoded bits; this results in those bits being encoded in the tile types of the paths beyond the encoding assemblies. The notion of bit-reading gadget was defined in~\cite{Polyominoes}. For completeness, we present the definition here and note that the definition applies even to systems of polygonal tiles.
Figure~\ref{fig:bit-gadget-definition} provides an intuitive overview of a temperature-1 system with a bit-reading gadget.  Essentially, depending on which bit is encoded by the assembly to be read, exactly one of two types of paths can complete growth past it, implicitly specifying the bit that was {\em read}.  It is important that the bit reading must be unambiguous, i.e., depending on the bit {\em written} by the pre-existing assembly, exactly one type of path (i.e., the one that denotes the bit that was written) can possibly
complete growth, with all paths not representing that bit being prevented. Furthermore, the correct type of path must always be able to grow.  Therefore, it cannot be the case that either all paths can be blocked from growth, or that any path not of the correct type can complete, regardless of whether a path of the correct type also completes, and these conditions must hold for any valid assembly sequence to guarantee correct computation.

The key to the correct functioning of a bit-reading gadget at temperature-1, where glue cooperation is not available and one source of ``input'' to the growing bit-reader must instead be provided by geometry, in the form of geometric hindrance which prevents exactly one path from continuing growth but allows another to proceed, is the fact that it must work when reading either of two different bit values.  Using Figure~\ref{fig:bit-gadget-definition} as a guide, one can see that it is easy to read the ``1'' bit in this example by blocking the blue path.  However, the difficulty which is encountered is in correctly blocking the yellow path while allowing the blue to continue in order to read a ``0'' bit.  With square tiles (and as we show, several others), this is in fact impossible.  However, with most polygonal tiles this can be accomplished by careful design of paths and blocking assemblies so that a gap remains between the blocked path and the blocking assembly in such a way that the other path can assemble through the gap.  The techniques for accomplishing this will be demonstrated throughout this paper.

\begin{figure*}[h!b]
\begin{center}
\vspace{-15pt}
\includegraphics[width=4.0in]{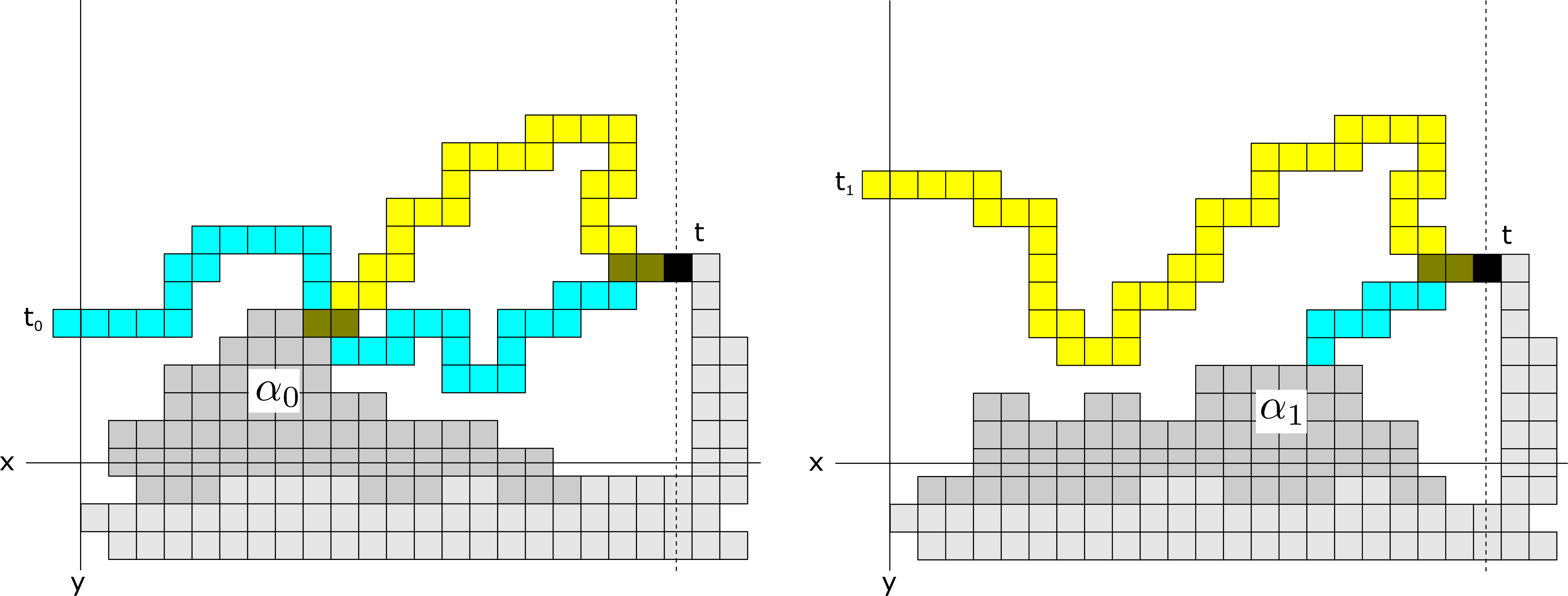}
\caption{Abstract schematic of a bit-reading gadget. (Left) The blue path grown from
$t$ ``reads'' the bit 0 from $\alpha_0$ (by being allowed to grow to $x < 0$ and
placing a tile $t_0 \in T_0$), while the yellow path (which could read a 1
bit) is blocked by $\alpha_0$. (Right) The yellow path grown from $t$ reads the bit 1
from $\alpha_1$, while the blue path that could potentially read a 0 is blocked by
$\alpha_1$.  Clearly, the specific geometry of the used polygonal tiles and
assemblies is important in allowing the yellow path in the left figure to be
blocked without also blocking the blue path.}
\label{fig:bit-gadget-definition}
\vspace{-20pt}
\end{center}
\end{figure*}

\ifabstract
\later{
\section{Formal Definition of Bit-Reading Gadget}

For the following definition is taken from \cite{Polyominoes} and modified slightly to account for the fact that polygonal tiles are placed in continuous, rather than discrete, space.  Here and throughout the paper, if we refer to a tile having an $x$ (or $y$) coordinate $i$, we are referring to its centroid being on the line $x=i$ (or $y=i$) for $i \in \mathbb{R}$.

\begin{definition}\label{def:bit-reader}
We say that a \emph{bit-reading gadget} exists for a tile assembly system
$\mathcal{T} = (T,\sigma,\tau)$, if the following hold.  Let $T_0 \subset T$ and
$T_1 \subset T$, with $T_0 \cap T_1 = \emptyset$, be subsets of tile types
which represent the bits $0$ and $1$, respectively.  For some producible
assembly $\alpha \in \prodasm{T}$, there exist two connected subassemblies,
$\alpha_0,\alpha_1 \sqsubseteq \alpha$ (with $w$ equal to the maximal width of
$\alpha_0$ and $\alpha_1$, i.e., the largest extent in $x$-direction spanned by either subassembly), such that if: \begin{enumerate}
    \item $\alpha$ is translated so that $\alpha_0$ has its minimal $y$-coordinate $\le 0$ and its minimal $x$-coordinate $\ge 0$,
    \item a tile of some type $t \in T$ is placed at $(w+n,h)$, where $n,h \ge 1$, and
    \item the tiles of $\alpha_0$ are the only tiles of $\alpha$ in the first quadrant to the left of $t$,
\end{enumerate}
then at least one path must grow from $t$ (staying strictly above the $x$-axis)
and place a tile of some type $t_0 \in T_0$ as the first tile with
$x$-coordinate $< 0$, while no such path can place a tile of type $t' \in (T
\setminus T_0)$ as the first tile to with $x$-coordinate $< 0$.  (This
constitutes the reading of a $0$ bit.)

Additionally, if $\alpha_1$ is used in place of $\alpha_0$ with the same
constraints on all tile placements, $t$ is placed in the same location as
before, and no other tiles of $\alpha$ are in the first quadrant to the left of
$t$, then at least one path must grow from $t$ and stay strictly above the
$x$-axis and strictly to the left of $t$, eventually placing a tile of some
type $t_1 \in T_1$ as the first tile with $x$-coordinate $< 0$, while no such
path can place a tile of type $t' \in (T \setminus T_1)$ as the first tile with
$x$-coordinate $< 0$.  (Thus constituting the reading of a $1$ bit.)
\end{definition}

We refer to $\alpha_0$ and $\alpha_1$ as the {\em bit writers}, and the paths
which grow from $t$ as the {\em bit readers}.  Also, note that while this
definition is specific to a bit-reader gadget in which the bit readers grow from
right to left, any rotation of a bit reader is valid by suitably rotating the
positions and directions of Definition~\ref{def:bit-reader}.
} %
\fi 

\vspace{-20pt}
\subsection{Grid assemblies}
\vspace{-5pt}
As we will see in Section~\ref{sec:TMsim}, our construction to simulate a Turing machine with a Polygonal TAM system consisting of the polygon $P$ will require us to string together several bit writers which we will then read with a series of bit readers.  In order to ensure that the path which is assembling the bit readers is placing the bit readers at the correct positions, we need to keep track of where the bit writers are located.  We accomplish this by constructing a lattice in the plane with $P$.  We can then place our bit writers at periodic positions in this lattice so that the path which is assembling the bit readers will know where to place the bit readers.
\vspace{-10pt}
\subsection{Turing machine simulation} \label{sec:TMsim}
\vspace{-8pt}

Let $M$ be a Turing machine and let $w$ be some input to $M$.  Figure~\ref{fig:tm_sim} shows a high-level schematic diagram of how a Polygonal TAM system simulates $M$ on input $w$.  The input $w$ is encoded as a sequence of bit writers with spacers placed in between them (shown at the bottom of Figure~\ref{fig:tm_sim} as shaded regions labeled with an ``s'').  These spacers allow for our bit readers to shift back on grid without encroaching on the territory of other bit writers.  As indicated by the arrows in Figure~\ref{fig:tm_sim}, bit readers then ``read'' the bit writers corresponding to the inputs.  Growth proceeds by growing to the north (shown as a dark unlabeled region in Figure~\ref{fig:tm_sim}, and then a bit writer is assembled depending on what was read by the bit reader.  After assembling the bit writer above, growth then continues by growing a path so that the next bit writer can be ``read'' (shown as the lightly shaded region labeled ``wr'' in Figure~\ref{fig:tm_sim}.  This growth continues until the last bit of the row is encountered at which point, the bit reader begins ``reading'' the next row.  Each row of bit writers can be thought of as representing the tape of $M$.  The symbols on the tape and location of the head of $M$ are all encoded in geometry as bit writers. If the head is not located at the set of bit writers the bit reader is currently reading, the symbols represented by the bit writers are simply rewritten as bit writers in the row above.  Otherwise, the transition may be carried out by writing the new symbol on the tape specified by the transition function of $M$ in the row above as a sequence of bit writers.  Also, bit writers are assembled to indicate that the head has moved as specified by the transition function.  See ~\cite{Polyominoes} for most complete exposition on this technique.

\begin{wrapfigure}{r}{0.55\textwidth}
\begin{center}
\vspace{-40pt}
\includegraphics[width=2.7in]{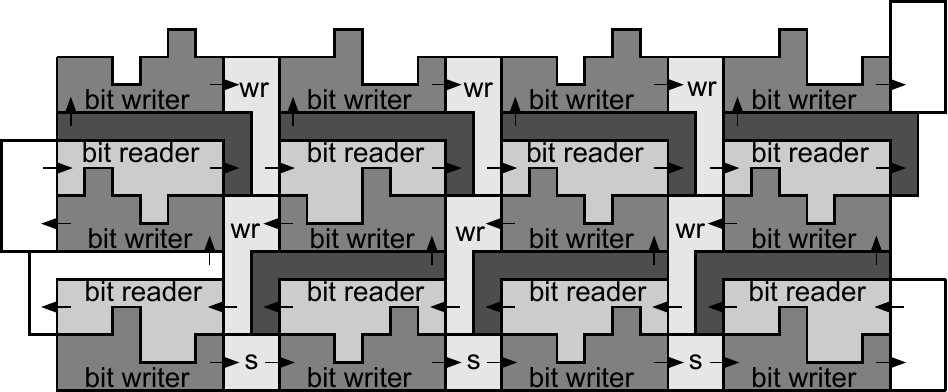}
\vspace{-10pt}
\caption{Schematic of simulating a Turing machine with bit-reading gadgets (from \cite{Polyominoes}).}
\label{fig:tm_sim}
\vspace{-25pt}
\end{center}
\end{wrapfigure}

Thus, to show our positive results, our task has become to 1) show that bit reading gadgets exist for the claimed systems and 2) show that we can string them together.  The first task is accomplished in Section~\ref{sec:bit-readers-main} and the grid which allows us to show the latter is shown in Section~\ref{sec:bit-grid-main}.

Given an $n$-sided regular polygon $P$ where $n>6$, a Turing machine $M$ and an input $w$, Algorithm~\ref{alg:high} shows a high-level schematic view of an algorithm that produces a single shape Polygonal TAM system which simulates the Turing machine $M$ on input $w$ and consists of tiles of shape $P$.  Note that here, we are abstracting the way in which the mathematical structures appearing in the algorithm are represented.  In Section~\ref{sec:bit-grid-main}, we give a construction which implicitly defines an algorithm which we call FORM\_GRID.  This algorithm takes an integer $n$ as input and returns a grid formed by the $n$-sided regular polygon. Given a grid $\mathcal{G}$ and an $n$-sided regular polygon, in Section~\ref{sec:technicalLemmas} our construction implicitly gives an algorithm which we call FORM\_GADGETS, that takes a grid $\mathcal{G}$ and an integer $n$, and produces a normalized bit-reading gadget.  Once we have a normalized bit reading gadget, we can use the algorithm implicitly described in Section 3.2 of \cite{Polyominoes}, which we call INITIALIZE, that produces a system, say $\mathcal{T}=(T, \sigma)$, which grows a geometric representation of the input $w$.  Finally, also in Section 3.2 of \cite{Polyominoes}, an algorithm is implicitly given, which we call TRANSITION\_TILES, that returns a set of tiles which are added to $T$ so that the system $\mathcal{T}$ is able to simulate a transition of the Turing machines $M$.

\begin{algorithm}[H]
 \KwData{$n$, $M$, $w$}
 \KwResult{Tile assembly system $\mathcal{T}$ which simulates $M$ on $w$}
 $\mathcal{G} \leftarrow$ FORM\_GRID(n)\;
 $\mathit{NRG} \leftarrow$ FORM\_GADGETS($\mathcal{G}$, n)\;
 $\mathcal{T} = (T, \sigma) \leftarrow$ INITIALIZE($n$, $M$, $w$, $\mathit{NRG}$)\; %
 $T \leftarrow T$ $\cup$ TRANSITION\_TILES($n$, $M$, $\mathit{NRG}$)\;
 \Return $\mathcal{T}$\;
 \caption{High level algorithm for constructing a system $\mathcal{T}$ which simulates $M$ on $w$.}
 \label{alg:high}
\end{algorithm} 

\vspace{-15pt}
\section{Regular Polygonal Tile Analysis With Complex Roots}\label{sec:roots-of-unity}
\vspace{-10pt}

In order to construct the grid assemblies and to show the correctness of the bit-reading gadgets we must show that the grid configurations and the bit-reading gadget configurations result in a valid assembly. In other words, we must show that the intersection of the interiors of any two distinct polygonal tiles in the configuration is empty. Moreover, in order to show that this assembly is indeed a valid bit-reading gadget we show that in the presence of the bit writer tiles, only one of two bit reading assemblies (representing either a $0$ or a $1$) can assemble depending on the bit writer tiles.

To prove that each bit-reading gadget configuration can be used to obtain a valid assembly, we must compute the distances from the center of a given polygon to the center of another polygon. For convenience, we assume that the length of the apothem (the line segment from the center of a polygon to the midpoint of one of its sides) of all of the regular polygons is $\frac{1}{2}$, so that the distance from the centers of abutting polygons is $1$. Then, let $T$ be a polygonal tile, and let $T'$ be a polygonal tile that abuts $T$. We say that a polygonal tile has the \emph{standard orientation} if after being translated so that it is centered at the origin, it has a side that corresponds to a vertical line $l$ segment with midpoint at $\left(\frac{1}{2}, 0\right)$. See Figure~\ref{fig:genericPolygonA} for a depiction of a polygonal tile with standard orientation that is also centered at the origin.
For a polygonal tile with an odd number of sides, we say that a polygonal tile has \emph{negated orientation} if after being translated so that it is centered at $(0,0)$, it is the reflection of a tile which has standard orientation across the imaginary axis. This is depicted in Figure~\ref{fig:genericPolygonB}.

\begin{wrapfigure}{r}{0.55\textwidth}
\centering
\vspace{-20pt}
  \subfloat[][A polygonal tile with \emph{standard orientation} with center $(0,0)$.]{%
        \label{fig:genericPolygonA}%
        \includegraphics[width=1.0in]{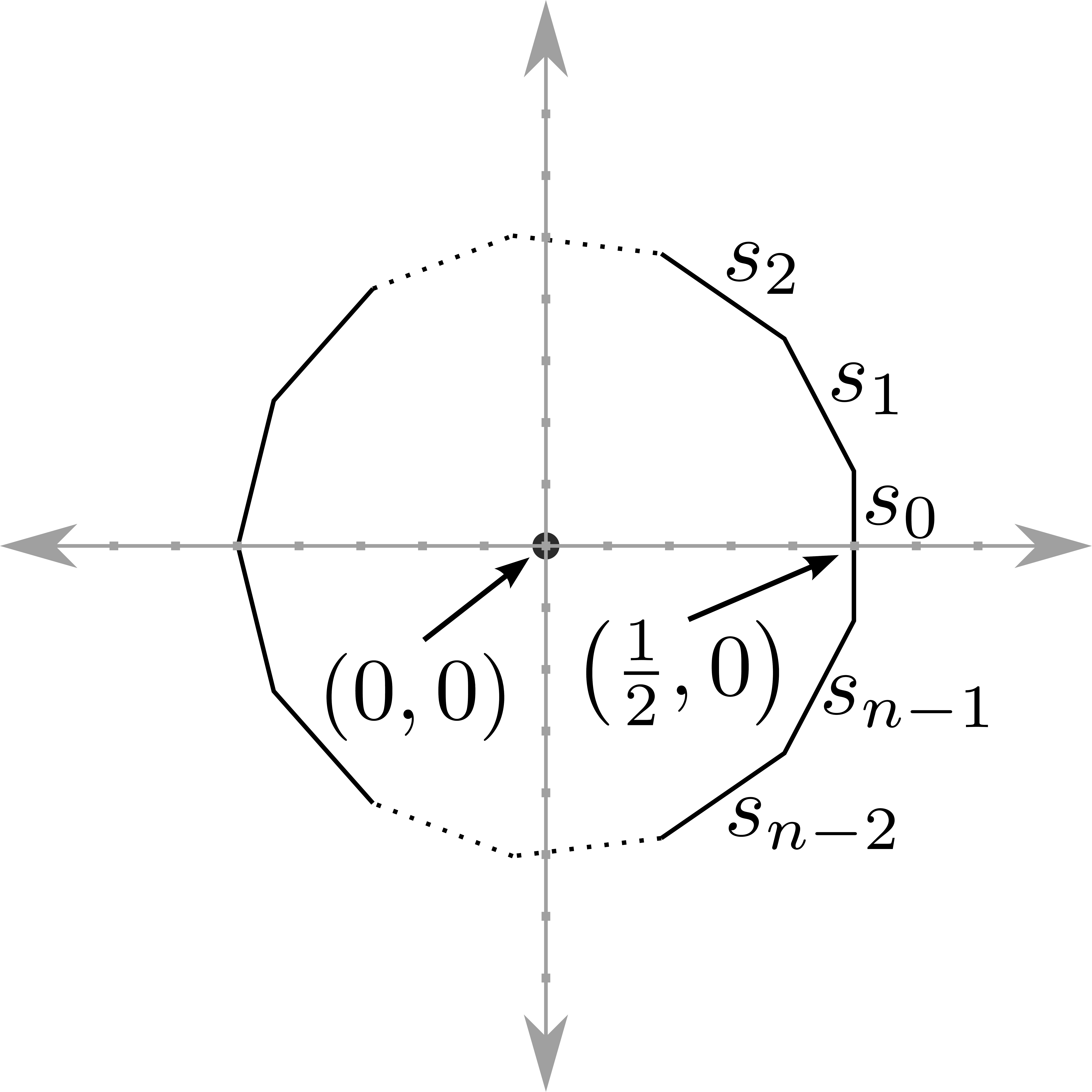}
        \vspace{-15pt}
        }%
        \quad
  \subfloat[][A polygonal tile with \emph{negated orientation} with center $(0,0)$.]{%
        \label{fig:genericPolygonB}%
        \includegraphics[width=1.0in]{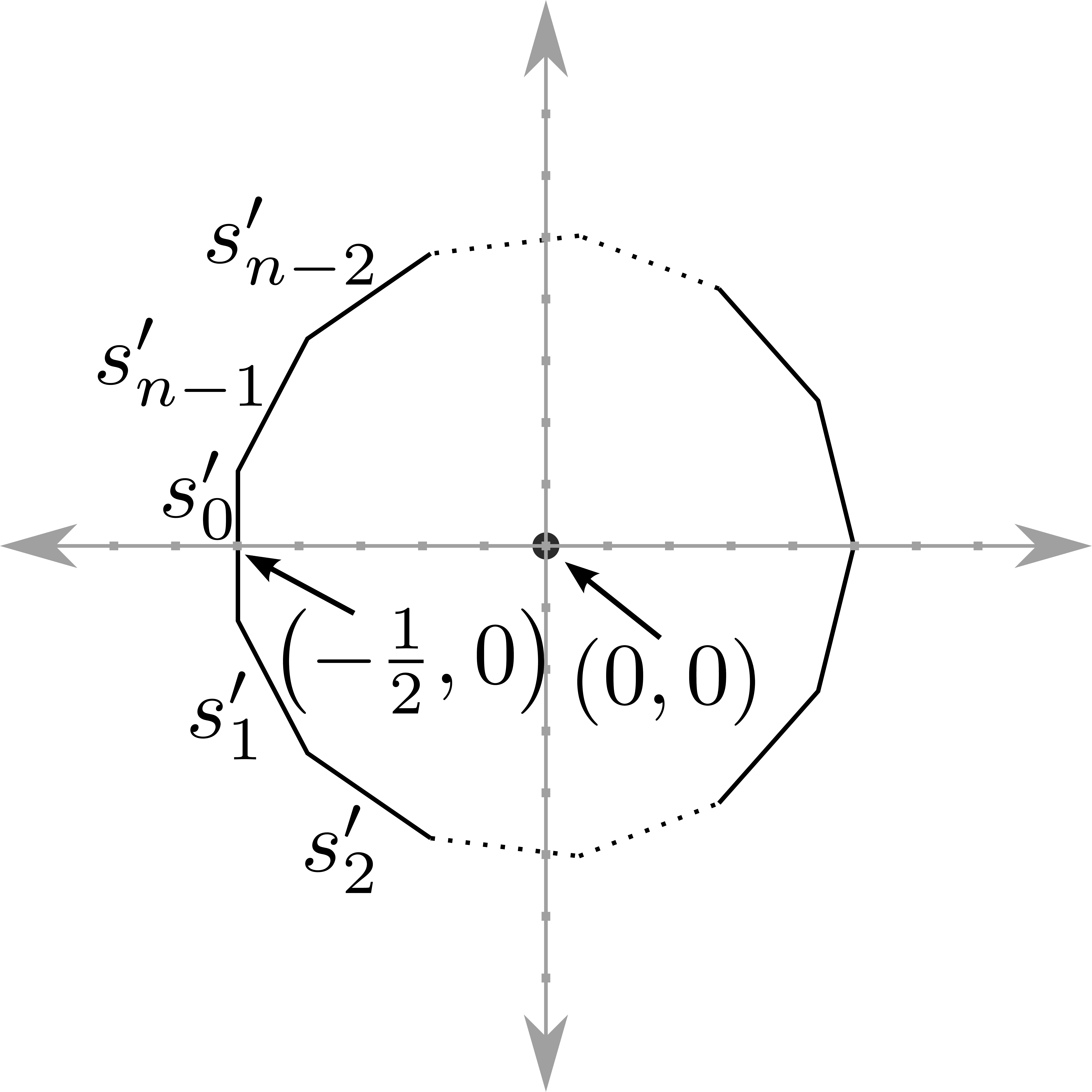}
        \vspace{-15pt}
        }%
  \vspace{-17pt}
  \caption{Regular polygonal tile orientations%
  \vspace{-20pt}}
  \label{fig:genericPolygon}
\end{wrapfigure}

We enumerate the sides of $T$ counter-clockwise starting from the side $s_0$ corresponding to $l$ and ending at $s_{n-1}$ where $n$ is the number of sides of $T$. Similarly, if $T$ has negated orientation, then we enumerate the sides as $\{s'_i\}_{i=0}^{n-1}$ as shown in Figure~\ref{fig:genericPolygonB}.  Then, relative to $T$, if $T'$ abuts $T$ along $s_0$, then the center of $T'$ is~$\left(1, 0\right)$. In general, for $\theta = \frac{2\pi}{n}$, if $T'$ abuts $T$ along $s_m$, then the center of $T'$ is $\left(\cos\left( m\theta \right), \sin\left( m\theta \right)\right)$. For the calculations in the following sections, it is convenient to identify $\R^2$ with the complex plane $\mathbb{C}$ so that $(x,y)$ is identified with $x + iy$. Then according to Euler's formula, $\left(\cos\left( m\theta \right), \sin\left( m\theta \right)\right)\in \R^2$ corresponds to the complex number $e^{mi\theta} = \cos\left( m\theta \right) + i\sin\left( m\theta \right)$. In other words, when $T$ has standard orientation, the centers of abutting polygons correspond to complex $n^{th}$ roots of unity, as the centers correspond to the roots of the complex polynomial $x^n + 1 = 0$ (recall that $n$ is the number of sides of $T$). Now let $\omega = e^{i\theta}$. Then these roots of unity are $\lbrace \omega^i \rbrace_{i=0}^{n-1}$. See Figure~\ref{fig:distances7sides} for an example in the heptagonal tile case. Finally, notice that if $T$ has negated orientation and $T'$ abuts $T$ along $s'_m$, then the center of $T'$ is $\left(-\cos\left( m\theta \right), -\sin\left( m\theta \right)\right)$, and so the center of $T'$ corresponds to $-\omega^m$.

\begin{wrapfigure}{r}{0.47\textwidth}
\centering
\vspace{-20pt}
	\includegraphics[width=1.0in]{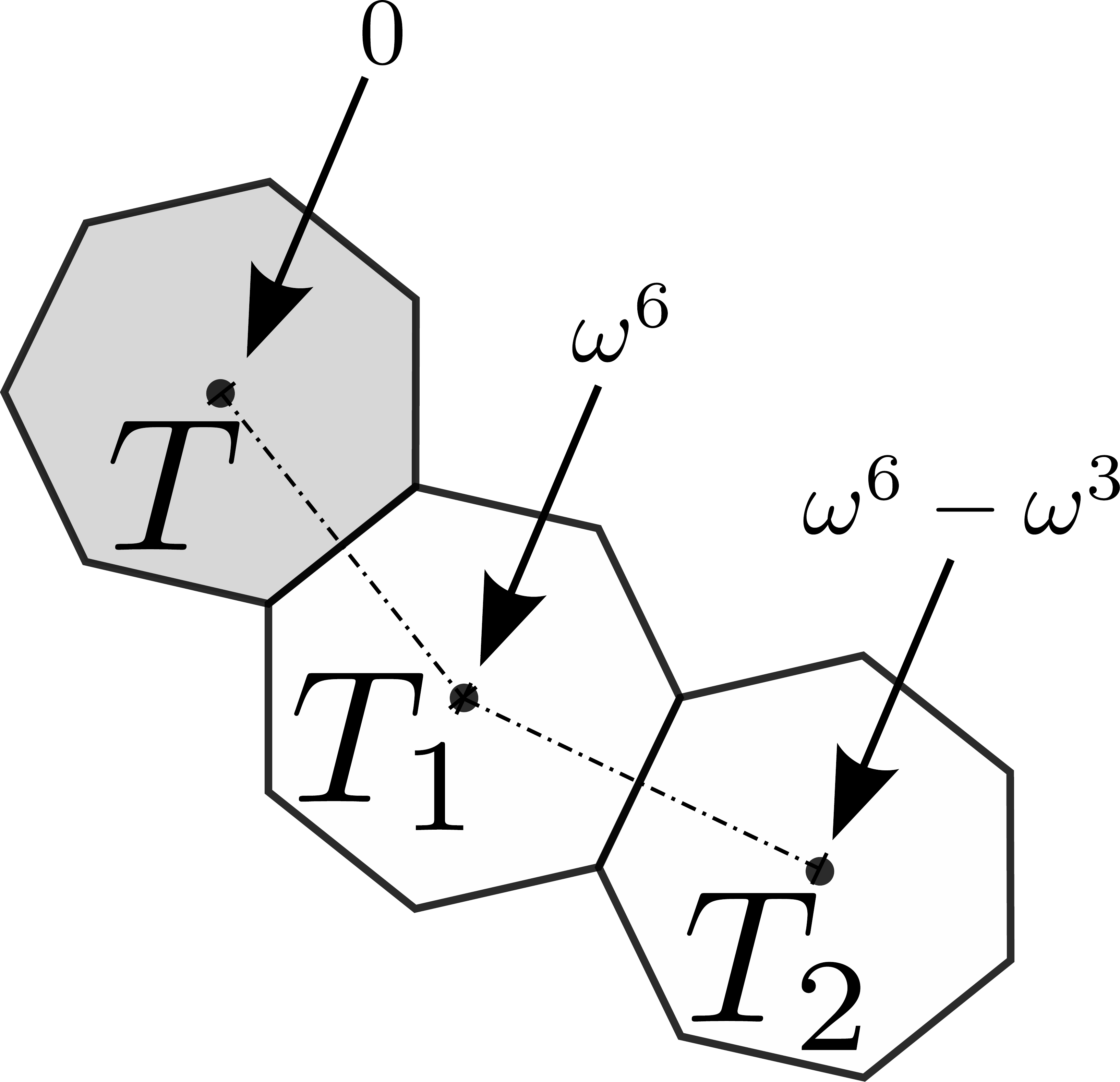}
    \vspace{-10pt}
	\caption{Relative to $T$, the center of $T_1$ corresponds to $\omega^6$ and the center of $T_2$ corresponds to $\omega^6 - \omega^3$.}
    \vspace{-20pt}
	\label{fig:heptagonSum-main}
\end{wrapfigure}

Now let $\calT$ be a TAS with tiles of a single regular polygon shape, and let $\alpha$ be an assembly in $\calT$ such that $\alpha$ contains a tile, $T$, with standard orientation and let $T'$ be any tile in $\alpha$ (including $T$). Then, since addition (respectively, subtraction) of complex numbers corresponds to vector addition (respectively, subtraction) in $\R^2$, the center of $T'$ corresponds to some polynomial in $\omega$ with integer coefficients. See Figure~\ref{fig:heptagonSum-main} for an example of the correspondence to the centers of heptagonal tiles to such polynomials.

\ifabstract
\later{
\section{Complex roots of unity example using heptagonal tiles}

In this section, we give example assemblies using heptagonal tiles by computing the distances of relevant tile centers using $7^{th}$ roots of unity. Let $\omega = e^{\frac{2\pi}{7}}$. For a polygonal tile $T$ with standard orientation, Figure~\ref{fig:distances7sides}(a) depicts the complex roots of unity corresponding to the centers of adjacent tiles. Similarly, for a polygonal tile $T$ with negated orientation, Figure~\ref{fig:distances7sides}(b) depicts the negated complex roots of unity corresponding to the centers of adjacent tiles.

\begin{figure}[htp]
\centering
  \subfloat[][]{%
        \label{fig:distances7sidesA}%
        \includegraphics[width=1.5in]{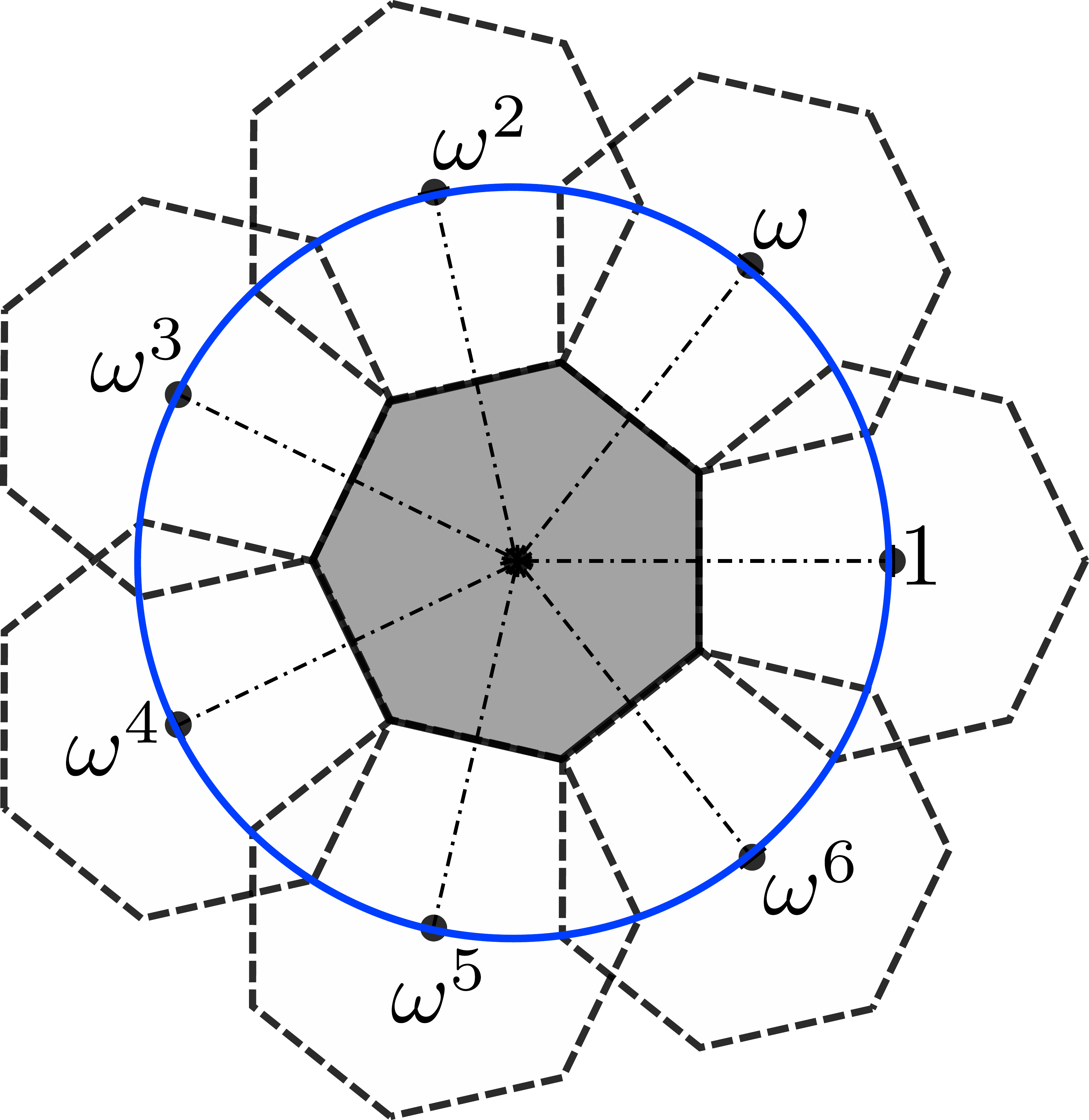}
        }%
        \quad
  \subfloat[][]{%
        \label{fig:distances7sidesB}%
        \includegraphics[width=1.5in]{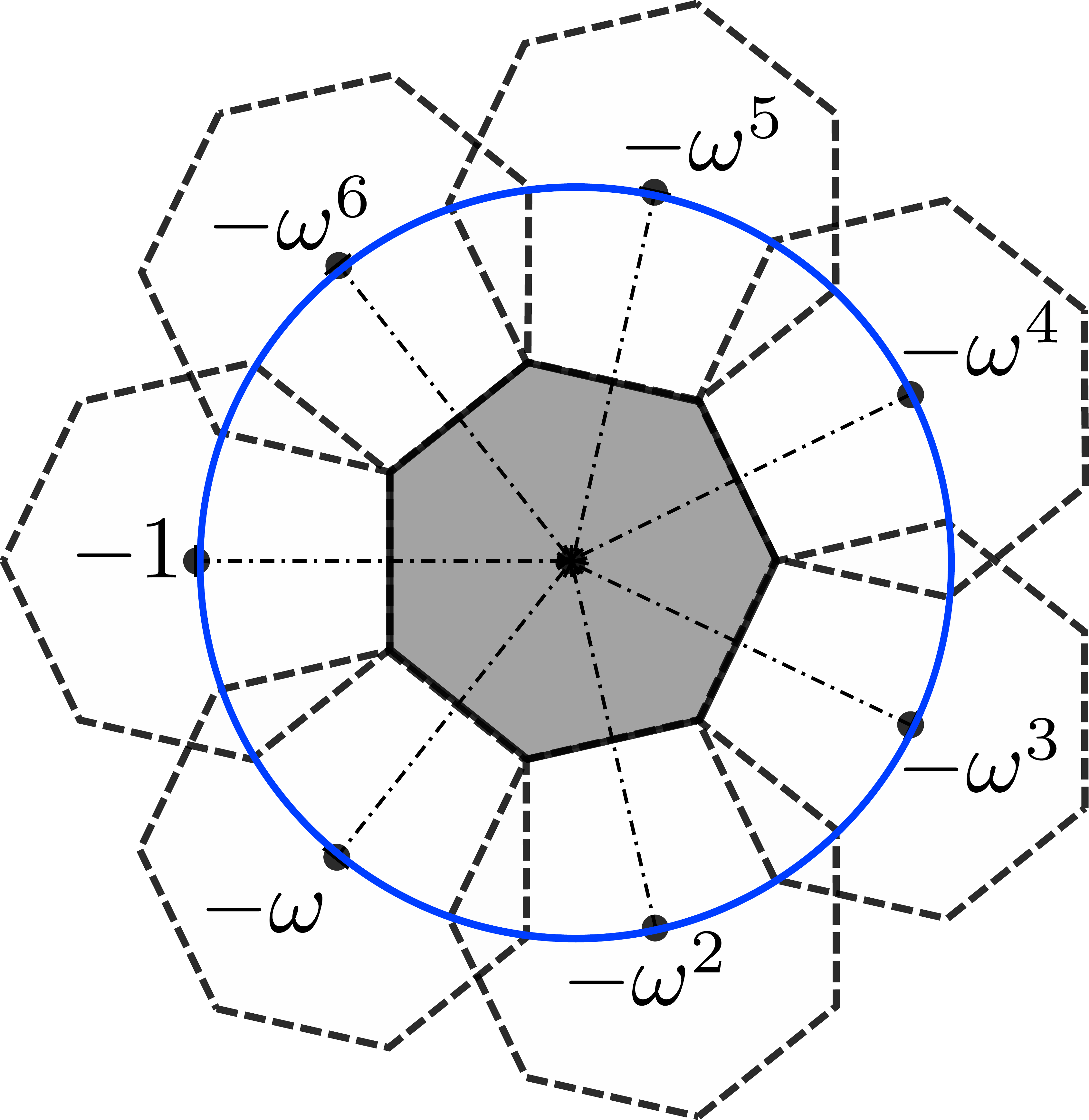}
        }%
  \caption{Representing the vector from the center of a heptagon (gray) to each center of an adjacent heptagon using the $7^{\text{th}}$ roots of unity.}
  \label{fig:distances7sides}
\end{figure}

\begin{wrapfigure}{r}{0.4\textwidth}
\centering
	\includegraphics[width=1.4in]{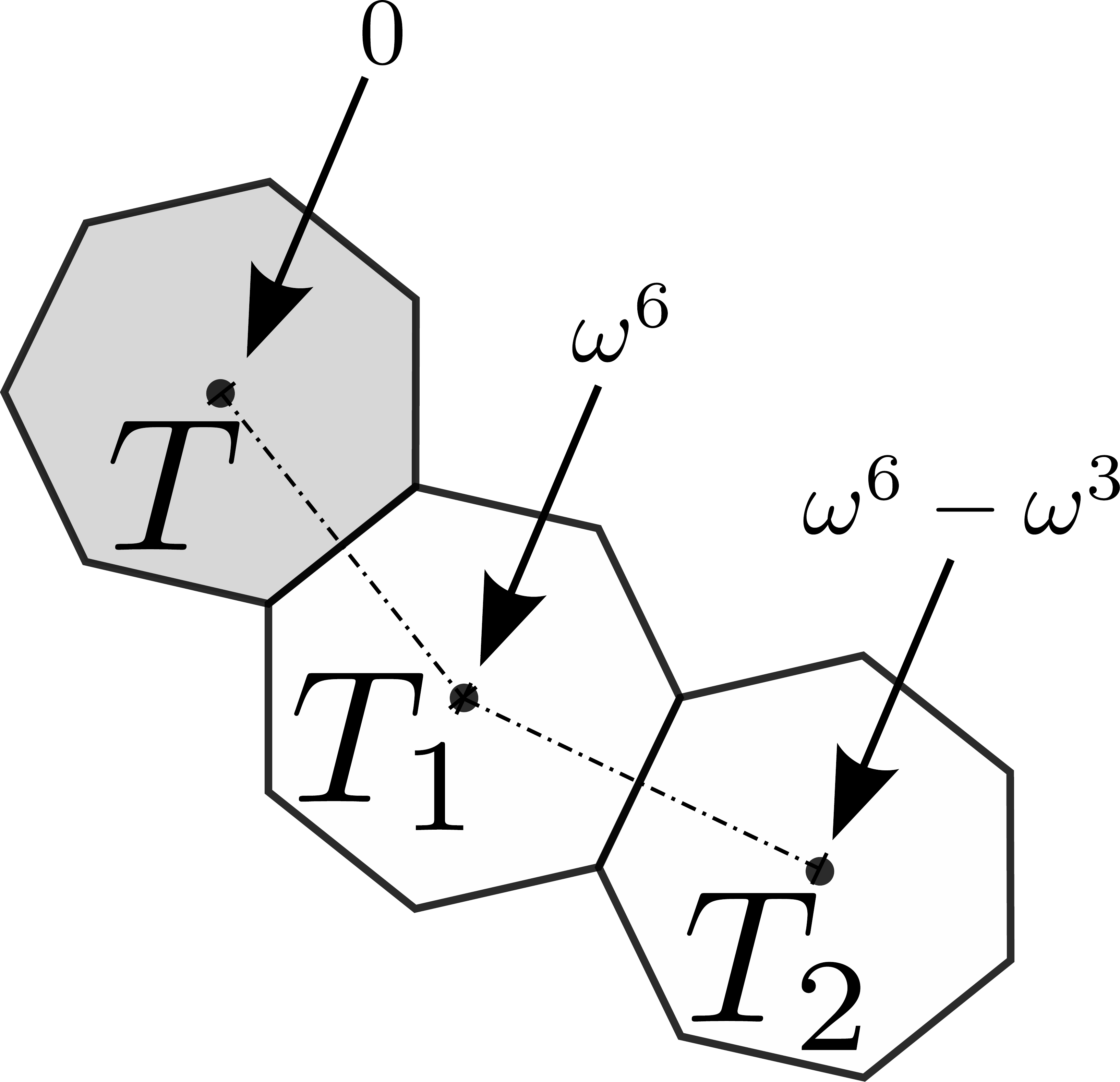}
	\caption{Relative to $T$, the center of $T_1$ corresponds to $\omega^6$ and the center of $T_2$ corresponds to $\omega^6 - \omega^3$.}
	\label{fig:heptagonSum}
\end{wrapfigure}
If $T$ denotes a polygonal tile with standard orientation (the case of negated orientation is similar) in an assembly $\alpha$ producible in a TAS $\calT$, we can compute the centers of any polygonal tile in $\alpha$ using complex addition and subtraction relative to the center of $T$. Figure~\ref{fig:heptagonSum} shows the complex numbers corresponding to the centers of tiles $T_1$ and $T_2$. First, $\omega^6$ corresponds to the center of $T_1$. Then note that relative to $T_1$, the center of $T_2$
corresponds to $-\omega^3$. Therefore, relative to $T$, the center of $T_2$ corresponds to $\omega^6 - \omega^3$. In a similar fashion, given any two polygonal tiles, $T$ and $T'$, the center of $T'$ relative to $T$ can be represented as a polynomial of $\omega$ with integer coefficients.

\begin{figure}[htp]
\centering
	\includegraphics[width=4in]{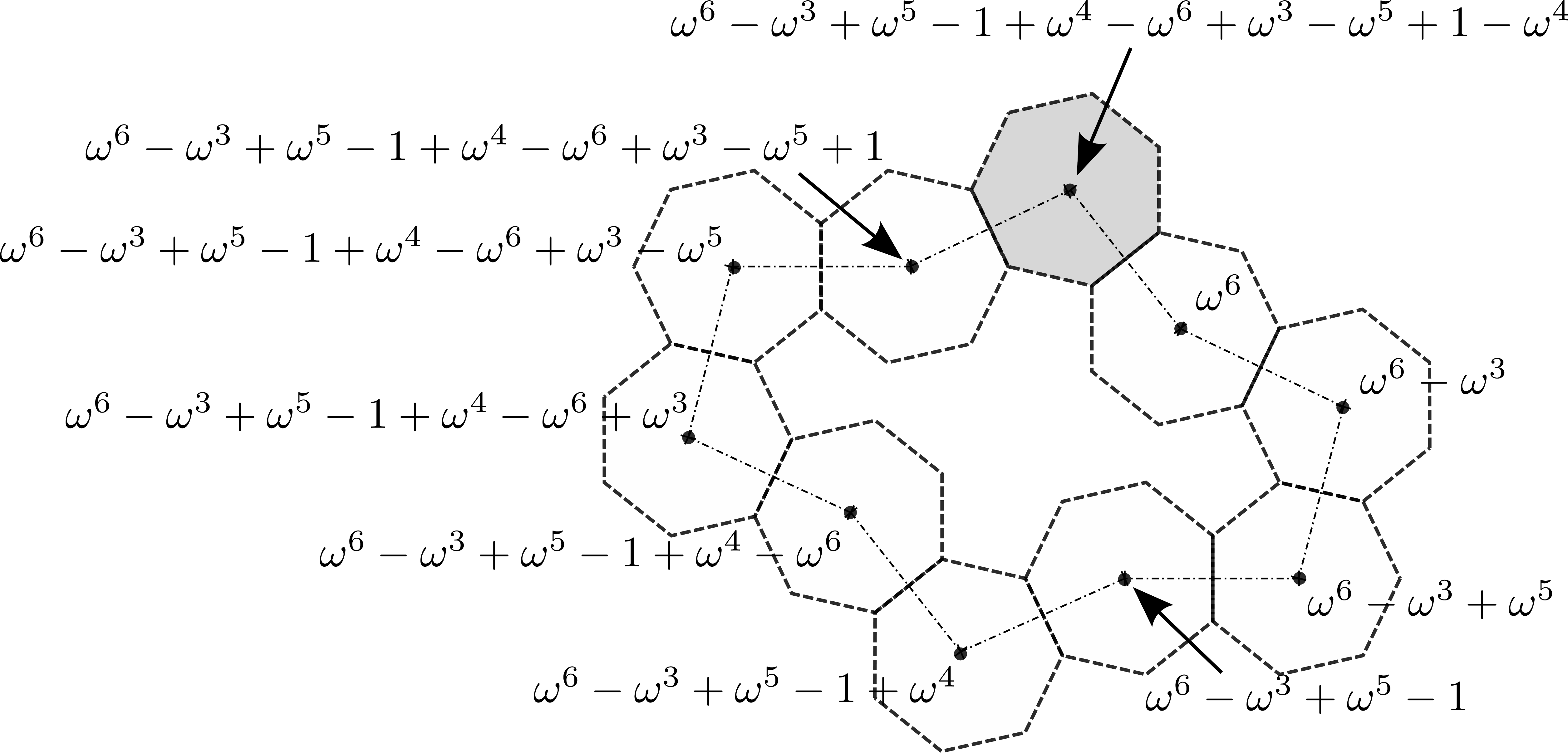}
	\caption{An example of computing the centers of heptagons using polynomials of complex roots of unity. The center of each heptagonal tile is labeled with a corresponding polynomial in $\omega$. Glue labels are not shown.}
	\label{fig:heptagonExample}
\end{figure}

For a more in depth example of computing the centers of heptagonal tiles, consider the following TAS.
Let $\calT$ be the polygonal tile assembly system consisting of $10$ tile types all with shape of a single regular heptagon. Moreover, suppose that each tile type has two edges with strength-1 glues, and that there are $10$ glues appropriately defined so that starting from a single seed tile (the gray tile in Figure~\ref{fig:heptagonExample}), the assemble proceeds until the closed ``loop'' of heptagonal tiles shown in Figure~\ref{fig:heptagonExample} assembles. At this point the assembly is terminal. Call this assembly $\alpha$.
Then let $T$ be the seed tile. Keeping Figure~\ref{fig:distances7sides} in mind, we can compute the centers of each polygonal tile in $\alpha$ relative to $T$. These are shown in Figure~\ref{fig:heptagonExample}. In fact, we can even compute that the center of $T$ to obtain the polynomial $\omega^6 - \omega^3 + \omega^5 -1 +\omega^4-\omega^6+\omega^3-\omega^5+1-\omega^4$, and note that this polynomial is $0$ reflecting the fact that $\alpha$ is a closed ``loop'' of heptagonal tiles.

} %
\fi

\vspace{-10pt}
\section{Overview of Polygonal Grid Construction} \label{sec:bit-grid-main}
\vspace{-10pt}

Given a regular polygon $P$, a \emph{junction polyform} $\mathcal{P}$ is constructed in the following manner.  We begin with a polygon in standard position centered at the origin.  Starting from side $s_0$, we traverse the sides of the polygon counterclockwise until we come across the edge $s_k$ where $k$ is such that $\Re(\omega^k) <= 0$ and $j \geq k$ for all $j \in \mathbb{Z}$ such that $\Re(\omega^j) <= 0$.  We place our next polygons of type $P$ in non-standard positions centered at locations $\omega^k$ and $\overline{\omega^k}$ as shown in Figure~\ref{fig:polyform_half_main}.  Call this shape $X$.  We create a new shape $X'$ by reflecting $X$ across the line $x=\frac{1}{2}$.  We then take the union of the shapes $X$ and $X'$ obtaining our junction polyform shown in Figure~\ref{fig:polyform_full_main}.

\begin{wrapfigure}{r}{3.2in}
\vspace{-35pt}
\centering
  \subfloat[][Left half]{%
        \label{fig:polyform_half_main}%
        \makebox[0.8in]{
        \includegraphics[width=0.6in]{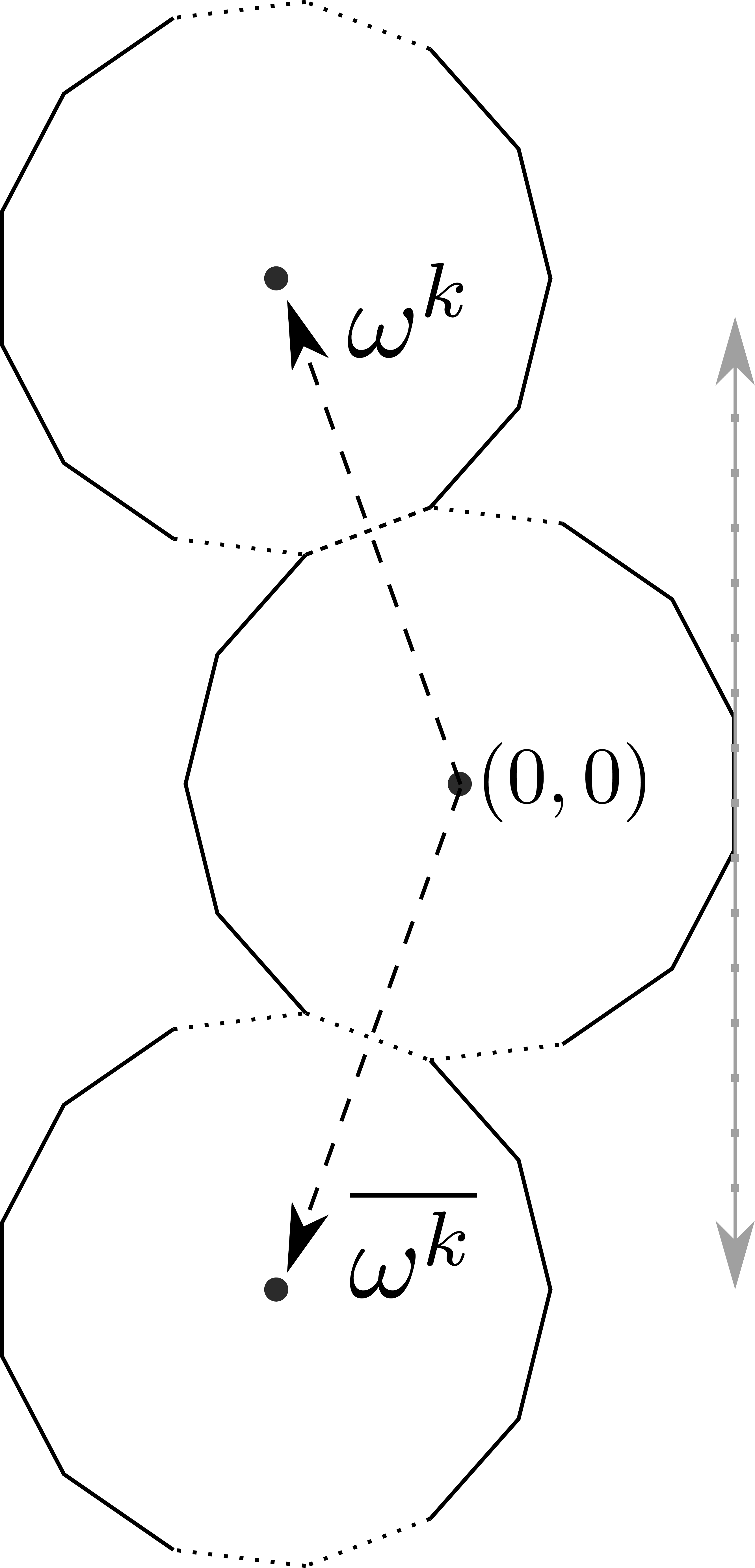}
        }
        }%
        \quad
  \subfloat[][Fully formed]{%
        \label{fig:polyform_full_main}%
        \makebox[1.0in]{
        \includegraphics[width=1.1in]{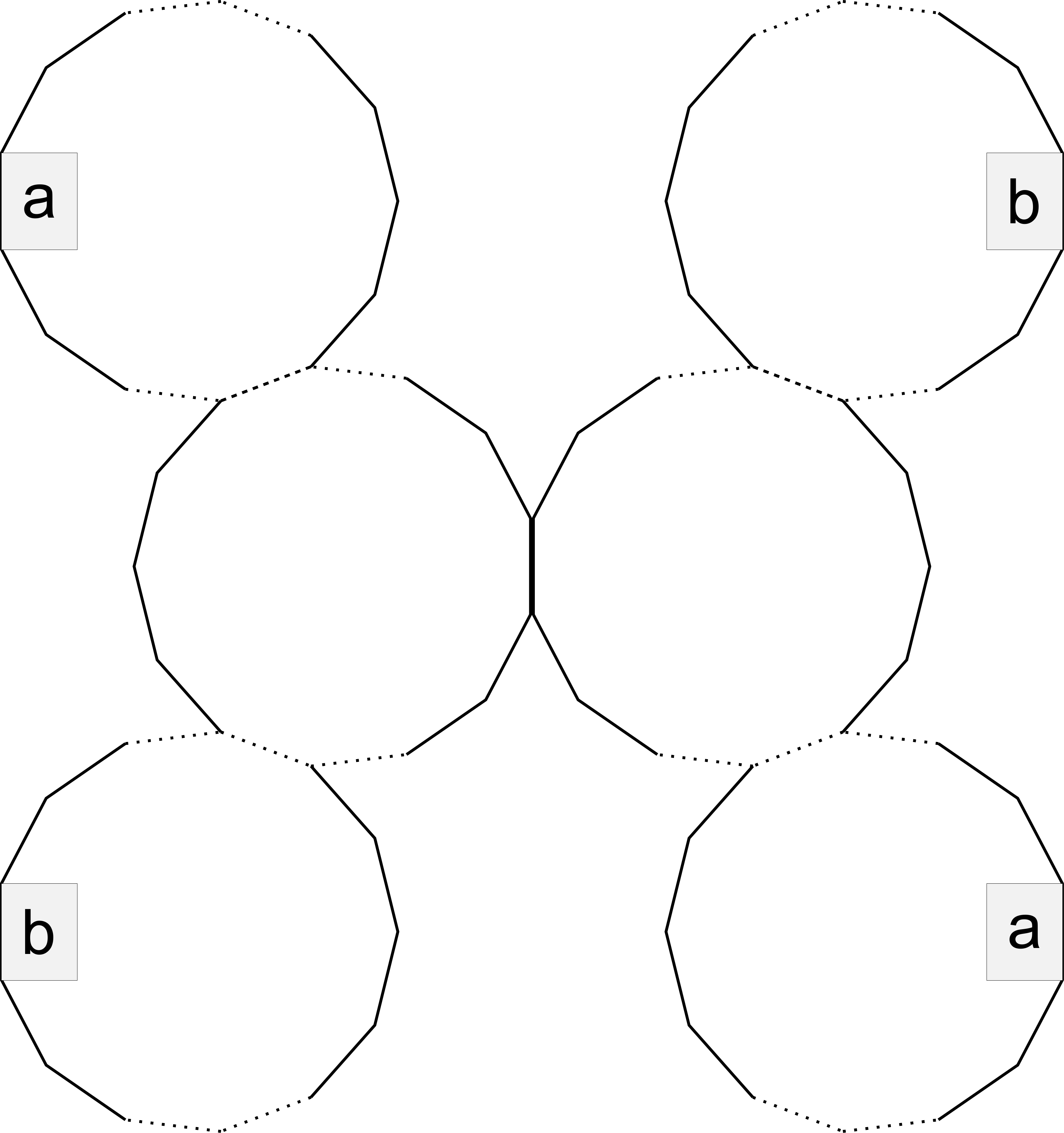}
        }
        }%
  \vspace{-22pt}
  \caption{
  Constructing a junction polyform.
  \vspace{-25pt}}
  \label{fig:polyform_main}
\end{wrapfigure}

We form a ``grid'' of these junction polyforms by attaching an infinite number of them to each other so that the polygons with sides labeled ``a'' are adjacent to each other and the polygons labeled ``b'' are adjacent to each other.

\ifabstract
\later{
\section{Polygonal Grid Construction}
Given a polygon $P$, we now show how to form a lattice consisting of $P$.  This grid will act as a coordinate system for our polygonal TAM systems and allow us to string several bit reading gadgets together so that we may simulate any Turing machine on any input.  In order to do this, we first show that we can construct a single polyform from $P$ which can ``grid'' the plane.  It will then follow that we can form a lattice in the plane with $P$ by placing polygons at the same locations and with same orientations as the polygons composing the grid formed with polyforms.

We begin by describing the construction of the polyform which we will use to construct our grid.  We then show that this is indeed a valid polyform.  Next, we shown that there exists a polygonal system which can tile the grid formed by the polyform.

Before we begin our construction, it is necessary to introduce a couple of definitions.

\begin{definition}
Let $P$ be a regular polygon.  A polyform $\mathcal{P}$ is a connected shape in the plane which is constructed by combining a finite number of copies of $P$ so that the following requirements are met:
\begin{enumerate}
\item the interior points of all instances of $P$ are disjoint
\item every instance of $P$ completely shares a common edge with some other instance of $P$.
\end{enumerate}
\end{definition}

The \emph{bounding rectangle} $B$ around a polyform $\mathcal{P}$ is the rectangle with minimal area that contains the interior points of $\mathcal{P}$.

\subsubsection{Junction Polyforms} \label{sec:polyform}
Given a regular polygon $P$, a \emph{junction polyform} $\mathcal{P}$ is constructed in the following manner.  We begin with a polygon which has standard orientation centered at the origin.  Starting from side $s_0$, we traverse the sides of the polygon counterclockwise until we come across the edge $s_k$ where $k$ is such that $\Re(\omega^k) <= 0$ and $j \geq k$ for all $j \in \mathbb{Z}$ such that $\Re(\omega^j) <= 0$.  We place our next polygons of type $P$ with negated orientations centered at locations $\omega^k$ and $\overline{\omega^k}$ as shown in Figure~\ref{fig:polyform_half}.  Call this shape $X$.  We create a new shape $X'$ by reflecting $X$ across the line $x=\frac{1}{2}$.  We then take the union of the shapes $X$ and $X'$ obtaining our junction polyform shown in Figure~\ref{fig:polyform_full}.  We call $k$ the polyform constant.

\begin{figure}[htp]
\centering
  \subfloat[][The left half of a junction polyform.]{%
        \label{fig:polyform_half}%
        \includegraphics[width=1.0in]{images/polyform_half}
        }%
        \quad
  \subfloat[][The fully formed junction polyform.]{%
        \label{fig:polyform_full}%
        \includegraphics[width=1.0in]{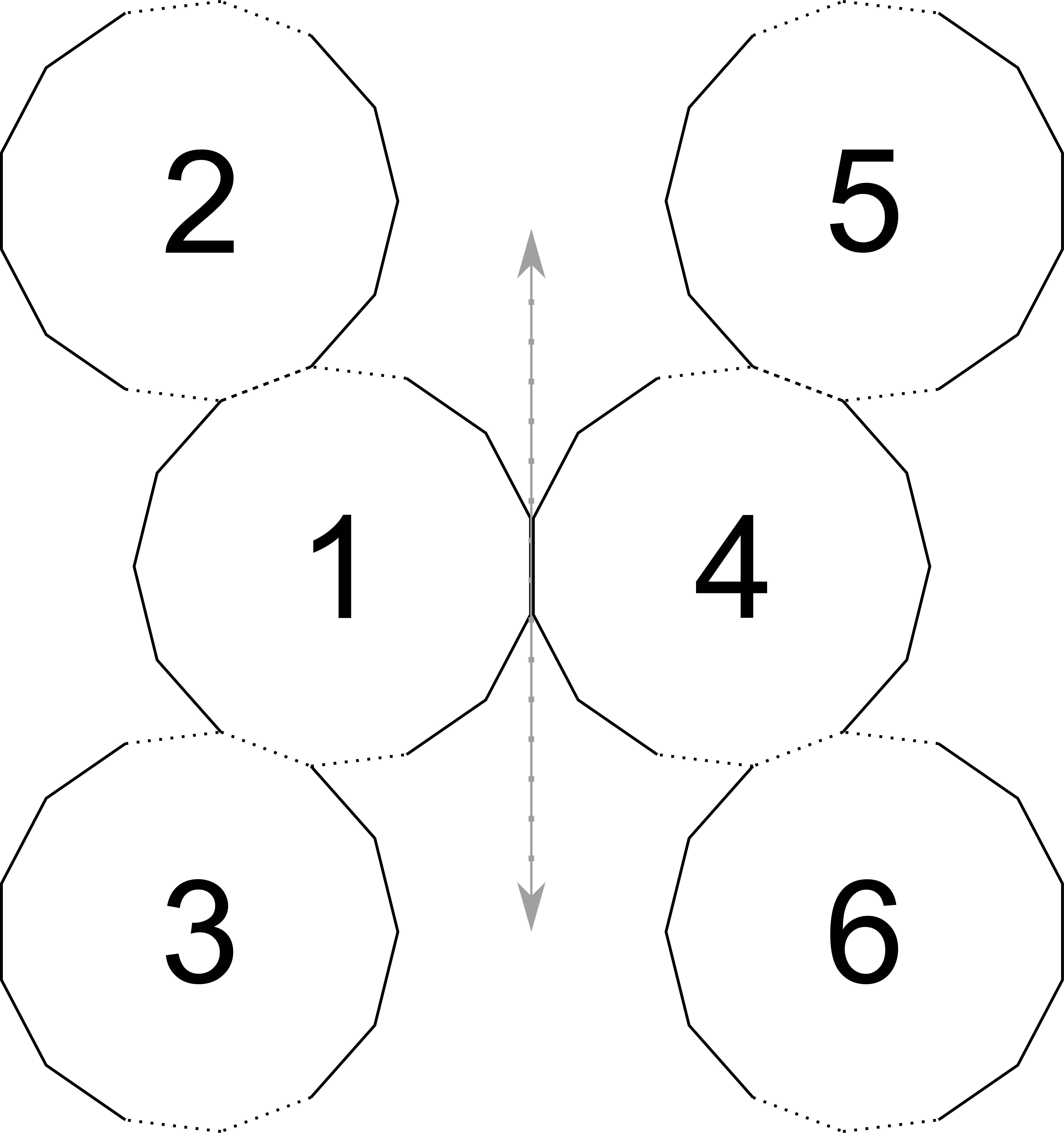}
        }%
  \caption{The construction of a junction polyform.}
  \label{fig:polyform}
\end{figure}

We now prove that this is indeed a valid polyform.  First, we begin with some observations.

\begin{observation} \label{ob:q2}
For any $n \in \mathbb{N}$ with $n > 2$, there exists a point $p$ in the $n^{th}$ roots of unity such that $-\frac{\sqrt{3}}{2} \leq \Re(p)\leq 0$.
\end{observation}

For $3 < n < 8$, this observation is mechanical.  If $n>=8$, the observation must hold since the $n^{th}$ roots of unity are evenly spaced around the unit circle.

\begin{observation} \label{ob:nonov0}
Let $P$ be a regular polygon with $n$ sides in standard orientation.  Also, let $k\in \mathbb{N}\cup\{0\}$ be such that $k \leq n$ and $\Im(-\omega^k)\leq 0$.  Denote the vertices that compose side $s_k$ by $\vec{v_l}$ and $\vec{v_r}$ where $\vec{v_l}$ is the counterclockwise most vertex and $\vec{v_r}$ is the clockwise most vertex.  Set $\vec{v} = \vec{v_r}-\vec{v_l}$.  Then the following hold:
\begin{enumerate}
\item if $\Re(-\omega^k) > 0$, then $\Im(\vec{v}) > 0$, and
\item if $\Re(-\omega^k) \leq 0$, then $\Im(\vec{v}) \leq 0$.
\end{enumerate}
\end{observation}
This observation falls out of the fact that $\omega^k$ and $\vec{v}$ must be orthogonal.
\begin{observation} \label{ob:nonov}
Let $k$ be the polyform constant for some polyform composed of regular polygons with $n$ sides.  Let $P$ be a regular polygon with $n$ sides centered at the origin in standard orientation.  Then
\begin{enumerate}
\item the clockwise most vertex that composes $s'_k$ is a southernmost point in $P$, and
\item the location of the counterclockwise most vertex that composes $s_k$, call this point $\vec{z}$, is such that $\Im(\vec{z}) \geq 0$.
\end{enumerate}
To see the first part of this observation, note that $\Re(-\omega^{k-1}) < 0$.  This along with the observation~\ref{ob:nonov0} implies that the clockwise most vertex of side $s'_{k-1}$ must lie to the north of the clockwise most vertex that composes $s'_k$.  Note that the clockwise most vertex of side $s'_{k+1}$ also must not lie to the south of the counterclockwise most vertex of side $s'_k$.  Consequently, because $P$ is convex, the clockwise most vertex that composes $s'_k$ is a southernmost point in $P$.

The second part of this observation follows from Observation~\ref{ob:q2} and the fact that a regular polygon in standard orientation centered at the origin will always have a vertex with an absent imaginary part and a real part that is less than 0.
\end{observation}

\begin{lemma} \label{lem:polyshape}
Let $P$ be a regular polygon, and let $k$ be the junction polyform constant obtained from the junction polyform composed of $P$.  Then the sets of interior points of the following polygons are pairwise disjoint: 1) the polygon in standard orientation centered at the origin, 2) the polygon with negated orientation centered at $\omega^k$, and 3) the polygon with negated orientation centered at $\overline{\omega^k}$.
\end{lemma}
\begin{proof}
It follows from the discussion in Section~\ref{sec:roots-of-unity} that the interior points of the polygon centered at the origin and the polygon centered at $\omega^k$ are disjoint.  Also since the complex conjugate of a root of unity is also a root of unity, it follows from the discussion in Section~\ref{sec:roots-of-unity} that the interior points of the polygon centered at the origin and the polygon centered at $\overline{\omega^k}$ are disjoint.

It is left to show that the interior points of the two polygons centered at the roots of unity are disjoint.  To see this, first note that it follows from Observation~\ref{ob:nonov} that no interior point of the polygon centered at the location $\omega^k$ has real part that is less than or equal to 0.  Indeed, first note that the clockwise most vertex of side $s'_k$ of the polygon centered at location $\omega^k$ will overlap the counterclockwise most vertex of side $s_k$ of the polygon centered at the origin by construction.  It follows immediately from Observation~\ref{ob:nonov} that all interior points in the polygon centered at $\omega^k$ have imaginary parts greater than $0$.  Since the polygon centered at $\overline{\omega^k}$ is a reflected copy of the polygon centered at $\overline{\omega^k}$, it follows that the interior points in this polygon have imaginary parts less than 0.  Consequently, the interior points of the two polygons are disjoint.

\end{proof}

\begin{lemma}
Given a regular polygon $P$, the junction polyform constructed above is indeed a valid polyform.
\end{lemma}

\begin{proof}
To see that the junction polyform constructed above is a valid polyform, we check that all of the requirements in the definition of polyform are met.  Since the center of polygons labeled ``2'' and ``3'' are each located at one of the $n^{th}$ roots of unity, it follows from the discussion in Section~\ref{sec:roots-of-unity} that polygons labeled ``1'' and ``2'' as well as polygons labeled ``1'' and ``3'' are joined along a common edge and share that edge entirely.  This same line of reasoning shows that the polygon labeled ``4'' is joined along a common edge and shares that edge entirely with the polygon labeled ``1''.  Since the shape formed by polygons labeled ``4'', ``5'' and ``6'' is a reflection of the left side of the shape, all of the polygons are joined along a common edge and shares that edge entirely. It is readily seen from this argument that our shape is also connected.

It is now left to show that no two polygons in the shape overlap.  We denote the polyform constant obtained from $P$ by $k$. It follows from Lemma~\ref{lem:polyshape} that the interior points of the polygons labeled ``1'', ``2'', and ``3'' are pairwise disjoint.  Since, the polygons labeled ``4'', ``5'', and ``6'' are a reflection of the polygons labeled ``1'', ``2'', and ``3'', they too are pairwise disjoint. To show that the polygons in the two reflected halves of the shape are pairwise disjoint, first observe that the centers of the polygons labeled ``2'' and ``3'' have real parts less than or equal to the real part of the polygon labeled ``1''.  Consequently, after the reflection and ``attachment'' of the two halves, the polygons labeled ``2'' and ``5'' and the polygons labeled ``3'' and ``6'' have no less distance between each other than the polygons labeled ``1'' and ``4''.  Since the polygons labeled ``1'' and ``4'' have disjoint interior points, it follows that the polygons mentioned above have disjoint interior points.  Consequently, no two polygons in the shape overlap.

\end{proof}

\subsubsection{Polygonal Grid Technical Lemmas}
The following lemma will assist us in proving Lemma~\ref{lem:grids}.  Informally, it states that the bounding rectangle of the junction polyform described above and shown in Figure~\ref{fig:polyform_full} will ``touch'' sides $s'_0$ of the polygons labeled ``2'' and ``3''  and sides $s_0$ of the polygons labeled ``5'' and ``6''.  This will imply that we can attach the polyform junctions by attaching sides $s_0$ of polygons labeled ``5'' and ``6'' to sides $s'_0$ of polygons labeled ``5'' and ``6''.

\begin{lemma}
Consider the polygons composing the junction polyform $\mathcal{P}$ constructed above from some regular polygon $P$ (shown in Figure~\ref{fig:polyform_full}). Also, let $B$ be the bounding rectangle around $\mathcal{P}$.  Let $E$ be the set of points consisting of the union of the following sets of points: 1) the set of boundary points on side $s'_0$ of the polygon labeled ``2'', 2) the set of boundary points on side $s'_0$ of the polygon labeled ``3'', 3) set of boundary points on side $s_0$ of the polygon labeled ``5'', and 4) the set of boundary points on side $s_0$ of the polygon labeled ``6''.  Then $E \subset E \cap B$.
\end{lemma}

\begin{proof}
We prove that the boundary points on side $s_0$ of the polygons labeled ``4'' and ``6'' in Figure~\ref{fig:polyform_full} lie on the bounding rectangle $B$.  The proof that the boundary points on side $s'_0$ of the polygons labeled ``2'' and ``3'' lie on the bounding rectangle will then follow from a similar argument.

First, observe that for a polygon $P$ with standard position centered at the origin, the boundary points on side $s_0$ are the easternmost points contained in the polygon.  Furthermore, all of these points lie on the line $x=\frac{1}{2}$.  Now note that by our construction of the junction polyform, one of the tiles labeled ``5'' and ``6'' will contain the easternmost point of the polyform.  Indeed, let $x_4$ be the real part of the point in the center of the polygon labeled ``4''.  Since our construction ensures the real part of the point in the center of the polygon labeled ``5'' is of the form $x_4+r$ for $r\in[0,\frac{\sqrt{3}}{2}]$, the polygon labeled ``5'' will contain a point as east or further east than the points in the polygon labeled ``4''.

We claim that the polygons labeled ``5'' and ``6'' have centers with equal real parts.  To see this, recall that the centers of the polygons labeled ``2'' and ``3''have the same real parts since they are conjugates of each other.  Since the polygons labeled ``5'' and ``6'' are in the same position relative to each other as the polygons labeled ``2'' and ``3'' just reflected across the line $y=\frac{1}{2}i$, it follows that the polygons labeled ``5'' and ``6'' have equal real parts.

From our construction of the junction polyform, it is clear that none of the polygons labeled ``1'', ``2'', or ``3'' have a point that is an easternmost point of the polyform.  Thus, the $s_0$ sides of the polygons labeled ``5'' and ``6'' are all easternmost points of the polyform.  Consequently, these points lie on the bounding box $B$.
\end{proof}

\begin{observation} \label{lem:box}
Let $P$ be a regular polygon, $\mathcal{P}$ be a polyform junction formed from $P$, $B$ be the bounding rectangle for $\mathcal{P}$, and let $k$ be the polyform constant.  Furthermore, let $h_b$ be the height of the bounding rectangle and let $h_w$ be the width of the bounding rectangle.  Then, the following constraints hold for $h_b$ and $h_w$: 1) $h_b \leq 4\Im(\omega^k)$ and 2) $h_w = 2\Re(-\omega^k + 1)$.
\end{observation}

Figure~\ref{fig:polyform_dims} shows the dimensions of the polyform.  Note that the width of the polyform is clearly $2\Re(-\omega^k + 1)$.  To see that $h_b \leq 4\Im(\omega^k)$, note that by the way we constructed the junction polyform no interior points of the polyform can lie on the dotted lines shown in Figure~\ref{fig:polyform_dims}.  Since the distance between these two dotted lines is $4\Im(\omega^k)$, it must be the case that $h_b \leq 4\Im(\omega^k)$.

\begin{figure}[htp]
\begin{center}
\includegraphics[width=4.0in]{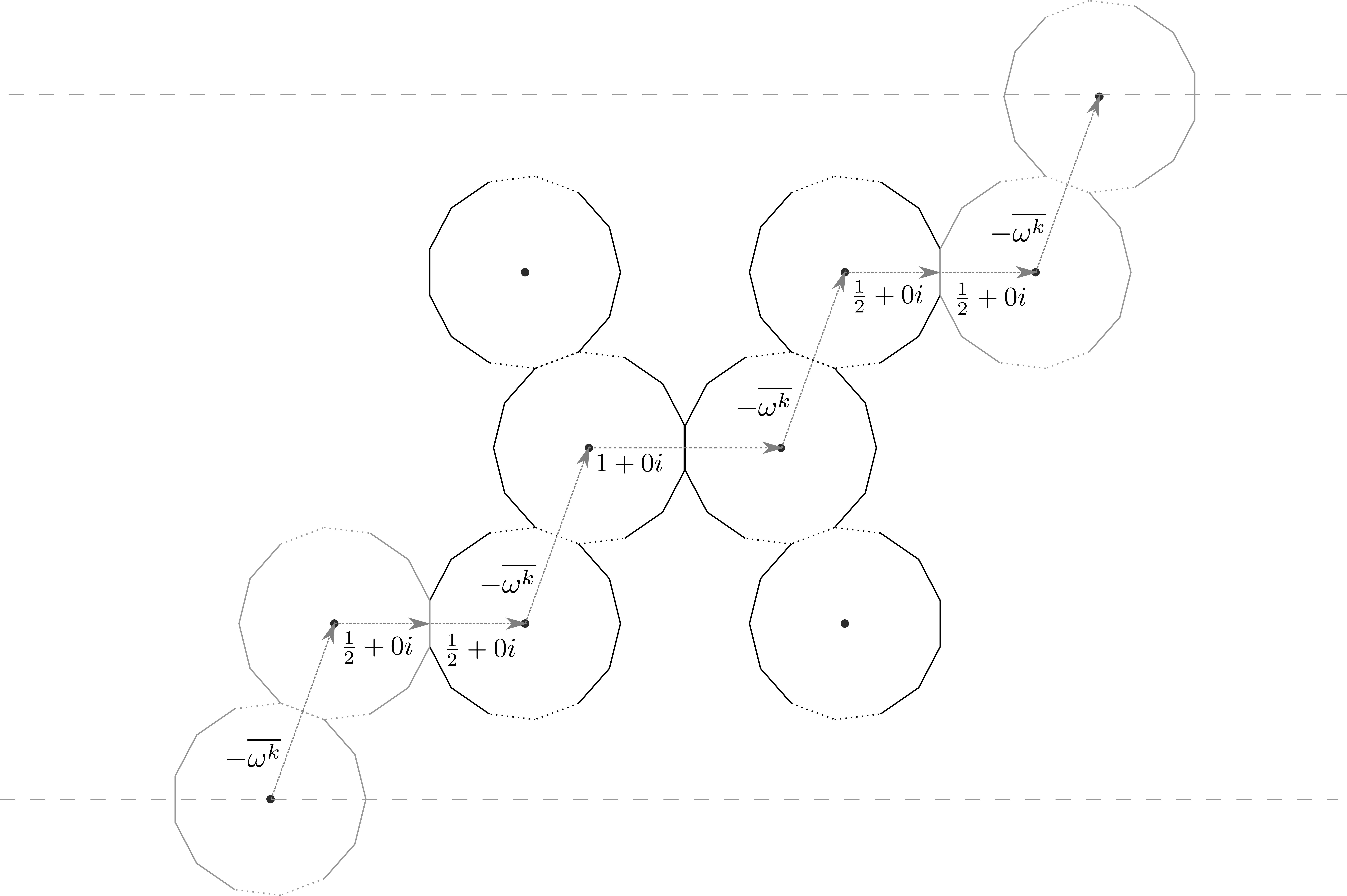}
\caption{The vectors showing the dimensions of the polyforms.}
\label{fig:polyform_dims}
\end{center}
\end{figure}

The next lemma states that given any regular polygon, we can form a a periodic grid of the plane.

\begin{figure}[htp]
\begin{center}
\includegraphics[width=2.0in]{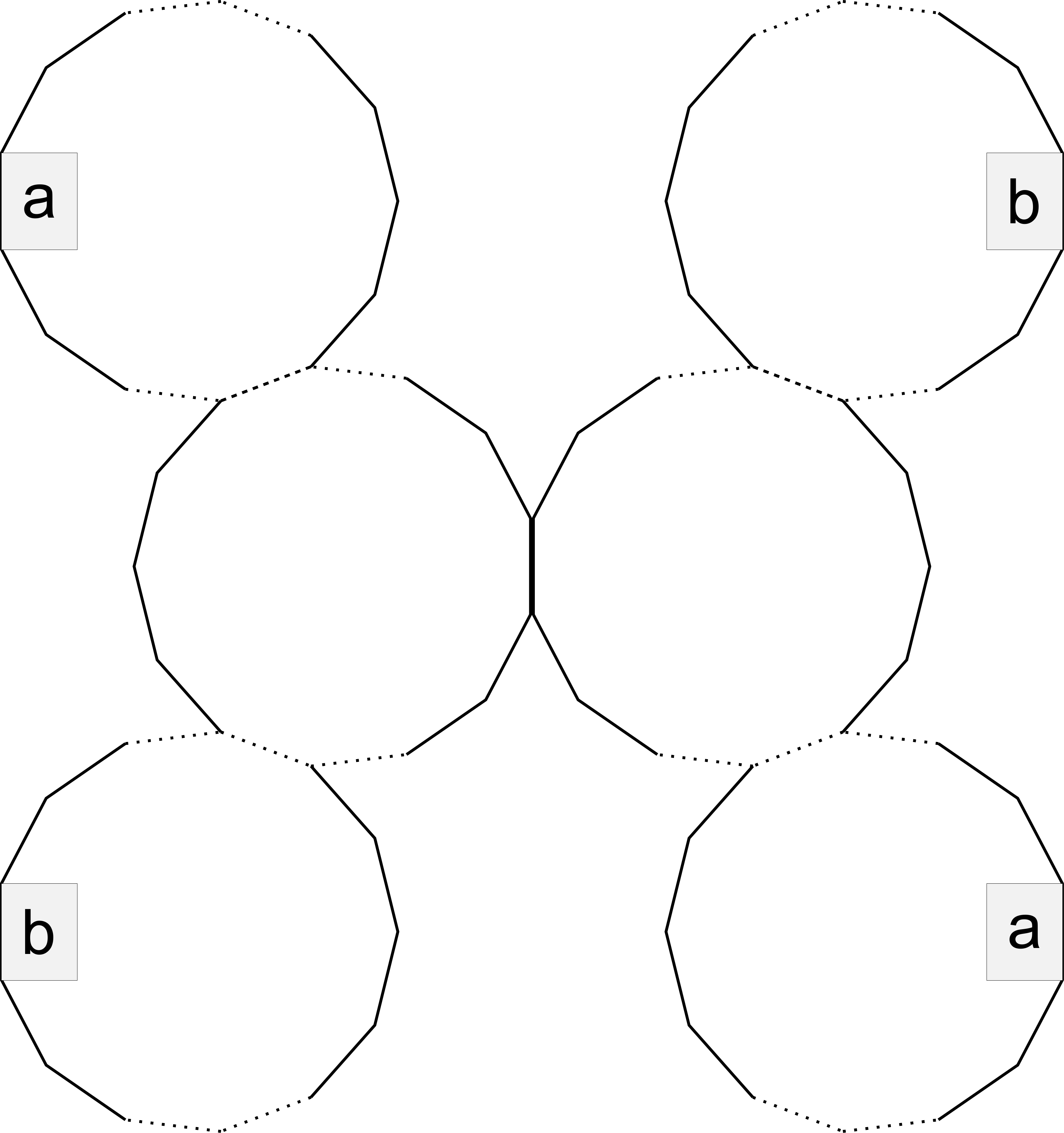}
\caption{The preformed assembly which is composed of the tile set of the system described in the proof of Lemma~\ref{lem:grids}.  The preformed assembly has two glues labeled ``a'' and ``b'' placed as shown.}
\label{fig:polyTile}
\end{center}
\end{figure}

\begin{figure}[htp]
\begin{center}
\includegraphics[width=4.5in]{images/polyTileGrid}
\caption{An assembly formed by the system described in the proof of Lemma~\ref{lem:grids}.}
\label{fig:polyTileGrid}
\end{center}
\end{figure}

\begin{figure}[htp]
\centering
  \subfloat[][The path of vectors which yields the vector $\vec{v}$.]{%
        \label{fig:polyformv}%
        \includegraphics[width=2.2in]{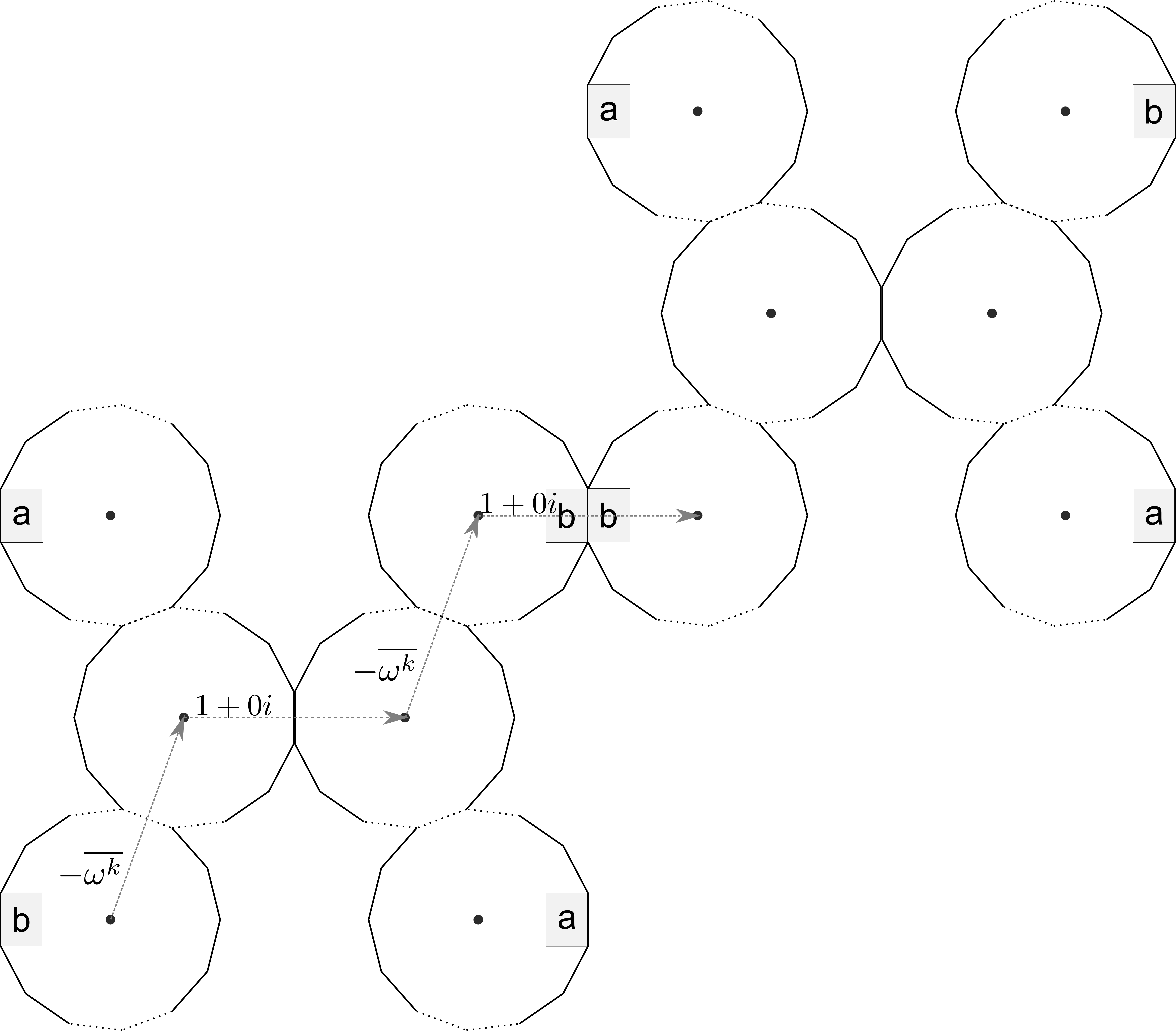}
        }%
        \quad
  \subfloat[][The path of vectors which yields the vector $\vec{w}$.]{%
        \label{fig:polyformw}%
        \includegraphics[width=2.2in]{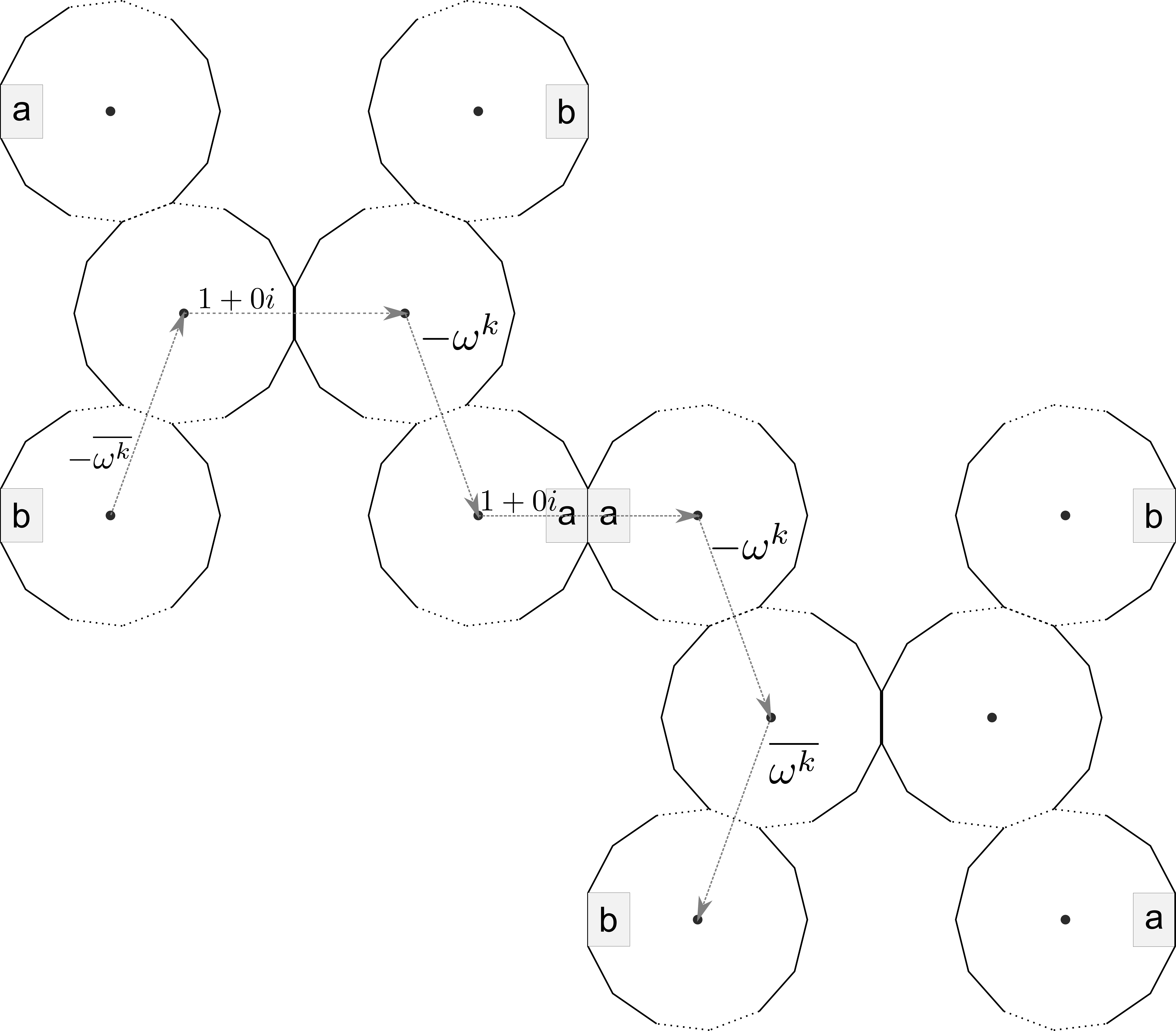}
        }%
  \caption{Choosing the vectors $\vec{v}$ and $\vec{w}$.}
  \label{fig:polyformvw}
\end{figure}

\subsubsection{Constructing the Polygonal Grid}

\begin{lemma} \label{lem:grids}
Given a regular polygon $P$, there exists a directed, polygonal tile system $\mathcal{T} = (T, \sigma)$ (where the seed is centered at location $(0,0)$ and the tile set $T$ contains a tile $t$) and vectors $\vec{v}, \vec{w} \in \mathbb{Z}^2$, such that $\mathcal{T}$ produces the terminal assembly $\alpha$, which we refer to as a \emph{grid}, with the following properties. (1) Every position in $\alpha$ of the form $c_1\vec{v} + c_2\vec{w}$, where $c_1, c_2 \in \mathbb{Z}$, is occupied by the tile $t$, and (2) for every $c_1, c_2\in \mathbb{Z}$, the position in $\Z^2$ of the form $c_1\vec{v} + c_2\vec{w}$ is occupied by the tile $t$.
\end{lemma}

\begin{proof}
For the first part of this proof, we think of our polygonal tile system as first forming the junction polyform $\mathcal{P}$ before attaching it to our assembly.  Later on in the proof, we will see that this is a valid assumption.  Our tile set $T$, will consist of tiles of shape $P$ that form the junction polyform with the glues labeled ``a'' and ``b'' exposed as shown in Figure~\ref{fig:polyTile}.  Note that for the first part of the proof we are essentially thinking of the assembly shown in Figure~\ref{fig:polyTile} as a tile.  Thus, we refer to the junction polyform as a tile and we refer to a polygon composing the polyform as a pixel.  More formally, a pixel in the polyform is a location in the complex plane given by the center of a tile in the polyform shown in Figure~\ref{fig:polyform} (where we assume that the center of the tile labeled ``1'' is placed at the origin).

To begin, we position our single seed so that the polygon labeled ``1'' in Figure~\ref{fig:polyform_full} is centered at the origin.  An assembly formed by such a system is shown in Figure~\ref{fig:polyTileGrid}.

Let $\mathcal{P}$ be a junction polyform composed of the polygon $P$ and let $k$ be the polyform constant as discussed in the construction of the junction polyform.  Set $\vec{v} = -\overline{\omega^k} + (1+0i) - -\overline{\omega^k} + (1+0i) = -2\overline{\omega^k} + 2(1+0i)$ and $\vec{w}=  -\overline{\omega^k} + (1+0i) - \omega^k + (1+0i)  - \omega^k + \overline{\omega^k} = -2\omega^k + 2(1+0i)$.  The intuition behind choosing these vectors is shown in Figure~\ref{fig:polyformv} and Figure~\ref{fig:polyformw}.

The following terminology is borrowed from \cite{Polyominoes}.  Define $\mathcal{P}[i,j] = p + i \cdot \vec{v} + j \cdot \vec{w}$ for $i,j \in \mathbb{Z}^2$.  Here, $p$ acts as a distinguished pixel that we use as a reference point.  Then, for two polyforms $\mathcal{P}[i,j]$ and $\mathcal{P}[k, l]$, we say that these polyforms are neighboring if $i=k$ and $|j-l|=1$ or $j=l$ and $|i-k| = 1$.

As in \cite{Polyominoes} we prove the following claim.

\emph{Claim:} $\mathcal{P}[i,j]$ for all $i, j \in \mathbb{Z}^2$ defines a grid of non-overlapping polyforms such that any two neighboring polyforms $\mathcal{P}[i,j]$ and $\mathcal{P}[k,l]$ contain pixels with a shared edge.  Such a grid of polyforms is shown in Figure~\ref{fig:polyTileGrid}.

To begin, we show that if $i \neq k$ or $j \neq l$, then the interior points of $\mathcal{P}[i,j]$ and $\mathcal{P}[k,l]$ are disjoint.  Let $a=(k-i)$ and $b=(l-j)$.  In order to show that $\mathcal{P}[i,j]$ does not overlap $\mathcal{P}[k,l]$, we show that 1)$|\Re(a\vec{v} + b\vec{w})| \geq |2\Re(-\omega^k + 1)|$ or 2) $|\Im(a\vec{v} + b\vec{w})| \geq |4\Im(\omega^k)|$.  Since, by Lemma~\ref{lem:box}, these are the dimensions of the bounding box of $\mathcal{P}$, it will then follow that their interiors are disjoint.

We consider three cases 1) $a+b > 0$, 2) $a+b=0$, and 3) $a+b < 0$.  First note that
\begin{eqnarray*}
a\vec{v}+b\vec{w} &=& a(-2\overline{\omega^k}+ 2(1+0i))+b(-2\omega^k+2(1+0i)) \\
                  &=& -2(a\overline{\omega^k}+b\omega^k)+2(a+b) 
\end{eqnarray*}

For case 1, observe that
\begin{eqnarray*}
|\Re(a\vec{v}+b\vec{w})| &=& |\Re(-2(a\overline{\omega^k}+b\omega^k) + 2(a+b))| \\
                        &=& |-2(a\Re(\overline{\omega^k})+b\Re(\omega^k))+2(a+b)| \\
                        &=& |-2\Re(\omega^k)(a+b)+2(a+b)| \\
                        &\geq& |-2\Re(\omega^k)+2|. 
\end{eqnarray*}

In the case that $a+b=0$, we have
\begin{eqnarray*}
|\Im(a\vec{v}+b\vec{w})| &=& |\Im(-2(a\overline{\omega^k}+b\omega^k)+2(a+b))| \\
                       &=& |\Im(-2(a\overline{\omega^k}+b\omega^k))| \\
                       &=& |\Im(-2((-b)\overline{\omega^k}+b\omega^k))| \\
                       &=& |-2(b)(\Im(-\overline{\omega^k})+\Im(\omega^k))| \\
                       &=& |-2(b)(2\Im(\omega^k))| \\
                       &\geq& |-4\Im(\omega^k)|.  
\end{eqnarray*}

Although case 3 is similar to case 1, we include it here for completeness.  If $a+b<0$, notice that
\begin{eqnarray*}
|\Re(a\vec{v}+b\vec{w})| &=& |\Re(-2(a\overline{\omega^k}+b\omega^k) + 2(a+b))| \\
                        &=& |-2(a\Re(\overline{\omega^k})+b\Re(\omega^k))+2(a+b)| \\
                        &=& |-2\Re(\omega^k)(a+b)+2(a+b)| \\
                        &\geq& |2\Re(\omega^k)-2|. 
\end{eqnarray*}

Now suppose that $\mathcal{P}[i,j]$ and $\mathcal{P}[k,l]$ are neighboring polyforms.  First, suppose that $i=k$ and $|j-l|=1$.  We consider the case where $l = j + 1$ and note that the case where $l=j-1$ is similar.  Consider the polygons in the lower left hand corner of the bounding rectangle of the polyforms and denote this polygon $p$.  Note that the polygon $p$ in $\mathcal{P}[k,l]$ lies at a position
\begin{eqnarray*}
(k\vec{v}+l\vec{w}) - (i\vec{v}+j\vec{w}) &=& (i\vec{v} + (j+1)\vec{w}) - (i\vec{v} + j\vec{w}) \\
                                          &=& \vec{w} 
\end{eqnarray*}
relative to the polygon $p$ in $\mathcal{P}[i,j]$.

Now, notice that $\mathcal{P}[i,j]$ has a polygon that lies at position $-\overline{\omega^k} + (1+0i) - \omega^k$ relative to $p$ in $\mathcal{P}[i,j]$(this is the polygon that lies in the bottom right hand corner of the bounding box), and $\mathcal{P}[k,l]$ has a polygon that lies at position $-\overline{\omega^k} + \omega^k$ relative to $p$ in $\mathcal{P}[k,l]$ (this is the polygon that lies in the top left hand corner of the bounding box).  Call the first pixel described $p'$ and the latter $p''$.   Observe that by the construction of the junction polyform, $p'$ has standard orientation and $p''$ has negated orientation.  Furthermore, observe that $p''$ lies at location
\begin{eqnarray*}
(\vec{w} + (-\overline{\omega^k} + \omega^k) - (-\overline{\omega^k} + (1+0i) - \omega^k) &=&  -2\omega^k + 2(1+0i) + (-\overline{\omega^k} + \omega^k) - (-\overline{\omega^k} + (1+0i) - \omega^k)\\
                                                                                          &=& (1+0i) 
\end{eqnarray*}
relative to $p'$.  Since $p'$ has standard orientation, $p''$ has negated orientation and $p''$ lies at position $(1+0i)$ relative to $p'$, it follows from the discussion in Section~\ref{sec:roots-of-unity} that polygon $p'$ and polygon $p''$ completely share a common edge.

Conversely, assume that $j=l$ and $|i-k| = 1$.  We consider the case where $k = i - 1$, and, once again, note that the case where $k = i + 1$ is similar.  Notice that the polygon $p$ in $\mathcal{P}[k,l]$ lies at a position
\begin{eqnarray*}
(k\vec{v}+l\vec{w}) - (i\vec{v}+j\vec{w}) &=& ((i-1)\vec{v} + j\vec{w}) - (i\vec{v} + j\vec{w}) \\
                                          &=& -\vec{v} 
\end{eqnarray*}
relative to the polygon $p$ in $\mathcal{P}[i,j]$.

Denote the polygon that lies at position $-2\overline{\omega^k} + (1+0i)$ relative to $p$ in $\mathcal{P}[k,l]$ by $p'$ (this is the polygon that lies in the top right hand corner of the bounding box).  Observe that, relative to polygon $p$ in $\mathcal{P}[i,j]$, the polygon $p'$ in $\mathcal{P}[k,l]$ lies at position
\begin{eqnarray*}
-\vec{v} + (-2\overline{\omega^k} + (1+0i)) &=& -(-2\overline{\omega^k} + 2(1+0i)) + (-2\overline{\omega^k} + (1+0i)) \\
                                            &=& -(1+0i). 
\end{eqnarray*}
Since $p$ in $\mathcal{P}[i,j]$ has negated orientation, $p'$ in $\mathcal{P}[k,l]$ has standard orientation, and $p'$ lies at a position $-(1+0i)$ relative to $p$, it follows from the discussion in Section~\ref{sec:roots-of-unity} that polygon $p$ and polygon $p''$ completely share a common edge.

Now, note that since none of the ``polyform junction tiles'' overlap, there are not any race conditions.  Consequently, we can build the assembly described above by attaching one polygon tile at a time (instead of an assembly of polygons).  The seed of our assembly will be the southwest tile of $\mathcal{P}[0,0]$.
\end{proof}

\subsection{Grid Notation}\label{sec:grid-notation}

For some polygon $P$, we let $g_{\alpha}$ denote the terminal assembly of the tile system given in Lemma~\ref{lem:grids} (i.e. the grid assembly obtained from $P$). Furthermore, for a tile system $\mathcal{T}$ of shape $P$, $\alpha \in \prodasm{\mathcal{T}}$, and $t$ a tile of $\alpha$ centered at the location $\vec{x}$, we say that $t$ is \emph{on grid} with respect to $g_{\alpha}$ if there exists a tile $t'\in g_{\alpha}$ such that $t'$ is centered at the location $\vec{x}$ and has the same orientation of $t$.  If there does not exist such a $t'$, then we say that $t$ is \emph{off grid} with respect to $g_{\alpha}$.

\subsection{\emph{Normalized} Bit-reading Gadgets}
 Let a bit reading gadget have the properties that: 1)the tile from which the bit writer begins growth is on grid, 2) the last tile to be placed in the bit writer is on grid, and 3) the tile $t$ from which the bit reader grows is also placed on grid.  We call such a bit-reading gadget an \emph{on grid bit-reading gadget}.  A pair of \emph{normalized bit-writers} $\alpha_{u0}$ and $\alpha_{u1}$ have the property that 1) $\alpha_{u0}$ and $\alpha_{u1}$ are the two bit writers for some bit reading gadget and 2) the location and position of the first tile placed in the two assemblies is the same as well as the location and position of the last tile placed.  A \emph{normalized bit-reading gadget} is an on grid bit-reading gadget with normalized bit-writers.
} %
\fi 

\ifabstract
\later{
\section{Polygons Which ``Can't Compute'' at Temperature 1}\label{sec:imposs-poly}

In this section, we prove Theorem~\ref{thm:cant-bit-read} by showing a set of polygons for which it is impossible to create bit-reading gadgets at $\tau=1$, namely regular polygons with less than 7 sides (i.e. equilateral triangles, regular pentagons, and regular hexagons), as this was already shown to be true for squares in \cite{Polyominoes}.  This provides a sharp dividing line, since we have shown that all regular polygons with $\ge 7$ sides can form bit reading gadgets, and thus are capable of universal computation, at $\tau=1$.

We now restate the Theorem for completeness and give its proof.

\begin{theorem}\label{thm:cant-bit-read-append}
Let $n\in \N$ be such that $3\leq n \leq 6$. Then, there exists no temperature 1 single-shaped polygonal tile assembly system $\mathcal{T} = (T,\sigma,1)$ where for all $t \in T$, $t$ is a regular polygon with $n$ sides, and a bit-reading gadget exists for $\mathcal{T}$.
\end{theorem}

To prove Theorem~\ref{thm:cant-bit-read-append}, we break it into two main cases and prove lemmas about (1) equilateral triangles and hexagons, and (2) pentagons.

\subsection{Equilateral triangles, squares, and regular hexagons}

Equilateral triangles, squares, and regular hexagons are all capable of tessellations of the plane.  That is, using tiles of only one of those shapes it is possible to tile the entire plane with no gaps.  (As a side note, these are the only regular polygons which can do so.)  In a system consisting of tiles of only one of those shapes, all tiles must be placed into positions aligning with a regular grid (i.e. no tile can be offset or rotated from the grid).  It was shown in \cite{Polyominoes} that squares cannot form bit-reading gadgets at $\tau=1$, and because of the tessellation ability of equilateral triangles and regular hexagons and their restriction to fixed grids, the proof of \cite{Polyominoes} can be extended in a straightforward way to also prove that equilateral triangles and regular hexagons cannot form bit reading gadgets at $\tau=1$.  Thus, the following proof is nearly identical to that for squares of \cite{Polyominoes}.

\begin{lemma}\label{lem:tri}
There exists no temperature 1 polygonal tile assembly system $\mathcal{T} = (T,\sigma,1)$ where for all $t \in T$, $t$ is an equilateral triangle, and a bit-reading gadget exists for $\mathcal{T}$.
\end{lemma}

\begin{lemma}\label{lem:hex}
There exists no temperature 1 polygonal tile assembly system $\mathcal{T} = (T,\sigma,1)$ where for all $t \in T$, $t$ is a regular hexagon, and a bit-reading gadget exists for $\mathcal{T}$.
\end{lemma}

\begin{proof}
We prove Lemmas~\ref{lem:tri} and \ref{lem:hex} by contradiction.  Also, since each will use exactly the same arguments, we will prove both simultaneously and note the single location in the proof where the shapes of the tiles is relevant.  Therefore, assume that there exists a single-shape system $\mathcal{T} = (T,\sigma,1)$ such that $\mathcal{T}$ has a bit-reading gadget.  (Without loss of generality, assume that the bit-reading gadget reads from right to left and has the same orientation as in Definition~\ref{def:bit-reader}.)  Let $(t_x,t_y)$ be the coordinate of the tile $t$ from which the bit-reading paths originate (recall that it is the same coordinate regardless of whether or not a $0$ or a $1$ is to be read from $\alpha_0$ or $\alpha_1$, respectively).  By Definition~\ref{def:bit-reader}, it must be the case that if $\alpha_0$ is the only portion of $\alpha$ in the first quadrant to the left of $t$, then at least one path can grow from $t$ to eventually place a tile from $T_0$ at $x=0$ (without placing a tile below $y=0$ or to the right of $(t_x-1)$).  We will define the set $P_0$ as the set of all such paths which can possibly grow.  Analogously, we will define the set of paths, $P_1$, as those which can grow in the presence of $\alpha_1$ and place a tile of a type in $T_1$ at $x=0$.  Note that by Definition~\ref{def:bit-reader}, neither $P_0$ nor $P_1$ can be empty.

Since all paths in $P_0$ and $P_1$ begin growth from $t$ at $(t_x,t_y)$ and must always be to the left of $t$, at least the first tile of each must be placed in location $(t_x-1,y)$.  We now consider a system where $t$ is placed at $(t_x,t_y)$ and is the only tile in the plane (i.e. neither $\alpha_0$ nor $\alpha_1$ exist to potentially block paths), and will inspect all paths in $P_0$ and $P_1$ in parallel.  If all paths follow exactly the same sequence of locations (i.e. they overlap completely) all the way to the first location where they place a tile at $x = 0$, we will select one that places a tile from $T_0$ as its first at $x=0$ and call this path $p_0$, and one which places a tile from $T_1$ as its first at $x = 0$ and call it $p_1$.  This situation will then be handled in Case (1) below.  In the case where all paths do not occupy the exact same locations, then there must be one or more locations where paths branch.  Since all paths begin from the same location, we move along them from $t$ in parallel, one tile at a time, until the first location where some path, or subset of paths, diverge.  At this point, we continue following only the path(s) which take the clockwise-most branch.  We continue in this manner, taking only clockwise-most branches and discarding other paths, until reaching the location of the first tile at $x = 0$.  (Figures~\ref{fig:tri-paths} and \ref{fig:hex-paths} show examples of this process.)  We now check to see which type(s) of tiles can be placed there, based on the path(s) which we are still following.  We again note that by Definition~\ref{def:bit-reader}, some path must make it this far, and must place a tile of a type either in $T_0$ or $T_1$ there.  If there is more than one path remaining, since they have all followed exactly the same sequence of locations, we randomly select one and call it $p'$.  If there is only one, call it $p'$.  Without loss of generality, assume that $p'$ can place a tile from $T_0$ at that location.  This puts us in Case (2) below.

\begin{figure}[htp]
\centering
  \subfloat[][Example sets $P_0$ and $P_1$, with $p'$ traced with a red line.  Red triangles represent branching points of paths, gold triangles represent overlapping points of different branches.]{%
        \label{fig:tri-paths}%
        \makebox[2.2in][c]{ \includegraphics[width=2.1in]{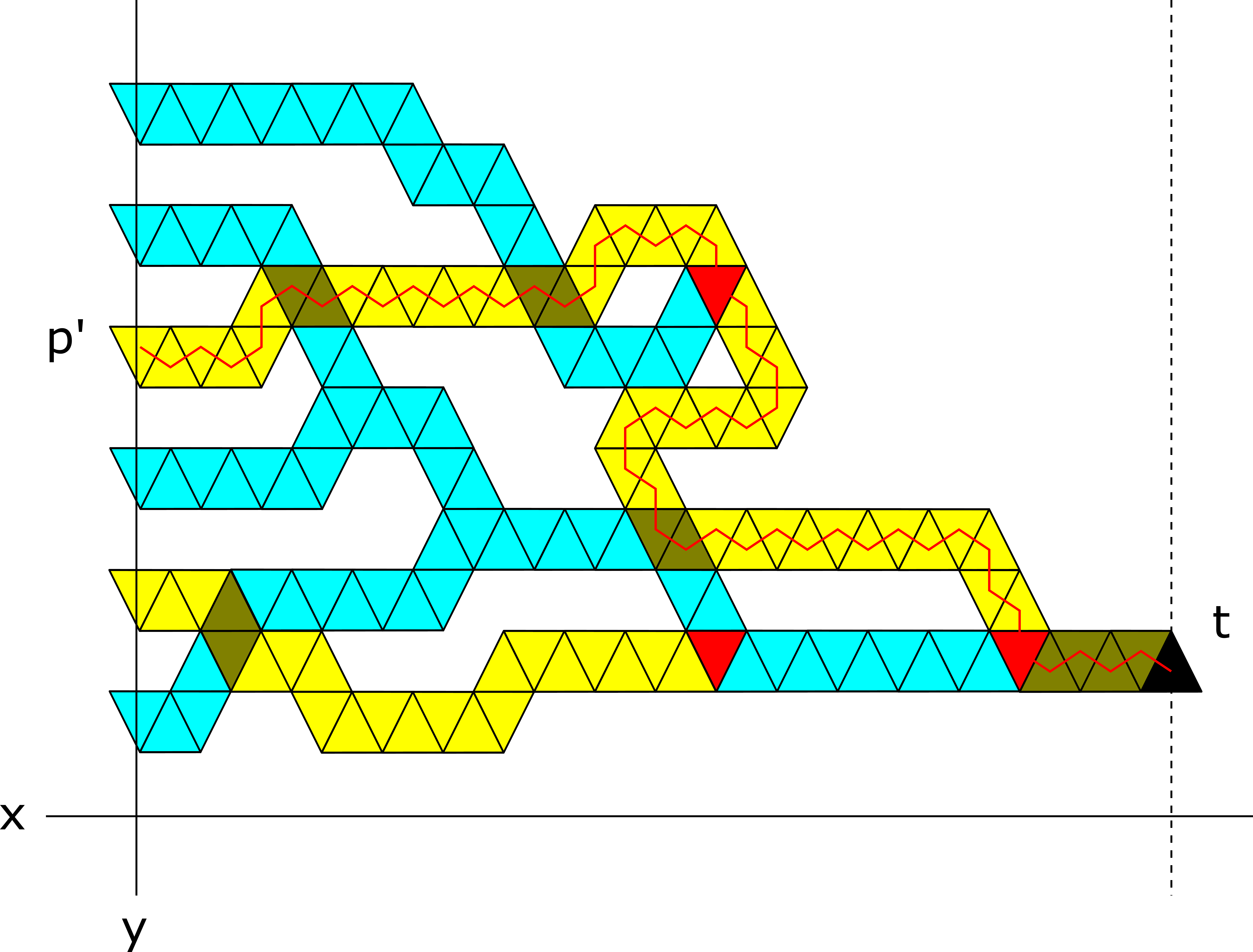}}
        }%
        \quad
  \subfloat[][An example of the growth of $p'$ (traced with a red line) blocked by $\alpha_1$.  By first letting as much of $p'$ grow as possible, it is guaranteed that all other paths must be blocked from reaching $x=0$.]{%
        \label{fig:tri-paths-blocked}%
        \makebox[2.2in][c]{\includegraphics[width=2.1in]{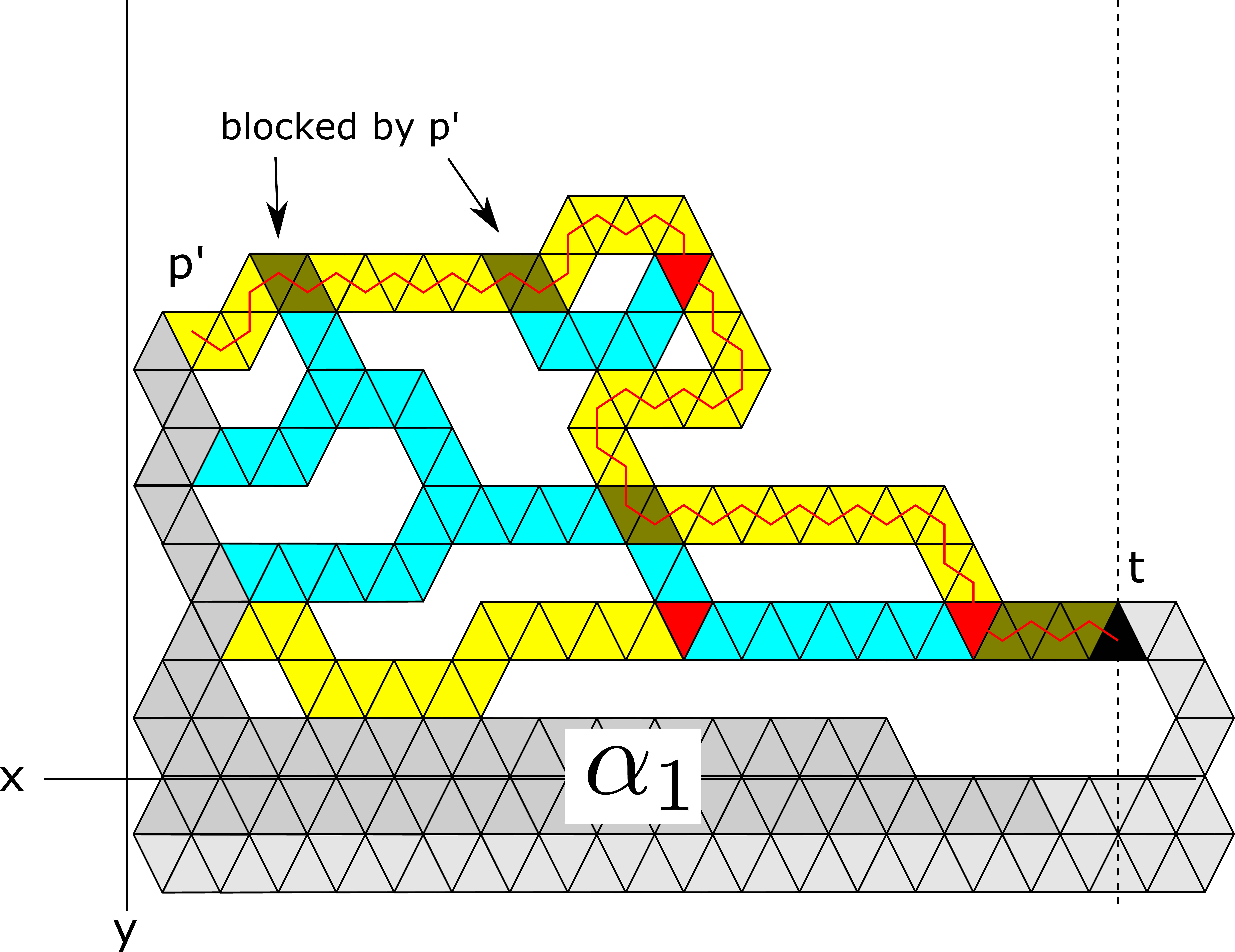}}
        }%
  \caption{Failed bit-readers with equilateral triangles.}
  \label{fig:failed-triangles}
\end{figure}

\begin{figure}[htp]
\centering
  \subfloat[][Example sets $P_0$ and $P_1$, with $p'$ traced with a red line.  Red hexagons represent branching points of paths, gold hexagons represent overlapping points of different branches.]{%
        \label{fig:hex-paths}%
        \makebox[2.2in][c]{ \includegraphics[width=2.1in]{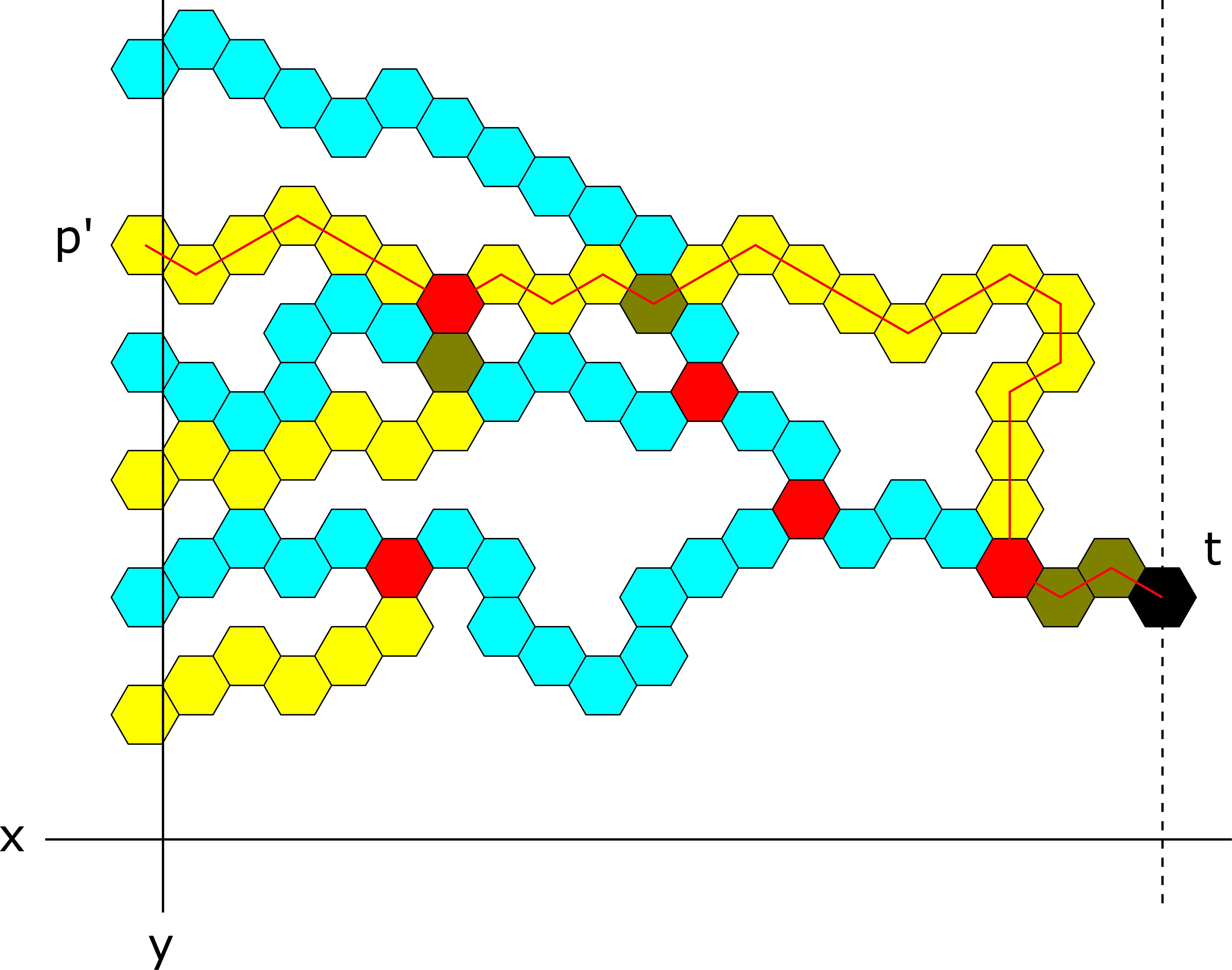}}
        }%
        \quad
  \subfloat[][An example of the growth of $p'$ (traced with a red line) blocked by $\alpha_1$.  By first letting as much of $p'$ grow as possible, it is guaranteed that all other paths must be blocked from reaching $x=0$.]{%
        \label{fig:hex-paths-blocked}%
        \makebox[2.2in][c]{\includegraphics[width=2.1in]{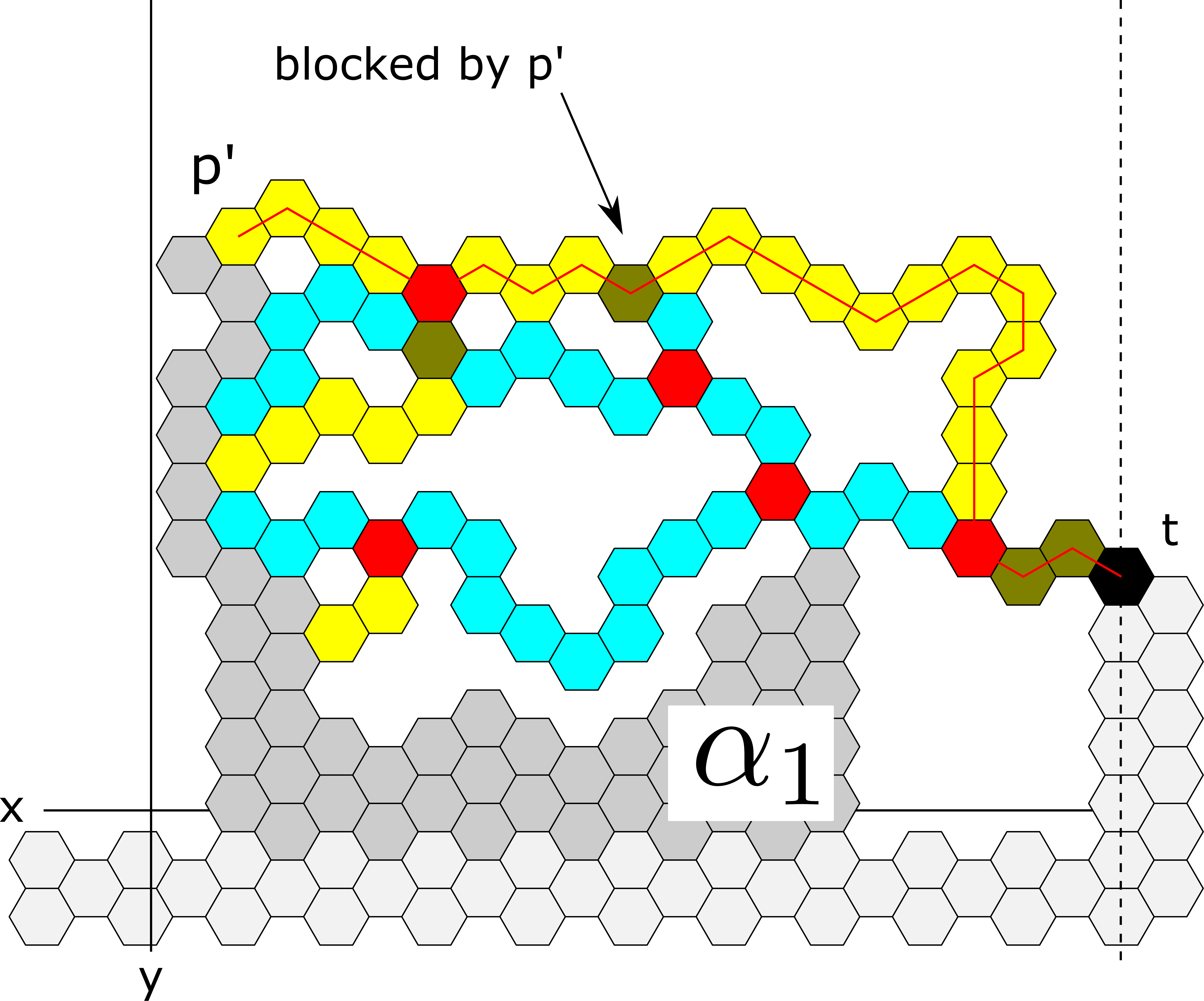}}
        }%
  \caption{Failed bit-readers with regular hexagons.}
  \label{fig:failed-hexagons}
\end{figure}

Case (1)  Paths $p_0$ and $p_1$ occupy the exact same locations through all tile positions and the placement of their first tiles at $x = 0$.  Also, there are no other paths which can grow from $t$, so, since by Definition~\ref{def:bit-reader} some path must be able to complete growth in the presence of $\alpha_0$, either must be able to.  Therefore, we place $\alpha_0$ appropriately and select an assembly sequence in which $p_1$ grows, placing a tile from $T_1$ as its first at $x = 0$.  This is a contradiction, and thus Case (1) cannot be true.

Case (2)  We now consider the scenario where $\alpha_1$ has been placed as the bit-writer according to Definition~\ref{def:bit-reader}, and with $t$ at $(t_x,t_y)$.  Note that path $p'$ must now always, in any valid assembly sequence, be prevented from growing to $x=0$ since it places a tile from $T_0$ at $x=0$, while some path from $T_1$ must always succeed.  We use the geometry of the paths of $T_1$ and path $p'$ to analyze possible assembly sequences.

We create a (valid) assembly sequence which attempts to first grow only $p'$ from $t$ (i.e. it places no tiles from any other branch).  If $p'$ reaches $x=0$, then this is not a valid bit-reader and thus a contradiction.  Therefore, $p'$ must not be able to reach $x=0$, and since the only way to stop it is for some location along $p'$ to be already occupied by a tile, then some tile of $\alpha_1$ must occupy such a location.  This means that we can extend our assembly sequence to include the placement of every tile along $p'$ up to the first tile of $p'$ occupied by $\alpha_1$, and note that by the definition of the regular grid of equilateral triangle tiles, or of regular hexagon tiles, some tile of $p'$ must now have a side adjacent to some tile of $\alpha_1$.  At this point, we can allow any paths from $P_1$ to attempt to grow.  However, by our choice of $p'$ as the ``outermost'' path due to always taking the clockwise-most branches, any path in $P_1$ (and also any other path in $P_0$ for that matter) must be surrounded in the plane by $p'$, $\alpha_1$, and the lines $y=0$ and $x=t_x$ (which they are not allowed to grow beyond), and thus cannot be connected and extend beyond that boundary. (Examples can be seen in Figures~\ref{fig:tri-paths-blocked} and \ref{fig:hex-paths-blocked}.)  Therefore, no path from $P_1$ can grow to a location where $x=0$ without colliding with a previously placed tile or violating the constraints of Definition~\ref{def:bit-reader}.  (This situation is analogous to a prematurely aborted computation which terminates in the middle of computational step.)  This is a contradiction that this is a bit-reader, and thus none must exist.
\qed
\end{proof}

\subsection{Regular pentagons}

Because regular pentagons don't tessellate the plane, the proof that they can't form bit-reading gadgets is slightly different than for equilateral triangles, squares, and regular hexagons.  However, the fact that they can only bind in two relative rotations and the ratio of their side lengths to perimeters ensure that they are still unable to form bit-reading gadgets due to the fact that it is still impossible for one path of regular pentagons to be blocked from continued growth without trapping all other paths on one side.  This means that the ``outermost'' path, along with any part of the bit-writer which blocks its full growth, can always prevent any inner paths from sufficient growth.

\begin{lemma}\label{lem:pent}
There exists no temperature 1 polygonal tile assembly system $\mathcal{T} = (T,\sigma,1)$ where for all $t \in T$, $t$ is a regular pentagon, and a bit-reading gadget exists for $\mathcal{T}$.
\end{lemma}

\begin{proof}
The proof of Lemma~\ref{lem:pent} is nearly identical to that for Lemmas~\ref{lem:tri} and \ref{lem:hex}, with the only slight change being due to the fact that regular pentagons aren't constrained to a single fixed grid.  First, because of this we will slightly adapt Definition~\ref{def:bit-reader} so that rather than requiring tiles to be at specific discrete coordinates, they instead are constrained by lines in $\mathbb{R}^2$.  For instance, we no longer require the bit-reader to grow a path to $x$-coordinate $0$, but instead just beyond a set vertical line $x=r$ for some $r \in \mathbb{R}$.  (without loss of generality we'll assume $x=r=0$ for that constraint.)  This change is merely a technicality and does not affect the proof, and therefore, we will use the previous proof up to the point where Case (2) makes the argument that the regular grid of tiles ensures that the last tile which can be placed along $p'$ must have an edge adjacent to a tile of $\alpha_1$.  Due to the lack of such a grid, we will now only be able to guarantee that some portion of the next position of $p'$, i.e. the location where $\alpha_1$ first prevents the addition of another tile (which we will now refer to as location $p'_b$), is occupied by a tile of $\alpha_1$ (whose location we will now refer to as $p'_{\alpha}$.  Referring to the location of the last tile which can be placed on $p'$ as $p'_{end}$, by the fact that $p'$ would have been a connected path which included $p'_b$, and that the tile at $p_{\alpha}$ prevents its placement, the location $p'_b$ must consist of the area of a tile oriented so that it has an edge adjacent to $p'_{end}$.  Also, although the tiles at $p'_{\alpha}$ and $p_{end}$ need not share an adjacent edge and there may in fact be a gap between them, $p'_{\alpha}$ must overlap with $p'_b$.  (See Figure~\ref{fig:pent-paths-blocked} for an example.)

\begin{figure}[htp]
\centering
  \subfloat[][Example sets $P_0$ and $P_1$, with $p'$ traced with a red line.  Red pentagons represent branching points of paths, gold hexagons represent overlapping points of different branches.]{%
        \label{fig:pent-paths}%
        \makebox[2.2in][c]{ \includegraphics[width=2.1in]{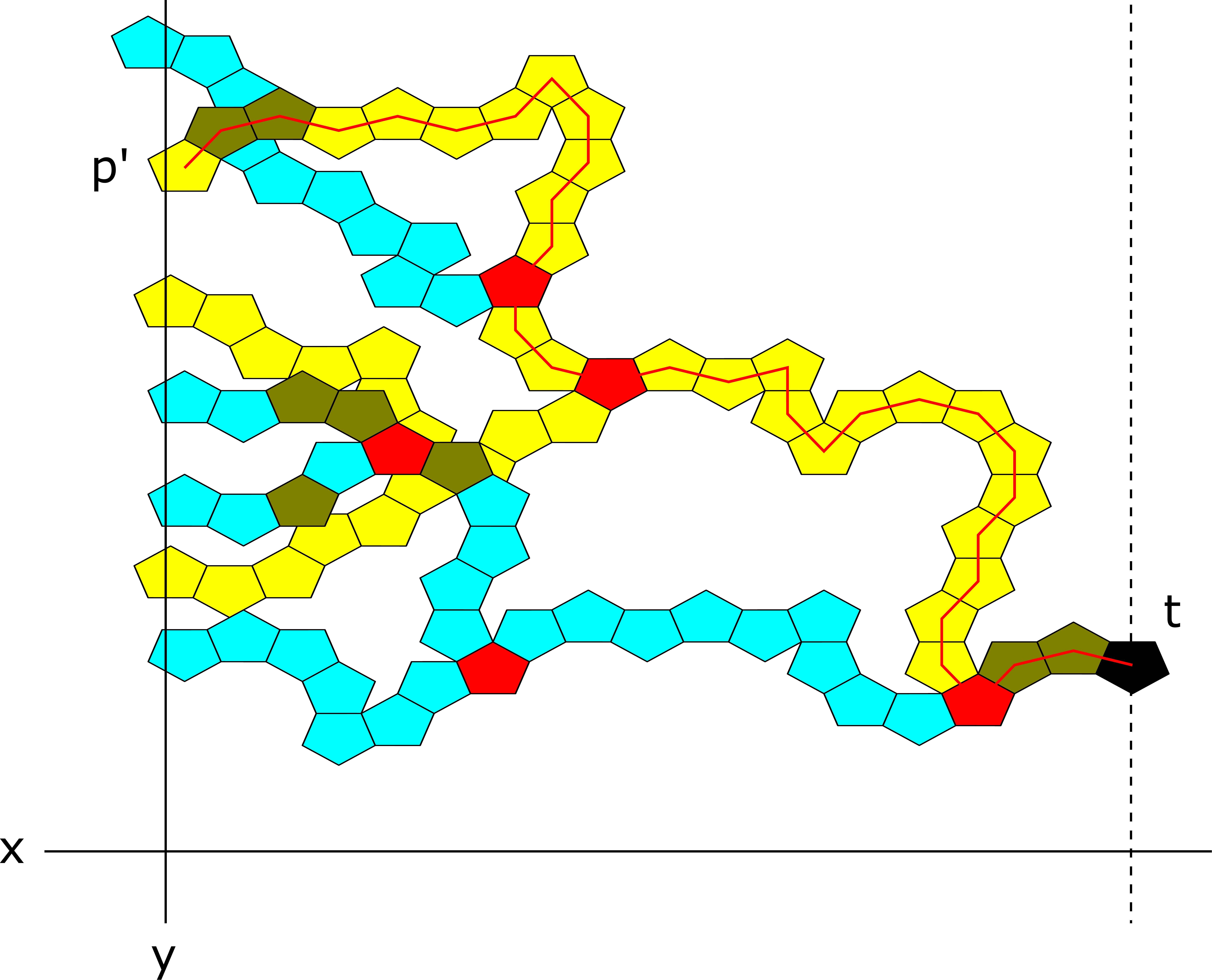}}
        }%
        \quad
  \subfloat[][An example of the growth of $p'$ (traced with a red line) blocked by $\alpha_1$.  By first letting as much of $p'$ grow as possible, it is guaranteed that all other paths must be blocked from reaching $x \leq 0$.  The location outlined in dashed red represents the location of the first tile of $p'$ which is blocked by $\alpha_1$.]{%
        \label{fig:pent-paths-blocked}%
        \makebox[2.2in][c]{\includegraphics[width=2.1in]{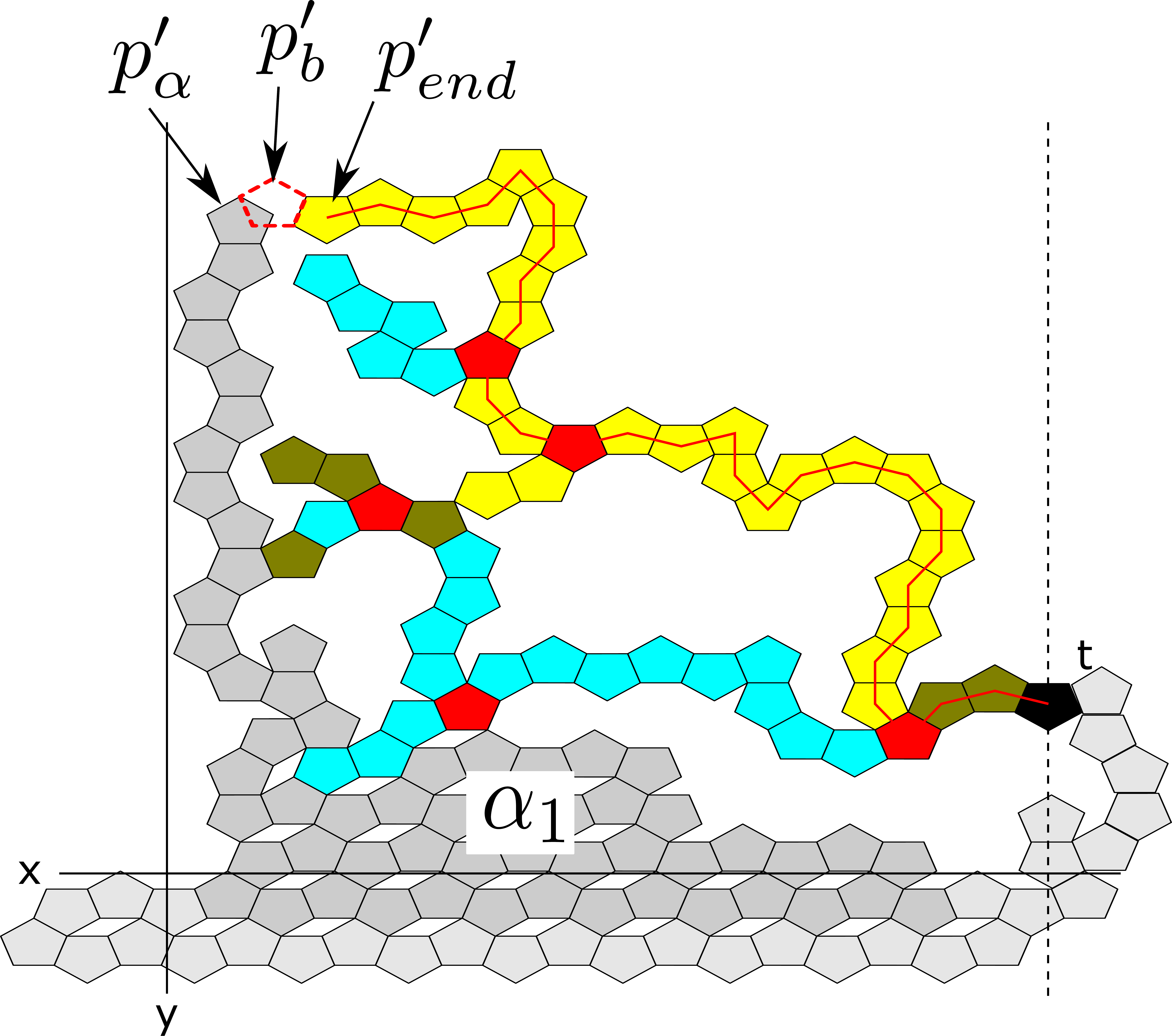}}
        }%
  \caption{Failed bit-readers with regular pentagons.}
  \label{fig:failed-pentagons}
\end{figure}

At this point, we continue the direction of the previous proof and allow any paths from $P_1$ to attempt to grow.  However, by our choice of $p'$ as the ``outermost'' path due to always taking the clockwise-most branches, any path in $P_1$ must be surrounded in the plane by $p'$, $\alpha_1$, and the lines $y=0$ and $x=t_x$ (which they are not allowed to grow beyond), with the only discontinuity being the possible gap consisting of the portion of $p'_b$ which is not occupied by the tile at $p'_{\alpha}$.  We prove that this gap must be insufficient to allow a path $p$ from $P_1$ to grow through using a simple case analysis.  A key feature of regular pentagonal tiles is the fact that although their relative offsets are not fixed on a grid, their relative rotations are constrained to a total of only two orientations while allowing them to be connected to the same assembly.

\begin{figure}[htp]
\centering
  \subfloat[][Tile 2 has the same orientation as $p'_b$ and attaches at an offset slightly below]{%
        \label{fig:pent-bad-same1}%
        \makebox[2.0in][c]{ \includegraphics[width=1.3in]{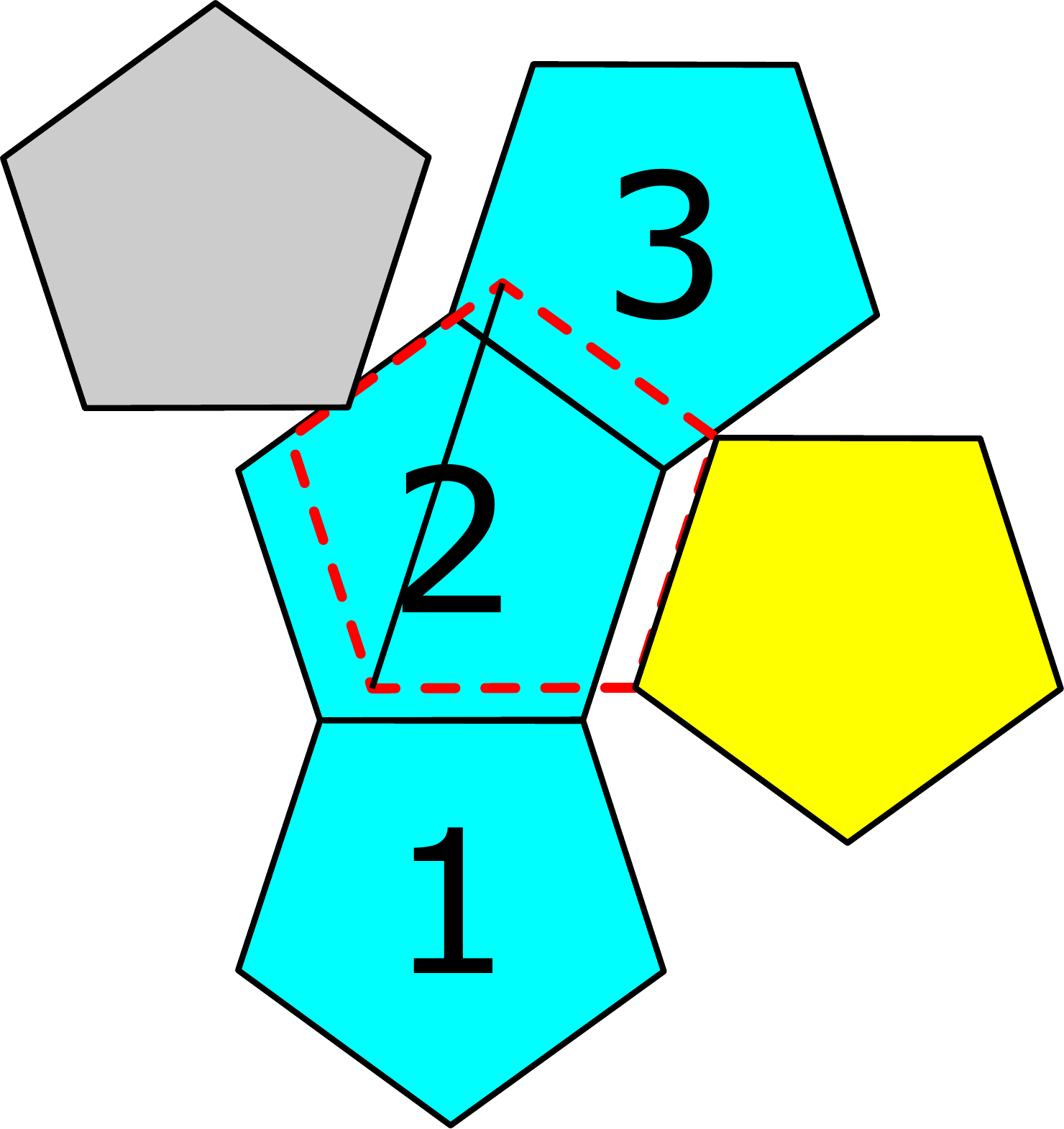}}
        }%
        \quad
  \subfloat[][Tile 2 has the same orientation as $p'_b$ and attaches at an offset slightly above]{%
        \label{fig:pent-bad-same2}%
        \makebox[2.0in][c]{\includegraphics[width=1.3in]{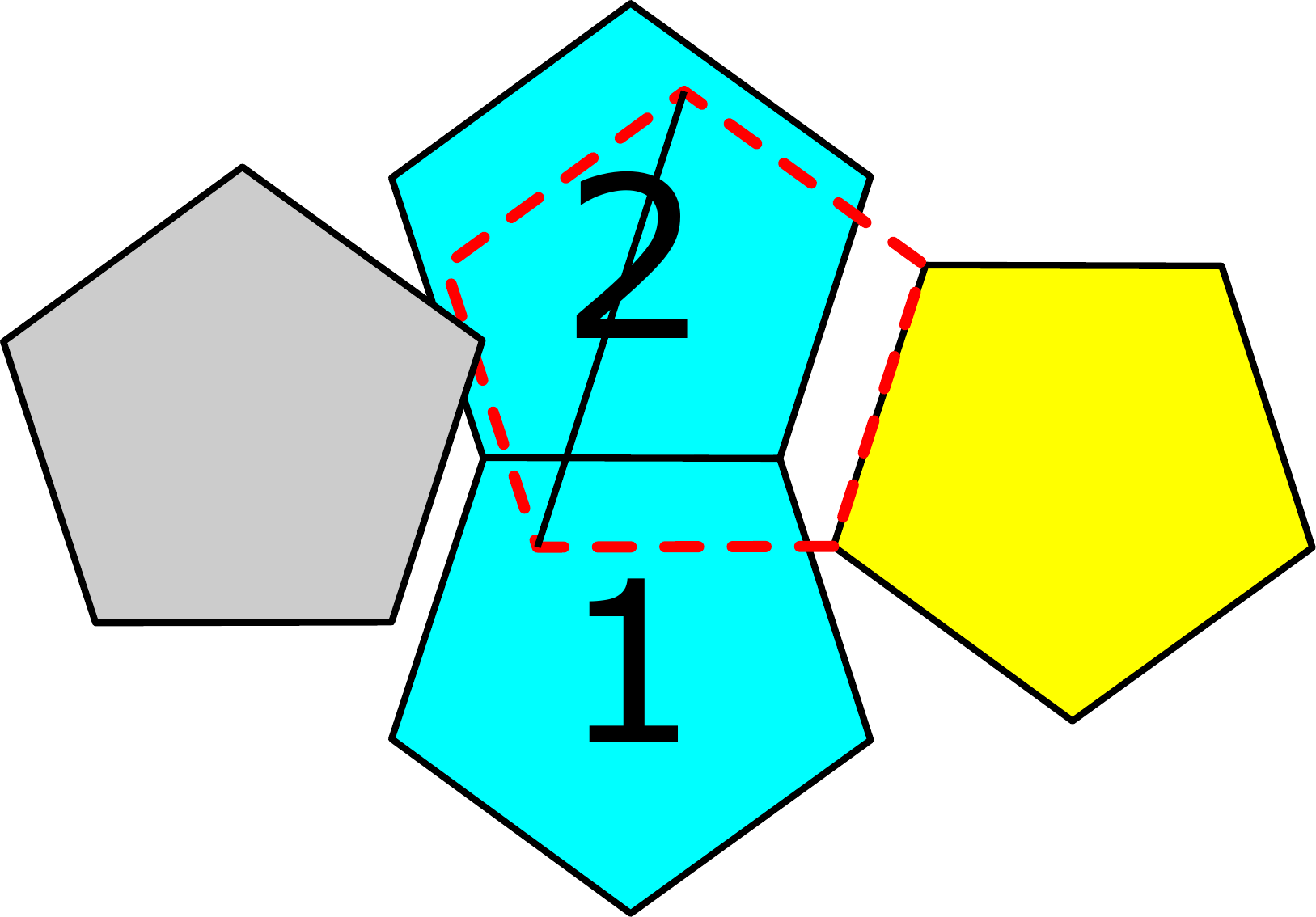}}
        }%
        \quad
  \subfloat[][Tile 2 has the opposite orientation as $p'_b$ and attaches with its nearest vertex below the bottom corner of $p'_{end}$ and $p'_b$]{%
        \label{fig:pent-bad-opp1}%
        \makebox[2.0in][c]{ \includegraphics[width=1.3in]{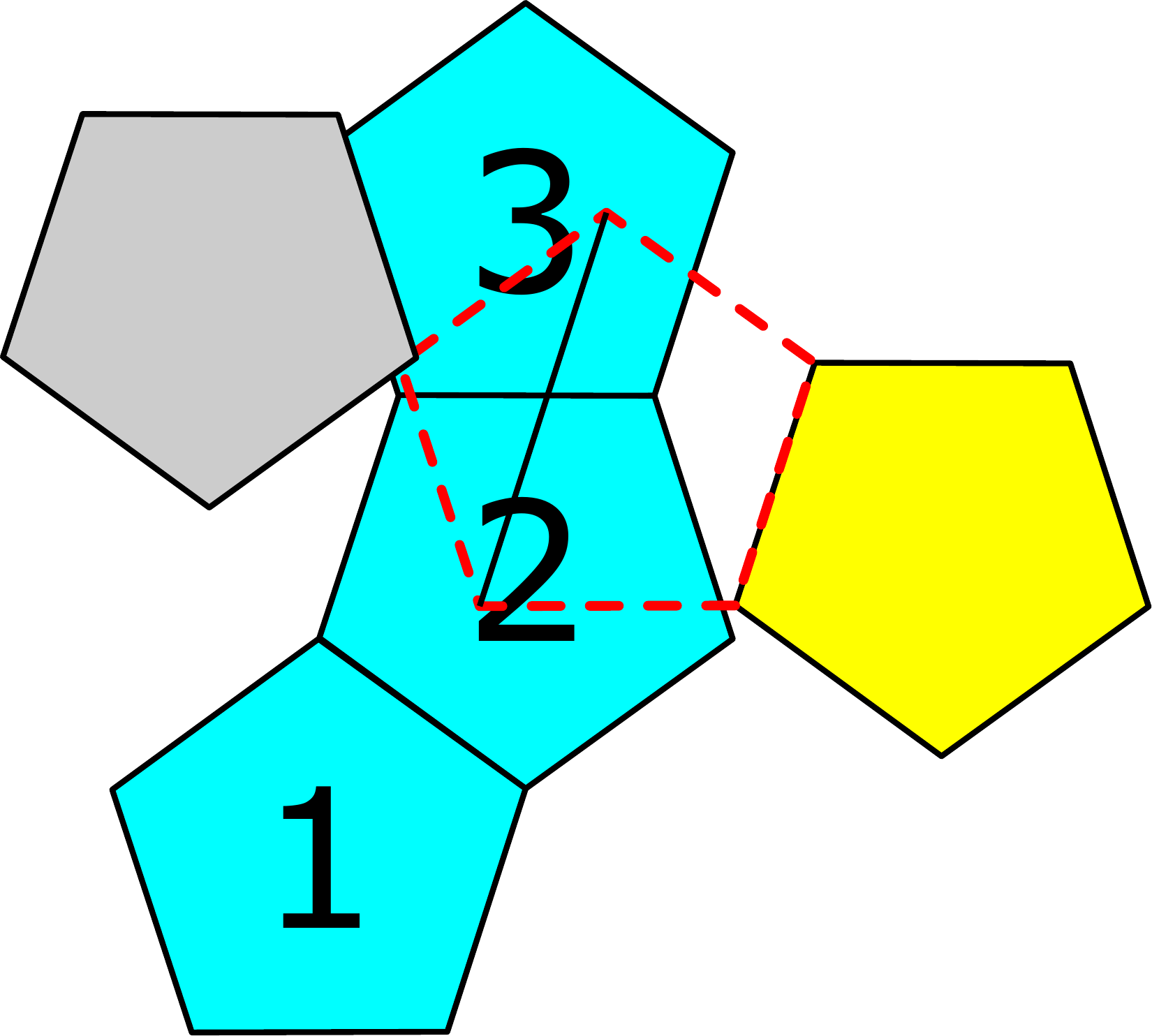}}
        }%
        \quad
  \subfloat[][Tile 2 has the opposite orientation as $p'_b$ and attaches with its nearest vertex above the bottom corner of $p'_{end}$ and $p'_b$]{%
        \label{fig:pent-bad-opp2}%
        \makebox[2.0in][c]{\includegraphics[width=1.3in]{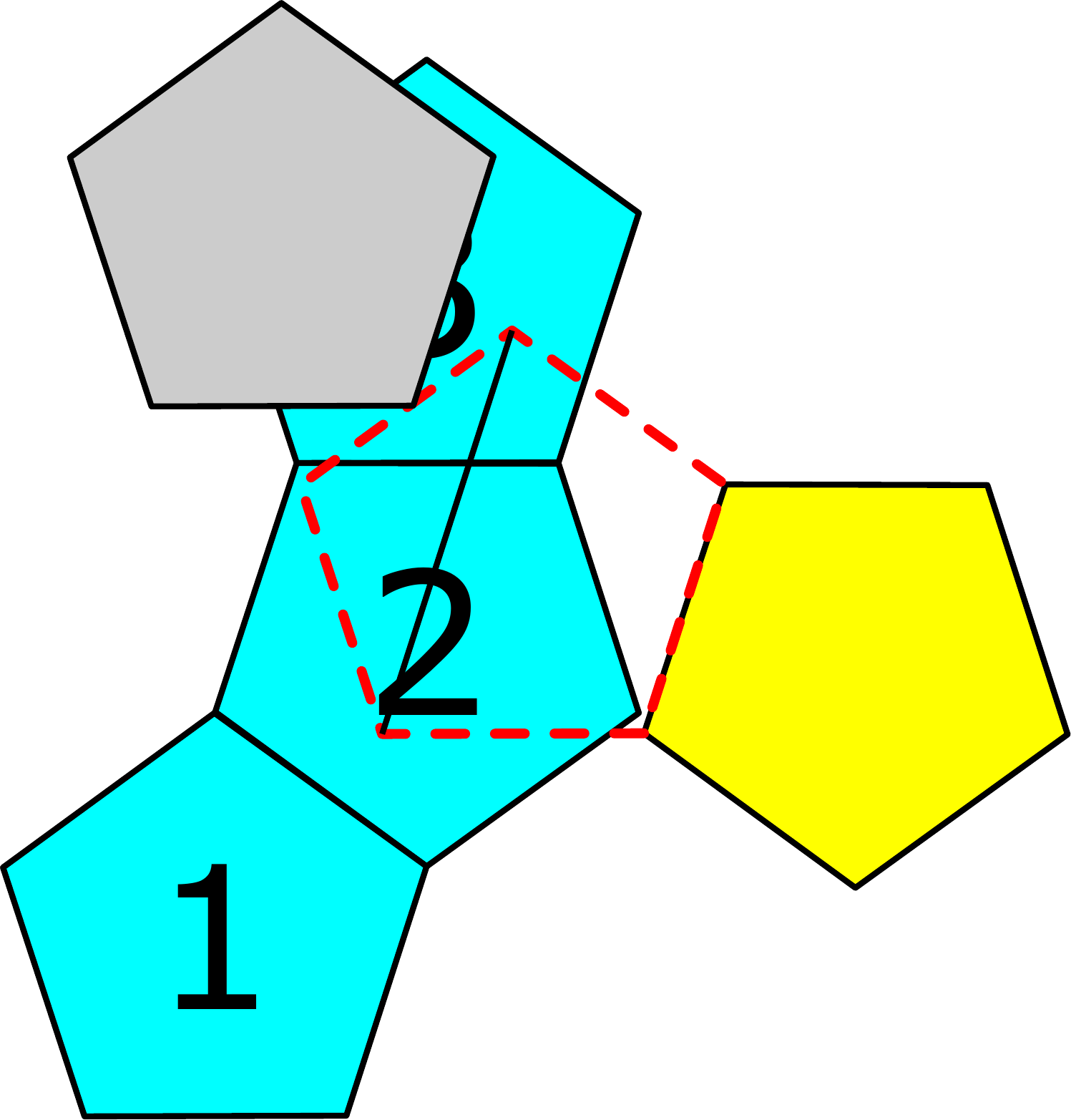}}
        }%
  \caption{A case analysis of why one path, $p'$, of regular pentagonal tiles cannot be blocked while allowing another, $p$, to grow through a gap.  Let the yellow tile be at $p'_{end}$, the location of the last placed tile of path $p'$, the grey represent some blocking tile of $\alpha_1$, at $p'_{\alpha}$ (which may be in either possible orientation), and the red dashed location be $p'_b$, the first tile of $p'$ prevented from being placed.  The blue tiles represent a portion of path $p$ which attempts to grow from below to above $p'$, with the tiles labeled in order of their placements.}
  \label{fig:pent-failed-block}
\end{figure}

To prove the gap is insufficient for $p$, we perform a case analysis as outlined in Figure~\ref{fig:pent-failed-block}. We first note that the blocker may never occupy space inside the black diagonal shown across a portion of the red dashed box, since that would leave a maximum distance of one side length for the gap throughout a portion of the gap, and for any pair of regular pentagonal tiles, the narrowest location is never less than that, occurring at the boundary of two adjacent tiles and immediately increasing on both sides of that.  We now analyze the various cases.

In Figure~\ref{fig:pent-bad-same1}, having the tile at position 2 at the same orientation but an offset below $p'_b$ requires that that tile fill the bottom edges of the location $p'_b$, leaving the blocker only the top left edge through which to block.  However, in order to allow the tile at location 3 to bind to the top right side, the tile at location 2 must be offset up and left in order not to collide with the yellow tile at $p'_{end}$ (since the width of the pair of adjacent tiles increases on both sides of their adjacent edge), forcing it (or a portion of tile 3) to overlap with the blocker. In order for the tile at location 3 to instead bind to the top left side of the tile at location 2, it would have to overlap with the blocker.  This means that $p$ would be blocked and the bit-reader fails, so this case must not be true.

In Figure~\ref{fig:pent-bad-same2}, having the tile at position 2 at the same orientation but an offset above $p'_b$ requires that that tile fill the entire right side of $p_b'$ in order to avoid the yellow tile at $p'_{\alpha}$, thus making it collide with the blocker and the bit-reader again fail.

In Figure~\ref{fig:pent-bad-opp1}, having the tile at position 2 at the opposite orientation as $p'_b$ but with its southeast corner below the southwest corner of the tile at $p'_{end}$ forces the tiles at positions 2 and 3 to cover all of the left side of $p'_b$ and thus collide with the blocker, once again making the bit-reader fail.

In Figure~\ref{fig:pent-bad-opp2}, having the tile at position 2 at the opposite orientation as $p'_b$ with with its southwestern corner above the southwest corner of the tile at $p'_{end}$ again requires that the tiles at positions 2 and 3 to cover the entire left side of $p'_b$, colliding with the blocker and making the bit-reader fail.

The cases discussed, along with all others which are the same up to rotation, prevent the growth of path $p$.  Therefore, no path from $P_1$ can grow to a location past the line $x\leq0$ without colliding with a previously placed tile or violating the constraints of Definition~\ref{def:bit-reader}.  (This situation is analogous to a prematurely aborted computation which terminates in the middle of computational step.)  This is a contradiction that this is a bit-reader, and thus none must exist.

\qed
\end{proof}

The combination of Lemmas~\ref{lem:tri}, \ref{lem:hex}, \ref{lem:pent}, and Theorem 6.1 of \cite{Polyominoes} suffice to prove Theorem~\ref{thm:cant-bit-read-append}.
} %
\fi 

\vspace{-15pt}
\section{Bit-reading Gadgets Overview}\label{sec:bit-readers-main}
\vspace{-10pt}

In the cases where tiles consist of regular polygons with $15$ or more sides, we give a general scheme for obtaining bit-reading gadgets for each case. Figure~\ref{fig:15+sides} depicts the bit-reading gadgets for each case. The others are handled explicitly in the technical appendix.  For the top configurations of Figure~\ref{fig:15+sides}, note that since each polygonal tile of these bit-reading gadgets is adjacent to another tile, we need only show that for each top configuration depicted in Figure~\ref{fig:15+sides}, of the two exposed glues, $g_0$ and $g_1$ of the tile $R$, $B$ prevents a tile from binding to $g_0$. In the bottom configurations of Figure~\ref{fig:15+sides}, we not only need to show that $B$ prevents a tile from binding to $g_1$, but we must also show that $B$ does not prevent a tile (the tile centered at $c_2$ in the bottom configurations for Figure~\ref{fig:15+sides}) from binding to the tile that binds to $g_0$. The latter statement ensures that when we use the bit-reading gadgets obtained from these configurations to simulate a Turing machine, in the case that a $0$ is read by attaching a tile to $g_0$, $B$ does not prevent further growth of an assembly.

\begin{figure}[htp]
\centering
\vspace{-20pt}
  \subfloat[][Pentadecagonal tiles]{%
        \label{fig:15+sidesA}%
    		\makebox[.30\textwidth]{
        \includegraphics[width=1.0in]{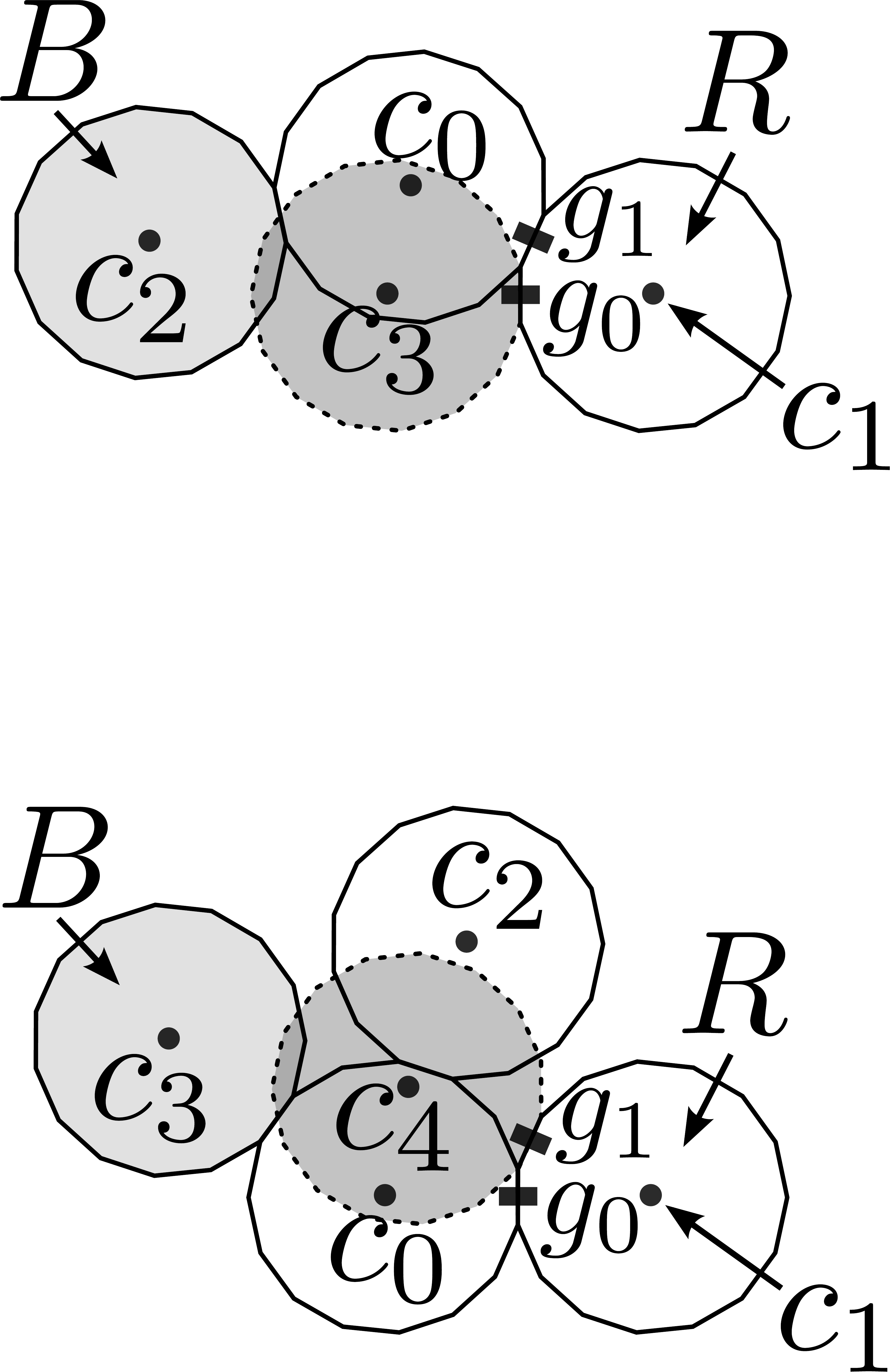}
        }
        }%
        \quad
  \subfloat[][Hexadecagonal tiles]{%
        \label{fig:15+sidesB}%
    		\makebox[.30\textwidth]{
        \includegraphics[width=1.0in]{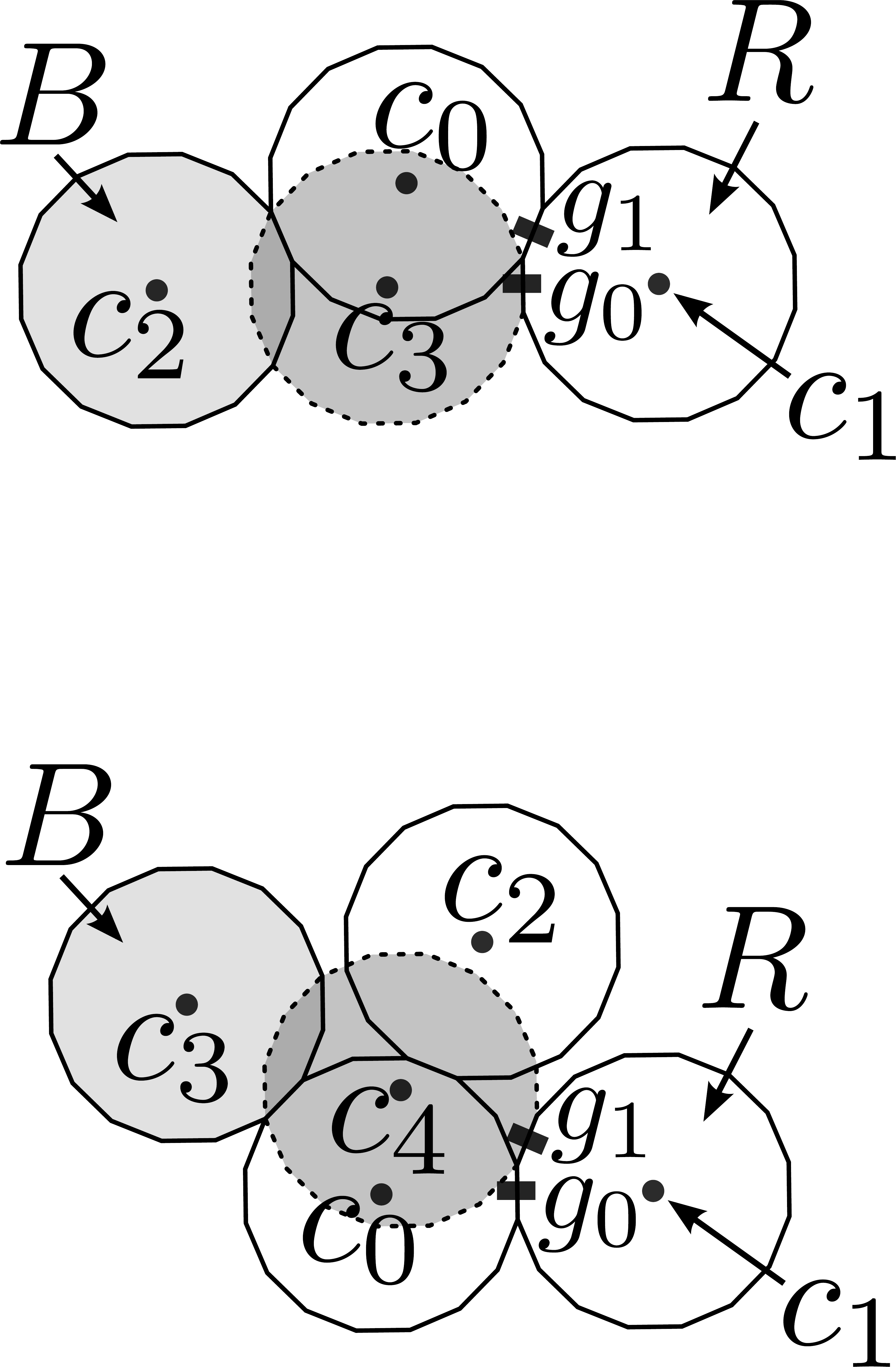}
        }
        }%
       \quad
  \subfloat[][Heptadecagonal tiles]{%
        \label{fig:15+sidesC}%
    		\makebox[.30\textwidth]{
        \includegraphics[width=1.0in]{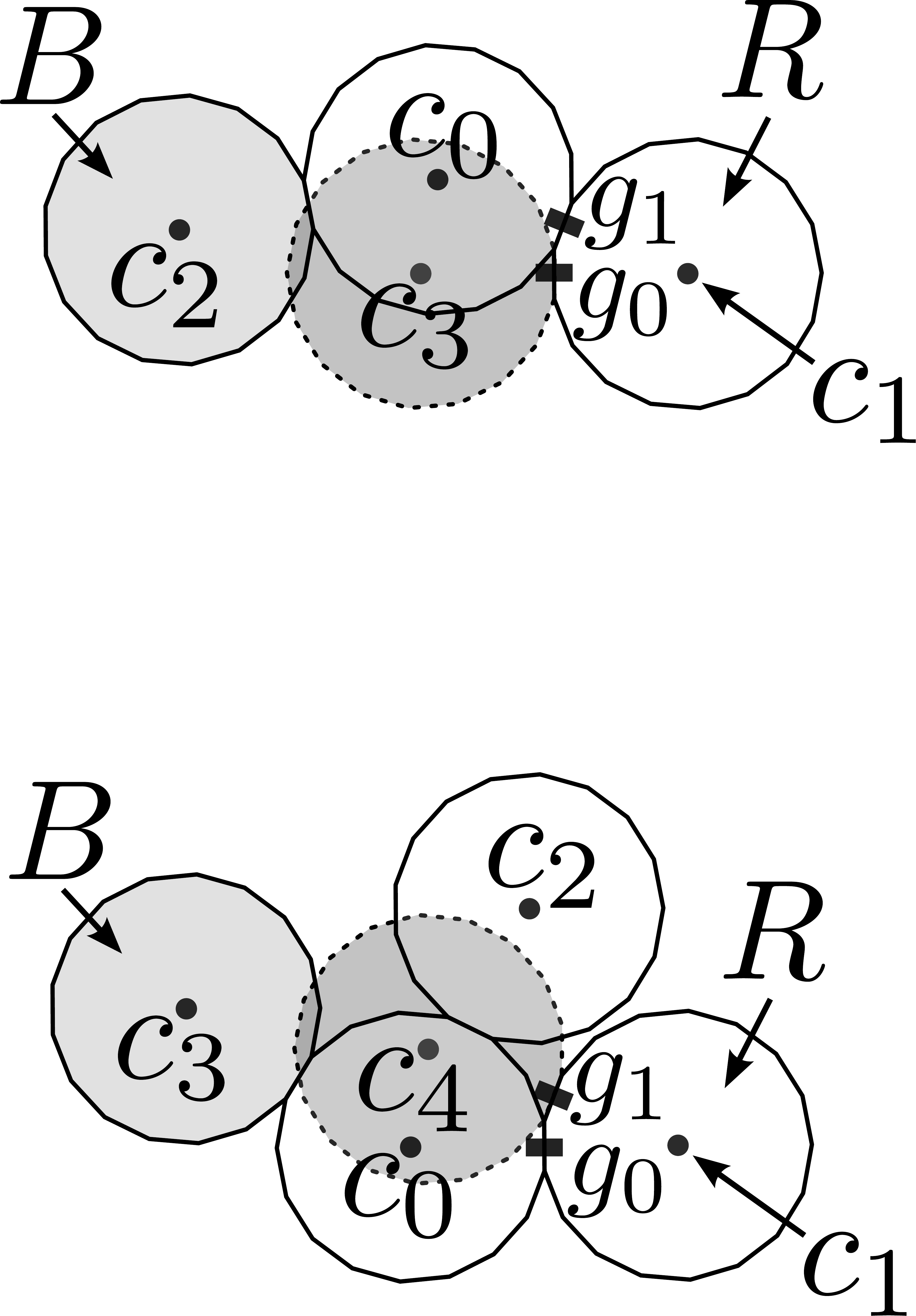}
        }
        }%
  \vspace{-10pt}
  \caption{Bit-reading gadget portions. (top) Reading a $1$ and preventing placement of a tile at $c_3$, (bottom) Reading a $0$ and preventing placement of a tile at $c_4$.
  \vspace{-20pt}}
  \label{fig:15+sides}
\end{figure}

Now, consider a polygon $P_n$ with $n\geq 15$ sides and let $\omega$ be the $n^{th}$ root of unity $e^{\frac{2\pi i}{n}}$. Then, the general scheme for constructing a bit-reading gadget falls into two cases. First, if $n$ is odd (the cases where $n$ is even are similar), relative to a tile with negated orientation (the polygon labeled $R$ in the configurations in Figure~\ref{fig:15+sides}), the two configurations that give rise to the bit-reading gadget are as follows. Let $k$ be such that $n=2k+1$ ($n=2k$ if $n$ is even). Referring to the top configurations of Figure~\ref{fig:15+sides}, to ``write'' a $1$, the configuration is obtained by centering a blocker tile with negated orientation, labeled $B$, at $-\omega^{n-1} + \omega^{k+1}$ (whether $n$ is even or odd) relative to $R$. Then to ``read'' a $1$, $R$ exposes two glues $g_1$ and $g_0$ such that if a tile binds to $g_1$, it will have standard orientation and be centered at $-\omega^{n-1}$ (whether $n$ is even or odd) and if a tile binds to $g_0$, it will have standard orientation and be centered at $-1$. We will show that $B$ will prevent this tile from binding. This gives the configuration depicted in the top figures of Figure~\ref{fig:15+sides}. Now, referring to the bottom configurations of Figure~\ref{fig:15+sides}, to ``write'' a $0$, the configuration is obtained by centering a blocker tile with negated orientation, labeled $B$, at $-1 + \omega^{k-1}$ ($-1 + \omega^{k-2}$ if $n$ is even) relative to $R$.  In this case, we will show that $B$ prevents a tile from binding to $g_1$. In addition, we place a glue on the tile that binds to $g_0$ that allows for another tile to bind to it so that its center is at $c_2 = -1 + \omega^{\lfloor \frac{k-1}{2} \rfloor}$ ($c_2 = -1 + \omega^{\frac{k-2}{2}}$ if $n$ is even) relative to $R$.  This gives the configuration depicted in the bottom figures of Figure~\ref{fig:15+sidesA} and Figure~\ref{fig:15+sidesC}. Moreover, we show that neither $R$ nor $B$ prevent the binding of this tile.

In order to perform the calculations used to show the correctness of these bit-reading gadgets, we consider the cases where $n$ is even and where $n$ is odd. Here we give brief version of the calculations that show that a regular polygon centered at $c_1$ and regular polygon centered at $c_2$ do not overlap when $n\geq 15$ is odd. For more detail and calculations for the case where $n$ is even, see Section~\ref{sec:technical-15+sides}.

Suppose that $n = 2k+1$. We now refer to the bottom configurations of Figure~\ref{fig:15+sidesA}. To show that a polygon centered at $c_1$ and a polygon centered at $c_2$ do not overlap, consider the case where $k$ is odd (the case where $k$ is even is similar). Note that relative to $c_0$, $c_1 = 1$ and $c_2 = \omega^{\frac{k-1}{2}}$. Then the distance $d_n$ from $c_1$ to $c_2$ satisfies the following equation.
{\fontsize{7pt}{1em}\selectfont
$$d_n^2 = \left(1-\cos\left(\frac{\left(k-1\right)\pi}{n}\right)\right)^2 + \sin^2\left(\frac{\left(k-1\right)\pi}{n}\right)$$
}%
Substituting $k = \frac{n-1}{2}$ for $k$ and simplifying, we obtain $d_n^2 = 2+2\sin\left(\frac{3\pi}{2n}\right)$. It is well known that for regular polygons with $n$ sides and apothem $\frac{1}{2}$, the circumradius is given by $\frac{1}{\cos\left( \frac{\pi}{n} \right)}$. Hence, to show that a polygon centered at $c_1$ and a polygon centered at $c_2$ do not overlap, we show that $d_{n}^2 > \frac{1}{\cos^2\left(\frac{\pi}{n}\right)}$ for $n \geq 15$. (See Section~\ref{sec:technical-15+sides}. It then follows that $d_{n} > \frac{1}{\cos\left(\frac{\pi}{n}\right)}$.  Therefore, $d_n$ is greater than twice the circumradius of our polygons. Hence, a polygon centered at $c_1$ and a polygon centered at $c_2$ do not overlap. We then perform similar calculations to show that for $n\geq15$, the configurations described in this above indeed give bit-reading gadgets.

Thus, we have shown that bit-reading gadgets can be formed, along with grids that allow bits to be written and read, using polygonal tiles with $\ge 15$ sides.  Combined with standard tile assembly techniques to simulate Turing machines, this proves that such systems are computationally universal.
\ifabstract
\later{
\section{Bit-reading Gadgets}\label{sec:bit_reading}

In this section, we give configurations that are then used to construct bit-reading gadgets for 1) single shape systems with regular polygonal tiles with $7$ or more sides (See Section~\ref{sec:bit-readers-regular}.), 2) 2-shaped systems with regular polygonal tiles for pairs of distinct polygons with $3$ to $6$ sides (See Section~\ref{sec:bit-readers-2shaped-regular}.), and 3) single shaped systems with equilateral polygonal tiles with $4$, $5$, or $6$ sides (See Section~\ref{sec:bit-readers-equilaterals}.). Finally, in Section~\ref{sec:bit-readers-triangle}, we give a bit-reading gadget for a single shaped system with tiles having the shape of an obtuse isosceles triangle.  All of the configurations presented here will be used to obtain bit-reading gadgets that read bits from right to left. It should be noted that for all of the polygons considered here, configurations that yield left to right  bit-reading gadgets can be obtained by simply reflecting the corresponding right to left configurations. 

\subsection{Single shape systems with regular polygonal tiles}\label{sec:bit-readers-regular}

In the following subsections, we give configurations that will be used to construct bit-reading gadgets for single shape systems with regular polygonal tiles with $7$ or more sides. While the configurations presented here do not technically fit Definition~\ref{def:bit-reader}, in Section~\ref{sec:technicalLemmas} we describe how to turn these configurations into bit-reading gadgets that do conform to that definition. In this section, we are concerned with showing how to use the geometry of polygonal tiles to ensure that our bit-reading gadgets properly read and write bits as described in Definition~\ref{def:bit-reader}. Therefore, combining the results of this section with Section~\ref{sec:technicalLemmas}, we show the following lemma.

\begin{lemma}\label{lem:regular-bit-gadget}
Let $P_n$ be a regular polygon with $n$ sides. Then, for all $n \geq 7$, there exists a single-shaped system $\mathcal{T}_n = (T_n, \sigma_n)$ with shape $P_n$ such that a bit-reading gadget exists for $\mathcal{T}_n$.
\end{lemma}

In order to prove Lemma~\ref{lem:regular-bit-gadget}, we first consider the cases where $n$ is $7,8,9,13,$ or $14$, since these cases are handled by giving a specific bit-reading gadget for each case. Second, we give bit-reading gadgets for the cases where $n$ is $10$, $11$, or $12$. These cases are simpler than the former cases and are handled using a more generic approach. Finally, we give the bit-reading gadgets for the cases where $n\geq 15$. These cases are handled by using a single generic scheme for constructing the bit-reading gadgets for each case.

\subsubsection{Tiles with $7,8,9,13,$ or $14$ sides}\label{sec:special-bit-readers}

In this section we give a description of the bit-reading gadget for heptagonal tiles and give a brief example of a calculation that shows that certain tiles do not overlap. Figure~\ref{fig:heptagonBitReader} gives a depiction of a bit-reader for heptagonal tiles.

\begin{figure}[htp]
\centering
  \subfloat[][A $0$ is read, and a $1$ cannot be read by mistake since the tile $B$ prevents a heptagonal tile from attaching via the glue labeled $g_1$.]{%
        \label{fig:heptagonBitReaderA}%
        \includegraphics[width=2.5in]{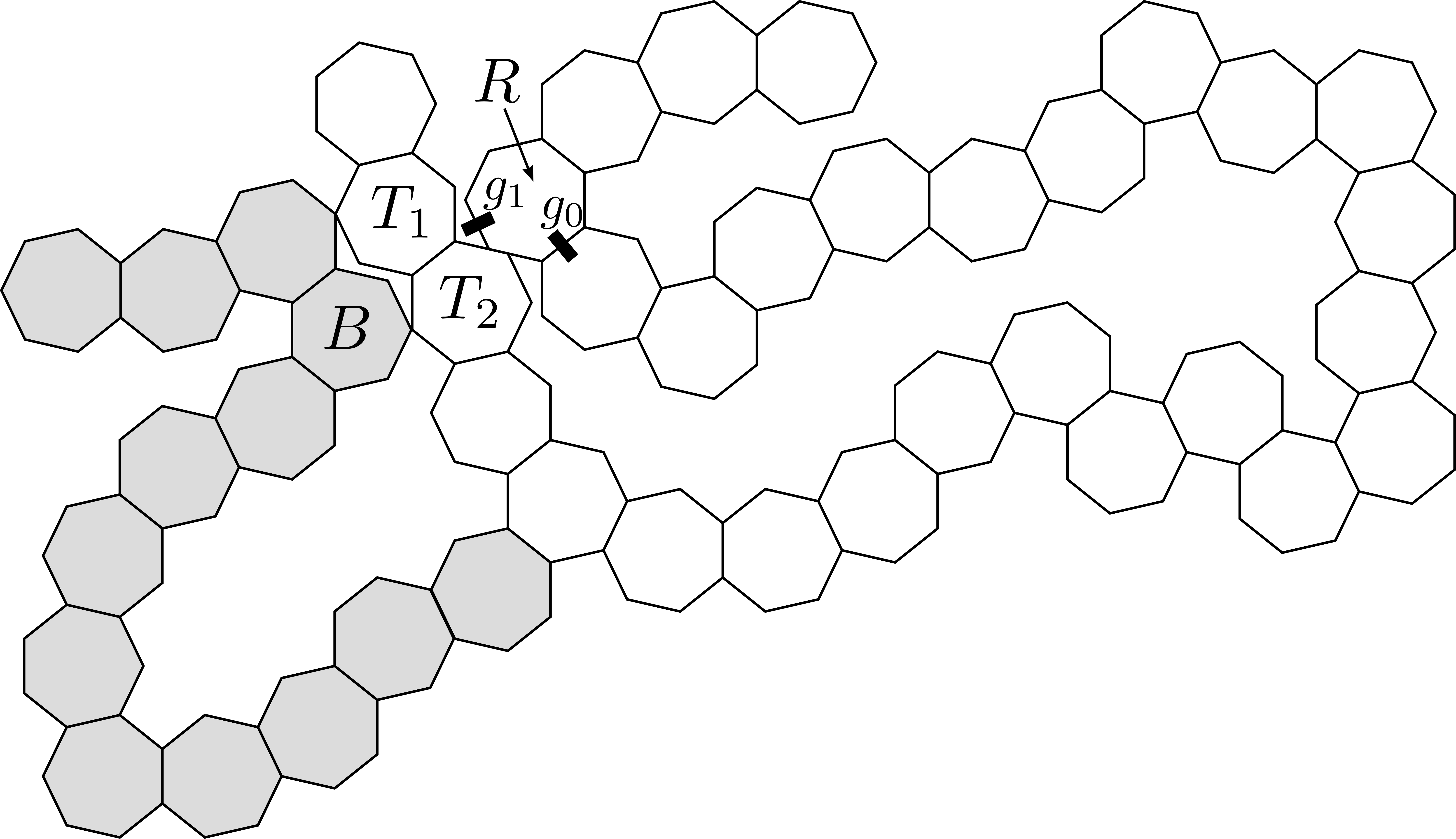}
        }%
        \quad
  \subfloat[][A $1$ is read. This time a $0$ cannot be read by mistake since the tile $B$ prevents growth of a path of heptagonal tiles that attach via the glue labeled $g_0$. Note that some of this path may form, but $B$ prevents the entire path from assembling, and thus prevents a $0$ from being read.]{%
        \label{fig:heptagonBitReaderB}%
        \includegraphics[width=2.5in]{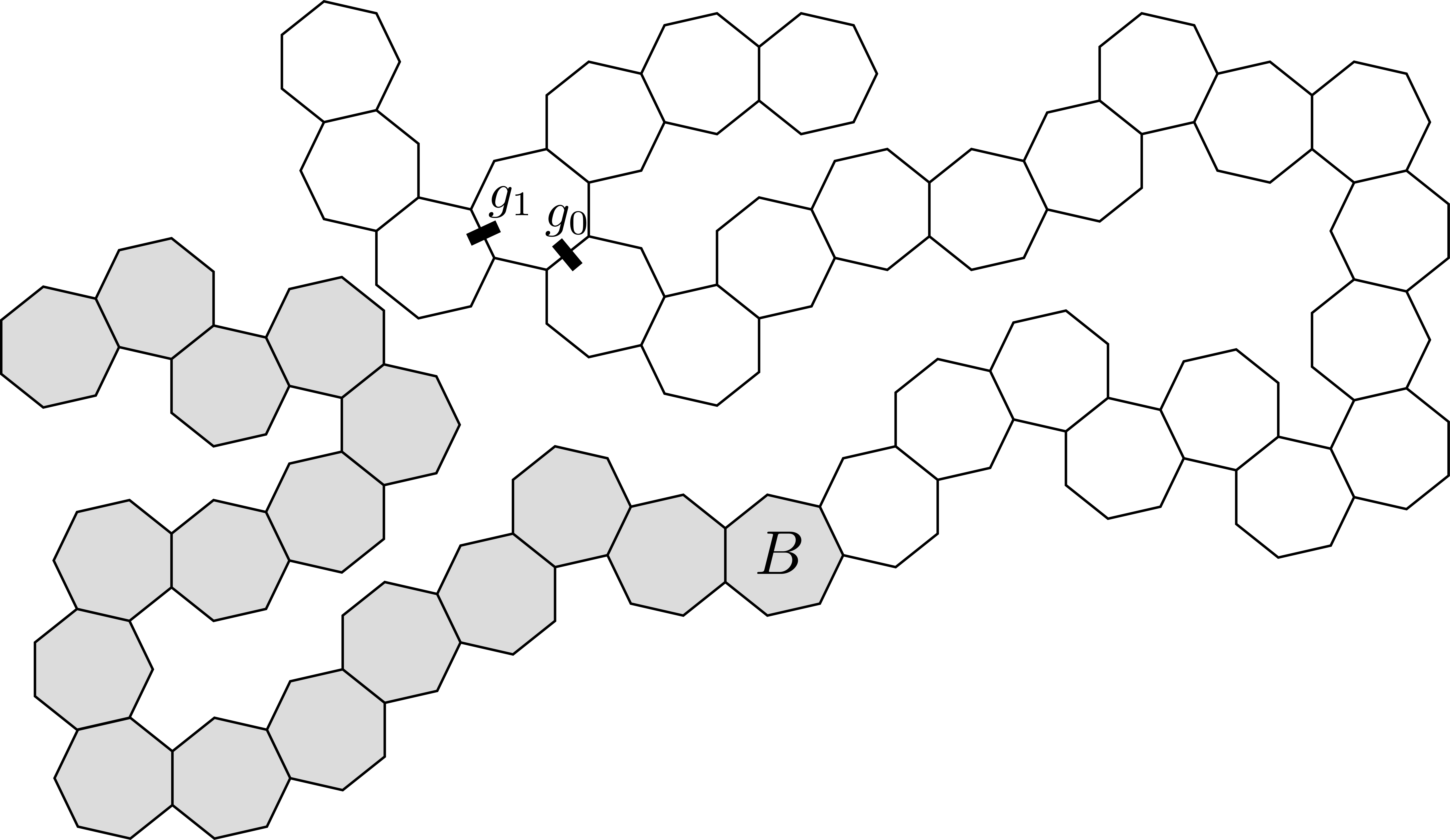}
        }%
  \caption{A connected bit-gadget consisting of heptagonal tiles.}
  \label{fig:heptagonBitReader}
\end{figure}

In Figure~\ref{fig:heptagonBitReader}, the gray tiles represent the bit writer tiles (representing either 0 or 1), while the white tiles are the bit reader tiles. In our construction, we ensure that we have an assembly sequence such that the gray tiles of a bit-reading gadget bind before any white tiles. Figure~\ref{fig:heptagonBitReaderB} depicts the case in which a $1$ has previously been written and is then read. In this case, we observe that the bit writer tiles prevent the formation of the path of tiles depicted in~\ref{fig:heptagonBitReaderA} from $R$ to $T_1$, ensuring that a tile is a tile binds to the glue $g_1$, resulting in a $1$ being read. Moreover, since the configuration of Figure~\ref{fig:heptagonBitReaderB} consists of abutting heptagonal tiles with non-overlapping interiors, we see that with appropriately defined glues, the bit writer configuration and the bit reader configuration are valid assemblies. We can also see that no two tiles of the bit writer configuration and the bit reader configuration have overlapping interiors; this ensures that these two assemblies can be part of the same larger assembly.

Similarly, Figure~\ref{fig:heptagonBitReaderA} depicts the case in which a $0$ has previously been written and is being read. Though much of this configuration consists of abutting heptagonal tiles with non-overlapping interiors, it is not clear that all of the heptagonal tiles have non-overlapping intersection. For example, it is indeed the case that $R$ and $T_2$ have non-overlapping intersection (It turns out that they do share a portion of an edge.) but it is not clear that the interiors of these tiles do not overlap on some tiny set of points. Moreover, it is not clear that a tile could not attach to the glue $g_1$. Therefore, we must calculate the distance between these tiles to show that, with appropriately defined glues, the bit reader configuration is a valid assembly, and that no two tiles of the bit writer configuration and the bit reader configuration have overlapping interiors.

Referring to Figure~\ref{fig:heptagonBitReaderA}, we will first show that the tile labeled $R$ does not prevent the binding of the tile labeled $T_1$ or the tile labeled $T_2$.
Let $c$ denote the center of the tile $R$, $c_1$ denote the center of $T_1$, and $c_2$ denote the center of $T_2$. Then, to calculate $c_1$ and $c_2$ relative to $c$, we assume that $R$ has standard orientation and is centered at the origin. Following the path of tiles from $R$ to $T_1$ and summing the appropriate roots of unity, we obtain the polynomials $c_1 = \omega^6 - \omega^3 + \omega - \omega^4 + 1 -\omega^4 + \omega - \omega^3 + 1 - \omega^2 + \omega^5 - \omega^2 + \omega^4 - \omega^6 + \omega^4 - \omega^6 + \omega^4 - \omega + \omega^4 - 1 + \omega^3 - \omega^6 + \omega^2$. Note that $c_2 = c_1 - \omega^6$. By simplifying $c_1$, we get $c_1 = 1+ \omega - \omega^2 - \omega^3  + 2\omega^4 + \omega^5 - 2\omega^6$.
Then, as multiplying by $\omega$ corresponds to rotating by $2\pi/7$, it is enough to show that $\Re(\omega^2 c_1) \geq 1$, and to see this, consider the following.

\begin{eqnarray*}
\omega^2c_1 &=& \omega^2 + \omega^3 - \omega^4 - \omega^5  + 2\omega^6 + \omega^7 - 2\omega^8\\
            &=& \omega^2 + \omega^3 - \omega^4 - \omega^5  + 2\omega^6 + 1 - 2\omega\\
            &=& \omega^2 + \omega^3 - \omega^{-3} - \omega^{-2}  + 2\omega^6 + 1 - 2\omega^{-6}\\
            &=& 1 + (\omega^2 - \omega^{-2}) + (\omega^3 - \omega^{-3})  + 2(\omega^6 - \omega^{-6})
\end{eqnarray*}

\noindent Finally, since $(\omega^2 - \omega^{-2})$, $(\omega^3 - \omega^{-3})$, and $2(\omega^6 - \omega^{-6})$ are purely imaginary, we see $\Re(\omega^2c_1) = 1$. It follows that the intersection of the interiors of $R$ and $T_1$ is empty. The remainder of the distance calculations are given in Section~\ref{sec:technical-heptagonal}. For tiles consisting of regular polygons with $8,9,13,$ or $14$ sides we give the bit-reading gadgets and calculations in Section~\ref{sec:technical-single-regular}.

\subsubsection{Tiles with $10,11,$ or $12$ sides}\label{sec:10-12sides-bit-readers}

In the cases where tiles consist of regular polygons with $10,11,$ or $12$ sides, bit-reading gadgets are relatively simple to construct. Figure~\ref{fig:10to12sidesBitReaders} depicts the configurations that we will use to construct our bit-reading gadgets for each case. Note that since each polygonal tile of these configurations is adjacent to another tile, we need only show that for each configuration depicted in Figure~\ref{fig:10to12sidesBitReaders}, of the two exposed glues, $g_0$ and $g_1$ of the tile $R$, a tile can only attach to one of these glues depending on the position of the tile $B$ in the figure. In other words, for each configuration depicted in Figure~\ref{fig:10to12sidesBitReaders}, we show that in the top configuration, $B$ prevents a tile from binding to $g_1$, and that in the bottom configuration, $B$ prevents a tile from binding to $g_0$.

\begin{figure}[htp]
\centering
  \subfloat[][Bit-reading gadget configuration for decagonal tiles.]{%
        \label{fig:10to12sidesBitReadersA}%
    		\makebox[.25\textwidth]{
        \includegraphics[width=.75in]{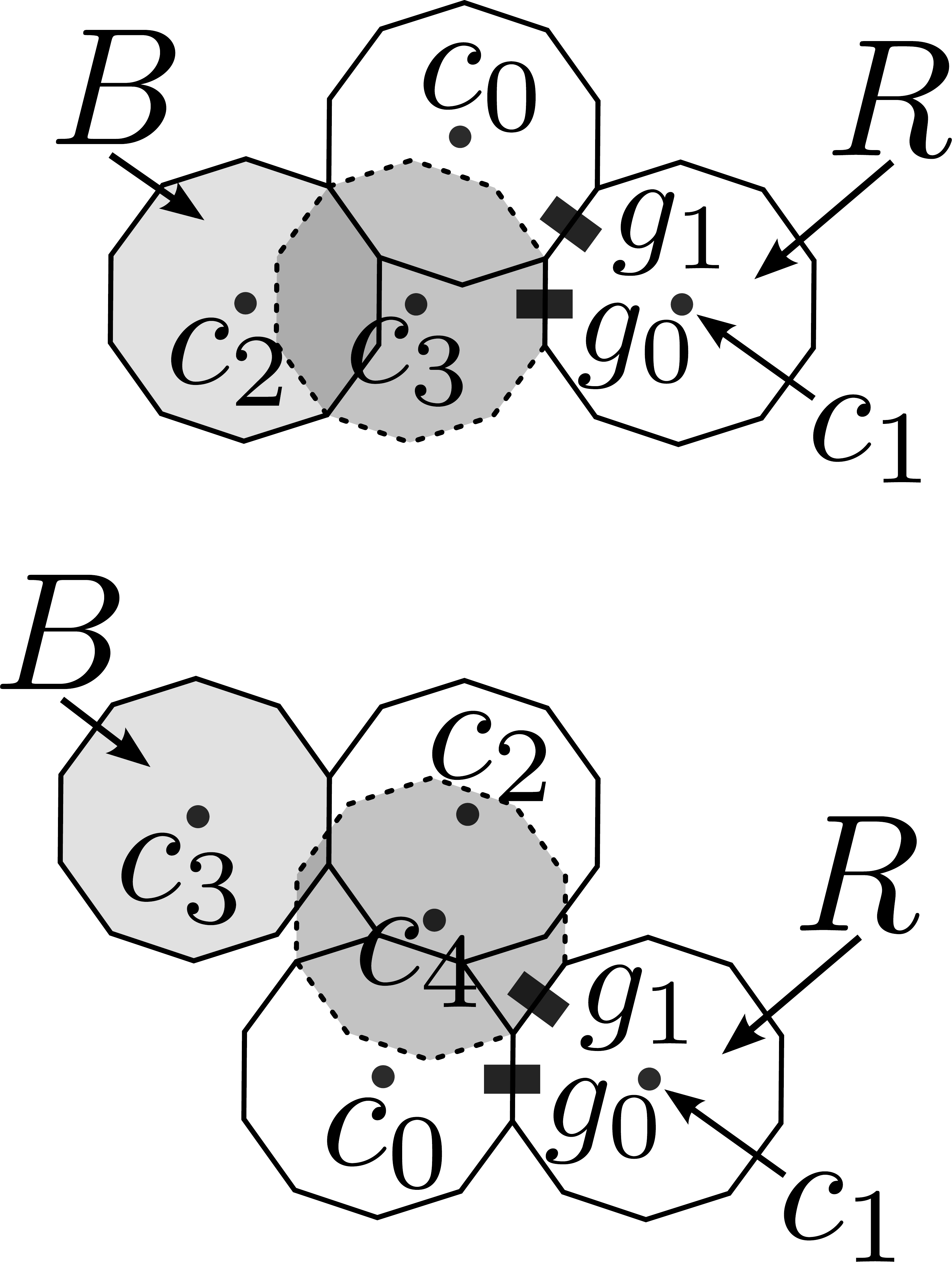}
        }
        }%
        \quad\quad
  \subfloat[][Bit-reading gadget configuration for hendecagonal tiles.]{%
        \label{fig:10to12sidesBitReadersB}%
    		\makebox[.25\textwidth]{
        \includegraphics[width=.75in]{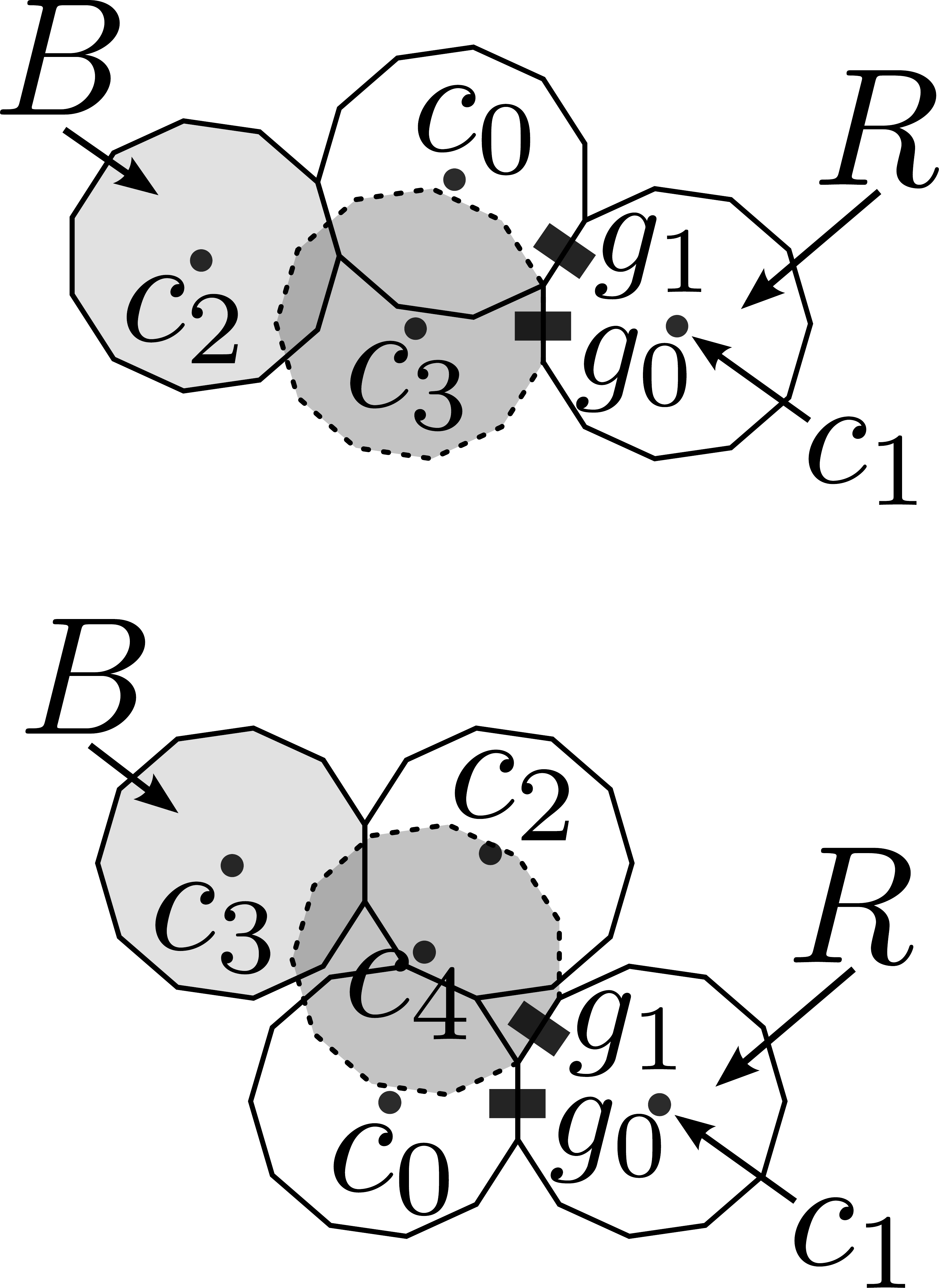}
        }
        }%
       \quad\quad
  \subfloat[][Bit-reading gadget configuration for dodecagonal tiles.]{%
        \label{fig:10to12sidesBitReadersC}%
    		\makebox[.25\textwidth]{
        \includegraphics[width=.75in]{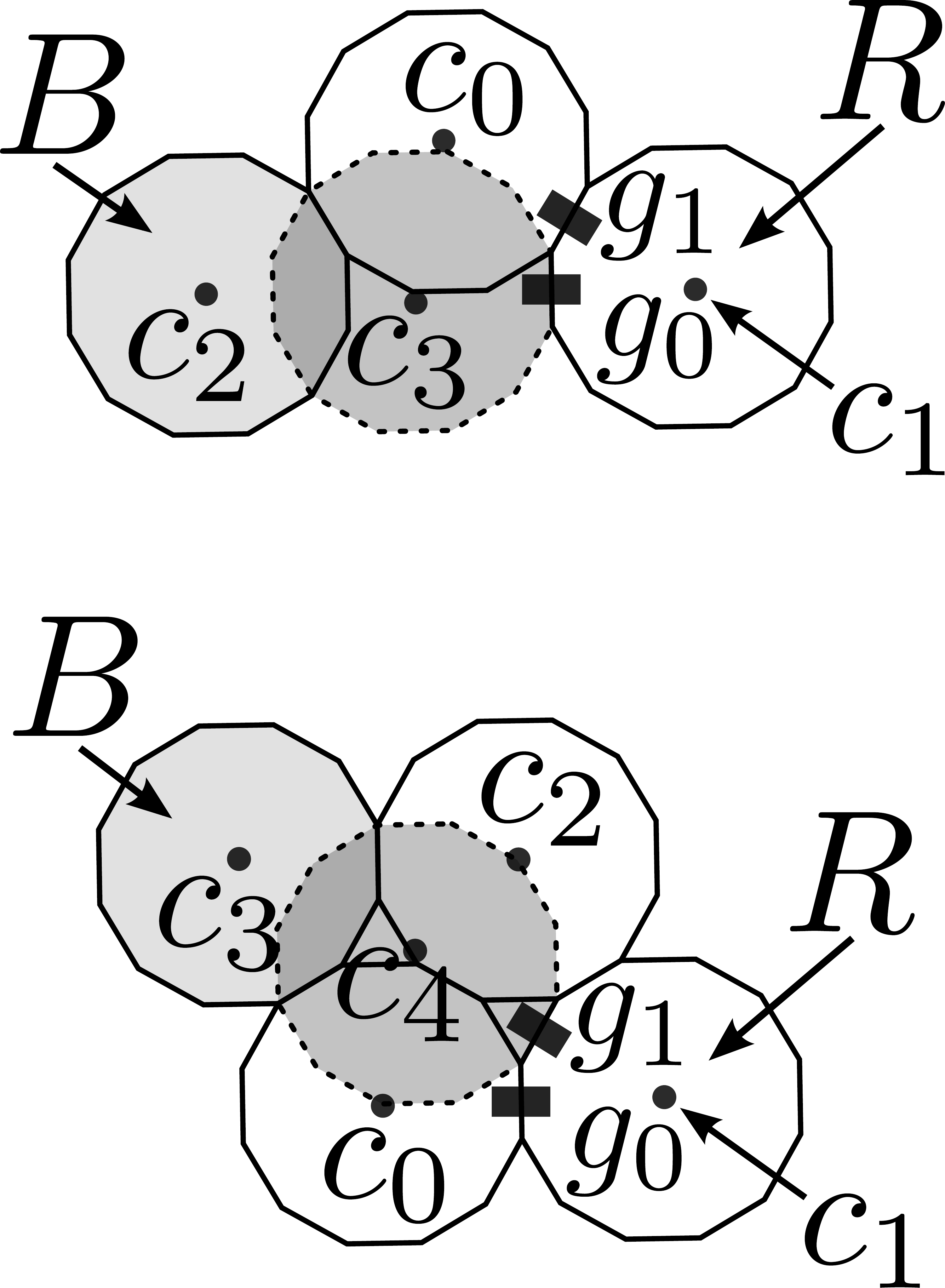}
        }
        }%
  \caption{(a), (b) and (c) each depict two configurations of polygonal tiles which represents either a $0$ (bottom) or a $1$ (top).}
  \label{fig:10to12sidesBitReaders}
\end{figure}

Like the bit-reading gadgets themselves, the calculations used to show the correctness of these bit-reading gadgets are relatively simple when compared to the previous cases. For example, for decagonal tiles,
in top configuration depicted in Figure~\ref{fig:10to12sidesBitReadersA}, to show that $B$ prevents a tile from binding to $g_0$, note that a polygon centered at $c_2$ and a polygon centered at $c_3$ overlap. Let $\omega$ be the $10^\text{th}$ root of unity $e^{\frac{2\pi i}{10}}$. Note that relative to $c_2$, $c_3 = \omega + \omega^9 - 1$. Hence, $c_3 = \omega + \omega^9 - 1 = 2\Re(\omega) - 1 = 2\cos\left( \frac{2\pi}{10} \right)$. Then the distance $d$ from $c_3$ to $c_2$ satisfies $d = |2\cos\left( \frac{2\pi}{10} \right) - 1| < .62$, and therefore the intersection of the interiors of a decagon centered at $c_3$ and a decagon centered at $c_2$ is nonempty. Hence, a decagonal tile cannot bind to the glue $g_0$. The remaining calculation for the decagonal tiles case as well as the calculations for the hendecagonal and dodecagonal cases are given in Section~\ref{sec:technical-10-12sides}.

\subsubsection{Tiles with $15$ or more sides}\label{sec:15+sides-bit-readers}

In the cases where tiles consist of regular polygons with $15$ or more sides, we give a general scheme for obtaining bit-reading gadgets for each case. Figure~\ref{fig:15+sides-append} depicts the bit-reading gadgets for each case. For the top configurations of Figure~\ref{fig:15+sides-append}, note that since each polygonal tile of these bit-reading gadgets is adjacent to another tile, we need only show that for each top configuration depicted in Figure~\ref{fig:15+sides-append}, of the two exposed glues, $g_0$ and $g_1$ of the tile $R$, $B$ prevents a tile from binding to $g_0$. In the bottom configurations of Figure~\ref{fig:15+sides-append}, we not only need to show that $B$ prevents a tile from binding to $g_1$, but we must also show that $B$ does not prevent a tile (the tile centered at $c_2$ in the bottom configurations for Figure~\ref{fig:15+sides-append}) from binding to the tile that binds to $g_0$. The latter statement ensures that when we use the bit-reading gadgets obtained from these configurations to simulate a Turing machine, in the case that a $0$ is read by attaching a tile to $g_0$, $B$ does not prevent further growth of an assembly.

\begin{figure}[htp]
\centering
  \subfloat[][Bit-reading gadget configuration for pentadecagonal tiles.]{%
        \label{fig:15+sidesA-append}%
    		\makebox[.25\textwidth]{
        \includegraphics[width=.8in]{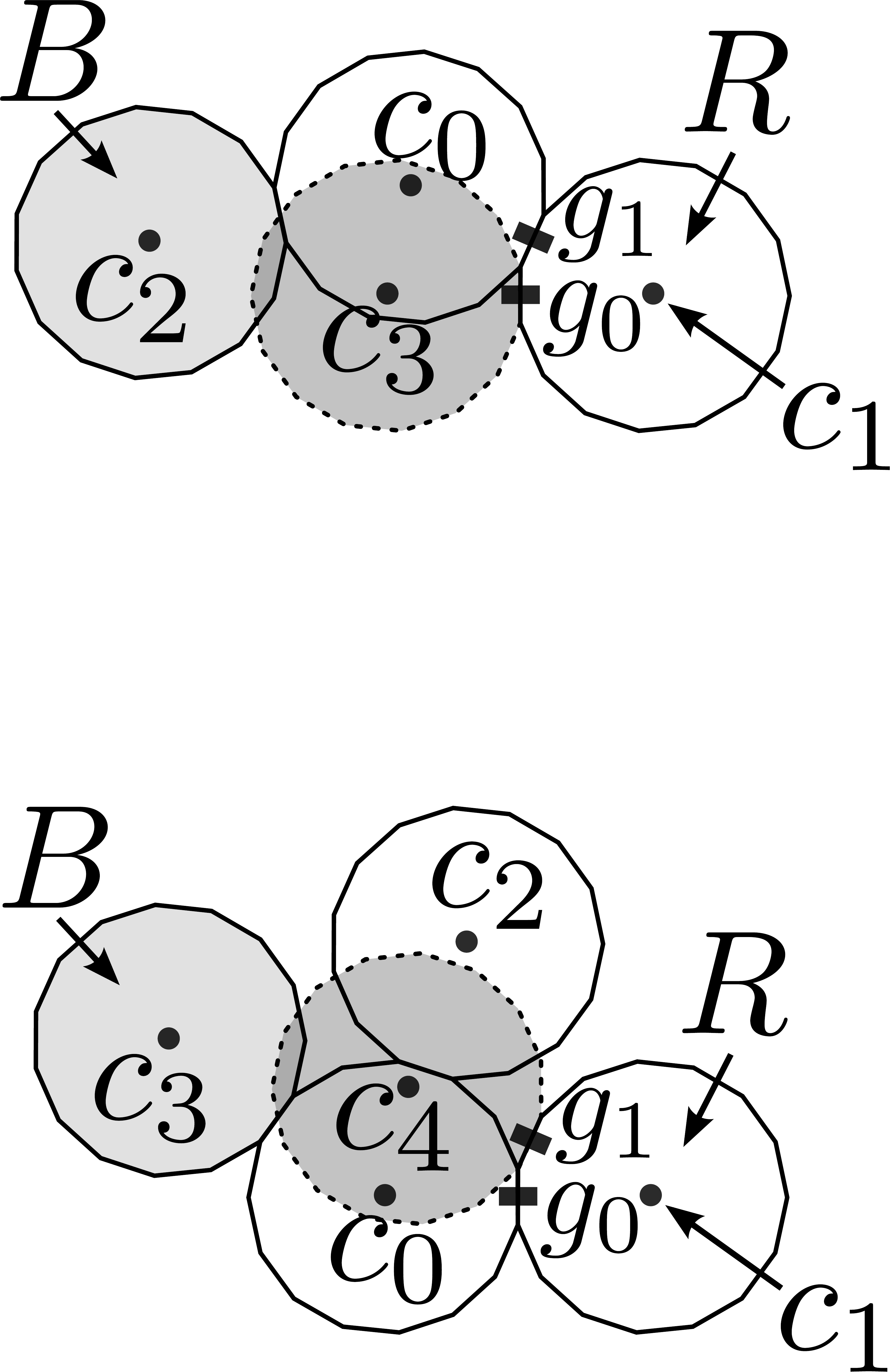}
        }
        }%
        \quad\quad
  \subfloat[][Bit-reading gadget configuration for hexadecagonal tiles.]{%
        \label{fig:15+sidesB-append}%
    		\makebox[.25\textwidth]{
        \includegraphics[width=.8in]{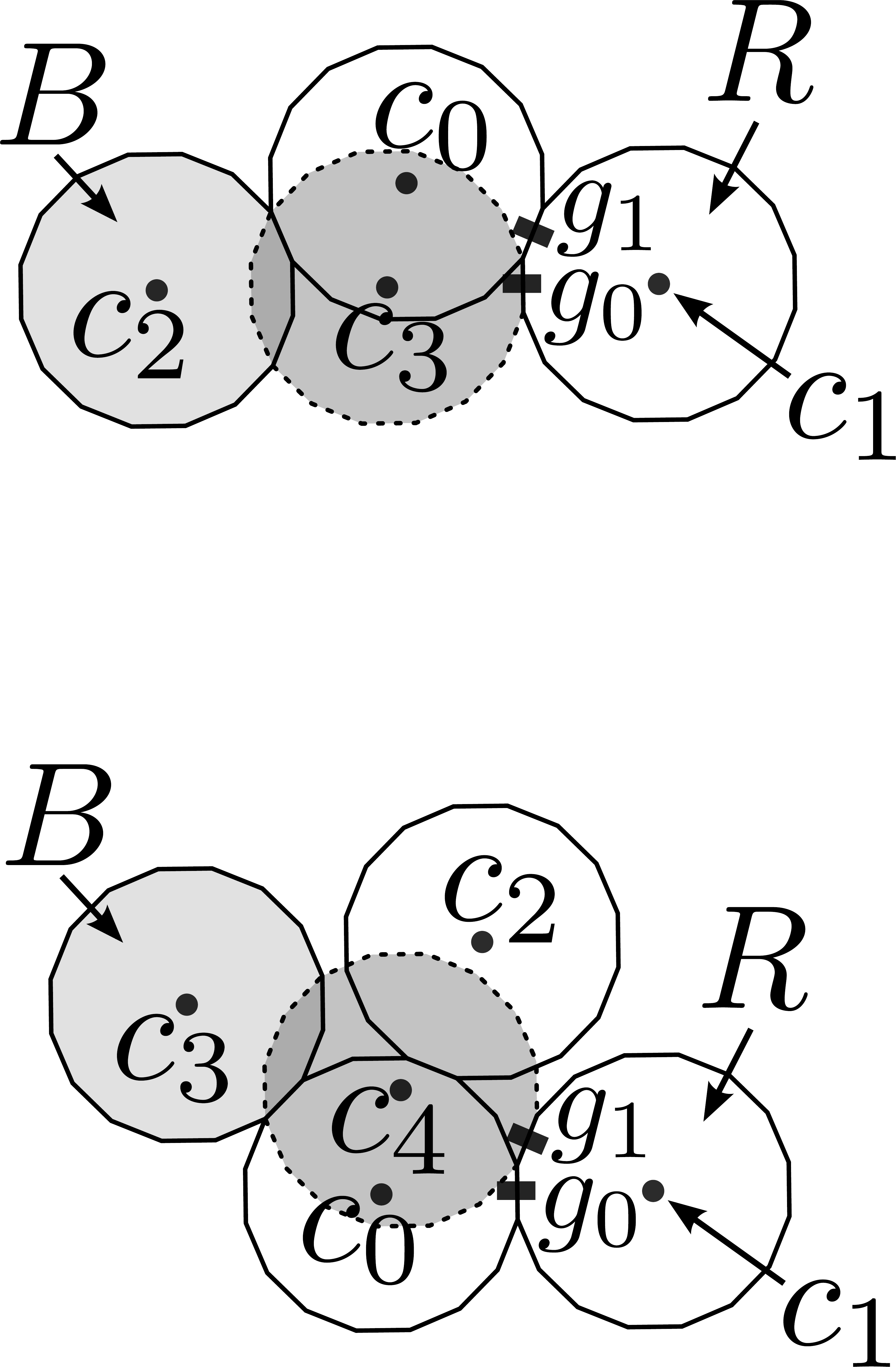}
        }
        }%
       \quad\quad
  \subfloat[][Bit-reading gadget configuration for heptadecagonal tiles.]{%
        \label{fig:15+sidesC-append}%
    		\makebox[.25\textwidth]{
        \includegraphics[width=.8in]{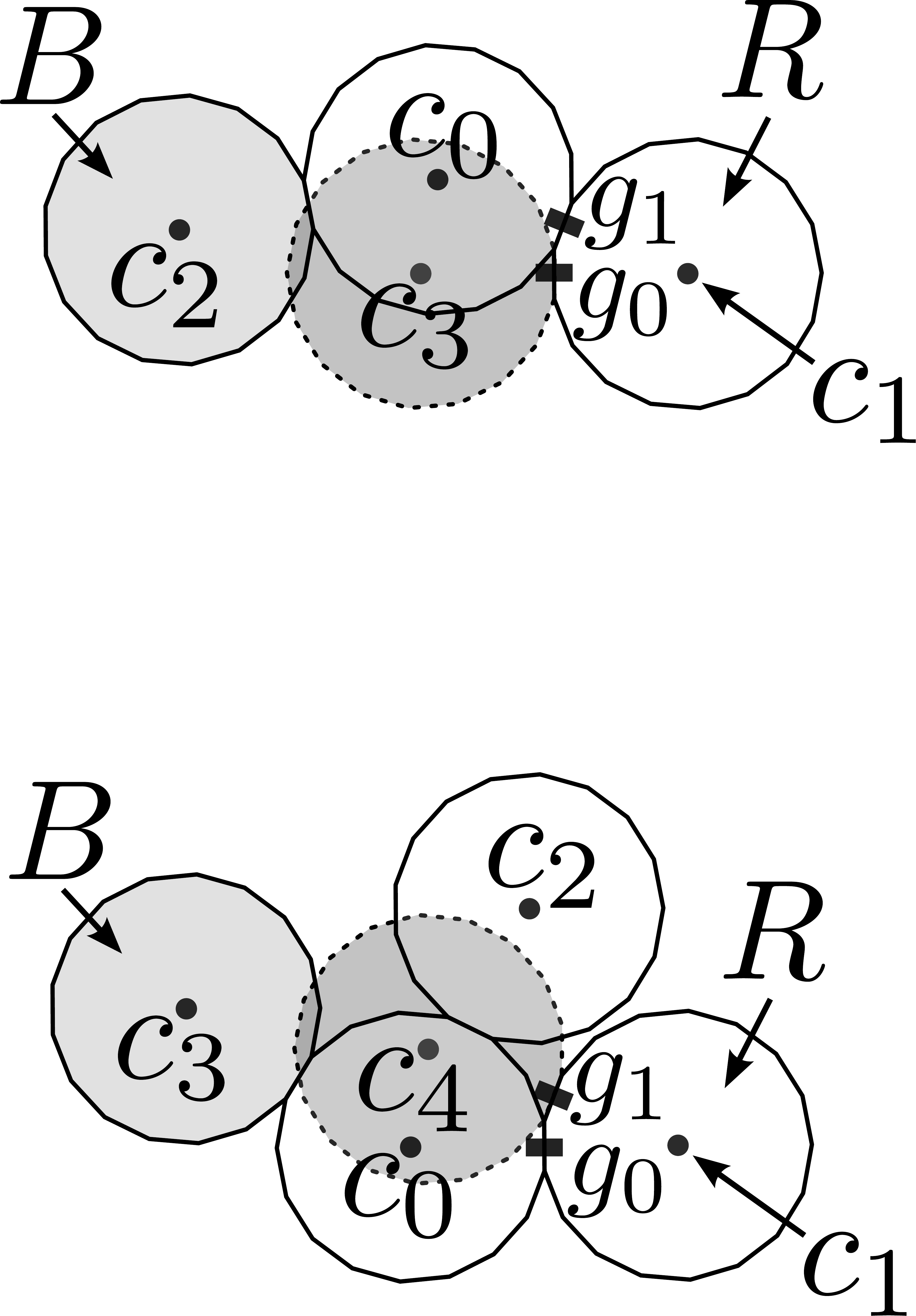}
        }
        }%
  \caption{(a), (b) and (c) each depict two configurations of polygonal tiles which represents either a $0$ (bottom) or a $1$ (top).}
  \label{fig:15+sides-append}
\end{figure}

Now, consider a polygon $P_n$ with $n\geq 15$ sides and let $\omega$ be the $n^{th}$ root of unity $e^{\frac{2\pi i}{n}}$. Then, the general scheme for constructing a bit-reading gadget falls into two cases. First, if $n$ is odd (the cases where $n$ is even are similar), relative to a tile with negated orientation (the polygon labeled $R$ in the configurations in Figure~\ref{fig:15+sides-append}), the two configurations that give rise to the bit-reading gadget are as follows. Let $k$ be such that $n=2k+1$ ($n=2k$ if $n$ is even). Referring to the top configurations of Figure~\ref{fig:15+sides-append}. To ``write'' a $1$, the configuration is obtained by centering a blocker tile with negated orientation, labeled $B$ in the top configurations of Figure~\ref{fig:15+sides-append}, at $-\omega^{n-1} + \omega^{k+1}$ (whether $n$ is even or odd) relative to $R$. Then to ``read'' a $1$, $R$ exposes two glues $g_1$ and $g_0$ such that if a tile binds to $g_1$, it will have standard orientation and be centered at $-\omega^{n-1}$ (whether $n$ is even or odd) and if a tile that binds to $g_0$, it will have standard orientation and be centered at $-1$. We will show that $B$ will prevent this tile from binding. This gives the configuration depicted in the top figures of Figure~\ref{fig:15+sides-append}. Now, referring to the bottom configurations of Figure~\ref{fig:15+sides-append}, to ``write'' a $0$, the configuration is obtained by centering a blocker tile with negated orientation, labeled $B$ in the bottom configuration of Figure~\ref{fig:15+sidesA-append}, at $-1 + \omega^{k-1}$ ($-1 + \omega^{k-2}$ if $n$ is even) relative to $R$.  In this case, we will show that $B$ prevents a tile from binding to $g_1$. In addition, we place a glue on the tile that binds to $g_0$ that allows for another tile to bind to it so that its center is at $c_2 = -1 + \omega^{\lfloor \frac{k-1}{2} \rfloor}$ ($c_2 = -1 + \omega^{\frac{k-2}{2}}$ if $n$ is even) relative to $R$.  This gives the configuration depicted in the bottom figures of Figure~\ref{fig:15+sidesA-append} and Figure~\ref{fig:15+sidesC-append}. Moreover, we show that neither $R$ nor $B$ prevent the binding of this tile.

In order to perform the calculations used to show the correctness of these bit-reading gadgets, we consider the cases where $n$ is even and where $n$ is odd. Here we give brief versions of the calculations when $n$ is odd. For more detail and calculations for the case where $n$ is even, see Section~\ref{sec:technical-15+sides}.

Suppose that $n = 2k+1$. We now refer to the bottom configurations of Figure~\ref{fig:15+sidesA-append}. To show that a polygon centered at $c_1$ and a polygon centered at $c_2$ do not overlap, consider the case where $k$ is odd (the case where $k$ is even is similar). Note that relative to $c_0$, $c_1 = 1$ and $c_2 = \omega^{\frac{k-1}{2}}$. Then the distance $d_n$ from $c_1$ to $c_2$ satisfies the following equation.
$$d_n^2 = \left(1-\cos\left(\frac{\left(k-1\right)\pi}{n}\right)\right)^2 + \sin^2\left(\frac{\left(k-1\right)\pi}{n}\right)$$
Substituting $k = \frac{n-1}{2}$ for $k$ and simplifying, we obtain $d_n^2 = 2+2\sin\left(\frac{3\pi}{2n}\right)$. It is well known that for regular polygons with $n$ sides and apothem $\frac{1}{2}$, the circumradius is given by $\frac{1}{\cos\left( \frac{\pi}{n} \right)}$. Hence, to show that a polygon centered at $c_1$ and a polygon centered at $c_2$ do not overlap, we show that $d_{n}^2 > \frac{1}{\cos^2\left(\frac{\pi}{n}\right)}$ for $n \geq 15$. To see this, note that $\cos^2\left(\frac{\pi}{n}\right)d_n^2 = 2\cos^2\left(\frac{\pi}{n}\right)\left(1+\sin\left(\frac{3\pi}{2n}\right)\right)$. Then for $n\geq15$, $2\cos^2\left(\frac{\pi}{n}\right)\left(1+\sin\left(\frac{3\pi}{2n}\right)\right) > 2\cos^2\left(\frac{\pi}{4}\right) =1$.
It then follows that $d_{n} > \frac{1}{\cos\left(\frac{\pi}{n}\right)}$.  Therefore, $d_n$ is greater than twice the circumradius of our polygons. Hence, a polygon centered at $c_1$ and a polygon centered at $c_2$ do not overlap.

To show that a polygon centered at $c_3$ and a polygon centered at $c_4$ overlap, note that relative to $c_1$, $c_3 = -1 + \omega^{k-1}$ and $c_4 = -\omega^{n-1}$. Therefore, the distance $d_n$ from $c_3$ to $c_4$ is satisfies the equation 
\begin{align*}
d_n^2 &= \left( -1 + \cos\left(\frac{2(k-1)\pi}{n}\right) + \cos\left( \frac{2(n-1)\pi}{n} \right) \right)^2 \\
      &= \ \ \ \ + \left( \sin\left( \frac{2(k-1)\pi}{n} \right) + \sin\left( \frac{2(n-1)\pi}{n} \right) \right)^2
\end{align*}
Substituting $k = \frac{n-1}{2}$ for $k$ and simplifying, we obtain, $$d_n^2 = 1 + 2\left(2\sin^2\left(\frac{\pi}{n}\right)\left(1-2\cos\left(\frac{\pi}{n}\right)\right)\right).$$ Note that for each $n>2$, $d_n^2 < 1$. To see this, it suffices to show that $$2\sin^2\left(\frac{\pi}{n}\right)\left(1-2\cos\left(\frac{\pi}{n}\right)\right) < 0.$$ This follows from the fact that $2\sin^2\left(\frac{\pi}{n}\right) > 0$ and $1-2\cos\left(\frac{\pi}{n}\right) < 0$ for $n >2$.
Now, since for each $n>2$, $d_n^2 < 1$, we see that $d_n < 1$. Since the length of the apothem for each tile is assumed to be $\frac{1}{2}$, we can conclude that a polygon centered at $c_3$ and a polygon centered at $c_4$ must overlap.

\subsection{2-shaped systems with regular polygonal tiles}\label{sec:bit-readers-2shaped-regular}

In this section we describe bit-reading gadgets for 2-shaped systems whose tileset consists of two distinct regular polygons. We assume that the edges of all polygonal tiles have the same length. The bit-reading gadgets that we give here are normalized on-grid bit-readers.

\begin{lemma}\label{lem:2shaped-bit-gadget}
Let $P_n$ and $Q_m$ be a regular polygons with $n$ and $m$ sides of equal length. Then, for all $n \geq 3$ and $m \geq 3$ such that $n\neq m$, there exists a 2-shaped system $\mathcal{T}_{n,m} = (T_{n,m}, \sigma_{n,m})$ with shapes $P_n$ and $Q_m$ such there a bit-reading gadget exists for $\mathcal{T}_{n,m}$.
\end{lemma}

\begin{figure}[htp]
\centering
  \subfloat[][A $1$ is read. This time a $0$ cannot be read by mistake since the tile $B$ prevents growth of a path of tiles that attach via the glue labeled $g_0$.]{%
        \label{fig:34BitReaderA}%
		\makebox[.45\textwidth]{
        \includegraphics[width=2in]{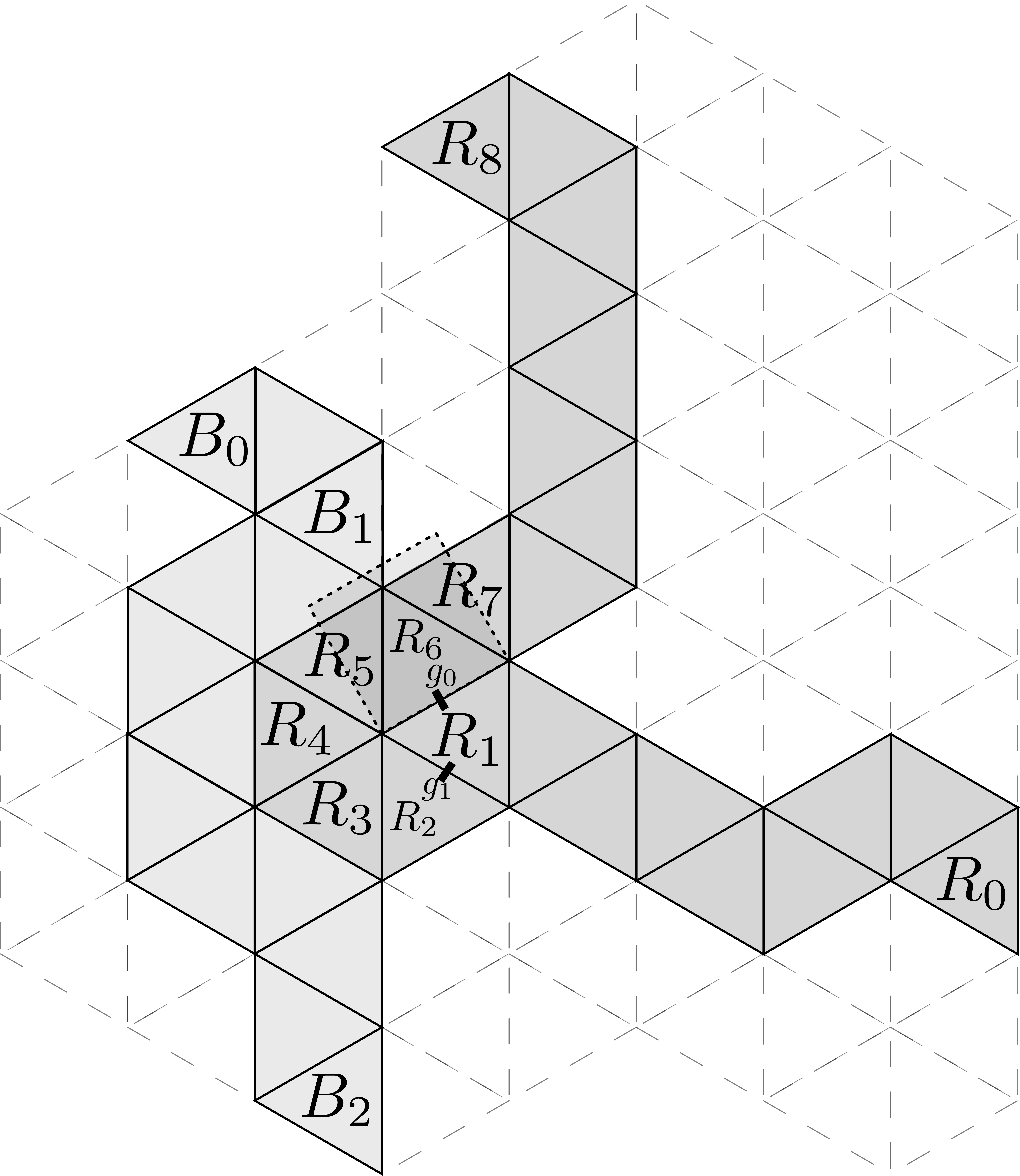}
        }}%
        \quad
  \subfloat[][A $0$ is read, and a $1$ cannot be read by mistake since the tile $B$ prevents a tile from attaching via the glue labeled $g_1$.]{%
        \label{fig:34BitReaderB}%
		\makebox[.45\textwidth]{
        \includegraphics[width=2in]{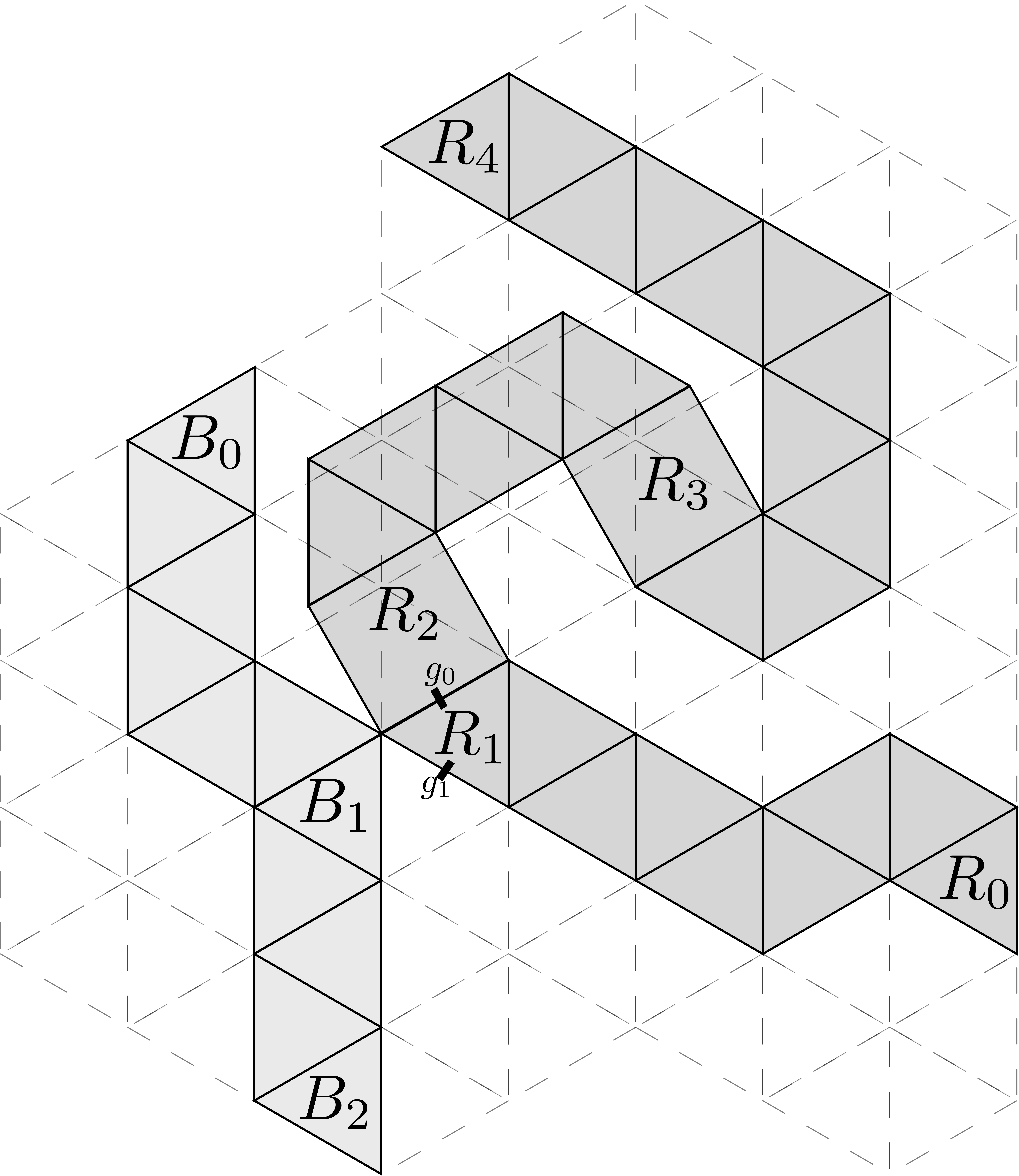}
        }}%
  \caption{A connected bit-gadget consisting of tiles shaped like a regular triangle or square.}
  \label{fig:34BitReader}
\end{figure}

(a) and (b) of Figure~\ref{fig:34BitReader} depict bit-reading gadgets which give a scheme for ``writing a bit'' as growth proceeds from left to right, and ``reading a bit'' as growth proceeds from right to left. To write a bit, we can define unique glues that enforce the assembly of the path of tiles (light gray tiles in (a) and (b) of Figure~\ref{fig:34BitReader}) starting from the tile labeled $B_0$ and ending at a tile labeled $B_2$ in (a) and (b). Assuming that the light gray tiles are part of an existing assembly, to read a bit, we define unique glues that enforce the (dark gray tiles in (a) and (b)) starting from the tile labeled $R_0$ and ending with the tile labeled $R_1$. Then, $R_1$ exposes two glues labeled $g_0$ and $g_1$ in Figure~\ref{fig:34BitReader}. Now, depending on whether an assembly which represents a $0$ is present or an assembly which represents a $1$ is present, either a triangular tile with a glue labeled $g_1$ binds to $R_1$ via the glue $g_1$ exposed by $R_1$ (depicted in Figure~\ref{fig:34BitReaderA}) or a square shaped tile with a glue labeled $g_0$ binds to $R_1$ via the glue $g_0$ exposed by $R_1$ (depicted in Figure~\ref{fig:34BitReaderB}).

In the former case, denote the triangular tile which binds to $R_1$ by $R_2$; this tile is labeled $R_2$ in Figure~\ref{fig:34BitReaderA}. Then, we can define glues that allow the tiles $R_i$ for $i = 3,4,5,6$ or $7$ to bind in that order as depicted in Figure~\ref{fig:34BitReaderA}. Finally, we define a set of tiles that form the path of tiles from $R_7$ to $R_8$. The latter case, depicted in Figure~\ref{fig:34BitReaderB}, is similar. In this case, a square tile (labeled $R_2$) binds to $g_0$. We define this tile such that the path of tiles from $R_2$ to $R_4$ assembles. Note that the square tile labeled $R_3$ ensures that the triangular tiles along this path of tiles from $R_3$ to $R_4$ are on-grid. In particular, $R_4$ is on-grid. Lastly, we refer to each configuration in Figure~\ref{fig:34BitReader} and note that relative to the underlying grid (shown as dashed lines), $B_0$, $B_2$, $R_0$ and $R_8$ in (a) and respectively $B_0$, $B_2$, $R_0$ and $R_4$ in (b) are on-grid and in the same location. It is straightforward to see that such configurations can be extended to give a normalized on-grid bit-reading gadget that conforms to Definition~\ref{def:bit-reader}.

Constructions for normalized on-grid bit-reading gadgets for pairs of regular polygons with sides $m$ and $n$ where $3\leq m \leq 6$, $5\leq n \leq 6$ and $m\neq n$ are similar and are given in Section~\ref{sec:technical-2shaped}.

\subsection{Single shaped systems with equilateral polygonal tiles}\label{sec:bit-readers-equilaterals}

In this section we describe bit-reading gadgets for single-shaped systems whose tileset consists of an equilateral polygon. The bit-reading gadgets that we give here are normalized on-grid bit-readers. Note that for polygons with $7$ or more sides, Lemma~\ref{lem:regular-bit-gadget} implies the following lemma. Hence, we need only show Lemma~\ref{lem:equilateral-bit-gadget} for equilateral polygons with $4$, $5$, or $6$ sides. Therefore, we give normalized on-grid bit-reading gadgets for all three cases showing the following lemma. It should be noted that while the general grid construction given in Section~\ref{sec:bit-grid-main} pertain to regular polygons. Similar techniques can be used to obtain grids for the equilateral polygonal tiles in this section. The grids themselves are depicted using dashed lines in the figures of this section.

\begin{lemma}\label{lem:equilateral-bit-gadget}
For all $n \geq 4$, there exists an equilateral polygon $P_n$ with $n$ sides and a single-shaped system $\mathcal{T}_n = (T_n, \sigma_n)$ shape $P_n$ such a bit-reading gadget exists for $\mathcal{T}_n$.
\end{lemma}

\begin{figure}[htp]
\centering
  \subfloat[][A $0$ is read, and a $1$ cannot be read by mistake since the tile $B$ prevents a quadrilateral tile from attaching via the glue labeled $g_1$.]{%
        \label{fig:equilateralBitReaderA}%
        \includegraphics[width=2.5in]{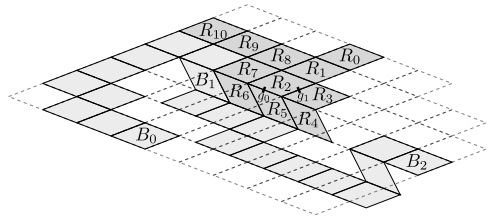}
        }%
        \quad
  \subfloat[][A $1$ is read. This time a $0$ cannot be read by mistake since the tile $B$ prevents growth of a path of quadrilateral tiles that attach via the glue labeled $g_0$.]{%
        \label{fig:equilateralBitReaderB}%
        \includegraphics[width=2.5in]{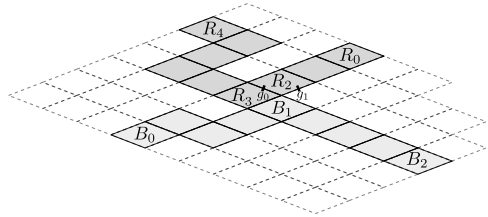}
        }%
  \caption{A connected bit-gadget consisting of quadrilateral tiles. This figure also depicts the tile shape.}
  \label{fig:equilateralBitReader}
\end{figure}

(a) and (b) of Figure~\ref{fig:equilateralBitReader} depict bit-reading gadgets which give a scheme for ``writing a bit'' as growth proceeds from left to right, and ``reading a bit'' as growth proceeds from right to left. To write a bit, we can define unique glues that enforce the assembly of the path of tiles (light gray tiles in (a) and (b) of Figure~\ref{fig:equilateralBitReader}) starting from the tile labeled $B_0$ and ending at a tile labeled $B_2$ in (a) and (b). Assuming that the light gray tiles are part of an existing assembly, to read a bit, we define unique glues that allow $R_0$, $R_1$ and $R_2$ to bind in that order. Then, $R_2$ exposes two glues labeled $g_0$ and $g_1$ in Figure~\ref{fig:equilateralBitReader}. Now, depending on whether an assembly which represents a $0$ is present or an assembly which represents a $1$ is present, either a quadrilateral tile with a glue labeled $g_1$ binds to $R_2$ via the glue $g_1$ exposed by $R_2$ (depicted in Figure~\ref{fig:equilateralBitReaderA}) or a quadrilateral tile with a glue labeled $g_0$ binds to $R_2$ via the glue $g_0$ exposed by $R_2$ (depicted in Figure~\ref{fig:equilateralBitReaderB}).
In the former case, denote the quadrilateral tile which binds to $R_2$ by $R_3$; this tile is labeled $R_3$ in Figure~\ref{fig:equilateralBitReaderA}. Then, we can define glues that allow the tiles $R_i$ for $3\leq i \leq 10$ to bind in that order as depicted in Figure~\ref{fig:equilateralBitReaderA}. Moreover, we note that $B_1$ prevents a tile from binding to $g_0$. The latter case, depicted in Figure~\ref{fig:equilateralBitReaderB}, is similar. In this case, a quadrilateral tile (labeled $R_2$) binds to $g_0$. We define this tile such that the path of tiles from $R_2$ to $R_4$ assembles. In the case of Figure~\ref{fig:equilateralBitReaderB}, we also note that $B_1$ prevents a tile from binding to $g_1$.  Finally, we note that relative to the underlying grid (shown as dashed lines in (a) and (b) of Figure~\ref{fig:equilateralBitReader}) this configuration can then be used to obtain a normalized on-grid bit-reading gadget.

Constructions for normalized on-grid bit-reading gadgets for equilateral polygons with sides $5$ or $6$ sides  are similar and are given in Section~\ref{sec:technical-equilaterals}.

\subsection{A single shaped system with triangular tiles}\label{sec:bit-readers-triangle}

In this section we describe bit-reading gadgets for single-shaped systems whose tile set consists of a particular obtuse isosceles triangle. We assume that the edges of all triangular tiles have the same length. The bit-reading gadgets that we give here are normalized on-grid bit-readers. Again, it should be noted that while the general grid construction given in Section~\ref{sec:bit-grid-main} pertain to regular polygons. Similar techniques can be used to obtain grids for the triangular tiles in this section. The grids themselves are depicted using dashed lines in the figures of this section.

\begin{lemma}\label{lem:triangular-bit-gadget}
There exists an obtuse isosceles triangle $P$ and a single-shaped system $\mathcal{T} = (T, \sigma)$ with shape $P$ such a bit-reading gadget exists for $\mathcal{T}$.
\end{lemma}

\begin{figure}[htp]
\centering
  \subfloat[][A configuration of polygonal tiles which represents a $0$]{%
        \label{fig:irregular-triangleA}%
    		\makebox[.53\textwidth]{
        \includegraphics[width=3in]{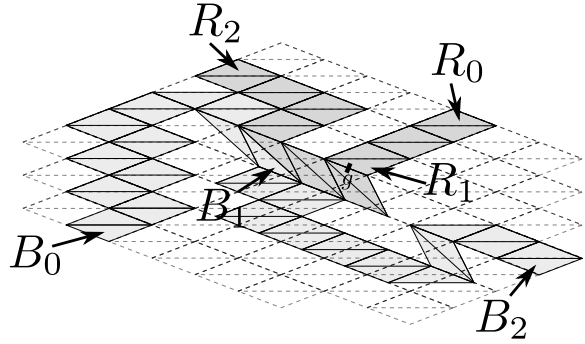}
        }
        }%
        \quad\quad\quad\quad\quad\quad\quad\quad
  \subfloat[][A configuration of polygonal tiles which represents a $1$.]{%
        \label{fig:irregular-triangleB}%
    		\makebox[.53\textwidth]{
        \includegraphics[width=3in]{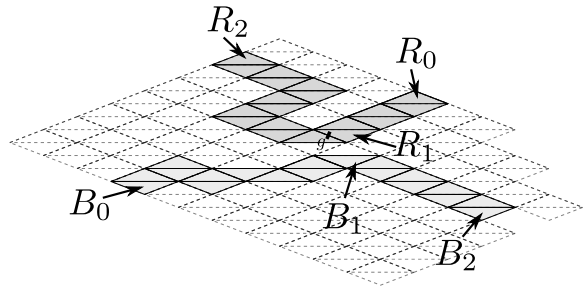}
        }
        }%
  \caption{Bit-reading gadget configuration for tiles with the shape of an irregular triangle.}
  \label{fig:irregular-triangles}
\end{figure}

Figure~\ref{fig:irregular-triangles} depicts the configurations that give rise to bit-reading gadgets for
single-shaped systems whose tiles have the shape of an obtuse isosceles triangle. As in Section~\ref{sec:bit-readers-equilaterals}, one can see that these configurations can be used to obtain a normalized on-grid bit-reading gadget.
} %
\fi

\ifabstract
\later{
\section{Building \emph{Normalized} Bit-reading Gadgets} \label{sec:technicalLemmas}
Let $P$ be a regular polygon with $7$ or more sides, and let $g_{\alpha}$ denote the terminal assembly of the tile system given in Lemma~\ref{lem:grids}.
We now show that given a bit-reading gadget from proceeding section corresponding to the regular polygon $P$, we can form an on grid bit-reading gadget (with respect to $g_{\alpha}$).  In order to show this, we first show how the individual bit-writers can be grown in an on grid manner (with the bit reader that reads these writers also on grid), and then we show how to find positions common to these bit-writers so that up to translation, the bit writer start and end in the same place. Before we begin our construction, we introduce a couple of definitions.
We denote the location of the center of a tile $P$ in the complex plane by $c(P)$.  We say that a tile $P$ is \emph{x-centered on grid $g_{\alpha}$} provided that $c(P) = c(P')$ and $P$ and $P'$ have the same orientation for some tile $P' \in g_{\alpha}$.

At a high-level, we construct a normalized bit-reading gadget from one of the gadgets presented in Section~\ref{sec:bit-readers-main} in the following way. Consider Figures~\ref{fig:heptagonal-gadget-0} and~\ref{fig:heptagonal-gadget-1} where normalized bit-reading gadgets are given in the case of heptagonal tiles. In those figures a bit is written from west to east and read from east to west. When writing a bit and starting from the southwest-most black tile, assembly proceeds via attachment of a single tile at a time on some fixed grid $g_\alpha$ (shown in the background in the figures as white heptagons). Then, the blue tiles ``shift'' off this grid and onto another grid, $g_{\alpha}'$ say. This shifting ensures that the tile labeled $R$ in those figures is on the grid $g_\alpha$. Then, the portion of a bit-reading gadget which encodes a 0 (Figure~\ref{fig:heptagonal-gadget-0}) or 1 (Figure~\ref{fig:heptagonal-gadget-1}) is assembled. The tiles which make up this portion are purple in the figures. Call the set of these tiles $S$. At this point, we are possibly on a grid $g_\alpha''$ which may or may not be distinct from $g_\alpha$ or $g_\alpha'$. Finally, we ``shift'' back onto the grid $g_\alpha$ by assembling the remaining portion of a bit-reading gadget (those tiles of the bit-reading gadget that are not in $S$). The tiles which produce this final shift are green in Figures~\ref{fig:heptagonal-gadget-0} and~\ref{fig:heptagonal-gadget-1}. At this point, a path of tiles (each of which is on $g_\alpha$) assemble until the southeast-most black tile in the figures attaches. This bit is read using the orange, red and yellow tiles. The tile $R$ is on grid $g_\alpha$. The red tiles are the path of unblocked tiles whose assembly indicates that the appropriate bit is read. The final red tile that is placed may not be on grid $g_\alpha$. Therefore, the yellow tiles (whose assembly sequence is essentially that of the red tiles in reverse order) ``shift'' back onto grid $g_\alpha$. Note that the black tiles and the end tiles of the reading path of tiles (the orange, yellow, and red tiles) in Figure~\ref{fig:heptagonal-gadget-0} have locations that match the locations of the respective tiles in Figure~\ref{fig:heptagonal-gadget-1}. This ensures that we can ``plug'' these gadgets into a zig-zag growth pattern to simulate a Turing machine.

\begin{figure}[htp]
	\includegraphics[width=\linewidth]{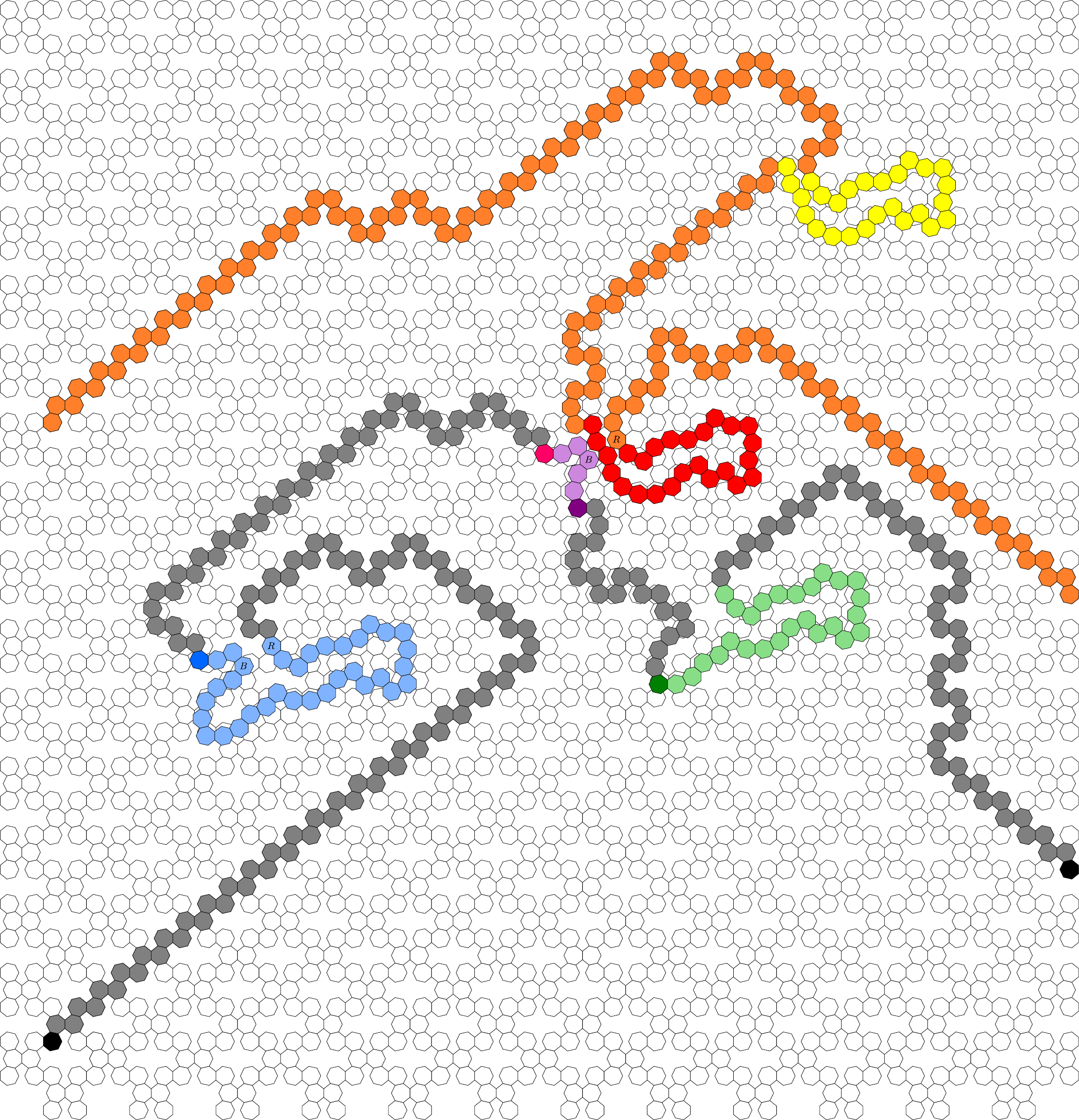}
    \caption{The completed bit-reading gadget for heptagons when a 0 is ``read''.  The grey tiles represent paths which connect the subconfigurations in the bit writer.  The blue tiles represent $C_{\alpha_w}$, and the dark blue tile represents $t_{ww}$.  The purple tiles represent $C_{\alpha}$, and the dark purple tile represents $t_s$.  The pink tile represents $t_w$.  The green tiles represent $C_{\alpha_e}$, and the dark green tile represents $t_{se}$. All other color of tiles represent tiles composing the bit reader. In this figure the bit is written from west to east and read from east to west.}
  \label{fig:heptagonal-gadget-0}
\end{figure}

\begin{figure}[htp]
	\includegraphics[width=\linewidth]{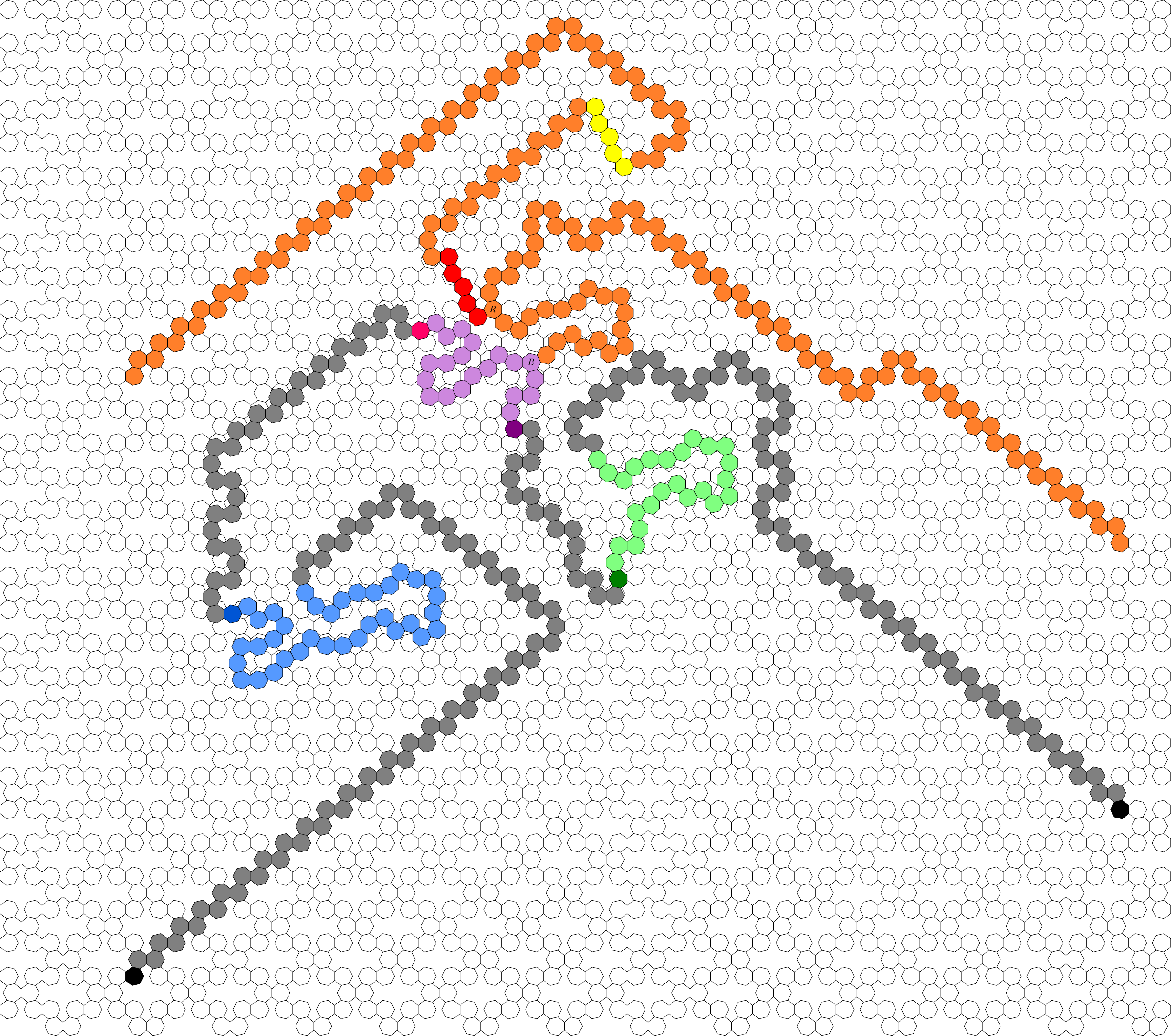}
    \caption{The completed bit-reading gadget gadget for heptagons when a 1 is ``read''.  The grey tiles represent paths which connect the subconfigurations in the bit writer.  The blue tiles represent $C_{\alpha_w}$, and the dark blue tile represents $t_{ww}$.  The purple tiles represent $C_{\alpha}$, and the dark purple tile represents $t_s$.  The pink tile represents $t_w$.  The green tiles represent $C_{\alpha_e}$, and the dark green tile represents $t_{se}$. All other color of tiles represent tiles composing the bit reader. In this figure the bit is written from west to east and read from east to west.}
  \label{fig:heptagonal-gadget-1}
\end{figure}

\subsection{Constructing On Grid Bit-writer Configurations} \label{sec:bitwon}
The $\alpha_0$ on grid bit writer will consist of three parts which we call: 1) a blocker subconfiguration, 2) an east shifting subconfiguration and 3) a west shifting subconfiguration.  The three subconfigurations are all formed by modifying a ``base'' configuration which we describe now.  The base configuration is formed by modifying the assembly obtained when the bit-reading gadget described in Section~\ref{sec:bit_reading} ``reads a 0''.

If the bit-reading gadget for $P$ is simple (e.g. those shown in Figures~\ref{fig:10to12sidesBitReaders} and~\ref{fig:15+sides}), we first extend the bit writing portion of the gadget in the following way.  To begin, observe that the bit writer portion of the bit reading gadget will consist of a tile with negated orientation which we will call $B$.  Note that by the construction of these simple bit reading gadgets, we can always place a tile which has standard orientation at a position of $\omega^{\lceil \frac{3k}{4} \rceil}$ relative to $B$.  After placing this tile, we continue placing tiles so that we form a path of tiles from $B$ such that the last tile placed in this path is the northernmost tile in the bit reading gadget configuration and has standard orientation (see Figure~\ref{fig:simple_bit_ext}).  Next we grow a path of tiles from $B$ that extends south so that the last tile placed in this path is the southernmost tile in the bit reading configuration as shown in Figure~\ref{fig:simple_bit_ext}.

\begin{figure}[htp]
\centering
  \subfloat[][A simple bit reading gadget formed from a regular polygon with 10 sides which has read a 0 (also shown in Figure~\ref{fig:10to12sidesBitReaders}.]{%
        \label{fig:simple_bit}%
        \includegraphics[width=0.75in]{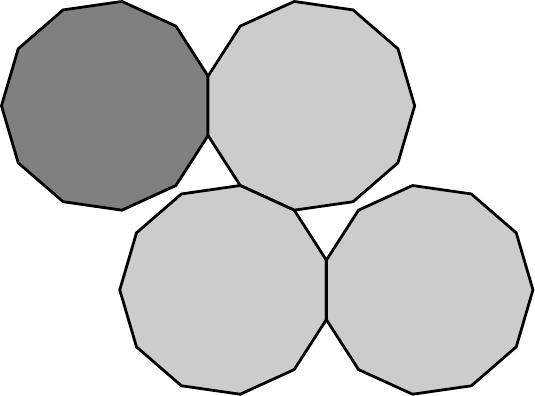}
        }%
        \quad
  \subfloat[][Extending the bit reading gadget to form our base configuration.]{%
        \label{fig:simple_bit_ext}%
        \includegraphics[width=1.5in]{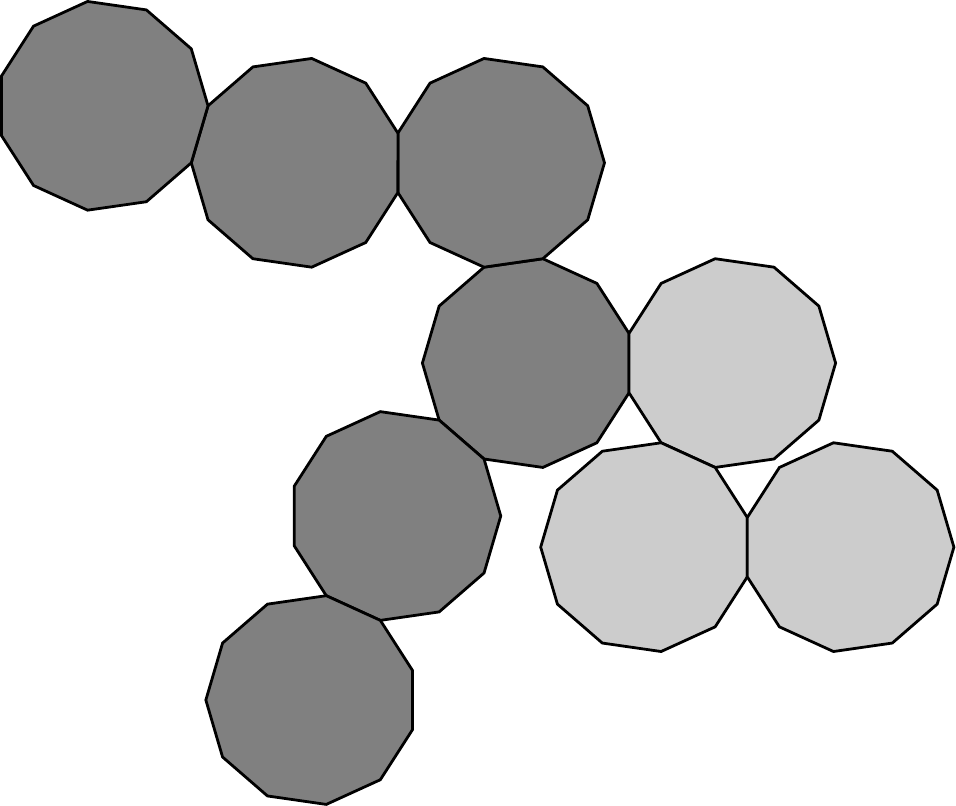}
        }%
  \caption{A simple bit reading gadget (which ``read a 0'') and its extension (which will form our base configuration).  The darkly shaded tiles are the bit writer portion of the bit reading gadget.}
  \label{fig:simple_ext}
\end{figure}

We say that the first tile to be placed in the bit-writer subconfiguration of the bit-reading gadget is the northernmost ``end tile'' in the path.  The other ``end tile'' in the path we refer to as the last tile to be placed in the bit-writer subconfiguration.  Also, recall that the tile $R$ in the bit reading gadget is the tile from which the bit reader grows.  To construct the base configuration, we simply remove the tiles in the configuration which do not lie on either the path from the first tile in the bit-writer to the tile $R$ or the path from the last tile in the bit-writer to the tile $R$.  Furthermore, we extend a path from the first tile to be placed in the bit-writer portion of the bit-reading gadget so that the last tile placed in this path has negated orientation and is the westernmost tile in the bit reading gadget configuration.  Call this configuration $C_{\alpha}$.

Let the tile $t_s$ represent the westernmost tile of the set of southernmost tiles in the bit-writer portion of the configuration $C_{\alpha}$.  We consider two cases: 1) the tile $t_s$ has negated orientation and 2) the tile $t_s$ has standard orientation.  In case 1, we add a tile in standard orientation to the configuration at location $-\omega^{\lfloor \frac{k}{4} \rfloor}$ relative to the tile $t_s$.  We know that this is still a valid configuration by the construction of the bit reading gadgets in the previous section and the assumption that $t_s$ is the westernmost tile of the set of southernmost tiles.  Note that after this modification we are now in case 2.  In the case that $t_s$ has standard orientation, we translate $C_{\alpha}$ so that the tile $t_s$ is $1$-centered on the grid $g_{\alpha}$. We denote the bounding box of $C_{\alpha}$ by $B_{\alpha}$ and the dimensions of $B_{\alpha}$ by $m_B \times n_B$.

Now, let $C_{\alpha_e}$ be the configuration obtained by taking a copy of $C_{\alpha}$ and removing all tiles which do not lie on the shortest path from $t_s$ to $R$.  For clarity we denote the tile $t_s$ in $C_{\alpha_e}$ by $t_{se}$.  Translate this configuration so that it has the following properties: 1) the tile $t_{se}$ is $1$-centered on the gird $g_{\alpha}$, 2) $\Re(c(t_{se})) - \Re(c(t_s)) \geq 5$, and 3) $\Im(c(t_s))-\Im(c(t_{se})) \geq n_B$.

Define the tile $t_w$ to be the westernmost tile of $C_{\alpha}$.  Translate the configurations $C_{\alpha}$ and $C_{\alpha_e}$ so that they remain in the same positions relative to each other and the tile $t_w$ is $4$-centered on the grid $g_{\alpha}$.  We now make a copy of configuration $C_{\alpha}$, which we call $C_{\alpha_w}$, and denote the tile $t_w$ in $C_{\alpha_w}$ by $t_{ww}$.  We translate the configuration $C_{\alpha_w}$ so that its location meets the following requirements: 1) the tile $t_{ww}$ is $4$-centered on the grid $g_{\alpha}$, 2) $\Re(c(t_w))-\Re(c(t_{ww}))\geq n_{B'}+5$, and 3) $\Im(c(t_w))-\Im(c(t_{ww}))\geq m_{B'}+5$.  Call this configuration $C''$.

The blocker subconfiguration consists of a modified version of the configuration $C_{\alpha}$.  Namely, it consists of the configuration $C_{\alpha}$ with all the tiles which do not lie on the minimal path from $t_w$ to $t_s$ removed.  We leave these extra tiles in the figures in the hopes that it will make the proof of correctness clearer.  The east shifting subconfiguration is given by $C_{\alpha_e}$ and the west shifting subconfiguration is given by $C_{\alpha_w}$.

\begin{figure}[htp]
\centering
  \subfloat[][Positioning the configuration $\alpha_e$ relative to the configuration $\alpha$.]{%
        \label{fig:bit-writer-ae}%
        \includegraphics[width=2.4in]{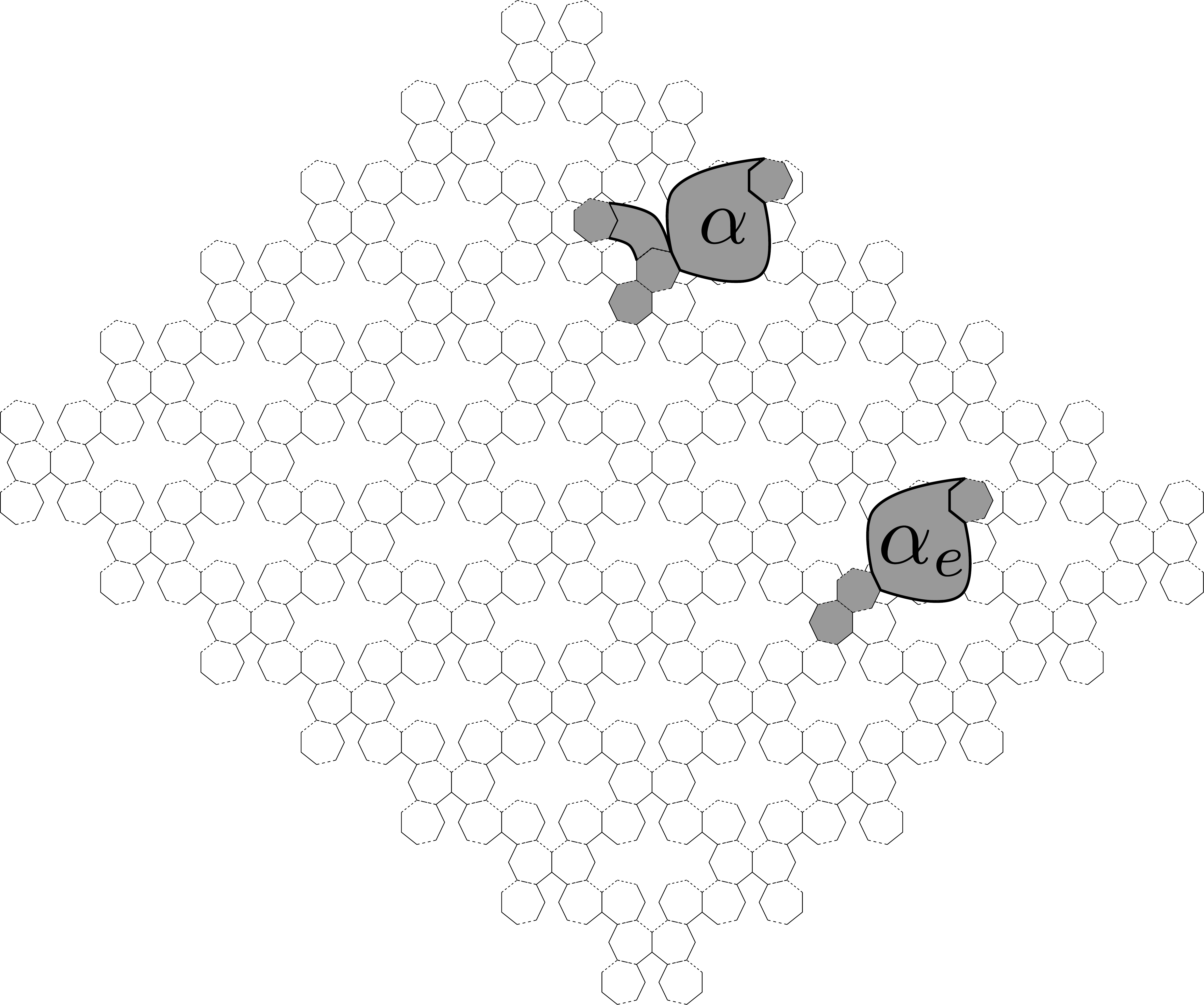}
        }%
        \quad
  \subfloat[][Positioning the configuration $\alpha_v$ relative to the configuration $\alpha$.]{%
        \label{fig:bit-writer-aw}%
        \includegraphics[width=2.4in]{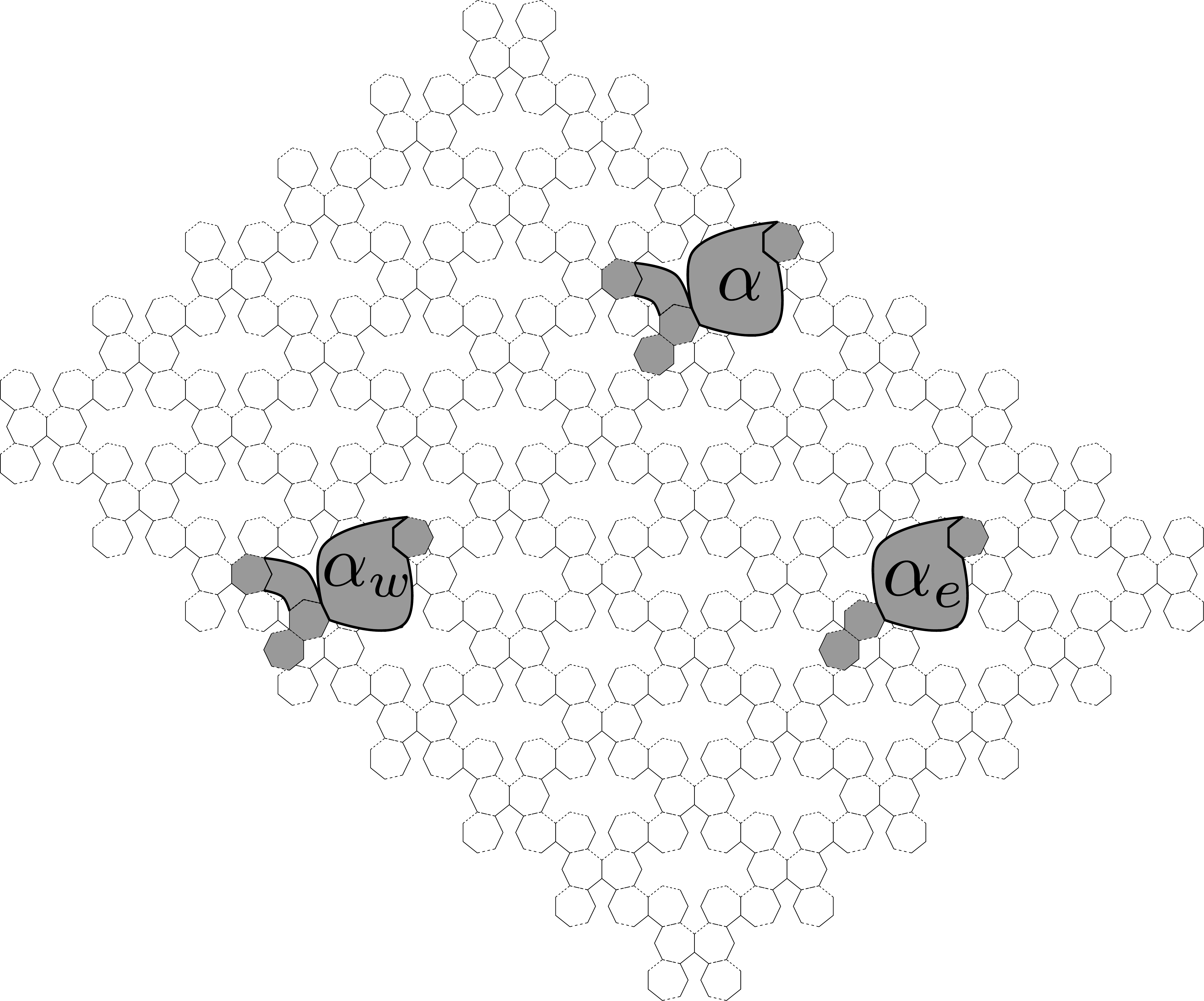}
        }%
  \caption{The three configurations and their positions relative to each other.}
  \label{fig:bit-writer-a}
\end{figure}

\subsection{Connecting the Bit-writer Subconfigurations}
Intuitively, we connect the blocker configuration to the east shifting configuration in the following way.  We shift the three configurations so that they remain in the same positions relative to each other, and the tile $t_s$ is $1$-centered.  Note that by construction, the tile $t_{se}$ lies at least 5 tiles to the right of $t_{s}$.  Thus, we can grow an almost straight line of tiles, which all lie on grid, until there is a tile which lies south of $t_s$ in the path.  Call this path $p_{se}$.  We then grow a path on grid from the last tile placed in $p_{se}$ that attaches to the southernmost side of the tile $t_s$.  An example of this can be found in Figure~\ref{fig:bit-writer-pe}.

Similarly, to attach the blocker configuration to the west shifting configuration, we first shift the two unconnected configurations (note there are now only two unconnected configurations now since the blocker configuration and the east shifting configuration are now attached) so that they remain in the same positions relative to each other, and the tile $t_w$ is $4$-centered on the grid.  Then we grow an on grid path of tiles from $t_n$ to the west (while keeping the path as straight as possible) until the path has tiles which lie to the west of $t_nw$ at which time the path turns (while still on grid) and grows south until it attaches to $t_nw$.  An example of this can be found in Figure~\ref{fig:bit-writer-pw}.

More formally, to connect the configurations $C_{\alpha}$ to $C_{\alpha_e}$ in the configuration $C''$, we grow a path in the following manner.  The first tile is placed with negated orientation and $3$-centered so that it completely shares a common edge with $t_s$.  We then grow a periodic path of tiles to the south with the tiles in the same positions and grid locations as the path of tiles in Figure~\ref{fig:bit-writer-pe}.  This repeats until a $1$-centered tile is placed so that it has the same imaginary part as tile $t_{se}$.  Once this occurs, we grow a periodic path of appropriately positioned tiles to the east in the same positions and grid locations as the path of east growing tiles in Figure~\ref{fig:bit-writer-pe}.  We do this until the $2$-centered tile completely shares an edge with the tile $t_{se}$ as shown in Figure~\ref{fig:bit-writer-pe}.  We call this configuration $C_e$.

To connect the configurations $C_{\alpha}$ to $C_{\alpha_w}$ in $C_e$, a path is grown from $C_{\alpha}$ to $C_{\alpha_w}$ as follows.  First, we shift the configuration $C_e$ so that the tile $t_w$ is $4$-centered.  Note that this also means the tile $t_ww$ is also $4$-centered.  We then grow a periodic path of $1$-centered, $2$-centered, $6$-centered, $4$-centered, $1$-centered, $3$-centered, $5$-centered, and $4$-centered tiles to the west (as shown in Figure~\ref{fig:bit-writer-pw}) until a $4$-centered tile is placed so that it has the same real part as the tile $t_{ww}$.  Once this occurs, we grow the periodic pattern south shown in Figure~\ref{fig:bit-writer-pw} until the $1$-centered tile in our path completely shares a common edge with $t_{ww}$.  Call this configuration $C'$.

\begin{figure}[htp]
\centering
  \subfloat[][Connecting $C_{\alpha_e}$ to $C_{\alpha}$.]{%
        \label{fig:bit-writer-pe}%
        \includegraphics[width=2.4in]{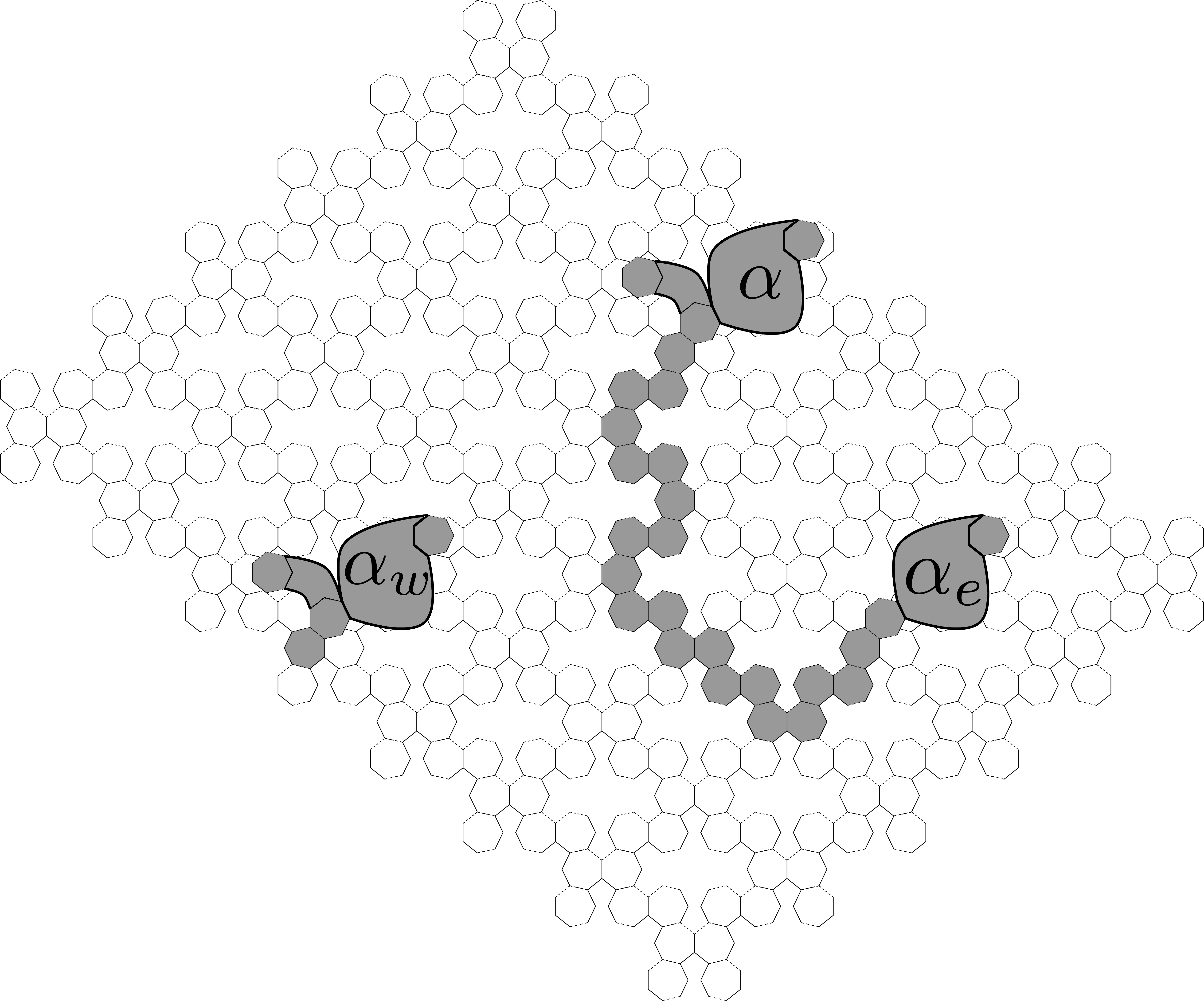}
        }%
        \quad
  \subfloat[][Connecting $C_{\alpha_w}$ to $C_{\alpha}$.]{%
        \label{fig:bit-writer-pw}%
        \includegraphics[width=2.4in]{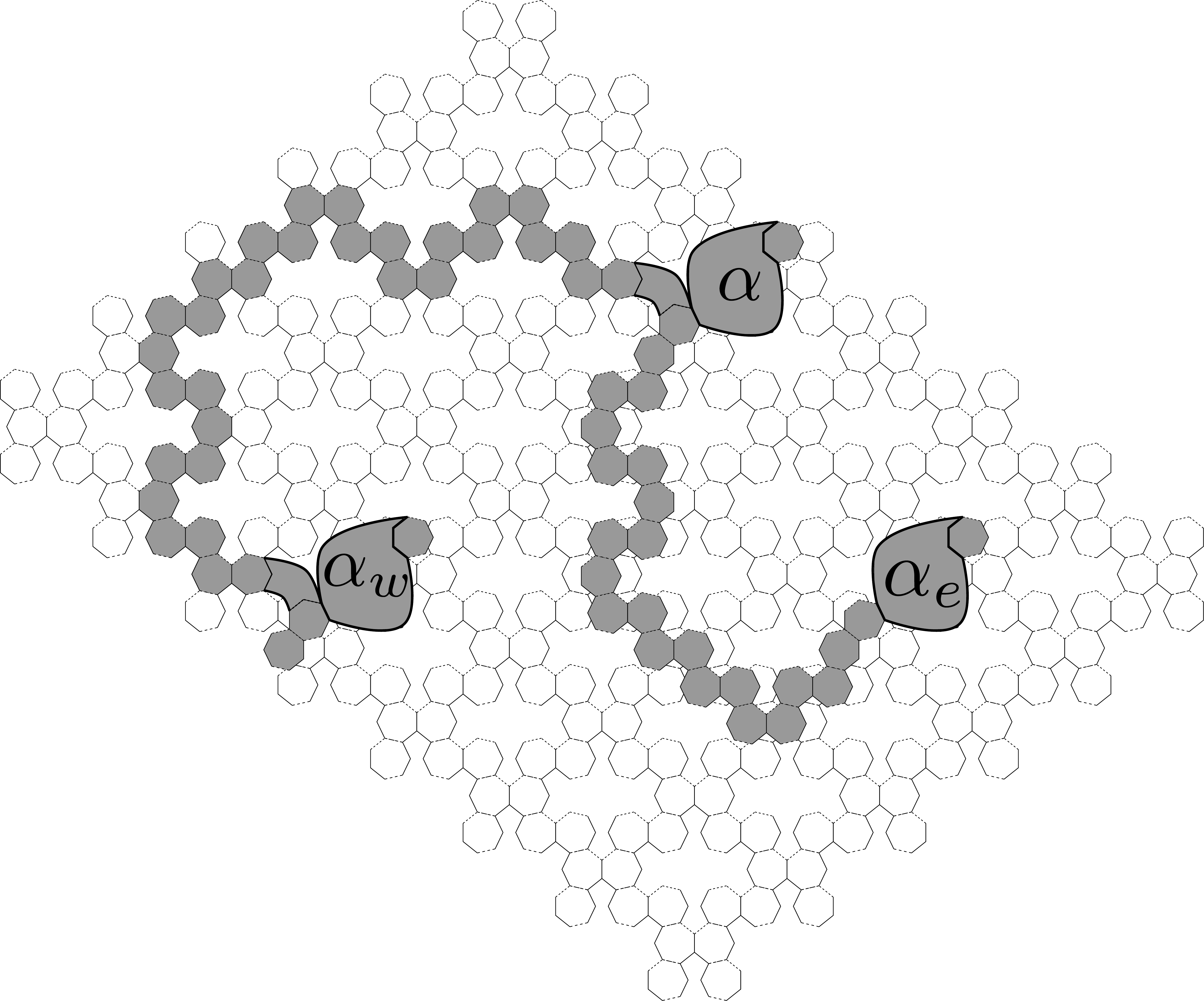}
        }%
  \caption{Connecting the configurations.}
  \label{fig:bit-writer-p}
\end{figure}

\begin{figure}[htp]
\centering
  \subfloat[][Positioning the configurations so that the $R$ tile is on grid.]{%
        \label{fig:bit-writer-complete0}%
        \includegraphics[width=2.4in]{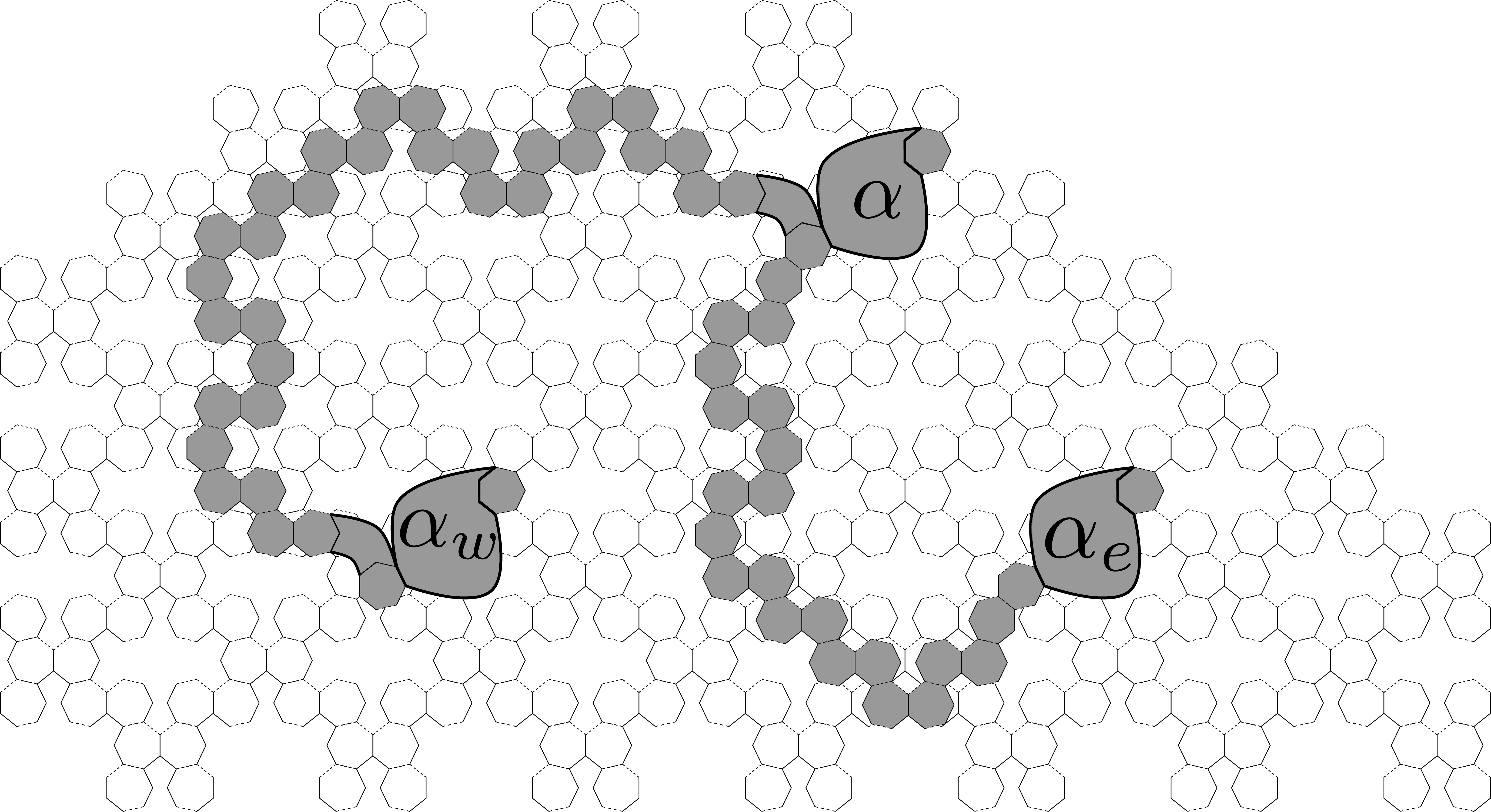}
        }%
        \quad
  \subfloat[][Growing the final paths of the bit writer configuration and removing the non-bit-writer portion of the configuration $\alpha$.  We refer to this configuration as the $\alpha_0$ bit-writer.]{%
        \label{fig:bit-writer-complete}%
        \includegraphics[width=2.4in]{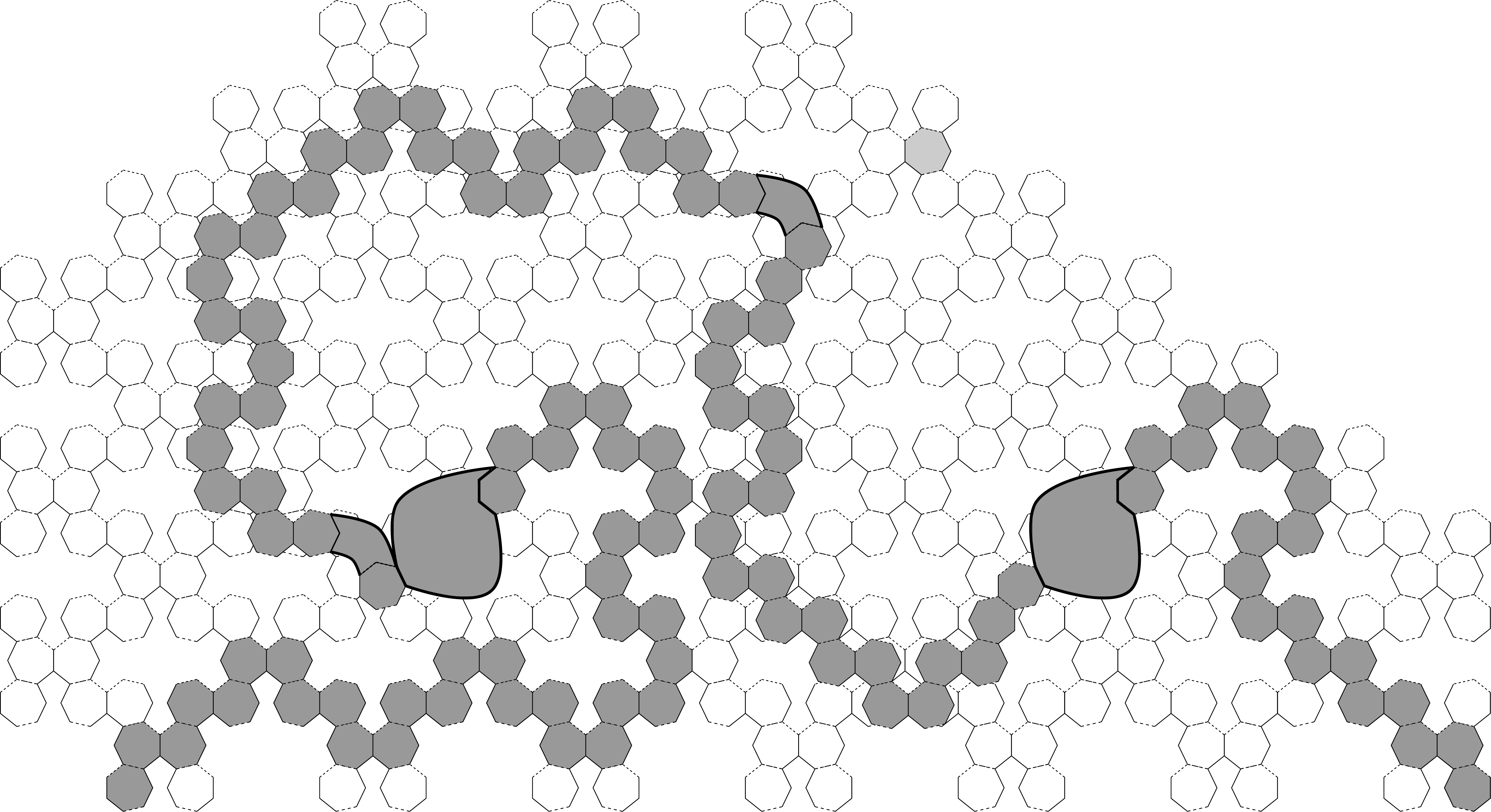}
        }%
  \caption{Completing the bit writer.}
  \label{fig:bit-writer-done}
\end{figure}

Now we describe how to grow out ``arms'' from the bit writer that are on grid which will allow the bit writers to connect to each other.  First, we translate the configuration $C'$ so that the tile $R$ in $C_{\alpha}$ is $4$ centered.  Note that this will imply the $R$ tiles in the configurations $C_{\alpha_e}$ and $C_{\alpha_w}$ are also $4$-centered as shown in Figure~\ref{fig:bit-writer-complete0}.  We then place a tile such that it has negated orientation, $3$-centered, and both the southernmost tile and east most tile in the configuration $C'$ (call this tile $t_{wb}$).  Likewise, we also place a tile such that it has standard orientation, it is $6$-centered, and both the southernmost tile and easternmost tile in the configuration $C'$ (call this tile $t_{eb}$).  Next, we place tiles on grid so that a path of tiles is formed from the $R$ tile in $C_{\alpha_w}$ to the $3$-centered tile placed above as shown in Figure~\ref{fig:bit-writer-complete}.  Similarly, we place tiles on grid so that a path of tiles is formed from the $R$ tile in $C_{\alpha_e}$ to the $6$-centered tile placed above.  Call this configuration $C$.

We construct $\alpha_1$ in a manner similar to our construction of $\alpha_0$.  The only difference in our construction of $\alpha_1$ will be that the configuration obtained from the bit reading gadget ``reading a 0'' will be used as our base configuration.

\subsection{Normalizing Bit-writers}
Now that we have constructed on grid bit reading gadgets, we can describe the construction of normalized bit-writers.

Construction of the normalized bit-writers begins by laying down the configurations $C_{\alpha_0}$ and $C_{\alpha_1}$ in the plane so that the tile labeled $R$ in each configuration (see above for the description of $R$) lies centered at the same point.  Next, we remove all tiles in the two configurations except for the tiles $t_{wb}$ and $t_{eb}$ in each bit writer.  We now place two new extremal tiles.  The first tile we place should be both the southernmost and westernmost tile in the configuration as well as $3$-centered.  Denote this tile $t_{\max w}$.  The location of the second tile's center should have the same imaginary part as the location of the center of the tile $t_{\max w}$.  In addition, this tile, which we denote $t_{\max e}$ should be the easternmost tile in the configuration.  See Figure~\ref{fig:au_st} for an example.

Now, we consider the configuration obtained above with all tiles contributed by $C_{\alpha_1}$ removed.  We place a connected path of tiles from the tile $t_{wb}$ to the tile $t_{\max w}$ as shown in Figure~\ref{fig:au_p0}.  Note that this path of tiles is such that none of the interior points of tiles overlap and every tile is connected to some other tile in the path by a completely shared edge. Similarly, we place tiles so as to form a path from the tile $t_{eb}$ to the tile $t_{\max w}$ which is also shown in Figure~\ref{fig:au_p0}.  These paths of tiles are then attached to configuration $C_{\alpha_0}$ in the same manner they are attached to the extremal tiles in the current configuration to form the configuration for the normalized $\alpha_0$ bit writer (shown in Figure~\ref{fig:au_a0}.

We also repeat this same process for the configuration obtained by considering the configuration in Figure~\ref{fig:au_st} with all tiles contributed by $C_{\alpha_0}$ removed which yields Figure~\ref{fig:au_p1}.  After ``copying and pasting'' these paths, we obtain the configuration for the normalized $\alpha_1$ bit writer which is shown in Figure~\ref{fig:au_a1}.

\begin{figure}[htp]
\centering
  \subfloat[][An example schematic of an $\alpha_0$ bit-writer.]{%
        \label{fig:au_a0_e}%
        \includegraphics[width=2.4in]{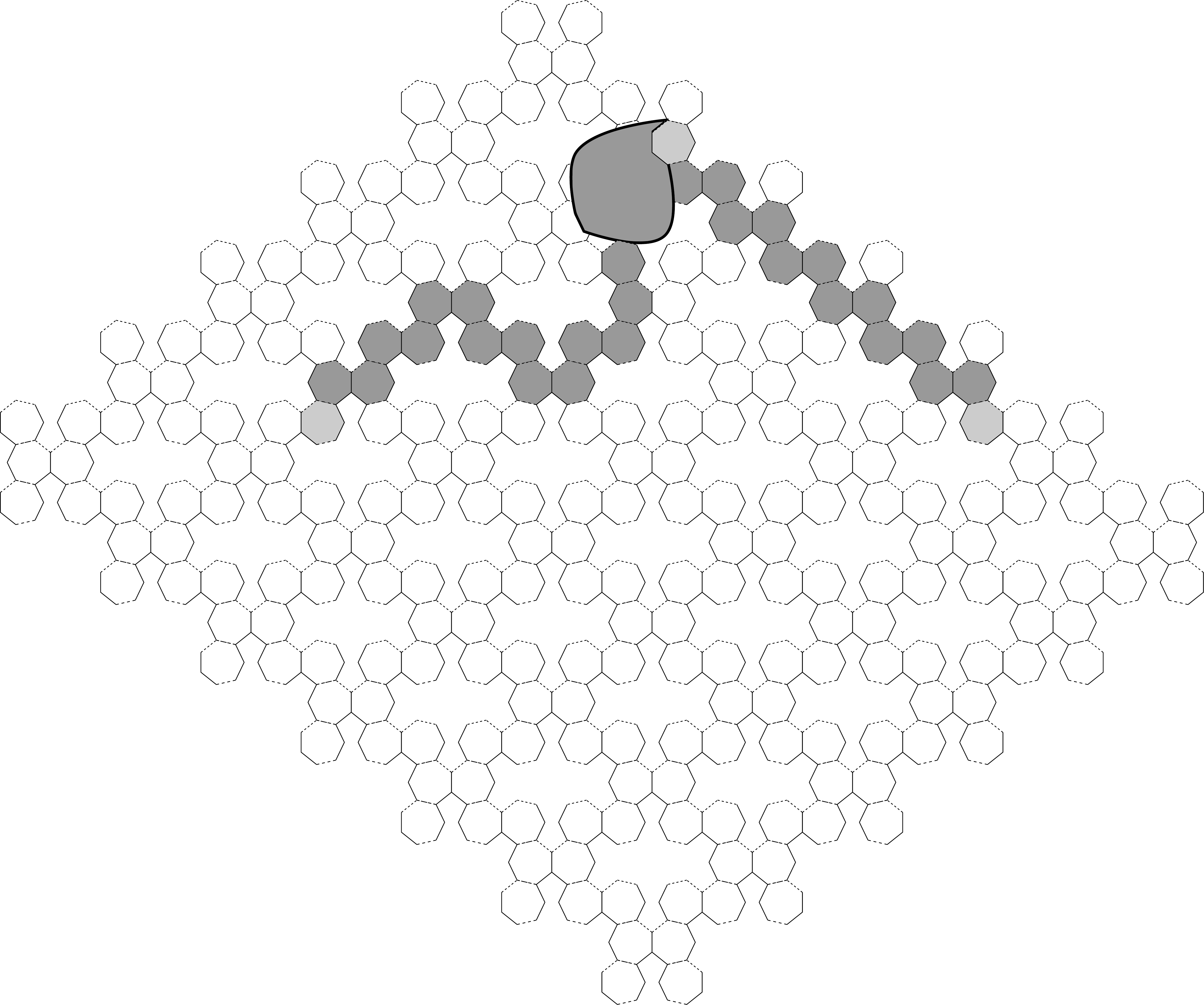}
        }%
        \quad
  \subfloat[][An example schematic of an $\alpha_1$ bit-writer.]{%
        \label{fig:au_a1_e}%
        \includegraphics[width=2.4in]{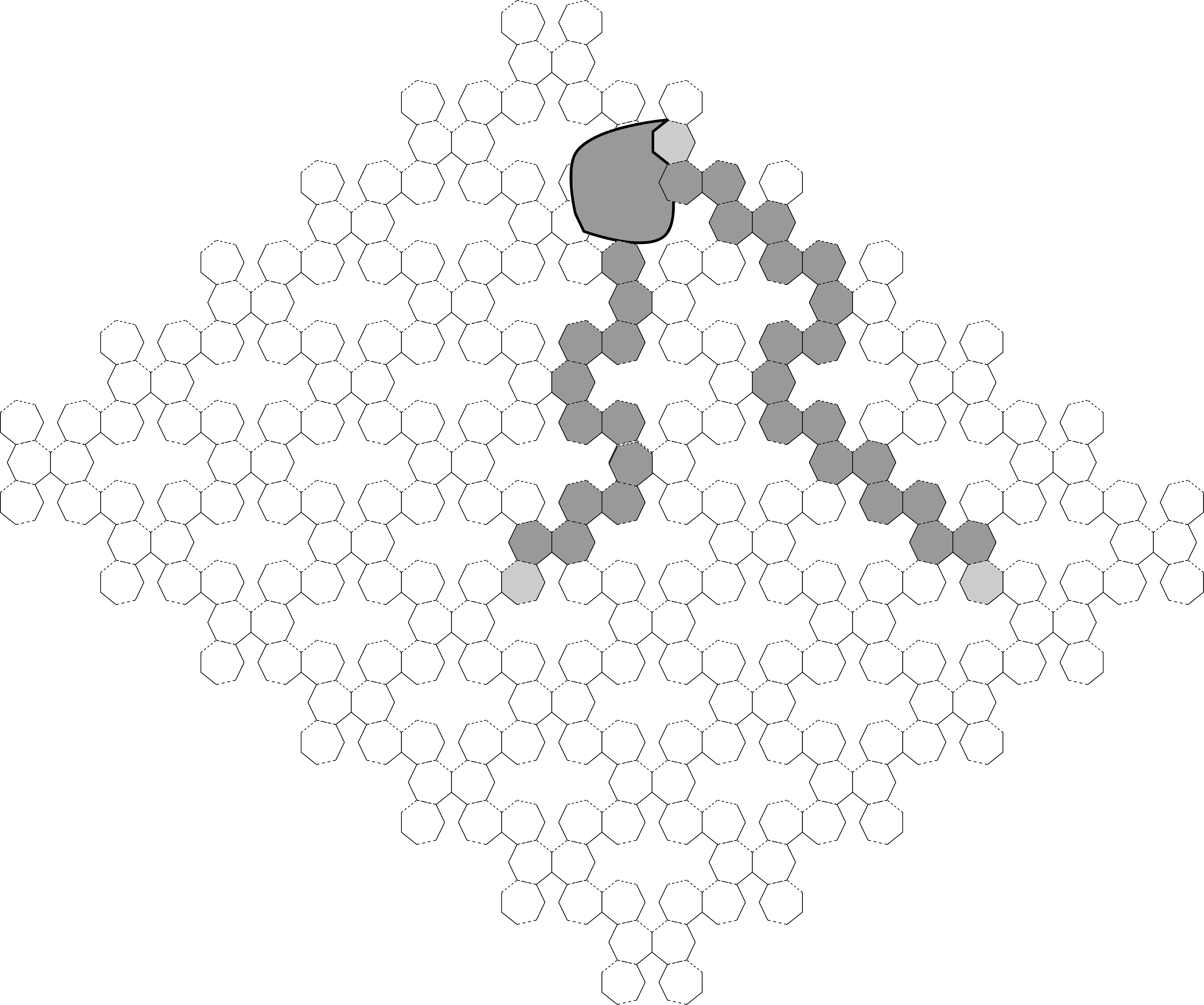}
        }%
  \caption{Completing the bit writer.}
  \label{fig:au_e}
\end{figure}

\begin{figure}[htp]
\begin{center}
\includegraphics[width=2.0in]{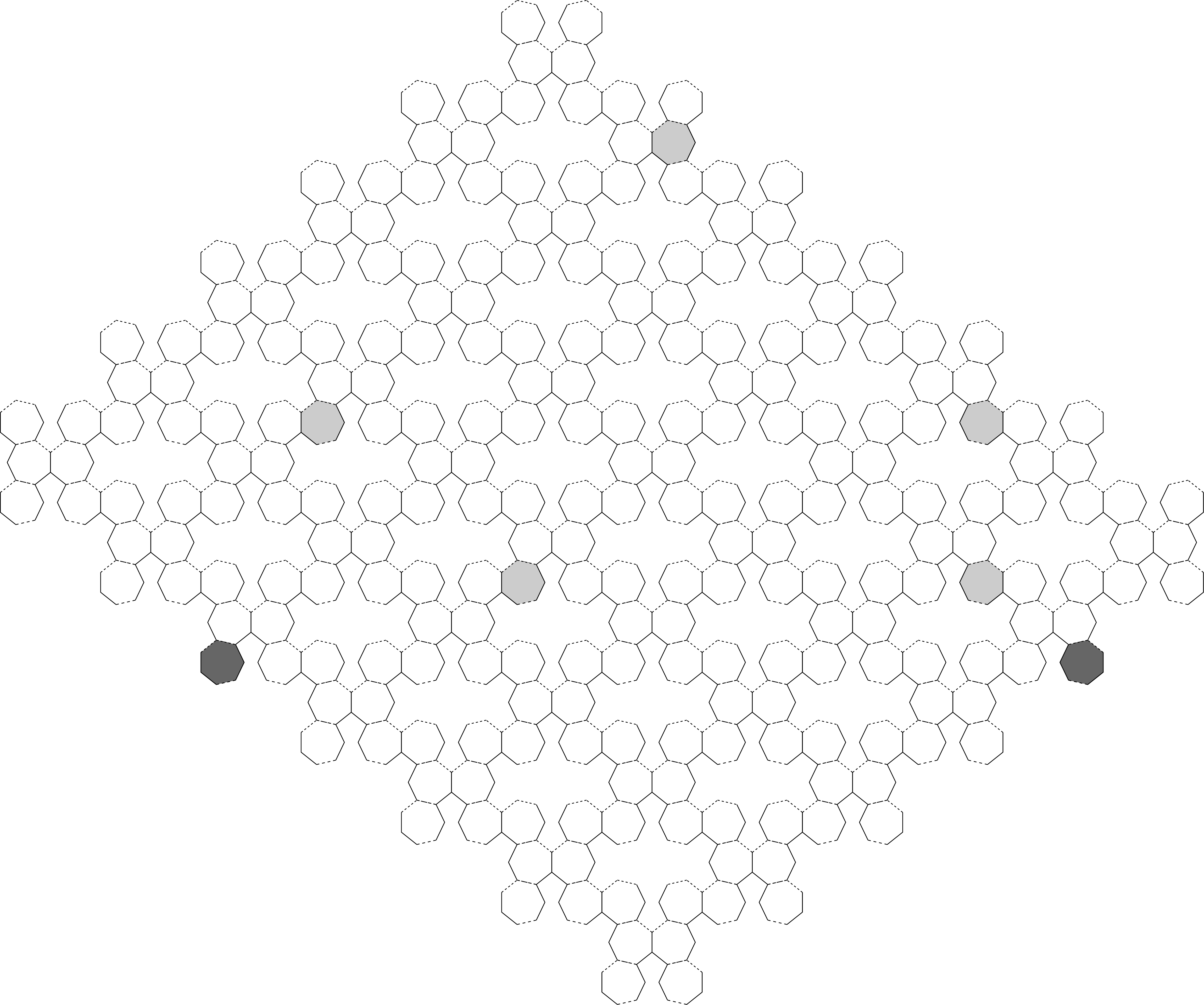}
\caption{The extremal points of the bit writer configurations (lightly shaded) and the newly created extremal points (darkly shaded).}
\label{fig:au_st}
\end{center}
\end{figure}

\begin{figure}[htp]
\centering
  \subfloat[][]{%
        \label{fig:au_p0}%
        \includegraphics[width=2.4in]{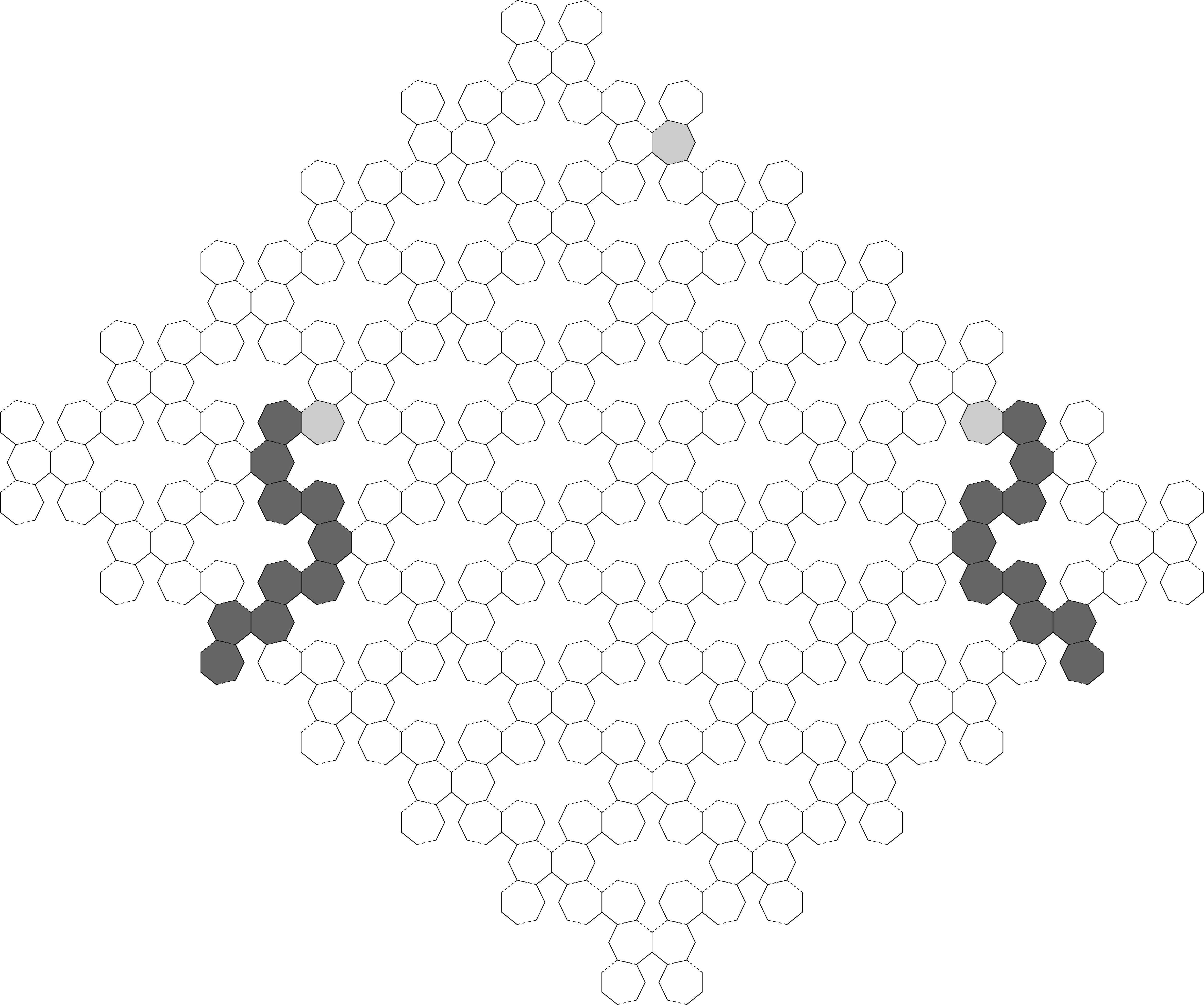}
        }%
        \quad
  \subfloat[][]{%
        \label{fig:au_p1}%
        \includegraphics[width=2.4in]{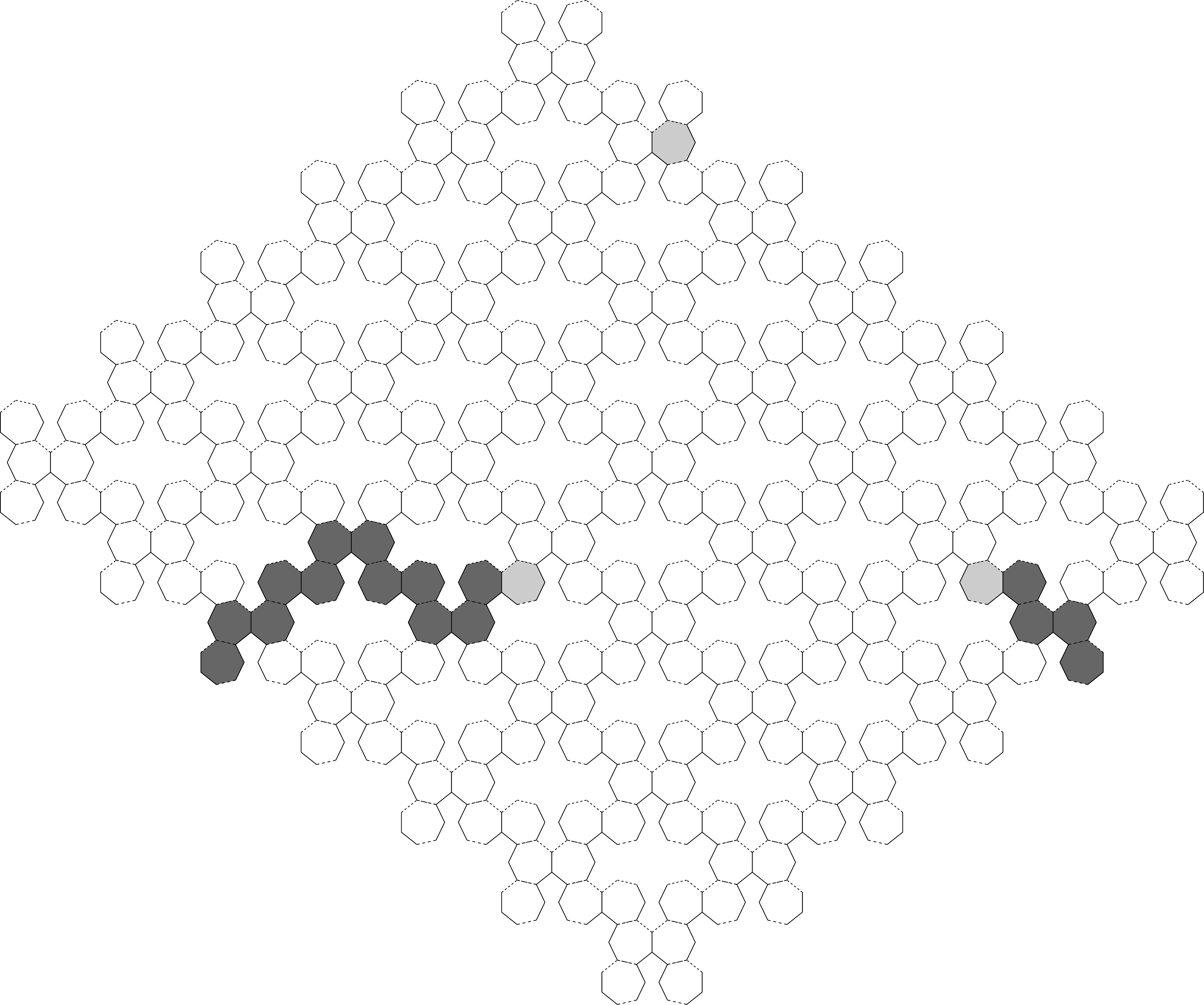}
        }%
  \caption{Growing a path of tiles from the old extremal points to the new ones.}
  \label{fig:au_p}
\end{figure}

\begin{figure}[htp]
\centering
  \subfloat[][]{%
        \label{fig:au_a0}%
        \includegraphics[width=2.4in]{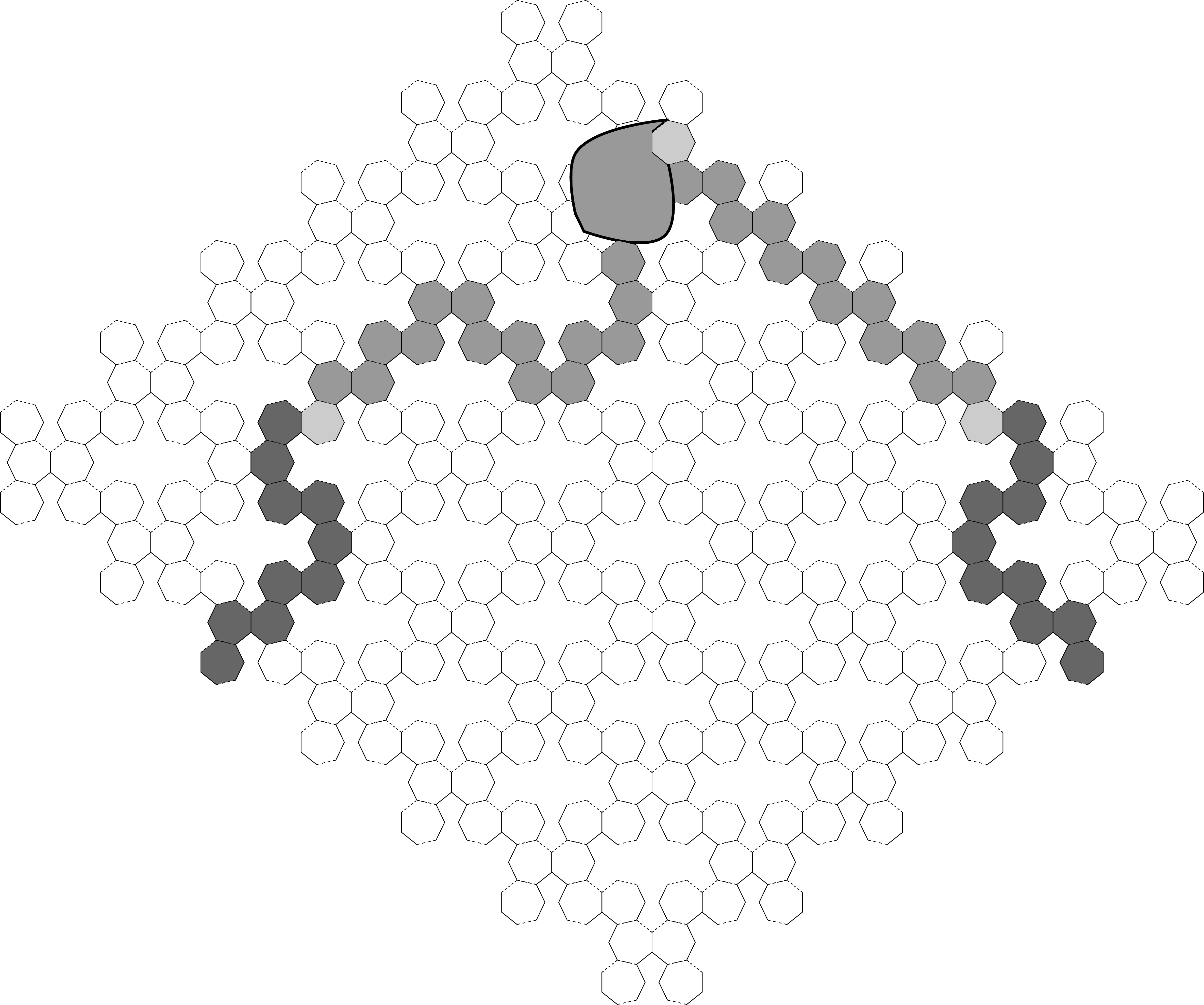}
        }%
        \quad
  \subfloat[][]{%
        \label{fig:au_a1}%
        \includegraphics[width=2.4in]{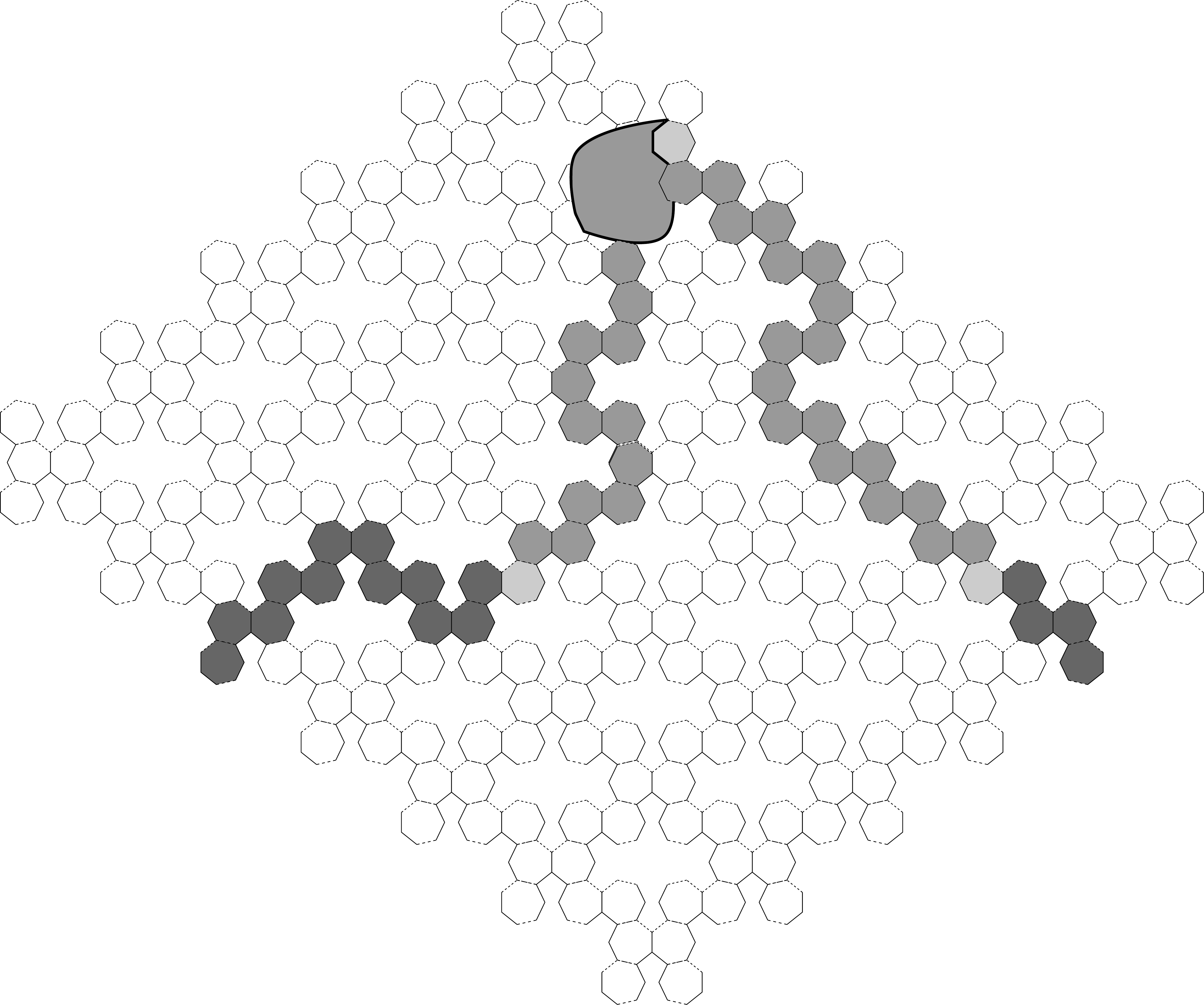}
        }%
  \caption{The normalized bit writers.}
  \label{fig:au_a}
\end{figure}
\subsection{Shifting on Grid after the Read}
Call the last tile placed by the bit reader $T_1$.  We now describe how the bit reader shifts back on grid after ``reading a bit''.  This part of the construction is very similar to the construction of the on grid bit writers in Section~\ref{sec:bitwon}.  For our shifting configuration we will use the configuration obtained by removing all tiles in the bit reading configuration except for the tiles that lie on the path from the tile $T_1$ to the tile $R$ (where tiles $T_1$ and $R$ are as described above in Section~\ref{sec:bit_reading}.  Without loss of generality, we assume that $T_1$ has standard orientation since if it is not, we can add one more tile to the path so that the last tile placed in the bit reader path is in standard orientation. We then construct an on grid bit reader in a manner similar to the way the on grid bit writers were constructed in Section~\ref{sec:bitwon}.

\subsection{Proof of Correctness}
To see that the configurations above, are indeed on grid bit reading gadgets (and thus assemblies) we make three claims: 1) every tile in the configuration completely shares an edge with another tile in the configuration and the configuration is connected, 2) the interiors of all the tiles in the configurations are pairwise disjoint, and 3) the beginning and end tiles are on grid as well as the $R$ tile.  After we see that these claims are true, then we can easily give a system which contains a bit reading gadget.

The first claim follows immediately from our construction.  The construction ensured that every tile placed was next to a pre-existing tile in the assembly and in the proper orientation.  The second claim also follows from the construction since we were careful to place subconfigurations sufficiently far away from each other so that there is no overlap and paths can be grown between them without overlapping.

To see claim 3, observe that the $R$ tiles in all of the subconfigurations lie in the same position relative to some polyform on the grid (see Figure~\ref{fig:bit-writer-complete0}.  Consequently, once we connect the subconfigurations and shift $R$ so that it is on grid, all of the $R$ tiles in the subconfigurations are on grid.  Thus when the ``arms'' of the bit writer are grown, they start and end on grid with respect to the tile $R$.  Hence, the beginning and end tiles are on grid as well as the $R$ tile.

To see that we can create a system which contains a bit reading gadget using our normalized bit writers, note that we can grow a path of tiles from the last tile placed in the normalized bit writer so that it starts the growth of a bit reader at an appropriate position in relation to the bit writer.  Using this notion, Figure~\ref{fig:complete_gadget} shows an example schematic of the complete bit reading gadget which results from reading a particular bit.  The system shown can be constructed by placing appropriate glues on the tiles so that they come together as shown.

\begin{figure}[htp]
\begin{center}
\includegraphics[width=2.0in]{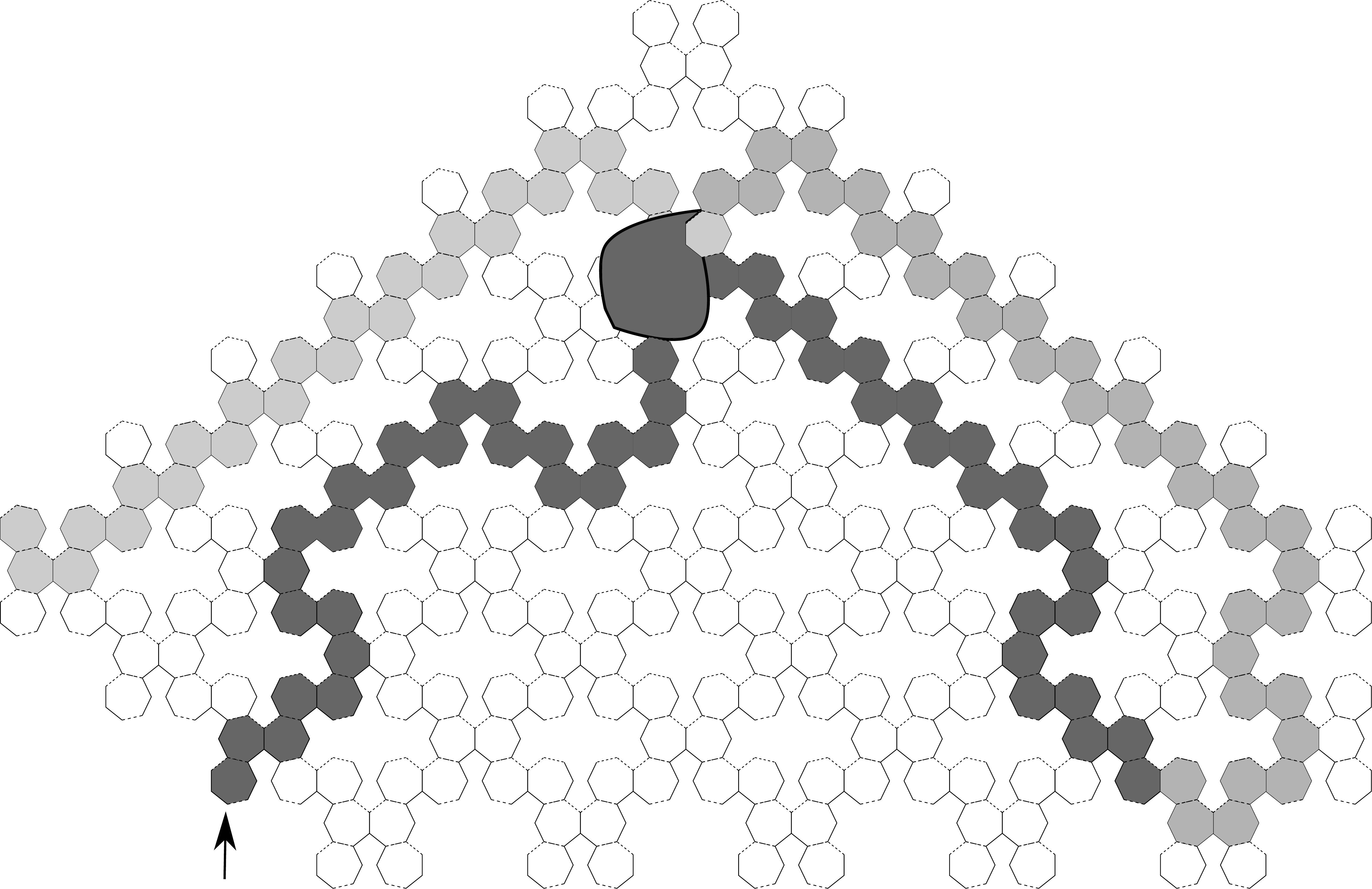}
\caption{The complete bit reading gadget reading a particular bit.  The arrow in this picture points to the seed of the system.}
\label{fig:complete_gadget}
\end{center}
\end{figure}
} %
\fi

\ifabstract
\later{
\section{Technical Appendix}\label{sec:technical-appendix}

In the following sections we will use this technique for computing the positions of the centers of polygonal tiles in order to show that the bit gadgets that we construct are indeed valid bit gadgets.

\subsection{Systems with tiles shaped like a single regular polygon}\label{sec:technical-single-regular}

In this section, we present the bit-reading gadgets for tiles shaped like a single regular polygon and the relevant calculation to show that these bit-reading gadgets are valid. Throughout this section we will use complex number to analyze configurations of polygonal tiles. This idea is presented in Section~\ref{sec:roots-of-unity}.
Many of these calculations rely on well known properties of complex numbers and regular polygons. In particular, for a complex number $z$, we use the equations $2\Re(z) = z + z^{-1}$ and $2\Im(z) = z - z^{-1}$. We also apply Euler's identity ($e^{i\theta} = \cos\left( \theta \right) + i\sin\left( \theta \right)$) when needed. Moreover, for a regular polygon $P_n$ with $n$ sides and apothem $.5$ (which we assume for all of the regular polygons considered here), the diameter $d_n$ of $P_n$ is given by the following equation which will often be used to show that two polygons do not overlap.. $$d_n = \frac{1}{2\cos\left( \frac{\pi}{n} \right)}$$

\subsubsection{A bit gadget for heptagonal tiles}\label{sec:technical-heptagonal}

Now that we have a means of computing the exact positions of the centers of polygonal tiles, we give a bit-reader gadget that works for heptagonal tiles at temperature-1. Figure~\ref{fig:technical-heptagonBitReader} gives a depiction of this bit-reader. Given this bit-reader, the burden of proof is two fold. 1) We must calculate the distances of the tiles in the bit-reader assembly in order to show that when a $1$ is read, a $0$ cannot be read and vice versa, and 2) we must show that these gadgets assemble in regular (grid-like) positions. In this section we will handle the first burden of proof, and show the second burden in Section~\ref{sec:technicalLemmas}.

\begin{figure}[htp]
\centering
  \subfloat[][A $0$ is read, and a $1$ cannot be read by mistake since the tile $B$ prevents a heptagonal tile from attaching via the glue labeled $g_1$.]{%
        \label{fig:technical-heptagonBitReaderA}%
        \includegraphics[width=2.5in]{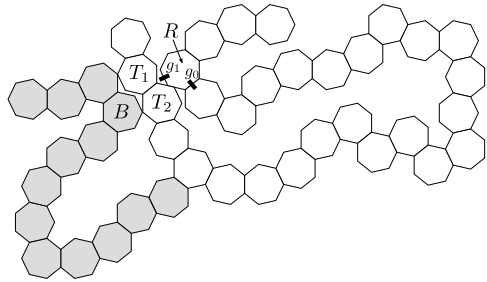}
        }%
        \quad
  \subfloat[][A $1$ is read. This time a $0$ cannot be read by mistake since the tile $B$ prevents growth of a path of heptagonal tiles that attach via the glue labeled $g_0$. Note that some of this path may form, but $B$ prevents the entire path from assembling, and thus prevents a $0$ from being read.]{%
        \label{fig:technical-heptagonBitReaderB}%
        \includegraphics[width=2.5in]{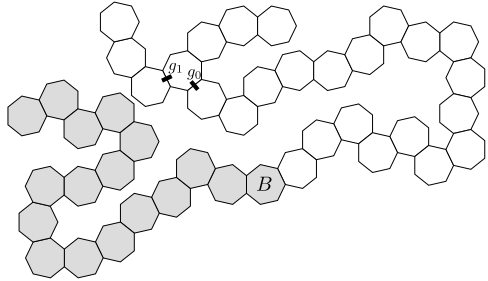}
        }%
  \caption{A connected bit-gadget consisting of heptagonal tiles.}
  \label{fig:technical-heptagonBitReader}
\end{figure}

In Figure~\ref{fig:technical-heptagonBitReader}, the gray tiles represent a ``written'' bit (either 0 or 1), while the white tiles are ``reading'' this bit. We ensure that the assembly sequence of a bit-gadget is such that all of the gray tiles bind before any white tiles. Referring to Figure~\ref{fig:technical-heptagonBitReaderA}, we will first show that the tile labeled $R$ in does not prevent the binding of the tile labeled $T_1$ or the tile labeled $T_2$.

Let $s$ denote the center of the tile $R$, $c_1$ denote the center of $T_1$, and $c_2$ denote the center of $T_2$. This is depicted in Figure~\ref{fig:technical-heptagonBitReader1}. Then, to calculate $c_1$ and $c_2$ relative to $s$, we assume that $R$ is in standard orientation. Following the path of tiles lying on the dotted line in Figure~\ref{fig:technical-heptagonBitReader1} and summing the appropriate roots of unity, we obtain the polynomials $c_1 = \omega^6 - \omega^3 + \omega - \omega^4 + 1 -\omega^4 + \omega - \omega^3 + 1 - \omega^2 + \omega^5 - \omega^2 + \omega^4 - \omega^6 + \omega^4 - \omega^6 + \omega^4 - \omega + \omega^4 - 1 + \omega^3 - \omega^6 + \omega^2$ and $c_2 = c_1 - \omega^6$. By simplifying $c_1$, we get $c_1 = 1+ \omega - \omega^2 - \omega^3  + 2\omega^4 + \omega^5 - 2\omega^6$.

\begin{figure}[htp]
\centering
	\includegraphics[width=3in]{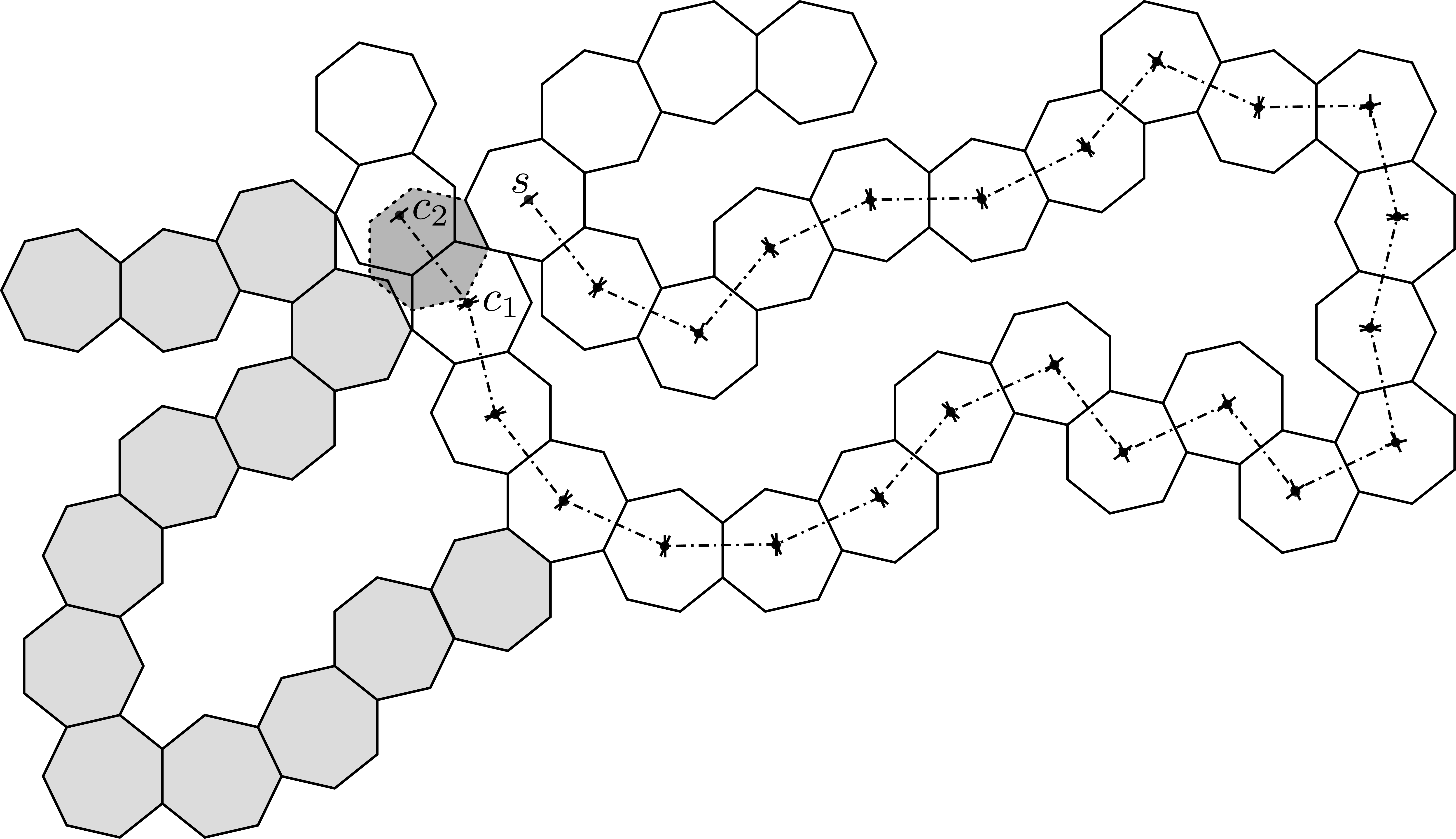}
	\caption{A possible configuration of the bit-reader given in Figure~\ref{fig:technical-heptagonBitReader}. We must show that the heptagonal tiles centered at $c_1$ and $c_2$ do not overlap the tile centered at $s$.}
	\label{fig:technical-heptagonBitReader1}
\end{figure}

First, as multiplying by $\omega$ is a rotation by $2\pi/7$, it is enough to show that $\Re(\omega^2 c_1) \geq 1$, and to see this, consider the following.

\begin{eqnarray*}
\omega^2c_1 &=& \omega^2 + \omega^3 - \omega^4 - \omega^5  + 2\omega^6 + \omega^7 - 2\omega^8\\
            &=& \omega^2 + \omega^3 - \omega^4 - \omega^5  + 2\omega^6 + 1 - 2\omega\\
            &=& \omega^2 + \omega^3 - \omega^{-3} - \omega^{-2}  + 2\omega^6 + 1 - 2\omega^{-6}\\
            &=& 1 + (\omega^2 - \omega^{-2}) + (\omega^3 - \omega^{-3})  + 2(\omega^6 - \omega^{-6})
\end{eqnarray*}

\noindent Then since $(\omega^2 - \omega^{-2})$, $(\omega^3 - \omega^{-3})$, and $2(\omega^6 - \omega^{-6})$ are imaginary, we see $\Re(\omega^2c_1) = 1$. Therefore, the heptagon with negated orientation centered at $c_1$ and the heptagon in standard orientation centered at $s$ do not overlap. Note that since $\Re(\omega^2c_1) = 1$, it may be that these two heptagons partially share an edge, however, the intersection of their interiors is empty.

To show that the heptagon, $H_{c_2}$, with standard orientation centered at $c_2$ and the heptagon, $H_{s}$, with negated orientation centered at $s$ do not overlap, note that $c_2 = 1 + \omega - \omega^2 - \omega^3  + 2\omega^4 + \omega^5 - 3\omega^6$. Then, one can approximate $|c_2|$ and observe that $|c_2| > 1.11 > \frac{1}{\cos\left( \frac{\pi}{7} \right)}$. Hence the distance from $s$ to $c_2$ is greater than twice the diameter of one of these heptagonal tiles. Therefore, $H_{c_2}$ and $H_{s}$ do not overlap.

\begin{figure}[htp]
\centering
	\includegraphics[width=3in]{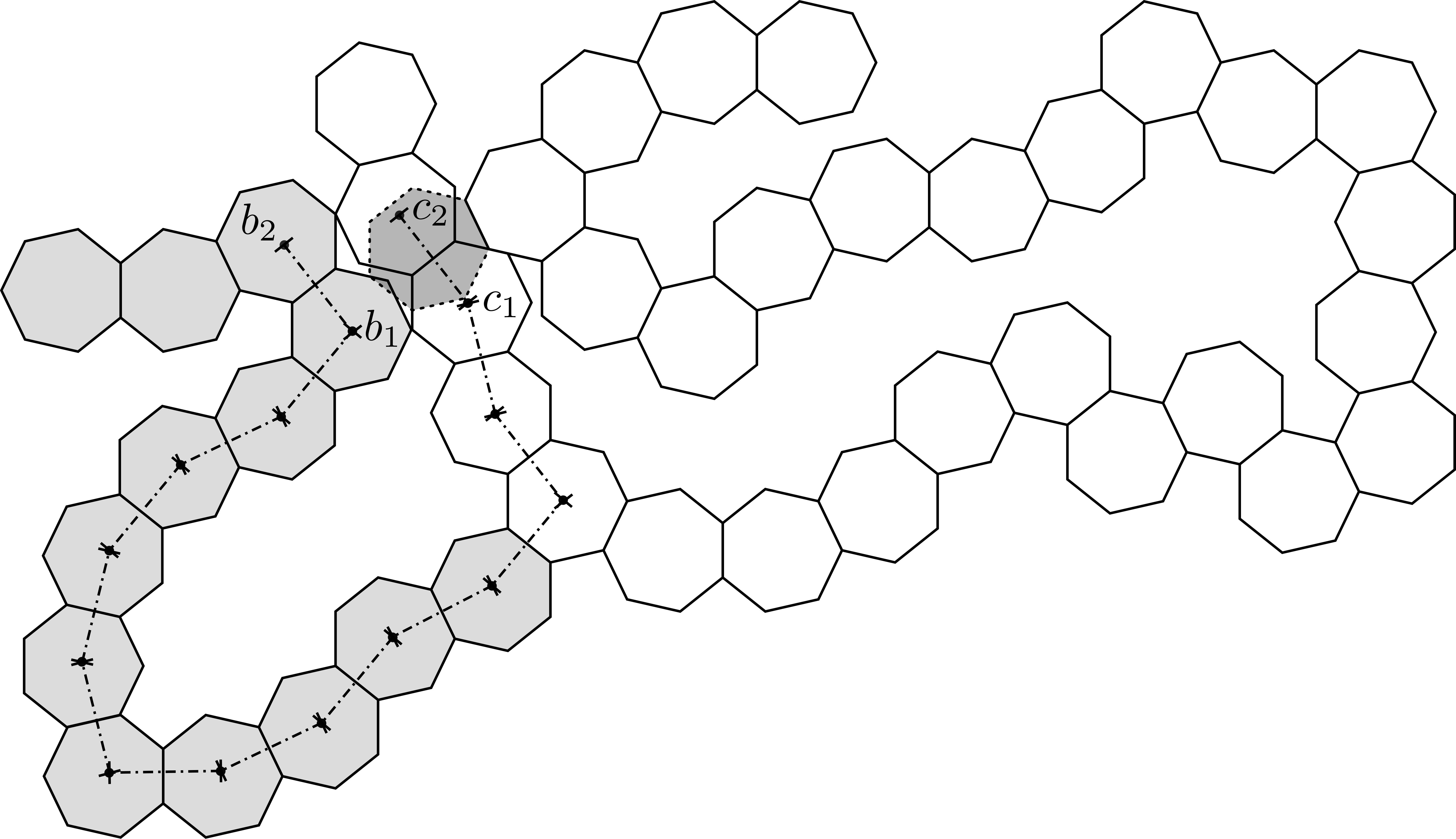}
	\caption{A possible configuration of the bit-reader given in Figure~\ref{fig:technical-heptagonBitReader}. We must show that the heptagonal tiles centered at $c_1$ and $c_2$ do not overlap those centered at $b_1$ and $b_2$.}
	\label{fig:technical-heptagonBitReader2}
\end{figure}

Referring to Figure~\ref{fig:technical-heptagonBitReader2}, relative to $b_1$, $c_1 = -\omega + \omega^4 - \omega + \omega^5 - \omega^2 + 1 - \omega^4 + \omega - \omega^4 + \omega - \omega^6 + \omega^2$, and $c_2 = c_1 - \omega^6$. Simplifying,
$c_1 = 1 - \omega^4 + \omega^5 - \omega^6$.

To show that a heptagonal tile in negated orientation centered at $b_1$ and a heptagonal tile in negated orientation centered at $c_1$ do not overlap, it suffices to show that $\Re(c_1) >= \frac{1}{2} + \frac{1}{2\cos\left(\frac{\pi}{7}\right)}$. Hence, it is enough to show that $\Re(c_1) = 1-\cos\left(\tfrac{8\pi}{7}\right) + \cos\left(\tfrac{10\pi}{7}\right) - \cos\left(\tfrac{12\pi}{7}\right) = \frac{1}{2} + \frac{1}{2\cos\left(\frac{\pi}{7}\right)}$. Equivalently, we show $1-2\cos\left(\tfrac{8\pi}{7} \right) + 2\cos\left(\tfrac{10\pi}{7}\right) - 2\cos\left( \tfrac{12\pi}{7} \right) = \frac{1}{\cos\left(\frac{\pi}{7}\right)}$. To see this, observe

\begin{eqnarray*}
2 &=& -2\cos\left(\pi\right)\\
  &=& -\left( e^{\pi i} + e^{-\pi i} \right)\\
  &=& - e^{\pi i} - e^{-\pi i} + e^{\frac{\pi i}{7}} - e^{\frac{\pi i}{7}} + e^{\frac{-\pi i}{7}} - e^{\frac{-\pi i}{7}}\\
  &=& e^{\frac{\pi i}{7}} - e^{-\pi i} - e^{\frac{-\pi i}{7}} + e^{\frac{-\pi i}{7}} - e^{\pi i} - e^{\frac{\pi i}{7}}\\
  &=& e^{\frac{\pi i}{7}} - e^{-\pi i} - e^{\frac{13\pi i}{7}} + e^{\frac{-\pi i}{7}} - e^{\pi i} - e^{-13\frac{\pi i}{7}}\\
  &=& e^{\frac{\pi i}{7}} - e^{\frac{9\pi i}{7}} - e^{\frac{-7\pi i}{7}} + e^{\frac{11\pi i}{7}} + e^{\frac{-9\pi i}{7}} - e^{\frac{13\pi i}{7}} - e^{\frac{-11\pi i}{7}}\\ &+&  e^{\frac{-\pi i}{7}} - e^{\frac{7\pi i}{7}} - e^{\frac{-9\pi i}{7}} + e^{\frac{9\pi i}{7}} + e^{\frac{-11\pi i}{7}} - e^{\frac{11\pi i}{7}} - e^{-13\frac{\pi i}{7}}\\
 &=& \left(e^{\frac{\pi i}{7}} + e^{\frac{-\pi i}{7}}\right)\left( 1 - e^{\frac{8\pi i}{7}} - e^{\frac{-8\pi i}{7}} + e^{\frac{10\pi i}{7}} + e^{\frac{-10\pi i}{7}} - e^{\frac{12\pi i}{7}} - e^{\frac{-12\pi i}{7}} \right)
\end{eqnarray*}

\noindent The last equality gives
$$1 - e^{\frac{8\pi i}{7}} - e^{\frac{-8\pi i}{7}} + e^{\frac{10\pi i}{7}} + e^{\frac{-10\pi i}{7}} - e^{\frac{12\pi i}{7}} - e^{\frac{-12\pi i}{7}} = \frac{2}{e^{\frac{\pi i}{7}} + e^{\frac{-\pi i}{7}}}$$

In other words,
$1 - 2\cos\left( \tfrac{8\pi}{7} \right) + 2\cos\left(\tfrac{10\pi}{7} \right) - 2\cos\left(\tfrac{12\pi}{7} \right) = \frac{1}{\cos\left(\frac{\pi}{7} \right)},$ which was what we wanted. Therefore, a heptagonal tile in negated orientation and centered $b_1$, and a heptagonal tile in negated orientation and centered at $c_1$ do not overlap.

From Figure~\ref{fig:technical-heptagonBitReader2}, it is now clear that a heptagonal tile in negated orientation and centered at $b_1$, and a heptagonal tile in standard orientation and centered at $c_2$ do not overlap, and that a heptagonal tile in standard orientation and centered $b_2$, and a heptagonal tile in standard orientation and centered at $c_2$ do not overlap.

\begin{figure}[htp]
\centering
	\includegraphics[width=3in]{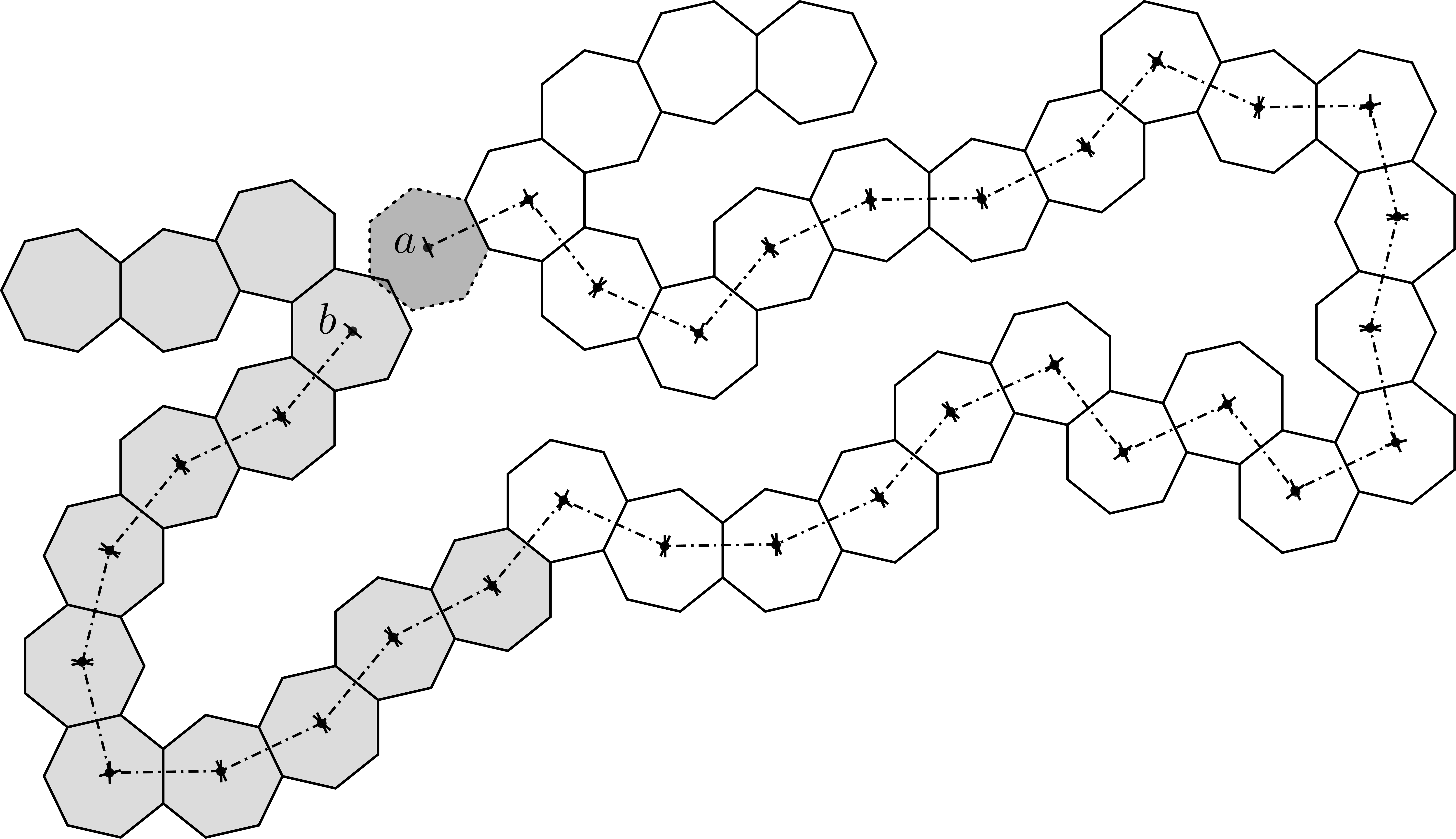}
	\caption{A possible configuration of the bit-reader given in Figure~\ref{fig:technical-heptagonBitReader}. We must show that the heptagonal tile centered at $a$ overlaps the tile centered at $b$.}
	\label{fig:technical-heptagonBitReader3}
\end{figure}

Now, referring to Figure~\ref{fig:technical-heptagonBitReader3}, we must show that a heptagonal tile, $H_a$, in negated orientation and centered $a$, and a heptagonal tile, $H_b$, in negated orientation and centered at $b$ overlap.
Note that relative to $a$, $b = -\omega^4 + \omega^6 - \omega^3 + \omega - \omega^4 + 1 -\omega^4 + \omega - \omega^3 + 1 - \omega^2 + \omega^5 - \omega^2 + \omega^4 - \omega^6 + \omega^4 - \omega^6 + \omega^4 - \omega + \omega^4 - 1 + \omega^3 - \omega + \omega^4 - \omega + \omega^4 - 1 + \omega^2 - \omega^5 + \omega - \omega^4 + \omega$. We can simplify $b$ to obtain $b = \omega - \omega^2 - \omega^3 + 2\omega^4 - \omega^6$. Then we approximate $|b|$ to show that $|b| < 1$. Therefore $H_a$ and $H_b$ must overlap. Given these calculations, we can obtain a bit-reading gadget for systems whose tiles have the shape of a heptagon.

\subsubsection{Octagonal Tile Assembly}\label{sec:technical-octagonal}

\begin{figure}[htp]
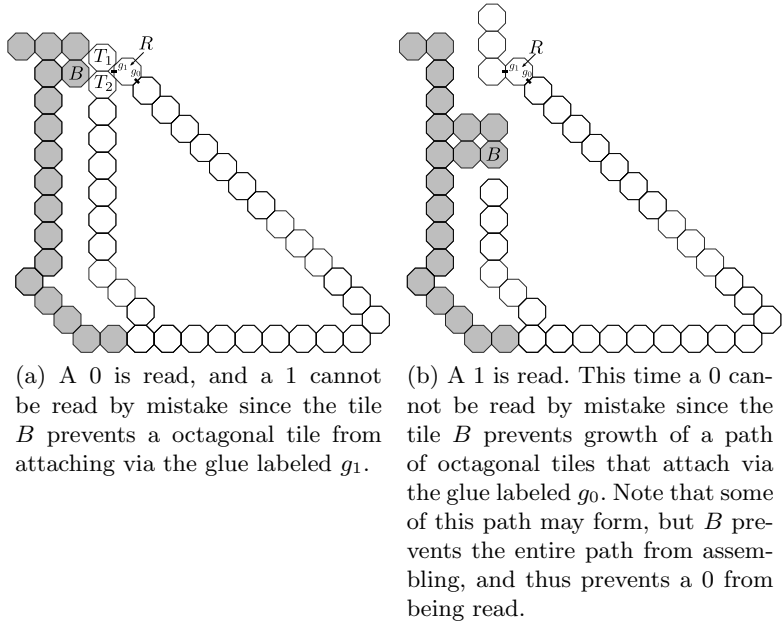

\centering
  \subfloat[][A $0$ is read, and a $1$ cannot be read by mistake since the tile $B$ prevents a octagonal tile from attaching via the glue labeled $g_1$.]{%
        \label{fig:technical-octagonBitReaderA}%
        \makebox[.4\textwidth]{
        \includegraphics[width=2in]{images/octagonBitReaderA}
        }}%
        \quad
  \subfloat[][A $1$ is read. This time a $0$ cannot be read by mistake since the tile $B$ prevents growth of a path of octagonal tiles that attach via the glue labeled $g_0$. Note that some of this path may form, but $B$ prevents the entire path from assembling, and thus prevents a $0$ from being read.]{%
        \label{fig:technical-octagonBitReaderB}%
        \makebox[.4\textwidth]{
        \includegraphics[width=2in]{images/octagonBitReaderB}
        }}%
  \caption{The configurations for a bit-reading gadget consisting of octagonal tiles.}
  \label{fig:technical-octagonBitReader}
\end{figure}

Figure~\ref{fig:technical-octagonBitReader} depicts two possible configurations of a bit-reading gadget construction for single-shaped systems with octagonal tiles, the gray tiles represent ``bit-writer'' tiles (representing either $0$ or $1$), while the white tiles are the ``bit-reader'' tiles. We ensure that the assembly sequence of a bit-gadget is such that all of the gray tiles bind before any white tiles. Referring to Figure~\ref{fig:technical-octagonBitReaderA}, we will first show that the tiles labeled $R$ and $B$ do not prevent the binding of the tile labeled $T_1$ or the tile labeled $T_2$. Then we will show that the tile labeled $B$ prevents an octagonal tile from binding to the glue labeled $g_1$.

\begin{figure}[htp]
\centering
	\includegraphics[width=1.5in]{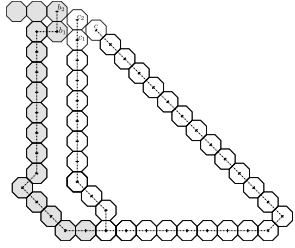}
	\caption{A possible configuration of the bit-reader given in Figure~\ref{fig:technical-octagonBitReader}. We must show that the octagonal tiles centered at $c_1$ and $c_2$ do not overlap those centered at $b_1$ and $b_2$.}
	\label{fig:octaBitReader1}
\end{figure}

When analyzing even sided polygons, note that we can always assume that each polygonal tile has the standard orientation. Let $T_{c}$ denote the octagonal tile  centered at $c$ and let $T_{c_1}$ denote the octagonal tile centered at $c_1$ as shown in Figure~\ref{fig:octaBitReader1}. To show that $T_{c}$ and $T_{c_1}$ do not overlap, let $\omega$ now denote $e^{\frac{2\pi}{8}}$ and note that relative to $c$, $c_1$ is given by the following equation.

\begin{eqnarray*}
c_1 &=& 13\omega^7 + \omega^5 + 8\omega^4 + \omega^2 + 2\omega^3 + 7\omega^2 \\
    &=& 8\omega^2 + 2\omega^3 + 8\omega^4 + \omega^5 + 13\omega^7
\end{eqnarray*}

Then, after multiplying by $\omega$ we need only show that $\Im(\omega c_1) \leq -1$. To see this, note that $\omega c_1 = 8\omega^3 + 2\omega^4 + 8\omega^5 + \omega^6 + 13\omega^8$. Then, since $\omega^8 = 1$, $\omega^2 = i$, $\omega^4 = -1$, and $\omega^3 = \omega^{-5}$, we see that $\omega c_1 = 11 + 8(2\Re(\omega^3)) - i$, and hence, $\Im(\omega c_1) = -1$. Therefore, $T_{c}$ and $T_{c_1}$ do not overlap.

To show that $T_{c}$ and $T_{c_2}$ do not overlap, note that $c_2 = c_1 + \omega^2$. Then, after multiplying $c_2$ by $\omega$ we need only show that $\Re(\omega c_2) \leq -1$. To see that this inequality holds, consider the following.

\begin{align*}
\Re\left(\omega c_2\right) &= \Re\left(\omega c_1 + \omega^3\right) = \Re\left(11 - i + 8(2\Re(\omega^3)) + \omega^3\right)\\
                &= \Re\left(11 + 17\Re(\omega^3)\right) = 11 - 17\dfrac{\sqrt{2}}{2}\\
                &< 11 - 17\dfrac{1.414}{2} < -1
\end{align*}

Then, since $\Re(\omega c_2) \leq -1$, we see that $T_{c}$ and $T_{c_2}$ do not overlap. Therefore, $T_c$ does not overlap $T_{c_1}$ and $T_{c_2}$ do not overlap. We now show that an octagonal tile, $T_{c_1}$ say, with center $c_1$ and an octagonal tile, $T_{b_1}$ say, with center $b_1$ do not overlap. It will then also be clear that an octagonal tile with center $c_1$ or $c_2$ and an octagonal tile with center $b_1$ or $b_2$ do not overlap.
To see that $T_{c_1}$ and $T_{c_1}$ do not overlap, note that relative to $c_1$, $b_1 = 7\omega^6 + 2\omega^7 + \omega^6 + 2\omega^4 + 3\omega^3 + \omega + 7\omega^2 + 1$. It suffices to show that $\Re(b_1) = -1$. Then we see that $\Re(b_1) = 1 + \frac{\sqrt{2}}{2} - 3\frac{\sqrt{2}}{2} - 2 + 2\frac{\sqrt{2}}{2}$.
Hence, $\Re(b_1) = -1$, and an octagonal tile with center $c_1$ and an octagonal tile with center $b_1$ do not overlap.

\begin{figure}[htp]
\centering
	\includegraphics[width=2.5in]{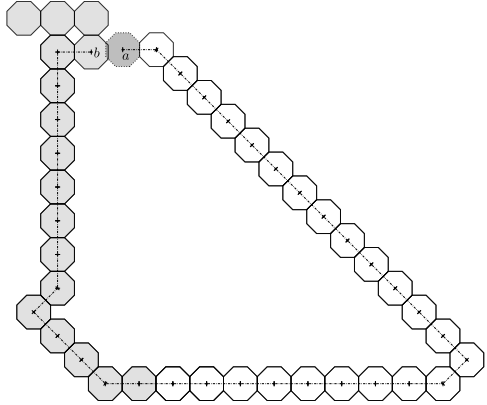}
	\caption{A configuration of the bit-reader given in Figure~\ref{fig:technical-octagonBitReader}. We must show that the octagonal tile centered at $a$ overlaps the tile centered at $b$.}
	\label{fig:octaBitReader2}
\end{figure}

Now, referring to Figure~\ref{fig:octaBitReader2}, in remains to be shown that an octagon with center $a$ and an octagon with center $b$ overlap. That is, an octagonal tile (in an existing assembly) centered at $b$ prevents the binding of an octagonal tile centered at $a$. To see this, note that relative to $a$, $b = 1 + 13\omega^7 + \omega^5 + 10\omega^4 + 3\omega^3 + \omega + 7\omega^2 + 1$. Simplifying $b$ gives $b = 2 + \omega + 7\omega^2 + 3\omega^3  + 10\omega^4 + \omega^5 + 13\omega^7$. Then one can check that $|b| < 1$. Given these calculations, we can obtain a bit-reading gadget for systems whose tiles have the shape of a octagon.

\subsubsection{Nonagonal Tile Assembly}\label{sec:technical-nonagonal}

Figure~\ref{fig:technical-nonagonBitReader} depicts two possible configurations of a bit-reading gadget construction for single-shaped systems with nonagonal tiles, the gray tiles represent a ``bit-writer'' tiles (representing either $0$ or $1$), while the white tiles are the ``bit-reader'' tiles. We ensure that the assembly sequence of a bit-gadget is such that all of the gray tiles bind before any white tiles. Referring to Figure~\ref{fig:technical-nonagonBitReaderA}, we will first show that the tiles labeled $R$ and $B$ do not prevent the binding of the tile labeled $T_1$ or the tile labeled $T_2$. Then we will show that the tile labeled $B$ prevents an octagonal tile from binding to the glue labeled $g_1$.

\begin{figure}[htp]
\centering
  \subfloat[][A $0$ is read, and a $1$ cannot be read by mistake since the tile $B$ prevents a nonagonal tile from attaching via the glue labeled $g_1$.]{%
        \label{fig:technical-nonagonBitReaderA}%
        \makebox[.4\textwidth]{
        \includegraphics[width=2in]{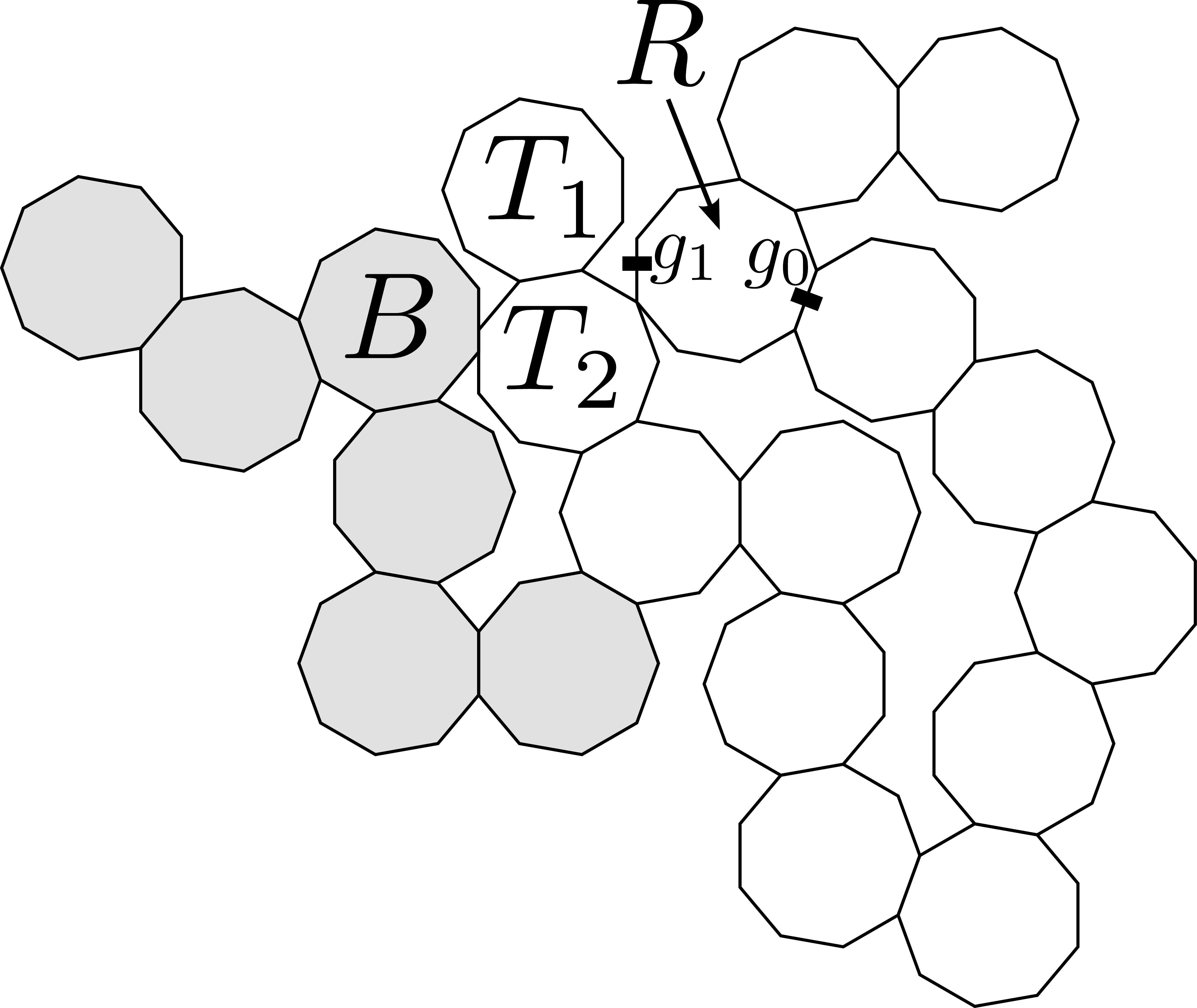}
        }}%
        \quad
  \subfloat[][A $1$ is read. This time a $0$ cannot be read since the tile $B$ prevents growth of a path of nonagonal tiles that attach via the glue labeled $g_0$. Note that some of this path may form, but $B$ prevents the entire path from assembling, and thus prevents a $0$ from being read.]{%
        \label{fig:technical-nonagonBitReaderB}%
        \makebox[.4\textwidth]{
        \includegraphics[width=2in]{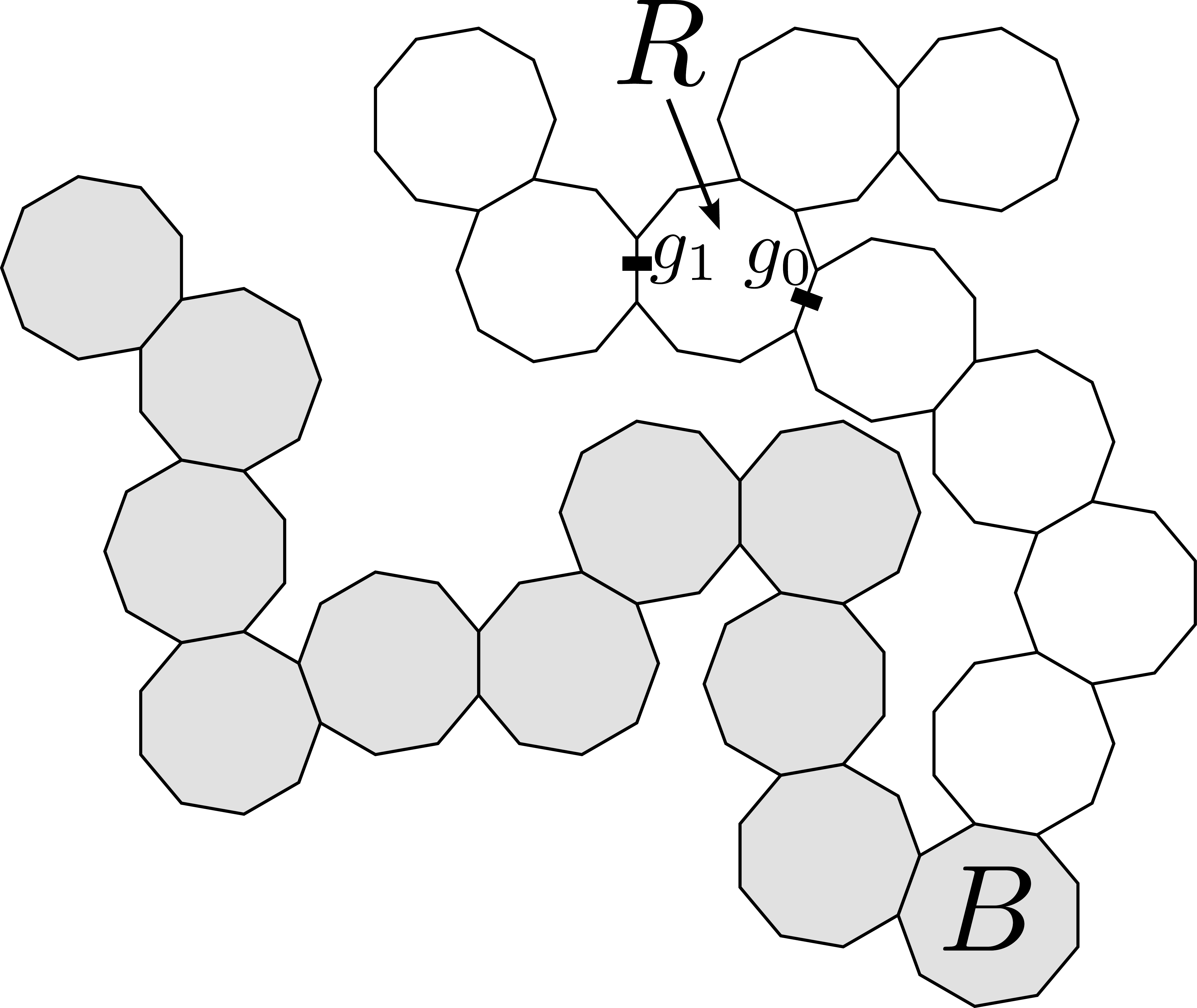}
        }}%
  \caption{The configurations for a bit-reading gadget consisting of nonagonal tiles.}
  \label{fig:technical-nonagonBitReader}
\end{figure}

Referring to Figure~\ref{fig:nonagonBitReader1}, let $T_c$ be a nonagonal tile with negated orientation centered at $c$ and let $T_{c_1}$ be a nonagonal tile with negated orientation centered at $c_1$. We first show that that $T_c$ and $T_{c_1}$ do not overlap. Let $\omega$ now denote $e^{\frac{2\pi}{9}}$.

\begin{figure}[htp]
\centering
	\includegraphics[width=2.5in]{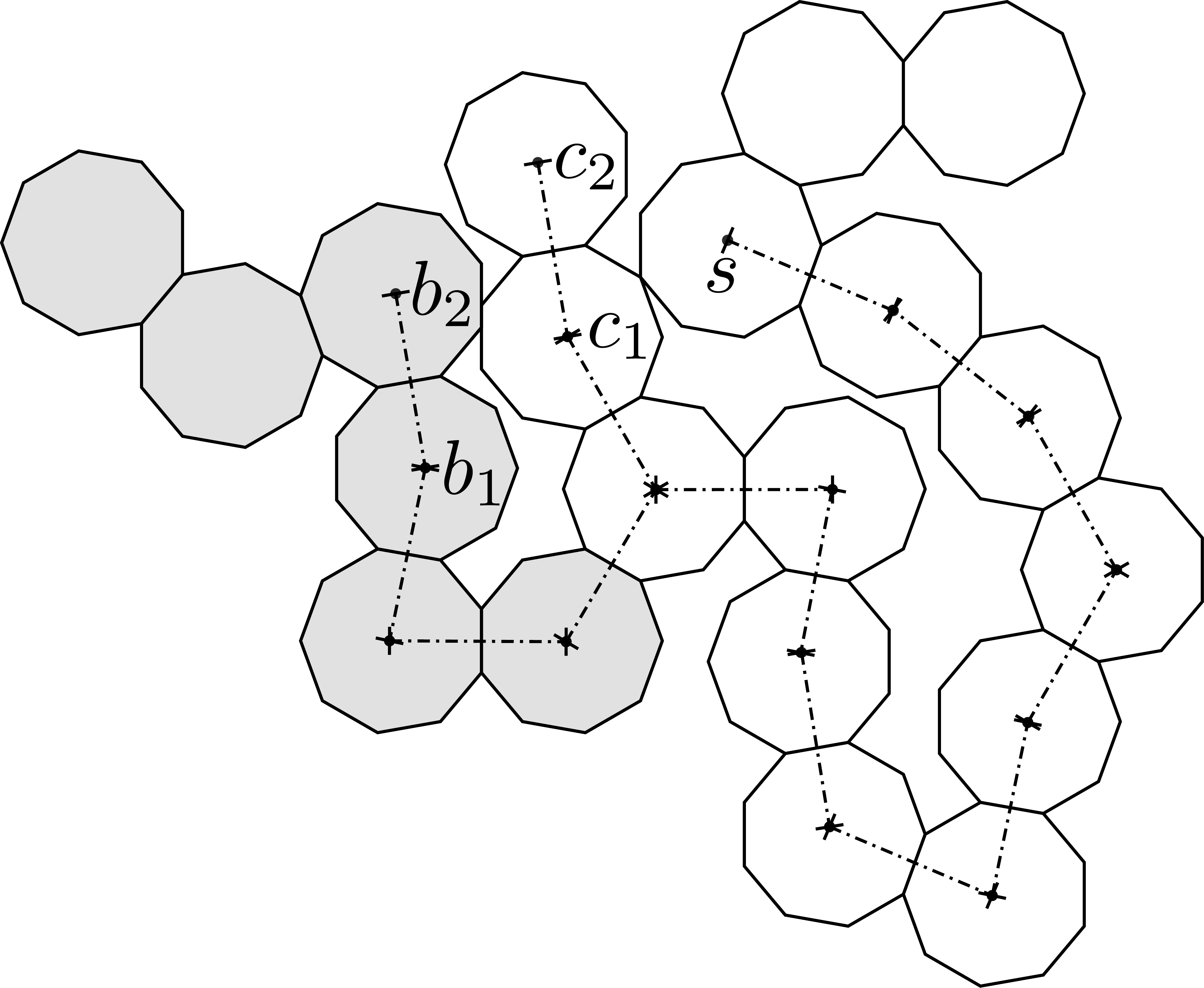}
	\caption{A possible configuration of the bit-reader given in Figure~\ref{fig:technical-nonagonBitReader}. We must show that the nonagonal tiles centered at $c_1$ and $c_2$ do not overlap those centered at $c$, $b_1$, and $b_2$.}
	\label{fig:nonagonBitReader1}
\end{figure}

Note that relative to $c$, $c_1 = -\omega^4 + \omega^8 - \omega^3 + \omega^6 - \omega^2 + \omega^4 - \omega^7 + \omega^2 - 1 + \omega^3$. Simplifying $c_1$ gives $c_1= -1 + \omega^8 + \omega^6 - \omega^7$. Then, after multiplying by $\omega^{-1}$ (which corresponds to rotating the Figure~\ref{fig:nonagonBitReader1} clockwise by $\frac{2\pi}{9}$), it suffices to show that $\Re(\omega^{-1}c_1) = -\frac{1}{2} - \frac{1}{2\cos\left(\frac{\pi}{9}\right)}$.
To see this, first note that $\omega^{-1}c_1 = \omega^5 - \omega^6 + \omega^7 - \omega^8$. Therefore, $\Re\left(\omega^{-1}c_1\right) = \cos\left(\tfrac{10\pi}{9}\right) - \cos\left(\tfrac{12\pi}{9}\right) + \cos\left(\tfrac{14\pi}{9}\right) - \cos\left(\tfrac{16\pi}{9}\right)$. Hence, it suffices to show that
$$\cos\left(\tfrac{10\pi}{9}\right) - \cos\left(\tfrac{12\pi}{9}\right) + \cos\left(\tfrac{14\pi}{9}\right) - \cos\left(\tfrac{16\pi}{9}\right) = -\frac{1}{2} - \frac{1}{2\cos\left(\frac{\pi}{9}\right)}.$$

\noindent To see this, consider the following equations.

\begin{align*}
-2 &= e^{\frac{\pi i}{9}} + e^{-\frac{\pi i}{9}} + e^{\pi i} + e^{-\pi i} - e^{-\frac{\pi i}{9}} - e^{\frac{\pi i}{9}}\\
   &= e^{\frac{\pi i}{9}} + e^{-\frac{\pi i}{9}} + e^{\frac{9\pi i}{9}} + e^{-\frac{9\pi i}{9}} - e^{\frac{17\pi i}{9}} - e^{-\frac{17\pi i}{9}}\\
   &= e^{\frac{\pi i}{9}} + e^{\frac{11\pi i}{9}} + e^{-\frac{9\pi i}{9}} - e^{\frac{13\pi i}{9}} - e^{-\frac{11\pi i}{9}} + e^{\frac{15\pi i}{9}} + e^{-\frac{13\pi i}{9}} - e^{\frac{17\pi i}{9}} - e^{-\frac{15\pi i}{9}}\\
   &\ \ + e^{-\frac{\pi i}{9}} + e^{\frac{9\pi i}{9}} + e^{-\frac{11\pi i}{9}} - e^{\frac{11\pi i}{9}} - e^{-\frac{13\pi i}{9}} + e^{\frac{13\pi i}{9}} + e^{-\frac{15\pi i}{9}} - e^{\frac{15\pi i}{9}} - e^{-\frac{17\pi i}{9}}\\
   &= \left(e^{\frac{\pi i}{9}} + e^{-\frac{\pi i}{9}}\right)\\
   &\ \ \times \left(1 + e^{\frac{10\pi i}{9}} + e^{-\frac{10\pi i}{9}} - e^{\frac{12\pi i}{9}} - e^{-\frac{12\pi i}{9}} + e^{\frac{14\pi i}{9}} + e^{-\frac{14\pi i}{9}} - e^{\frac{16\pi i}{9}} - e^{-\frac{16\pi i}{9}}\right)
\end{align*}

\noindent Therefore, $-\frac{2}{e^{\pi i/9} + e^{-\pi i/9}} = 1 + e^{10\pi i/9} + e^{-10\pi i/9} - e^{12\pi i/9} - e^{-12\pi i/9} + e^{14\pi i/9} + e^{-14\pi i/9} - e^{16\pi i/9} - e^{-16\pi i/9}$. Using the identity $\cos(\theta) = \frac{e^{i\theta} + e^{-i\theta}}{2}$, we can see that
$-1-\frac{1}{\cos\left(\frac{\pi i}{9}\right)} = 2\cos\left(\frac{10\pi i}{9}\right) - 2\cos\left(\frac{12\pi i}{9}\right) + 2\cos\left(\frac{14\pi i}{9}\right) - 2\cos\left(\frac{16\pi i}{9}\right)$. Therefore, $T_c$ and $T_{c_1}$ do not overlap.

Now we let $T_{b_1}$ denote a nonagonal tile with negated orientation centered at $b_1$ and show that $T_{b_1}$ and $T_{c_1}$ do not overlap. Relative to $c_1$, $b_1 = -\omega^3 + \omega^6 -1 + \omega^2$. It suffices to show that $\Re(\omega^{-1}b_1) < -\frac{1}{2} - \frac{1}{2\cos(\pi/9)}$, which we can numerically verify is true by approximating each side of the inequality.

Similarly, we let $T_{b_2}$ denote a nonagonal tile with standard orientation centered at $b_2$ and show that $T_{b_2}$ and $T_{c_1}$ do not overlap. Relative to $c_1$, $b_2 = -\omega^3 + \omega^6 -1 + \omega^2 - \omega^7$. Then, it suffices to show that $\Re(b_2) = -1$. To see this, note that
$b_2 =  -1 -(\omega^3 - \omega^{-3}) + (\omega^2 - \omega^{-2})$. Since $\omega^3 - \omega^{-3}$ and $\omega^2 - \omega^{-2}$ are imaginary, $\Re(b_2) = -1$.
Therefore, $T_{c_1}$ and $T_{b_2}$ do not overlap. Similarly, we can see that a nonagonal tile centered at $c_2$ that is in standard orientation and a nonagonal tile centered at $b_2$ that is in standard orientation do not overlap.

\begin{figure}[htp]
\centering
	\includegraphics[width=2.5in]{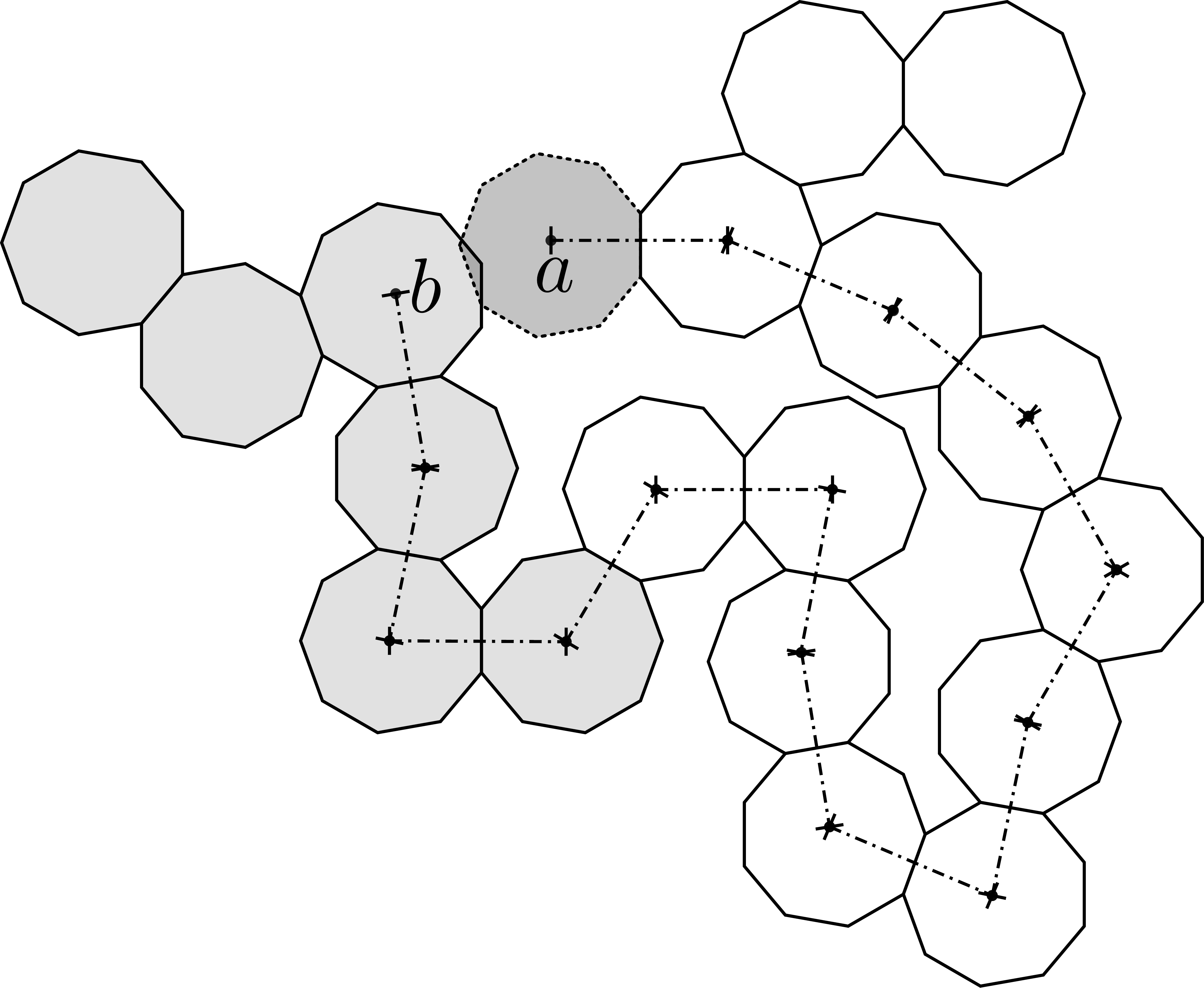}
	\caption{A configuration of the bit-reader given in Figure~\ref{fig:technical-nonagonBitReader}. We must show that the nonagonal tile centered at $a$ overlaps the tile centered at $b$.}
	\label{fig:nonagonBitReader2}
\end{figure}

Now, referring to Figure~\ref{fig:nonagonBitReader2}, it remains to be shown that a nonagonal tile, which we will denote by $T_a$, centered at $a$ that is in standard orientation and a nonagonal tile, which we will denote by $T_b$, centered at $b$ that is in standard orientation overlap. Relative to $a$, $b = 1 - \omega^4 + \omega^8 - \omega^3 + \omega^6 - \omega^2 + \omega^4 - \omega^7 + \omega^2 -1 + \omega^6 - 1 + \omega^2 - \omega^7$. Simplifying $b$ gives $b = -1 + \omega^2 - \omega^3 + 2\omega^6  - 2\omega^7 + \omega^8$. Then we can approximate $|b|$ to see that $|b|< 1$. Therefore, $T_a$ and $T_b$ overlap.

\subsubsection{Polygonal Tile Assembly with $10$, $11$, or $12$ Sided Regular Polygonal Tiles}\label{sec:technical-10-12sides}

In the cases where tiles consist of regular polygons with $10,11,$ or $12$ sides, bit-reading gadgets are relatively simple to construct. Figure~\ref{fig:technical-10to12sidesBitReaders} depicts the bit-reading gadgets for each case. Note that since each polygonal tile of these bit-reading gadgets abuts another tile, we need only show that for each configuration depicted in Figure~\ref{fig:technical-10to12sidesBitReaders}, of the two exposed glues, $g_0$ and $g_1$ of the tile $R$, a tile can only attach to one of these glues depending on the position of the tile $B$ in the figure. In other words, for each configuration depicted in Figure~\ref{fig:technical-10to12sidesBitReaders}, we show that the intersection of the interiors of a polygon with the same shape, position and orientation as $B$ and a polygon with the same shape, position and orientation of the gray tile's position and orientation.

\begin{figure}[htp]
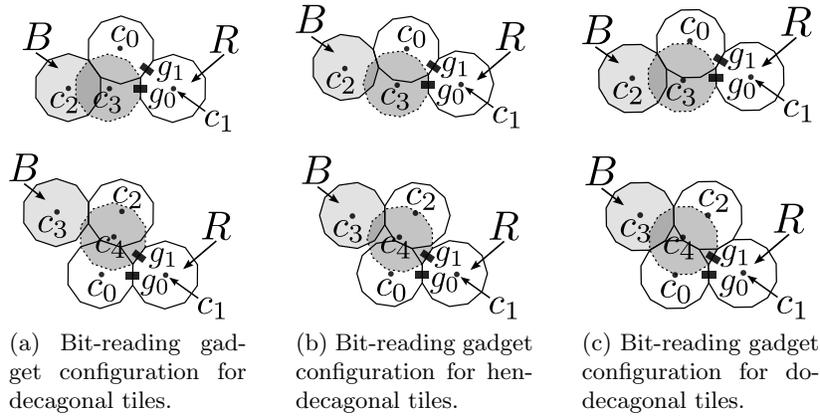

\centering
  \subfloat[][Bit-reading gadget configuration for decagonal tiles.]{%
        \label{fig:technical-10to12sidesBitReadersA}%
    		\makebox[.25\textwidth]{
        \includegraphics[width=1.2in]{images/10-12sidesBitReadersA}
        }
        }%
        \quad\quad
  \subfloat[][Bit-reading gadget configuration for hendecagonal tiles.]{%
        \label{fig:technical-10to12sidesBitReadersB}%
    		\makebox[.25\textwidth]{
        \includegraphics[width=1.2in]{images/10-12sidesBitReadersB}
        }
        }%
       \quad\quad
  \subfloat[][Bit-reading gadget configuration for dodecagonal tiles.]{%
        \label{fig:technical-10to12sidesBitReadersC}%
    		\makebox[.25\textwidth]{
        \includegraphics[width=1.2in]{images/10-12sidesBitReadersC}
        }
        }%
  \caption{(a), (b) and (c) each depict two configurations of polygonal tiles which represents either a $0$ (bottom) or a $1$ (top).}
  \label{fig:technical-10to12sidesBitReaders}
\end{figure}

For decagonal tiles, let $\omega = e^{\frac{2\pi i}{10}}$ and consider Figure~\ref{fig:technical-10to12sidesBitReadersA}. To show that this gives a valid bit-reader, we first show that using the top assembly depicted in the top figure of Figure~\ref{fig:technical-10to12sidesBitReadersA}, a polygon centered at $c_2$ and a polygon centered at $c_3$ overlap. Note that relative to $c_1$, $c_3 = -1$ and $c_2 = \omega^{4} + \omega^{6}$. Hence, $c_2 = \omega^{4} + \omega^{-4} = 2\Re\left(\omega^{4}\right) = 2\cos\left( \frac{8\pi}{10} \right)$. Then the distance $d$ from $c_3$ to $c_2$ satisfies $d = |-1 - 2\cos\left( \frac{8\pi}{10} \right)| < .62$.

Secondly, we show that in the bottom figure of Figure~\ref{fig:technical-10to12sidesBitReadersA}, a polygon centered at $c_3$ and a polygon centered at $c_4$ overlap. Relative to $c_4$, $c_3 = \omega^9 - 1 + \omega^2 - 1$. Hence, $c_3 = -2 + \omega^2 + \omega^9$. Then, $|c_3| = \left(-2 + \cos\left(\frac{4\pi}{10}\right) + \cos\left(\frac{18\pi}{10}\right) \right)^2 + \left(\sin\left(\frac{4\pi}{10}\right) + \sin\left(\frac{18\pi}{10}\right) \right)^2 < .91$.

For hendecagonal tiles, let $\omega = e^{\frac{2\pi i}{11}}$ and consider  Figure~\ref{fig:technical-10to12sidesBitReadersB}. To show that this gives a valid bit-reader, we first show that using the top assembly depicted in the top figure of Figure~\ref{fig:technical-10to12sidesBitReadersB}, a polygon centered at $c_2$ and a polygon centered at $c_3$ overlap. Note that relative to $c_3$, $c_2 = 1 - \omega^{10} + \omega^{6}$. Hence, $|c_2|^2 = \left(1 - \cos\left(\frac{20\pi}{11}\right) + \cos\left(\frac{12\pi}{11}\right) \right)^2 + \left(-\sin\left(\frac{20\pi}{11}\right) + \sin\left(\frac{12\pi}{11}\right) \right)^2 < .71$

Secondly, we show that in the bottom figure of Figure~\ref{fig:technical-10to12sidesBitReadersB}, a polygon centered at $c_2$ and a polygon centered at $c_1$ do not overlap. Relative to $c_1$, $c_2 = - 1 + \omega^2$.  Then, $|c_2|^2 = \left(-1 + \cos\left(\frac{4\pi}{11}\right)\right)^2 + \sin^2\left(\frac{4\pi}{11}\right) > \frac{1}{\cos\left(\frac{\pi}{11}\right)}$.

For dodecagonal tiles, let $\omega = e^{\frac{2\pi i}{12}}$ and consider Figure~\ref{fig:technical-10to12sidesBitReadersC}. To show that this gives a valid bit-reader, we first show that using the top assembly depicted in the top figure of Figure~\ref{fig:technical-10to12sidesBitReadersC}, a polygon centered at $c_2$ and a polygon centered at $c_3$ overlap. Note that relative to $c_3$, $c_2 = 1 + \omega^{5} + \omega^{7}$. Hence, $|c_2|^2 = \left(1 + \cos\left(\frac{10\pi}{12}\right) + \cos\left(\frac{14\pi}{12}\right) \right)^2 + \left(\sin\left(\frac{10\pi}{12}\right) + \sin\left(\frac{14\pi}{12}\right) \right)^2 < .54$

Secondly, we show that in the bottom figure of Figure~\ref{fig:technical-10to12sidesBitReadersC}, a polygon centered at $c_2$ and a polygon centered at $c_1$ do not overlap. Relative to $c_1$, $c_2 = - 1 + \omega^2$.  Then, it suffices to show that $\omega^2 c_2 = -1$. Note that

\begin{align*}
\omega^2 c_2 &= -\omega^2 + \omega^4 = \omega^{-4} + \omega^4\\
             &= 2\Re\left( \omega^4 \right) = 2\cos\left(\tfrac{8 \pi}{12}\right)\\
             &= 2\cos\left(\tfrac{2 \pi}{3}\right) = -1
\end{align*}

\subsubsection{Tridecagonal Tile Assembly}\label{sec:technical-tridecagonal}

\begin{figure}[htp]
\centering
  \subfloat[][A $0$ is read, and a $1$ cannot be read by mistake since the tile $B$ prevents a tridecagonal tile from attaching via the glue labeled $g_1$.]{%
        \label{fig:technical-tridecagonBitReaderA}%
        \makebox[.4\textwidth]{
        \includegraphics[width=2in]{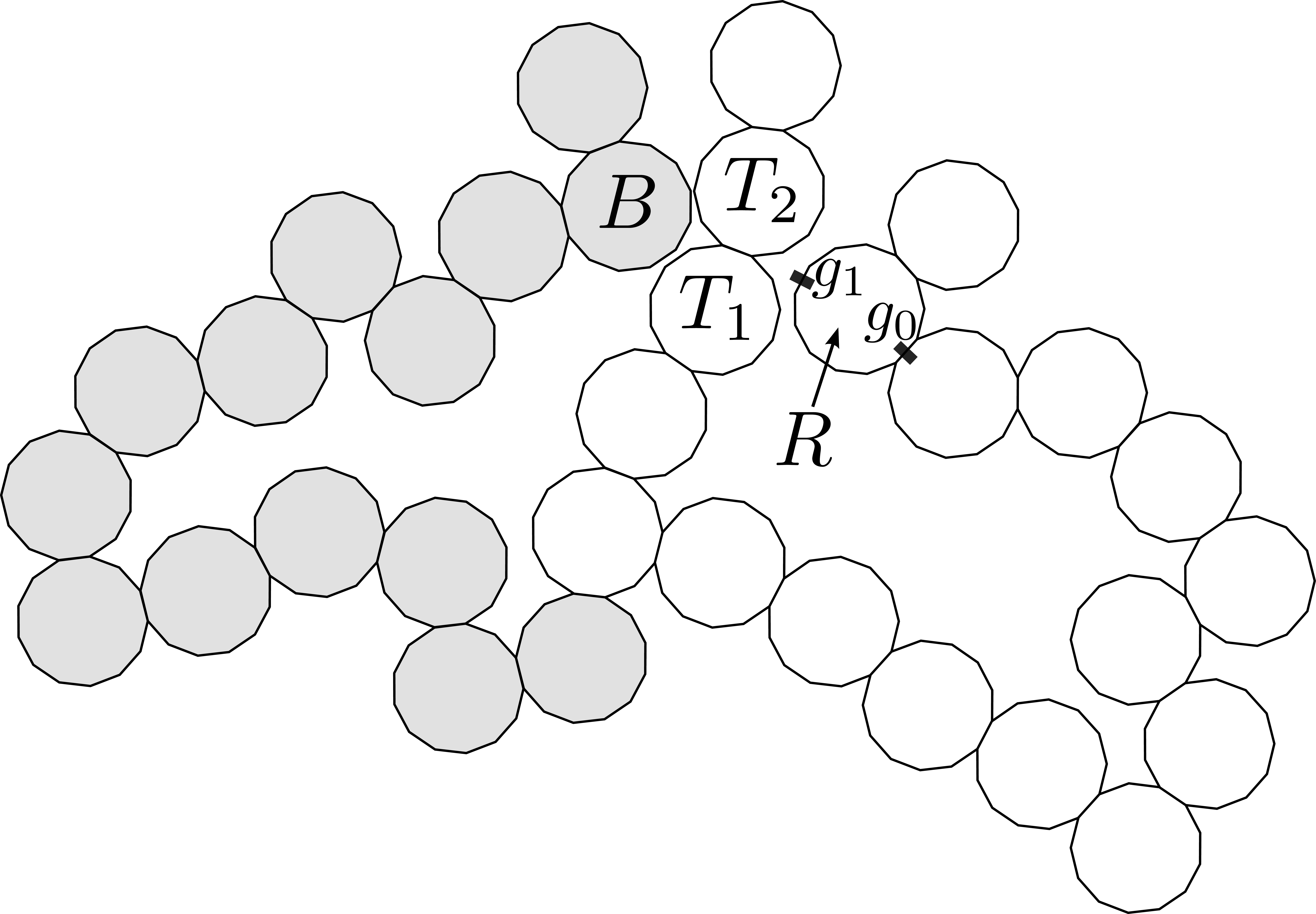}
        }}%
        \quad
  \subfloat[][A $1$ is read. This time a $0$ cannot be read by mistake since the tile $B$ prevents growth of a path of tridecagonal tiles that attach via the glue labeled $g_0$. Note that some of this path may form, but $B$ prevents the entire path from assembling, and thus prevents a $0$ from being read.]{%
        \label{fig:technical-octagonBitReaderB}%
        \makebox[.4\textwidth]{
        \includegraphics[width=2in]{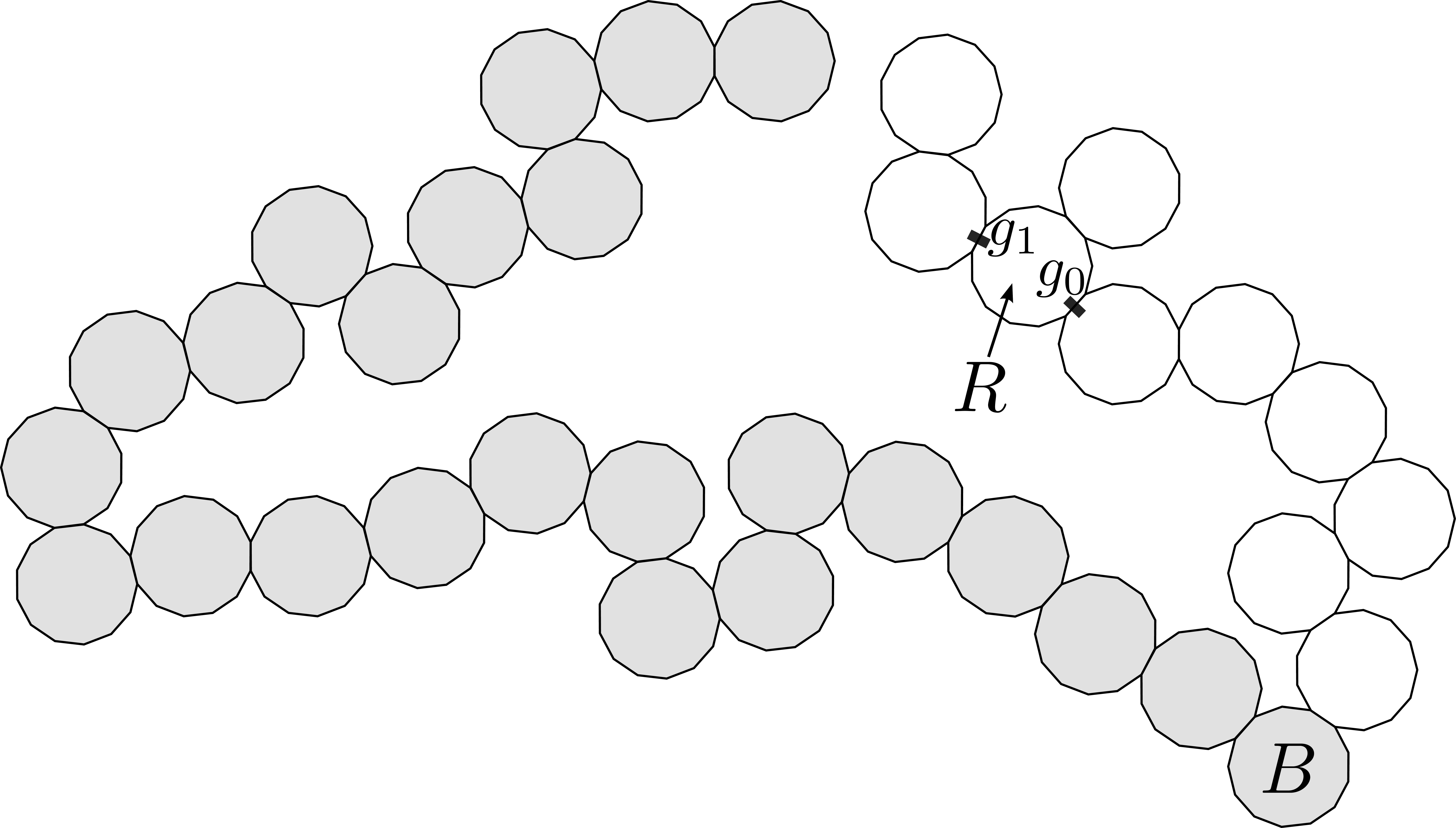}
        }}%
  \caption{The configurations for a bit-reading gadget consisting of tridecagonal tiles.}
  \label{fig:technical-tridecagonBitReader}
\end{figure}

Figure~\ref{fig:technical-tridecagonBitReader} depicts two possible configurations of a bit-reading gadget construction for single-shaped systems with tridecagonal tiles, the gray tiles represent a ``bit-writer'' tiles (representing either $0$ or $1$), while the white tiles are the ``bit-reader'' tiles. We ensure that the assembly sequence of a bit-gadget is such that all of the gray tiles bind before any white tiles. Referring to Figure~\ref{fig:technical-tridecagonBitReaderA}, we will first show that the tiles labeled $R$ and $B$ do not prevent the binding of the tile labeled $T_1$ or the tile labeled $T_2$. Then we will show that the tile labeled $B$ prevents an octagonal tile from binding to the glue labeled $g_1$.

\begin{figure}[htp]
\centering
	\includegraphics[width=2.5in]{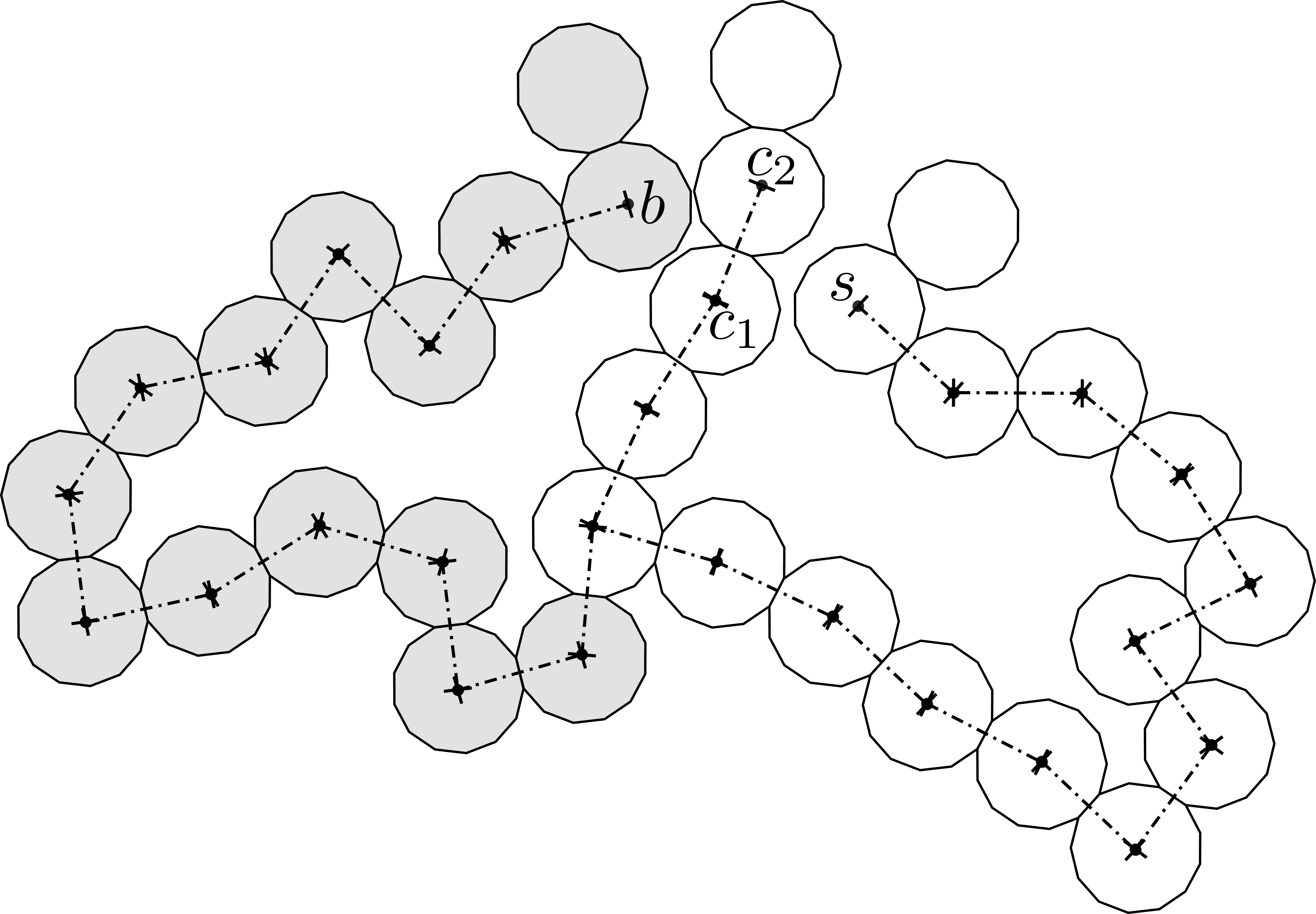}
	\caption{A possible configuration of the bit-reader given in Figure~\ref{fig:technical-tridecagonBitReader}. We must show that the nonagonal tiles centered at $c_1$ and $c_2$ do not overlap those centered at $c$ and $b$.}
	\label{fig:tridecagonBitReader1}
\end{figure}

We now refer to Figure~\ref{fig:tridecagonBitReader1} and let $\omega$ be $e^{\frac{2\pi i}{13}}$.
Let $T_{c}$ denote the tridecagonal tile with negated orientation centered at $c$ and let $T_{c_1}$ denote the tridecagonal tile with negated orientation centered at $c_1$. To show that $T_{c}$ and $T_{c_1}$ do not overlap, note that relative to $c$, $c_1$ is given by $c_1 = -\omega^{5} + 1 - \omega^{5} + \omega^{11} - \omega^{1} + \omega^{11} - \omega^{2} + \omega^{5} - \omega^{12} + \omega^{5} - \omega^{12} + \omega^{6} - \omega^{9} + \omega^{2}$ and $c_2 = c_1 - \omega^{9}$.
Simplifying $c_1$, we obtain $c_1 = -2 \omega ^{12}+2 \omega ^{11}-\omega ^9+\omega ^6-\omega +1$. Then by approximating $|c_1|$ we can see that $|c_1|>1.13>\frac{1}{\cos\left(\frac{\pi}{13}\right)}$. Therefore, $T_{c}$ and $T_{c_1}$ do not overlap.

Now let $T_{c_2}$ denote the tridecagonal tile with standard orientation centered at $c_2$. Since $c_2 = c_1 -\omega^9$, we see that $c_2 = -2 \omega ^{12}+2 \omega ^{11}-2\omega ^9+\omega ^6-\omega +1$. Then we approximate $|c_2|$ to show that $|c_2| > 1.21>\frac{1}{\cos\left(\frac{\pi}{13}\right)}$. Therefore, $T_{c}$ and $T_{c_2}$ do not overlap.

Let $T_b$ denote the tridecagonal tile with standard orientation centered at $b$. Then, relative to $b$, $c_1 = \omega^{7} - \omega^{2} + \omega^{5} - \omega^{2} + \omega^{7} - \omega^{2} + \omega^{10} - \omega^{7} + \omega - \omega^{6} + \omega^{10} - \omega^{7} + \omega^{3} - \omega^{9} + \omega^{2}$. We can simplify $c_1$ to obtain $c_1 = 2 \omega ^{10}-\omega ^9-\omega ^6+\omega ^5+\omega ^3-2 \omega ^2+\omega$. Also note that $c_2 = c_1 - \omega^9$.

Then, we approximate $|c_1|$ to show that $|c_1| > 1.06 > \frac{1}{\cos\left(\frac{\pi}{13}\right)}$. Therefore, $T_{b}$ and $T_{c_1}$ do not overlap.  Similarly, $|c_2| > 1.04 > \frac{1}{2} + \frac{1}{2\cos(\pi/13)}$, and so $T_{b}$ and $T_{c_1}$ do not overlap.

\begin{figure}[htp]
\centering
	\includegraphics[width=2.5in]{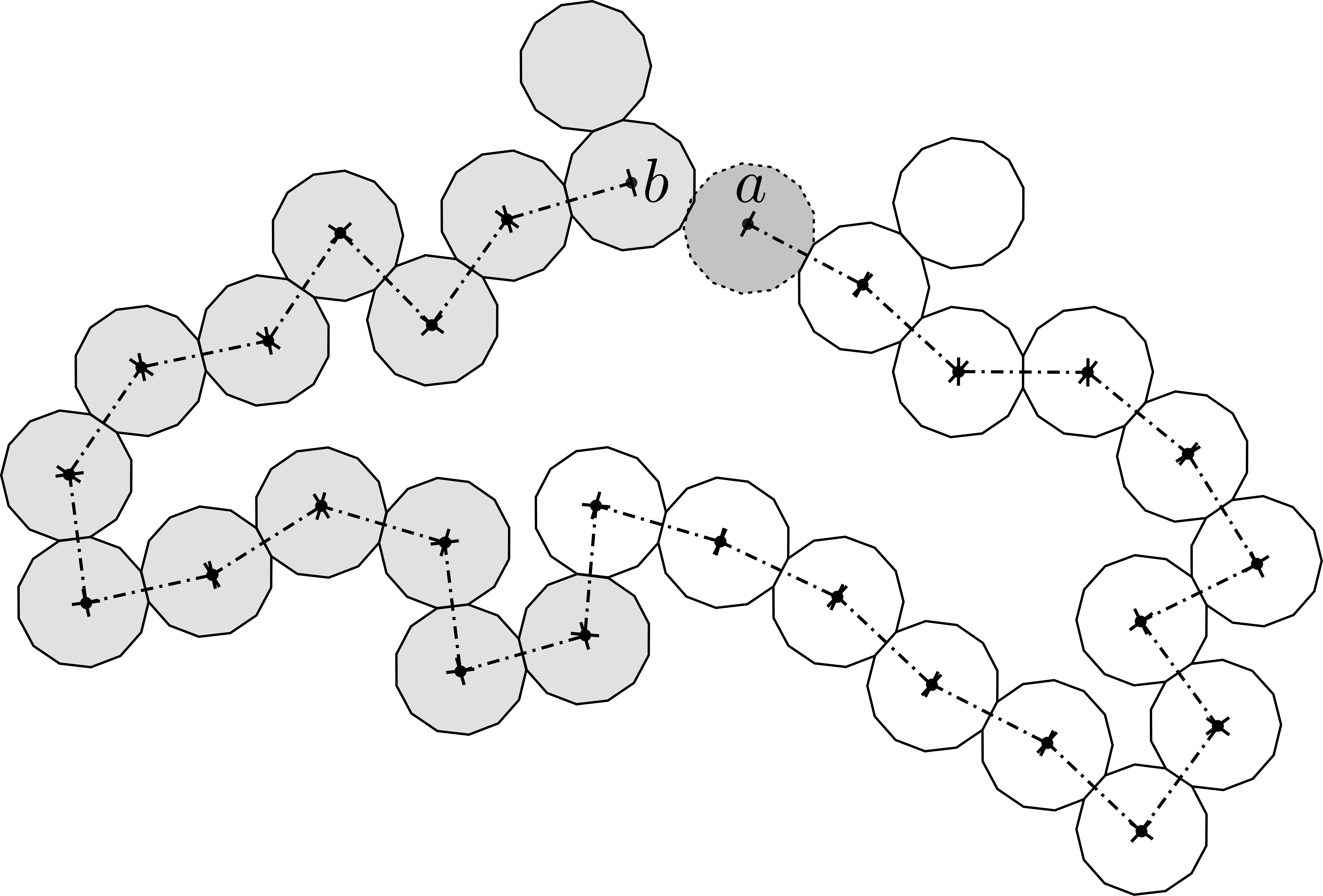}
	\caption{A configuration of the bit-reader given in Figure~\ref{fig:technical-tridecagonBitReader}. We must show that the tridecagonal tile centered at $a$ overlaps the tile centered at $b$.}
	\label{fig:tridecagonBitReader2}
\end{figure}

Now, referring to Figure~\ref{fig:tridecagonBitReader2}, it remains to be shown that a tridecagonal tile, which we will denote by $T_a$, centered at $a$ that is in standard orientation and a tridecagonal tile, which we will denote by $T_b$, centered at $b$ that is in standard orientation overlap. Relative to $b$,

\begin{align*}
a &= \omega^{7} - \omega^{2} + \omega^{5} - \omega^{2} + \omega^{7} - \omega^{2} + \omega^{10} - \omega^{7} + \omega - \omega^{6} + \omega^{10} - \omega^{7} + \omega^{3} - \omega^{6} + \omega^{12} - \omega^{5}   \\
  &\quad + \omega^{12} - \omega^{5}  + \omega^{2} - \omega^{11} + \omega - \omega^{11} + \omega^{5} - 1 + \omega^{5} - \omega^{12} \\
  &= -1 + 2\omega - 2\omega^{2} + \omega^{3} + \omega^{5} - 2\omega^{6}  + 2\omega^{10} - 2\omega^{11} + \omega^{12}
\end{align*}

\noindent Then we can approximate $|a|$ to see that $|a|< 1$. Therefore, $T_a$ and $T_b$ overlap.

\subsubsection{Tetradecagonal Tile Assembly}\label{sec:technical-tetradecagonal}

\begin{figure}[htp]
\centering
  \subfloat[][A $0$ is read, and a $1$ cannot be read by mistake since the tile $B$ prevents a tetradecagonal tile from attaching via the glue labeled $g_1$.]{%
        \label{fig:technical-tetradecagonBitReaderA}%
        \makebox[.4\textwidth]{
        \includegraphics[width=1.5in]{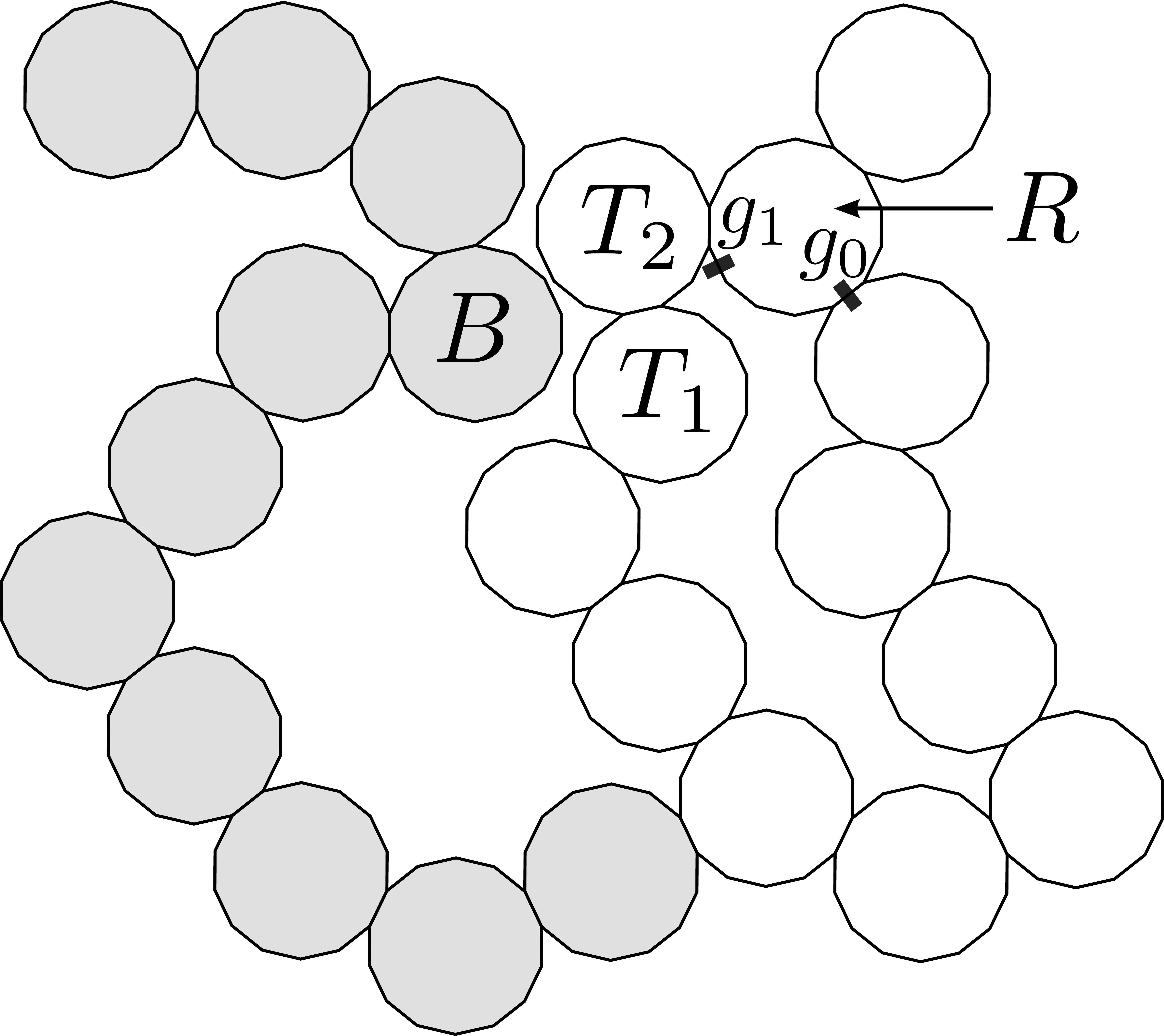}
        }}%
        \quad
  \subfloat[][A $1$ is read. This time a $0$ cannot be read by mistake since the tile $B$ prevents growth of a path of tetradecagonal tiles that attach via the glue labeled $g_0$. Note that some of this path may form, but $B$ prevents the entire path from assembling, and thus prevents a $0$ from being read.]{%
        \label{fig:technical-tetradecgonBitReaderB}%
        \makebox[.4\textwidth]{
        \includegraphics[width=1.5in]{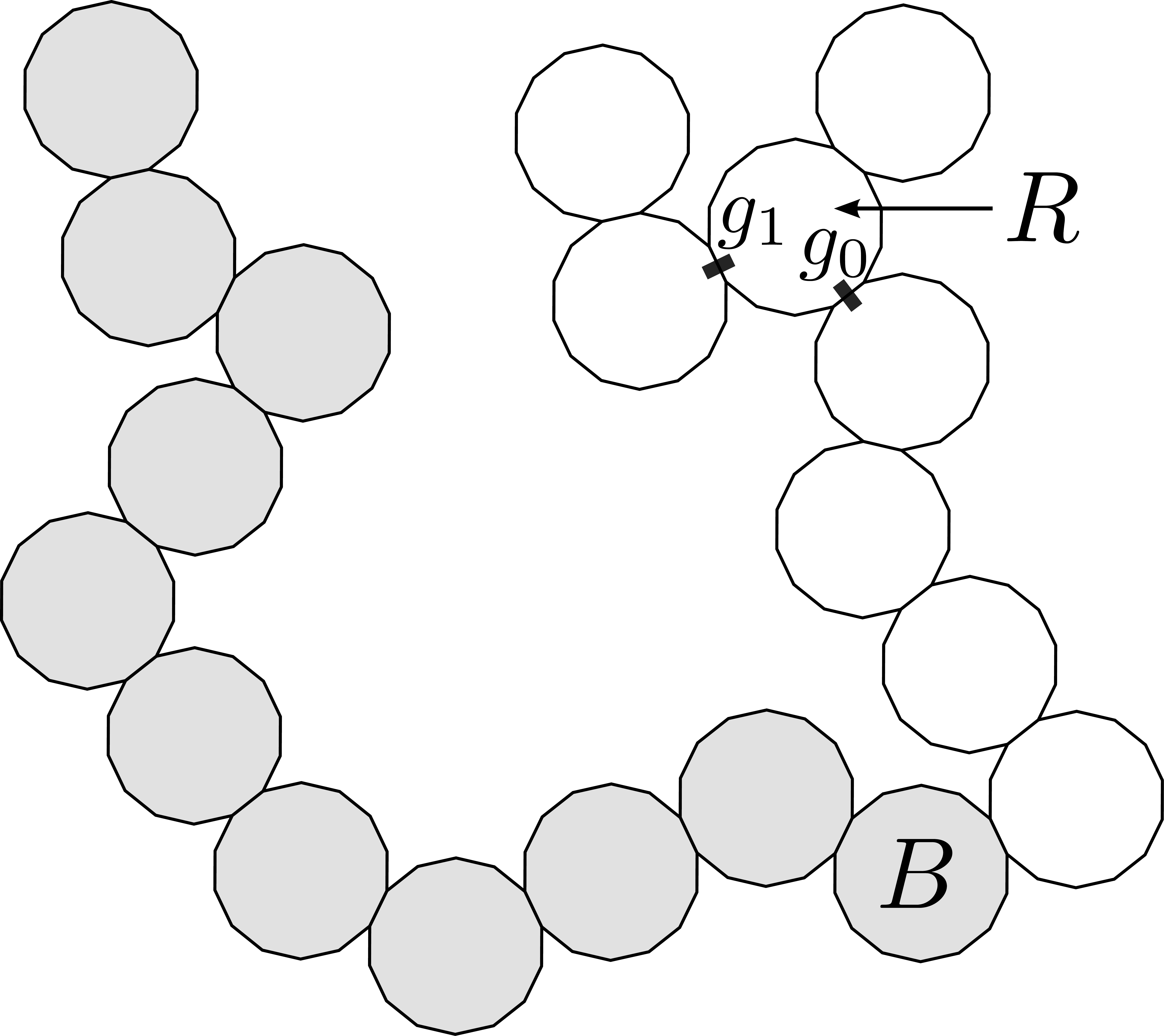}
        }}%
  \caption{The configurations for a bit-reading gadget consisting of tetradecagonal tiles.}
  \label{fig:technical-tetradecagonBitReader}
\end{figure}

Figure~\ref{fig:technical-tetradecagonBitReader} depicts two possible configurations of a bit-reading gadget construction for single-shaped systems with tetradecagonal tiles, the gray tiles represent a ``bit-writer'' tiles (representing either $0$ or $1$), while the white tiles are the ``bit-reader'' tiles. We ensure that the assembly sequence of a bit-gadget is such that all of the gray tiles bind before any white tiles. Referring to Figure~\ref{fig:technical-tetradecagonBitReaderA}, we will first show that the tiles labeled $R$ and $B$ do not prevent the binding of the tile labeled $T_1$ or the tile labeled $T_2$. Then we will show that the tile labeled $B$ prevents an octagonal tile from binding to the glue labeled $g_1$.

\begin{figure}[htp]
\centering
	\includegraphics[width=2.5in]{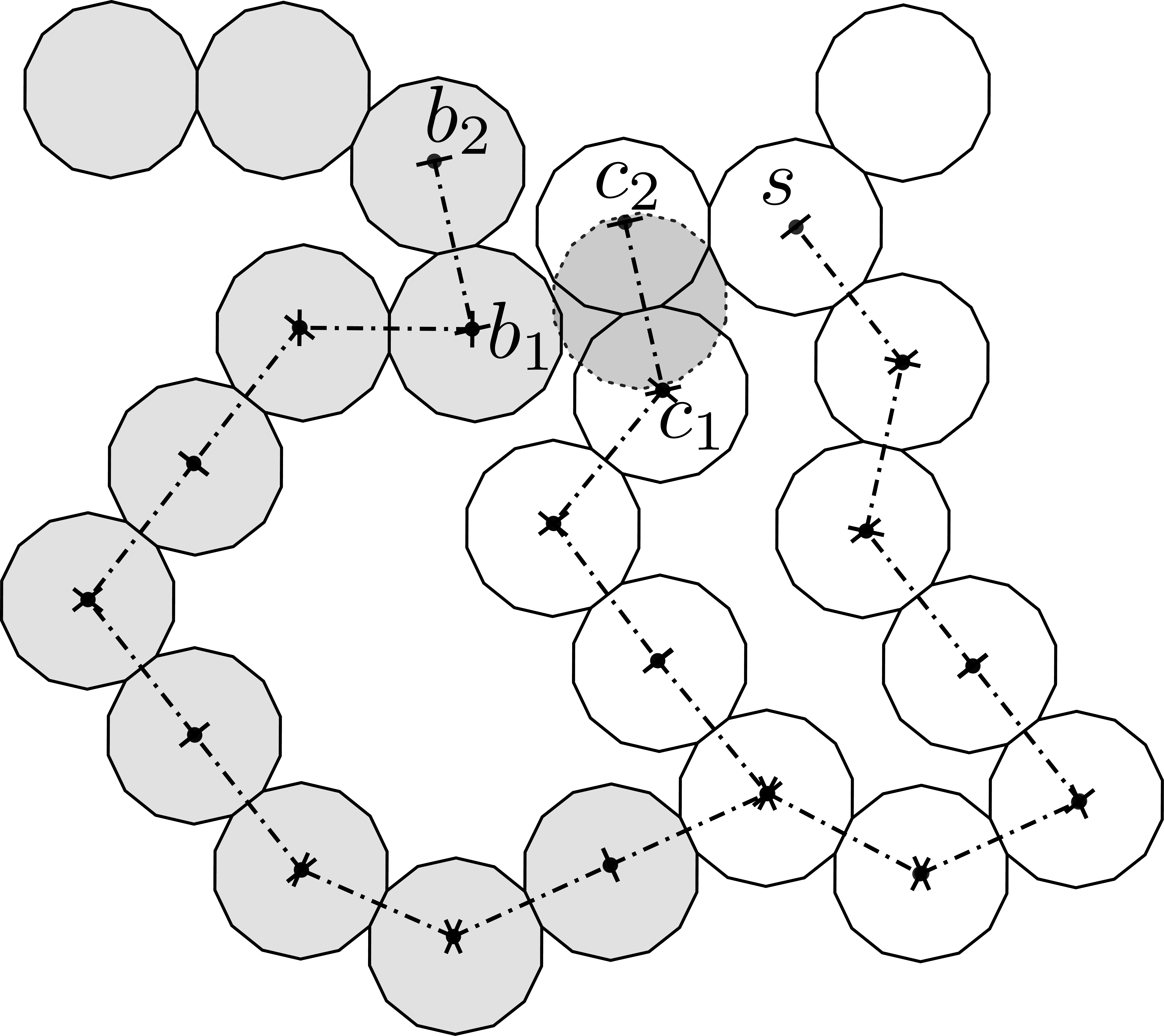}
	\caption{A possible configuration of the bit-reader given in Figure~\ref{fig:technical-tetradecagonBitReader}. We must show that the nonagonal tiles centered at $c_1$ and $c_2$ do not overlap those centered at $c$, $b_1$ and $b_2$.}
	\label{fig:tetradecagonBitReader1}
\end{figure}

We now refer to Figure~\ref{fig:tetradecagonBitReader1} and let $\omega$ be $e^{\frac{2\pi i}{13}}$.
Let $T_{c}$, $T_{c_1}$, $T_{c_2}$, $T_{b_1}$, and $T_{b_2}$ denote the tridecagonal tile with standard orientation centered at $c$, $c_1$, $c_2$, $b_1$, and $b_2$ respectively. Then, to show that $T_{c}$ and $T_{c_1}$ do not overlap, note that relative to $c$, $c_1$ is given by $c_1 = c_1 = 3\omega^{12} + \omega^{10}  + \omega^{8} + \omega^{6} + 2\omega^{5} + \omega^{2}$. Then by approximating $|c_1|$ we can see that $|c_1|>1.2>\frac{1}{\cos\left(\frac{\pi}{14}\right)}$. Therefore, $T_{c}$ and $T_{c_1}$ do not overlap. Moreover, $c_2 = c_1 + \omega^4$. Then, consider the following.

\begin{align*}
c_2 &= 3\omega^{12} + \omega^{10} + \omega^{8} + \omega^{6} + 2\omega^{5} + \omega^{4} + \omega^{2}\\
    &= (2\omega^{12} + 2\omega^{5}) + (\omega^{12} + \omega^{10} + \omega^{8} + \omega^{6} + \omega^{4} + \omega^{2}) \numberthis \label{eqn:line1}\\
    &= (2\omega^{12} - 2\omega^{12}) + (\omega^{12} + \omega^{10} + \omega^{8} + \omega^{6} + \omega^{4} + \omega^{2} + 1 - 1)\numberthis \label{eqn:line2}\\
    &= -1
\end{align*}

\noindent Equation~\eqref{eqn:line1} follows from the following equalities that $$\omega^5 = e^{\left(\frac{10\pi i}{14}\right)} = -e^{\left(\frac{10\pi i}{14} + \frac{14\pi i}{14}\right)} = -e^{\left(\frac{24\pi i}{14}\right)} = -\omega^{12}.$$ Equation~\eqref{eqn:line2} follows from the fact that $\omega^{12} + \omega^{10} + \omega^{8} + \omega^{6} + \omega^{4} + \omega^{2} + 1 = 0.$ To see this, note that $\omega^{12} + \omega^{10} + \omega^{8} + \omega^{6} + \omega^{4} + \omega^{2} + 1 = \omega^2\left(\omega^{12} + \omega^{10} + \omega^{8} + \omega^{6} + \omega^{4} + \omega^{2} + 1\right)$, and so, $$(\omega^2 - 1)\left(\omega^{12} + \omega^{10} + \omega^{8} + \omega^{6} + \omega^{4} + \omega^{2} + 1\right) = 0.$$ Then, since $\omega^2 - 1 \neq 0$, it follows that $\omega^{12} + \omega^{10} + \omega^{8} + \omega^{6} + \omega^{4} + \omega^{2} + 1 = 0$. Therefore, $T_{c}$ and $T_{c_2}$ do not overlap.

Now, to show that $T_{b_1}$ does not overlap $T_{c_1}$ or $T_{c_2}$, note that relative to $b_1$, $c_1 = -1 + 2\omega^9 + 2\omega^{12} + \omega^{13} + 2\omega + 2\omega^5 + \omega^2$. Simplifying $c_1$, we obtain $c_1 = -1 + \omega^9 + \omega^{13} + 2\omega$. Then we can approximate $|c_1|$ to see that $|c_1| > 1.1 > \frac{1}{\cos(\pi/14)}$. Similarly, relative to $b_1$, $|c_2| > 1.06 > \frac{1}{\cos(\pi/14)}$. Therefore, $T_{b_1}$ does not overlap $T_{c_1}$ or $T_{c_2}$. This also shows that $T_{b_2}$ and $T_{c_2}$ do not overlap since relative to $b_2$, $c_2 = -\omega^4 + c_1 + \omega^4$.

\begin{figure}[htp]
\centering
	\includegraphics[width=2.5in]{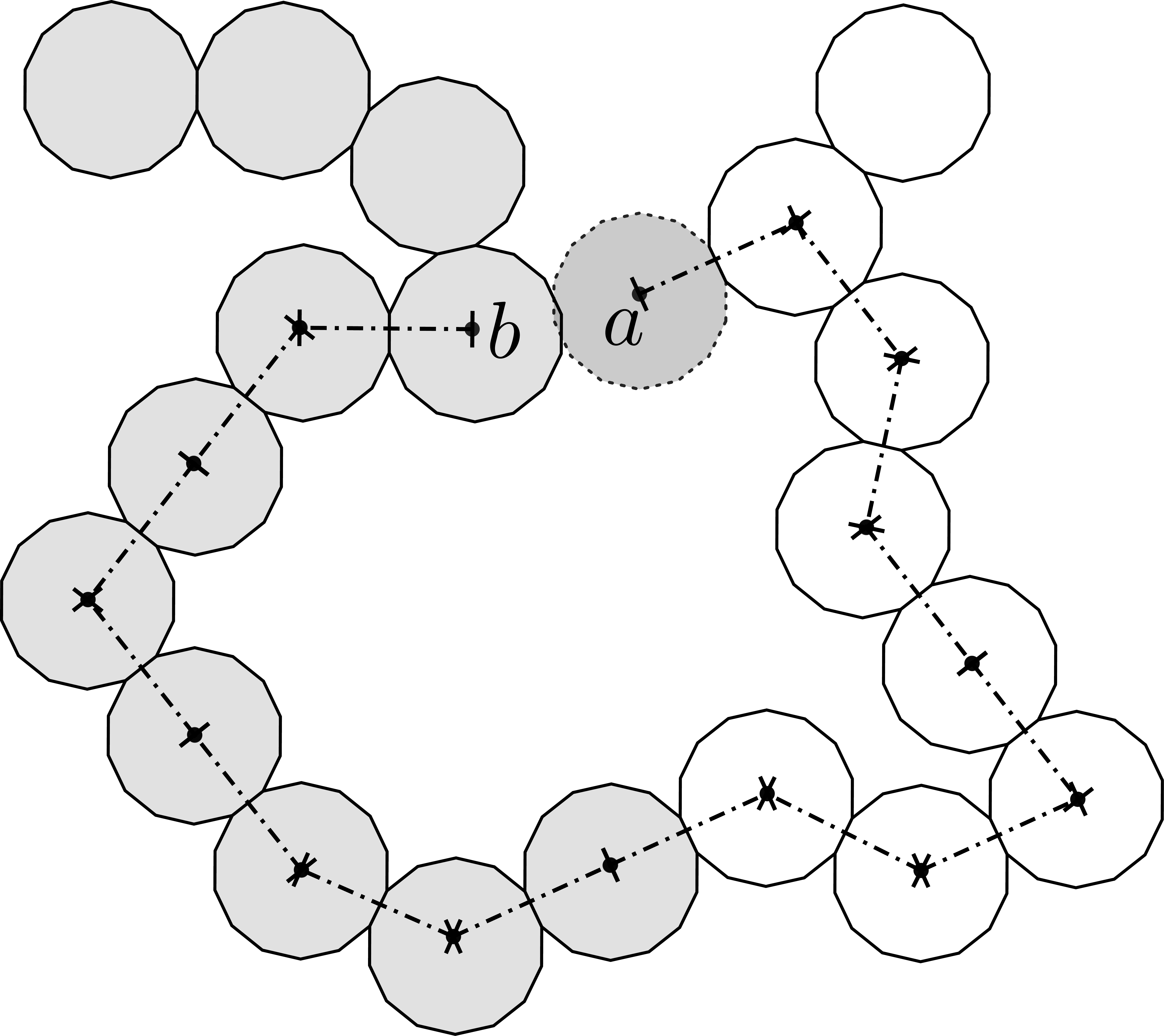}
	\caption{A configuration of the bit-reader given in Figure~\ref{fig:technical-tetradecagonBitReader}. We must show that the tetradecagonal tile centered at $a$ overlaps the tile centered at $b$.}
	\label{fig:tetradecagonBitReader2}
\end{figure}

Now, referring to Figure~\ref{fig:tetradecagonBitReader2}, it remains to be shown that a tetradecagonal tile, which we will denote by $T_a$, centered at $a$ that is in standard orientation and a tetradecagonal tile, which we will denote by $T_b$, centered at $b$ that is in standard orientation overlap. Relative to $b$,
$a = -1 + 2\omega^9 + 2\omega^{12} + \omega^{13} + 2\omega + \omega^{13} + \omega + 2\omega^{5} + \omega^{3} + \omega^{5} + \omega^{8}$. Simplifying $a$, we see $ a= -1 + 3\omega + \omega^{3} + \omega^{5} + \omega^{8} + 2\omega^9 + 2\omega^{13}$. Then we can approximate $|a|$ to see that $|a|< 1$. Therefore, $T_a$ and $T_b$ overlap.

\subsubsection{Polygonal Tile Assembly with Regular Polygonal Tiles with $15$ or More Sides}\label{sec:technical-15+sides}

In the cases where tiles consist of regular polygons with $15$ or more sides, we give a general scheme for obtaining bit-reading gadgets for each case. Figure~\ref{fig:technical-15+sides} depicts the bit-reading gadgets for each case. Note that since each polygonal tile of these bit-reading gadgets abuts another tile, we need only show that for each configuration depicted in Figure~\ref{fig:technical-15+sides}, of the two exposed glues, $g_0$ and $g_1$ of the tile $R$, a tile can only attach to one of these glues depending on the position of the tile $B$ in the figure. In other words, for each configuration depicted in Figure~\ref{fig:technical-15+sides}, we show that the intersection of the interiors of a polygon with the same shape, position and orientation as $B$ and a polygon with the same shape, position and orientation of the gray tile's position and orientation.

\begin{figure}[htp]
\centering
  \subfloat[][Bit-reading gadget configuration for pentadecagonal tiles.]{%
        \label{fig:technical-15+sidesA}%
    		\makebox[.25\textwidth]{
        \includegraphics[width=1.2in]{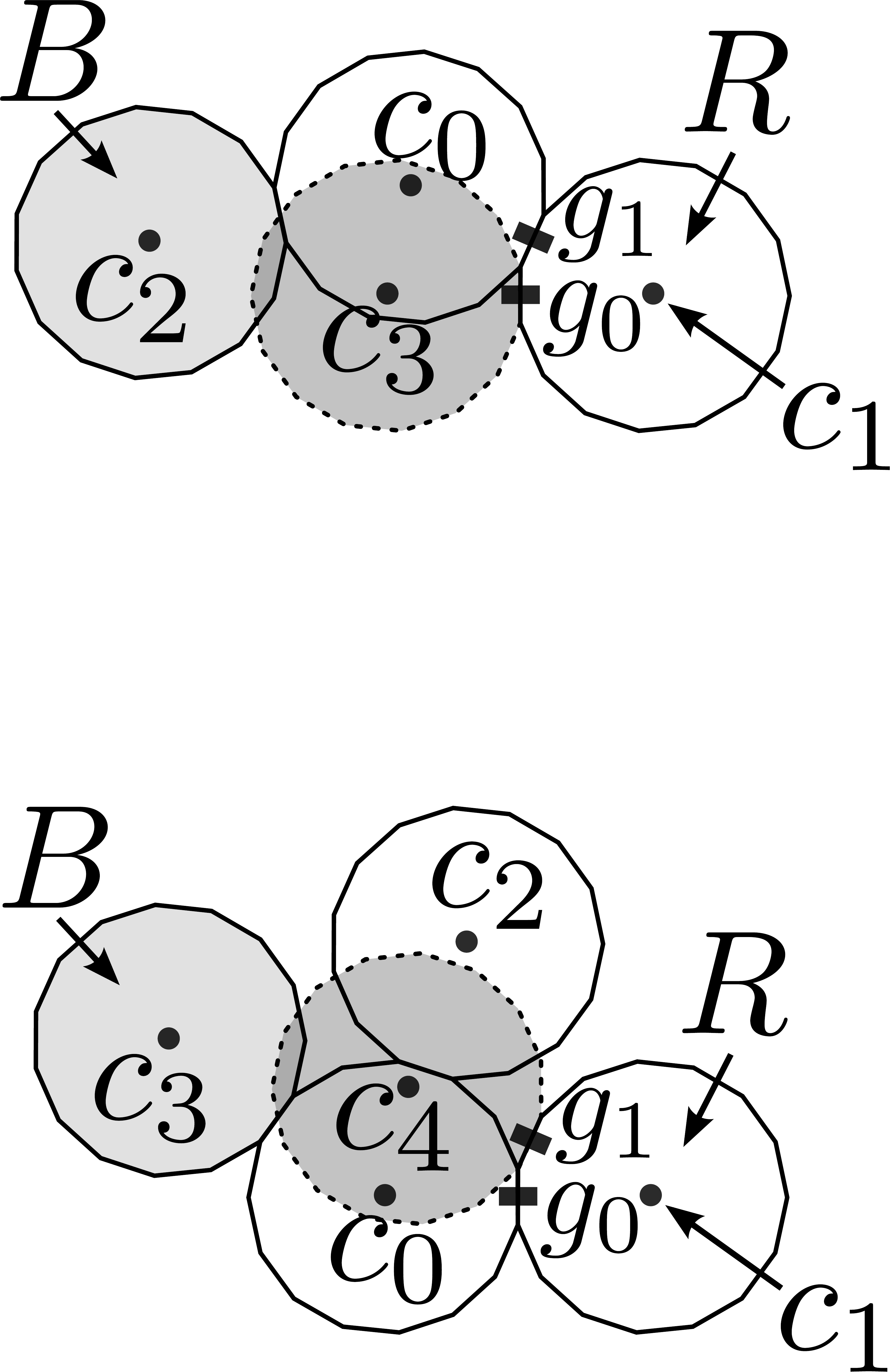}
        }
        }%
        \quad\quad
  \subfloat[][Bit-reading gadget configuration for hexadecagonal tiles.]{%
        \label{fig:technical-15+sidesB}%
    		\makebox[.25\textwidth]{
        \includegraphics[width=1.2in]{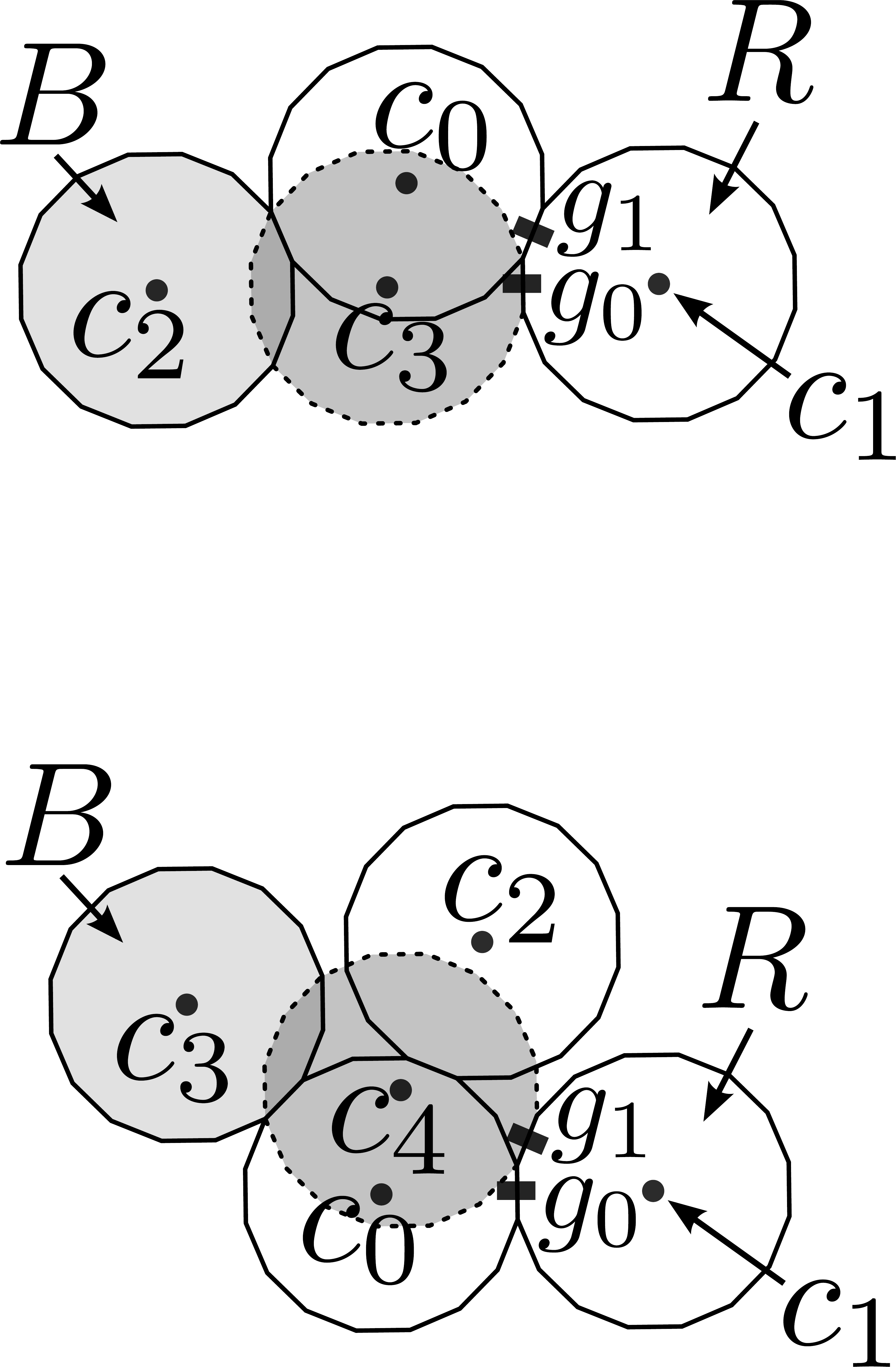}
        }
        }%
       \quad\quad
  \subfloat[][Bit-reading gadget configuration for heptadecagonal tiles.]{%
        \label{fig:technical-15+sidesC}%
    		\makebox[.25\textwidth]{
        \includegraphics[width=1.2in]{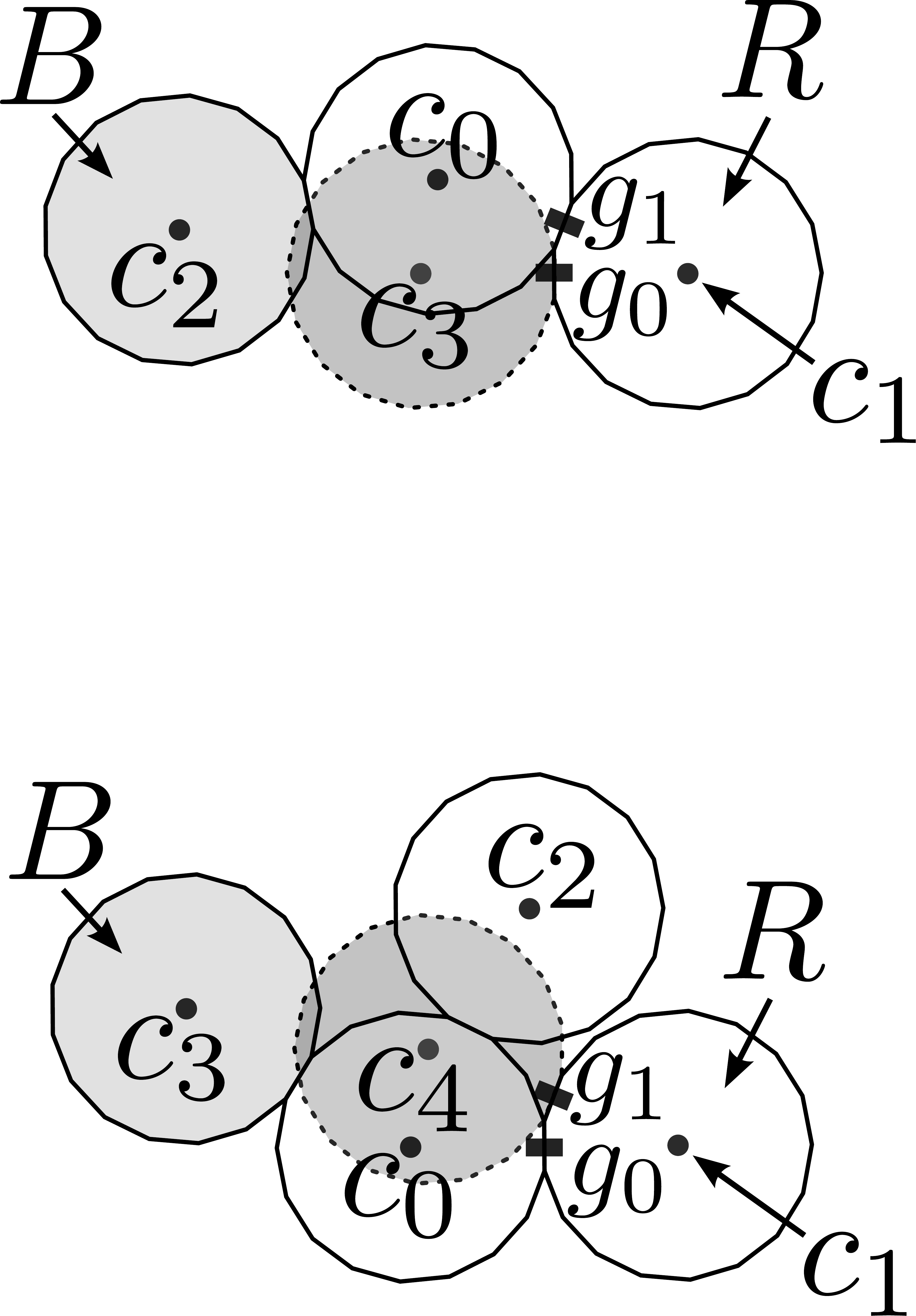}
        }
        }%
  \caption{(a), (b) and (c) each depict two configurations of polygonal tiles which represents either a $0$ (bottom) or a $1$ (top).}
  \label{fig:technical-15+sides}
\end{figure}

Now, consider a polygon $P_n$ with $n\geq 15$ sides and let $\omega$ be the $n^{th}$ root of unity $e^{\frac{2\pi i}{n}}$. Then, the general scheme for constructing a bit-reading gadget falls into two cases. First, if $n$ is odd (the cases where $n$ is even are similar), relative to a tile with negated orientation (the polygon labeled $R$ in the configurations in Figure~\ref{fig:technical-15+sides}), the two configurations that give rise to the bit-reading gadget are as follows. Let $k$ be such that $n=2k+1$ ($n=2k$ if $n$ is even). To ``read'' a $1$, the configuration is obtained by centering a blocker tile with negated orientation, labeled $B$ in the top configurations of Figure~\ref{fig:technical-15+sides}, at $-\omega^{n-1} + \omega^{k+1}$ (whether $n$ is even or odd). Then $R$ exposes two glues $g_1$ and $g_0$ such that if a tile binds to $g_1$, it will have standard orientation and be centered at $-\omega^{n-1}$ (whether $n$ is even or odd) and if a tile that binds to $g_0$, it will have standard orientation and be centered at $-1$. We will show that $B$ will prevent this tile from binding. This gives the configuration depicted in the top figures of Figure~\ref{fig:technical-15+sides}. Similarly, to ``read'' a $0$, the configuration is obtained by centering a blocker tile with negated orientation, labeled $B$ in the bottom configuration of Figure~\ref{fig:technical-15+sidesA}, at $-1 + \omega^{k-1}$ ($-1 + \omega^{k-2}$ if $n$ is even) relative to $R$.  In this case, we will show that $B$ prevents a tile from binding to $g_1$. In addition, we place a glue on the tile that binds to $g_0$ that allows for another tile to bind to it so that its center is at $c_2 = -1 + \omega^{\lfloor \frac{k-1}{2} \rfloor}$ ($c_2 = -1 + \omega^{\frac{k-2}{2}}$ if $n$ is even).  This gives the configuration depicted in the bottom figures of Figure~\ref{fig:technical-15+sidesA} and Figure~\ref{fig:technical-15+sidesC}. Moreover, we show that neither $R$ nor $B$ prevent the binding of this tile.

In order to perform the calculations used to show the correctness of these bit-reading gadgets, we consider the cases where $n$ is even and where $n$ is odd.

\subsubsection*{Case 1: ($n$ is odd)}

Suppose that $n = 2k+1$ for some $k$. To show that a polygon centered at $c_1$ and a polygon centered at $c_2$ do not overlap, consider the case where $k$ is odd. Note that relative to $c_0$, $c_1 = 1$ and $c_2 = \omega^{\frac{k-1}{2}}$. Then the distance $d_n$ from $c_1$ to $c_2$ satisfies the following equation.
$$d_n^2 = \left(1-\cos\left(\frac{\left(k-1\right)\pi}{n}\right)\right)^2 + \sin^2\left(\frac{\left(k-1\right)\pi}{n}\right)$$
Substituting $k = \frac{n-1}{2}$ for $k$ and simplifying, we obtain $d_n^2 = 2+2\sin\left(\frac{3\pi}{2n}\right)$. Now to show that a polygon centered at $c_1$ and a polygon centered at $c_2$ do not overlap, we show that $d_{n}^2 > \frac{1}{\cos^2\left(\frac{\pi}{n}\right)}$ for $n \geq 15$. To see this, note that $\cos^2\left(\frac{\pi}{n}\right)d_n^2 = 2\cos^2\left(\frac{\pi}{n}\right)\left(1+\sin\left(\frac{3\pi}{2n}\right)\right)$. Then for $n\geq15$, $2\cos^2\left(\frac{\pi}{n}\right)\left(1+\sin\left(\frac{3\pi}{2n}\right)\right) > 2\cos^2\left(\frac{\pi}{4}\right) =1$.

It then follows that $d_{n} > \frac{1}{\cos\left(\frac{\pi}{n}\right)}$, and therefore $d_n$ is greater than twice the circumradius of our polygons. Hence, a polygon centered at $c_1$ and a polygon centered at $c_2$ do not overlap.

In the case where $k$ is even, let $m$ be such that $k = 2m$. Then relative to $c_0$, $c_1 = 1$ and $c_2 = \omega^{m-1}$. In this case, $d_n$ satisfies the following equation.

$$d_n^2 = \left(1-\cos\left(\frac{\left(2m-2\right)\pi}{n}\right)\right)^2 + \sin^2\left(\frac{\left(2m-2\right)\pi}{n}\right)$$

\noindent Substituting $m = \frac{k}{2}$ for $m$ and $k = \frac{n-1}{2}$ for $k$ we obtain $d_n^2 = 2+2\sin\left(\frac{5\pi}{2n}\right)$. Now to show that a polygon centered at $c_1$ and a polygon centered at $c_2$ do not overlap, we show that $d_{n}^2 > \frac{1}{\cos^2\left(\frac{\pi}{n}\right)}$ for $n \geq 15$. To see this, note that $\cos^2\left(\frac{\pi}{n}\right)d_n^2 = 2\cos^2\left(\frac{\pi}{n}\right)\left(1+\sin\left(\frac{5\pi}{2n}\right)\right)$. Then for $n\geq15$, $2\cos^2\left(\frac{\pi}{n}\right)\left(1+\sin\left(\frac{5\pi}{2n}\right)\right) > 2\cos^2\left(\frac{\pi}{4}\right) =1$.

It then follows in the case where $k$ is even, $d_{n} > \frac{1}{\cos\left(\frac{\pi}{n}\right)}$, and therefore $d_n$ is greater than twice the circumradius of our polygons. Hence, a polygon centered at $c_1$ and a polygon centered at $c_2$ do not overlap.

Now, to show that a polygon centered at $c_3$ and a polygon centered at $c_4$ overlap, note that relative to $c_1$, $c_3 = -1 + \omega^{k-1}$ and $c_4 = -\omega^{n-1}$. Therefore, the distance $d_n$ from $c_3$ to $c_4$ is satisfies the following equation.
\begin{align*}
d_n^2 &= \left( -1 + \cos\left(\frac{2(k-1)\pi}{n}\right) + \cos\left( \frac{2(n-1)\pi}{n} \right) \right)^2\\
      &\ \ \ \ + \left( \sin\left( \frac{2(k-1)\pi}{n} \right) + \sin\left( \frac{2(n-1)\pi}{n} \right) \right)^2
\end{align*}

\noindent Substituting $k = \frac{n-1}{2}$ for $k$ and simplifying, we obtain, $$d_n^2 = 1 + 2\left(2\sin^2\left(\frac{\pi}{n}\right)\left(1-2\cos\left(\frac{\pi}{n}\right)\right)\right)$$. Note that for each $n>2$, $d_n^2 < 1$. To see this, it suffices to show that $$2\sin^2\left(\frac{\pi}{n}\right)\left(1-2\cos\left(\frac{\pi}{n}\right)\right) < 0$$. This follows from the fact that $2\sin^2\left(\frac{\pi}{n}\right) > 0$ and $1-2\cos\left(\frac{\pi}{n}\right) < 0$ for $n >2$.

Now, since for each $n>2$, $d_n^2 < 1$, we see that $d_n < 1$. Since the length of the apothem for each tile is assumed to be $\frac{1}{2}$, we can conclude that a polygon centered at $c_3$ and a polygon centered at $c_4$ must overlap.

\subsubsection*{Case 2: ($n$ is even)}

Let $k$ be such that $n = 2k$. Then, relative to $c_0$, $c_1 = 1$ and $c_2 = \omega^{\lfloor\frac{k-2}{2}\rfloor}$. Then the distance, $d_n$ say, from $c_1$ to $c_2$ satisfies the following equation
$$d_n = \left(1-\cos\left(\frac{\left(k-2\right)\pi}{n}\right)\right)^2 + \sin^2\left(\frac{\left(k-2\right)\pi}{n}\right)$$.

\noindent Substituting $k = \frac{n}{2}$ for $k$ and simplifying, we obtain $d_n^2 = 2+2\sin\left(\frac{2\pi}{n}\right)$. To show that $c_0$ and $c_1$ do not overlap, it suffices to show that $d_n^2 > \frac{1}{\cos^2\left( \frac{\pi}{n} \right)}$.  To see this, note that $\cos^2\left(\frac{\pi}{n}\right)d_n^2 = 2\cos^2\left(\frac{\pi}{n}\right)\left(1+\sin\left(\frac{2\pi}{n}\right)\right)$. Then for $n\geq16$, $2\cos^2\left(\frac{\pi}{n}\right)\left(1+\sin\left(\frac{2\pi}{n}\right)\right) > 2\cos^2\left(\frac{\pi}{4}\right) = 1$.

As in the case where $n$ is odd, it then follows that in the case where $n$ is even, $d_{n} > \frac{1}{\cos\left(\frac{\pi}{n}\right)}$, and therefore $d_n$ is greater than twice the circumradius of our polygons. Hence, a polygon centered at $c_1$ and a polygon centered at $c_2$ do not overlap.

Now, to show that a polygon centered at $c_3$ and a polygon centered at $c_4$ overlap, note that relative to $c_1$, $c_3 = -1 + \omega^{k-2}$ and $c_4 = -\omega^{n-1}$. Therefore, the distance $d_n$ from $c_3$ to $c_4$ is satisfies the following equation.

\begin{align*}
d_n^2 &= \left( -1 + \cos\left(\frac{2(k-2)\pi}{n}\right) + \cos\left( \frac{2(n-1)\pi}{n} \right) \right)^2 \\
      &\ \ \ \ + \left( \sin\left( \frac{2(k-2)\pi}{n} \right) + \sin\left( \frac{2(n-1)\pi}{n} \right) \right)^2
\end{align*}

\noindent Substituting $k = \frac{n}{2}$ for $k$ and simplifying, we obtain,
$d_n^2 = 1-8\left(\sin^2\left(\frac{\pi}{n}\right)\cos\left(\frac{2\pi}{n}\right)\right)$.
Note that for each $n\geq 16$, $d_n^2 < 1$. To see this, it suffices to show that $-8\left(\sin^2\left(\frac{\pi}{n}\right)\cos\left(\frac{2\pi}{n}\right)\right) < 0$. This follows from the fact that $8\sin^2\left(\frac{\pi}{n}\right) > 0$ and $\cos\left(\frac{2\pi}{n}\right) > 0$ for $n > 16$.

Now, since for each $n\geq 16$, $d_n^2 < 1$, we see that $d_n < 1$. Since the length of the apothem for each tile is assumed to be $\frac{1}{2}$, we can conclude that a polygon centered at $c_3$ and a polygon centered at $c_4$ must overlap.

\subsection{2-shaped systems with regular polygonal tiles}\label{sec:technical-2shaped}

The following figures give configurations for normalized on-grid bit-reading gadgets that can be used to obtain bit-reading assemblies for 2-shaped systems where the tiles of the system have the shape of two different regular polygons. Note that the grid construction techniques from Section~\ref{sec:bit-grid-main} can be used to obtain the grids shown using dashed lines in the figures below. 

\newcolumntype{M}{>{\centering\arraybackslash}m{\dimexpr.45\linewidth-2\tabcolsep}}
\begin{table}[htp]
\centering
\begin{tabular}{| M | M |}
	\hline
    \includegraphics[width=1.6in]{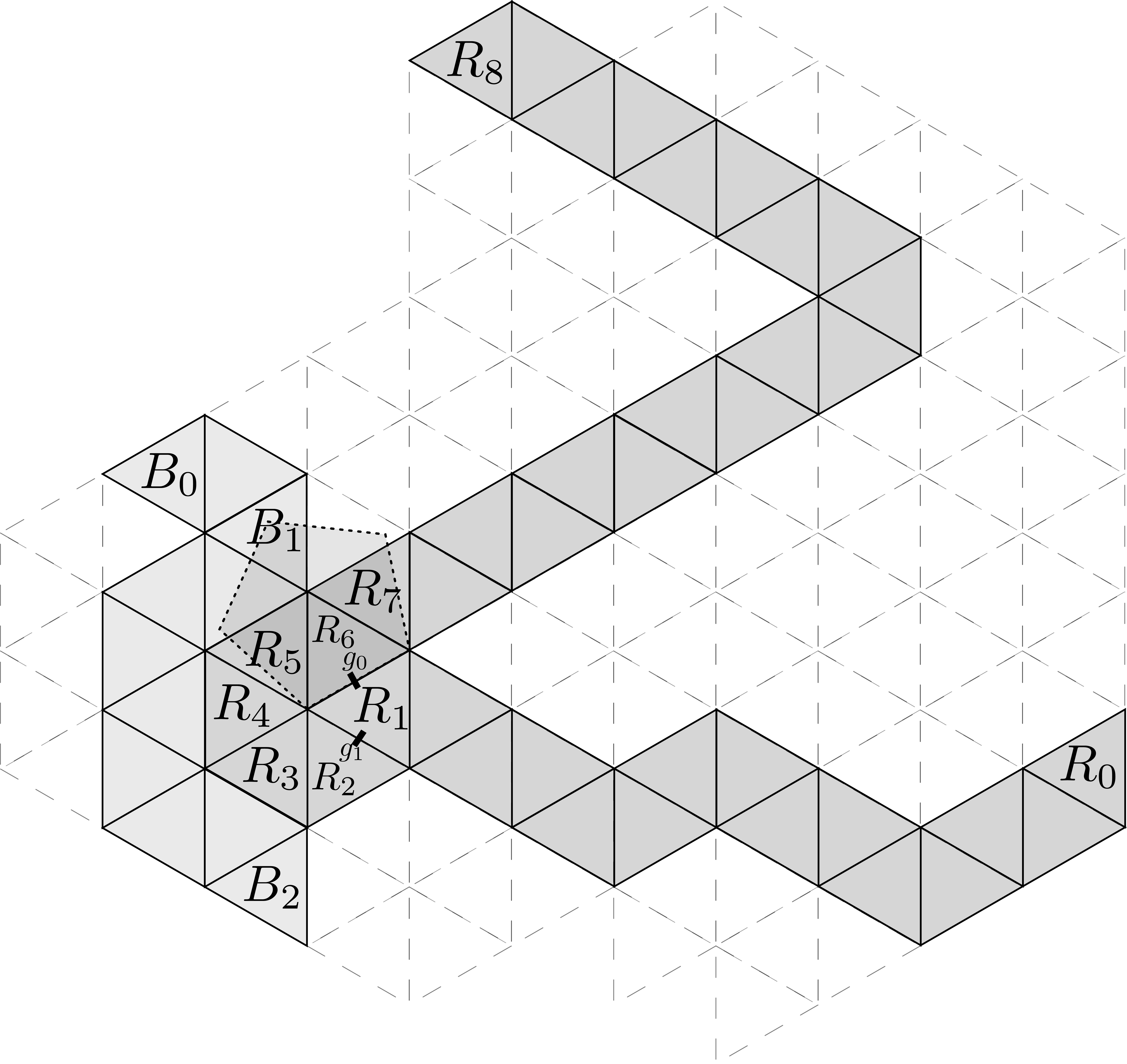} &
	\includegraphics[width=1.6in]{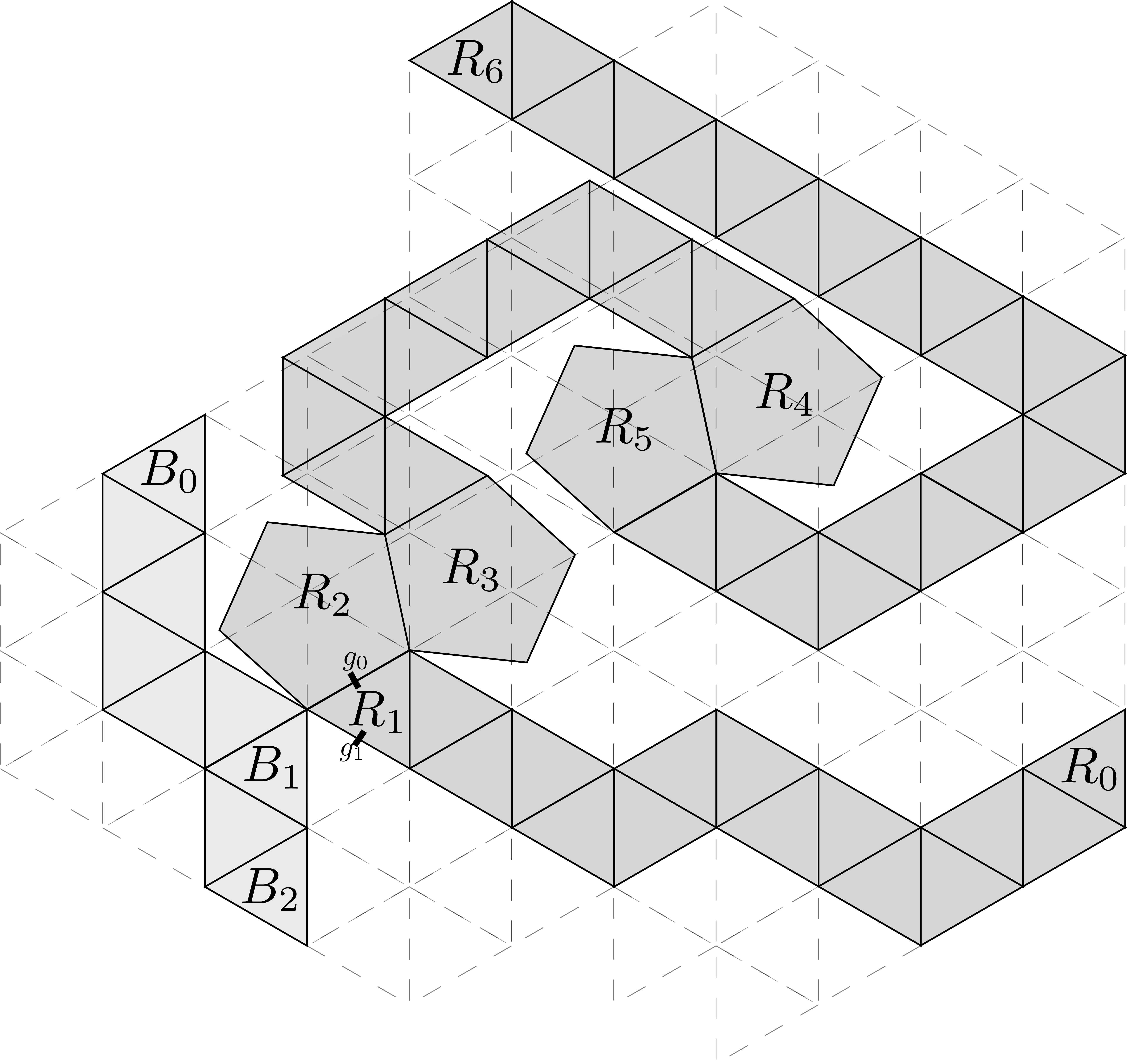}  \\
	(a) & (b) \\\hline
\end{tabular}\vspace{1ex}\caption{Configurations for normalized on-grid bit-reading gadgets that can be used for 2-shaped systems using whose tiles have the shape of a regular triangle and a regular pentagon. (a) represents a $0$, and (b) represents a $1$.}\label{tbl:mixed3s}\vspace{-20pt}
\end{table}

\newcolumntype{M}{>{\centering\arraybackslash}m{\dimexpr.45\linewidth-2\tabcolsep}}
\begin{table}[htp]
\centering
\begin{tabular}{| M | M |}
	\hline
    \includegraphics[width=1.2in]{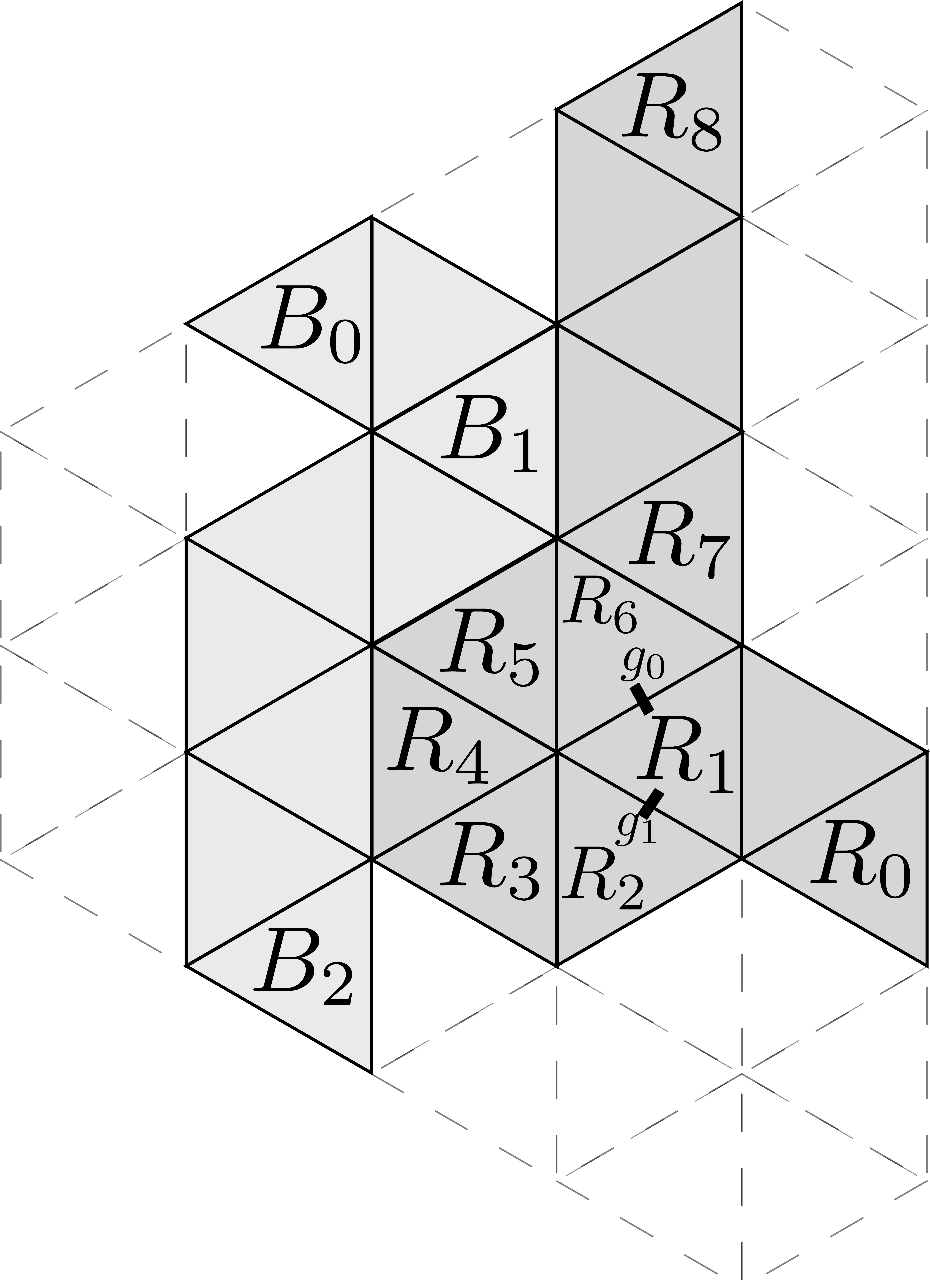} &
	\includegraphics[width=1.2in]{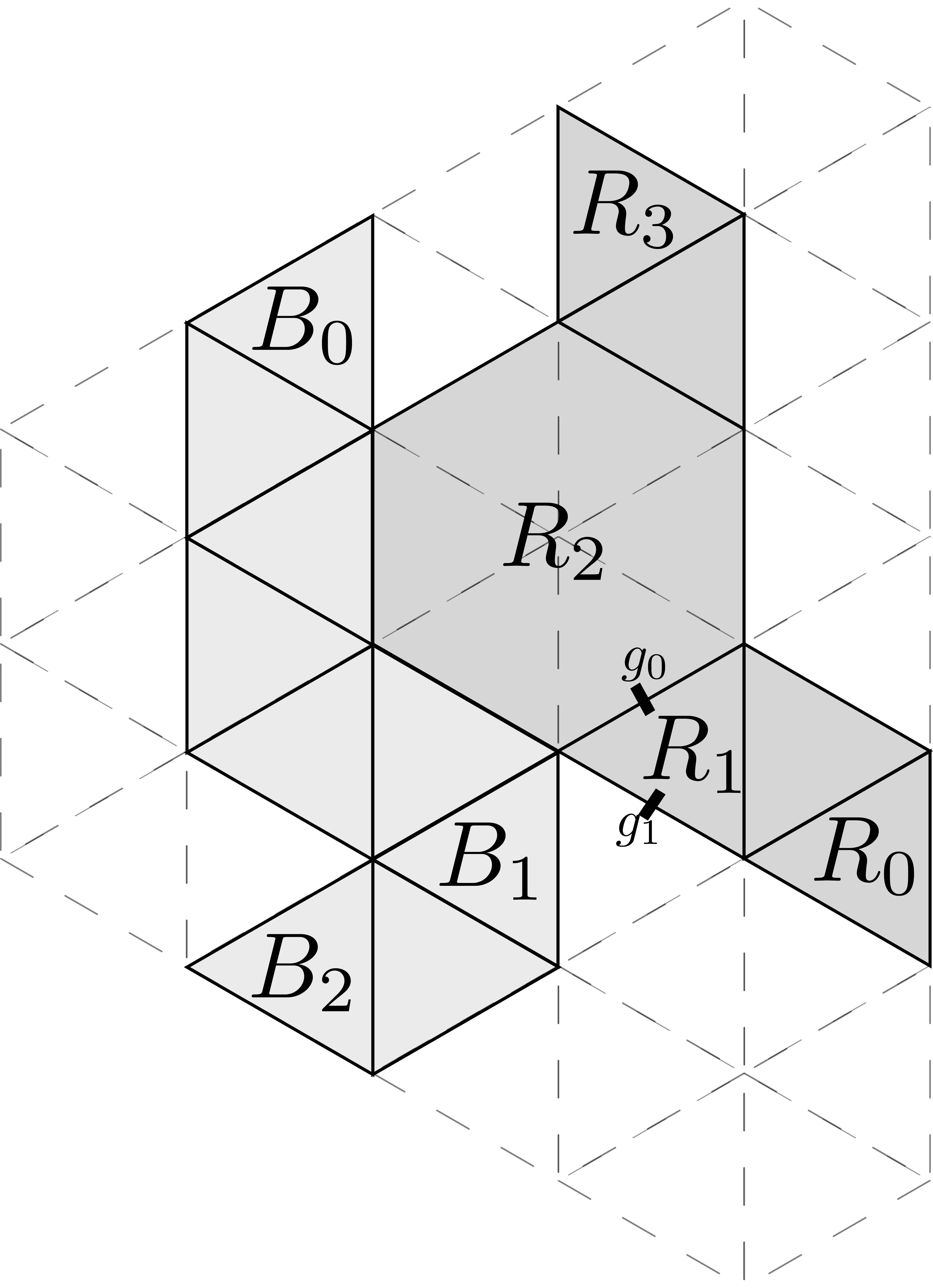}  \\
	(a) & (b) \\\hline
\end{tabular}\vspace{1ex}\caption{Configurations for normalized on-grid bit-reading gadgets that can be used for 2-shaped systems using whose tiles have the shape of a regular triangle and a regular hexagon. (a) represents a $0$, and (b) represents a $1$.}\label{tbl:mixed3s}\vspace{-20pt}
\end{table}

\bgroup
\newcolumntype{M}{>{\centering\arraybackslash}m{\dimexpr.5\linewidth-2\tabcolsep}}
\begin{table}[htp]
\centering
\begin{tabular}{| M | M |}
	\hline
	\includegraphics[trim=0 0 0 -150,width=2in]{images/mixed45A}  & \includegraphics[trim=0 0 0 -150,width=1.84in]{images/mixed45B} \\
	(a) & (b)  \\\hline
\end{tabular}\vspace{1ex}\caption{Configurations of for normalized on-grid bit-reading gadgets that can be used for 2-shaped systems using whose tiles have the shape of a square and a regular pentagon. (a) represents a $0$, and (b) represents a $1$.}\label{tbl:mixed4s}\vspace{-20pt}
\end{table}
\egroup

\bgroup
\newcolumntype{M}{>{\centering\arraybackslash}m{\dimexpr.5\linewidth-2\tabcolsep}}
\begin{table}[htp]
\centering
\begin{tabular}{| M | M |}
	\hline
	\includegraphics[trim=0 0 0 -150,width=2in]{images/mixed46A}  & \includegraphics[trim=0 0 0 -150,width=1.84in]{images/mixed46B} \\
	(d) & (e)  \\\hline
\end{tabular}\vspace{1ex}\caption{Configurations for normalized on-grid bit-reading gadgets that can be used for 2-shaped systems using whose tiles have the shape of a square and a regular hexagon. (a) represents a $0$, and (b) represents a $1$.}\label{tbl:mixed4s}\vspace{-20pt}
\end{table}
\egroup

\begin{figure}[htp]
\centering
        \includegraphics[width=3in]{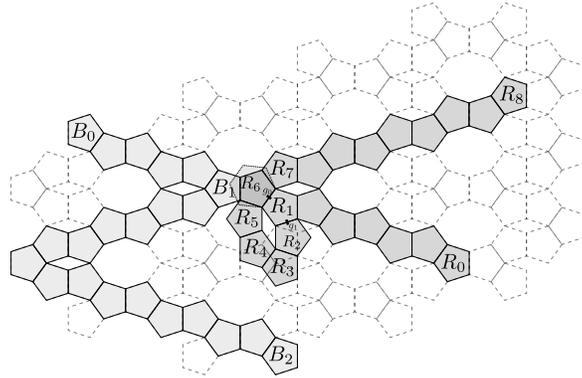}
  \caption{Bit-reading gadget configuration for tiles with the shape of either a pentagon or a hexagon. This figure depicts a configuration of polygonal tiles which represents a $0$, while Figure~\ref{fig:mixed56B} depicts a configuration of polygonal tiles which represents a $1$.}
  \label{fig:mixed56A}
\end{figure}

\begin{figure}[htp]
\centering
        \includegraphics[width=3in]{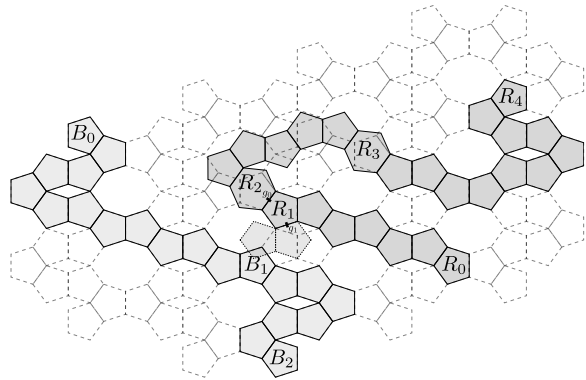}
  \caption{This figure depicts a configuration of polygonal tiles with the shape of either a pentagon or a hexagon which represents a $0$.}
  \label{fig:mixed56B}
\end{figure}

\bigskip\bigskip\bigskip\bigskip\bigskip\bigskip\bigskip\bigskip\bigskip %
\subsection{Single shaped systems with equilateral polygonal tiles}\label{sec:technical-equilaterals}

\newcolumntype{M}{>{\centering\arraybackslash}m{\dimexpr.45\linewidth-2\tabcolsep}}
\begin{table}[!htp]
\centering
\begin{tabular}{| M | M |}
	\hline
    \includegraphics[width=1.2in]{images/equilateral-pentagonA} &
	\includegraphics[width=1in]{images/equilateral-hexagonA}  \\
	(a) & (c) \\\hline
	\includegraphics[width=1.2in]{images/equilateral-pentagonB} &
	\includegraphics[width=1in]{images/equilateral-hexagonB}  \\
	(b) & (d) \\\hline
\end{tabular}\vspace{1ex}\caption{Configurations of for normalized on-grid bit-reading gadgets that can be used for 1-shaped systems using whose tiles have the shape of a particular equilateral pentagon ((a) and (b)) or a particular equilateral hexagon ((c) and (d)).}\label{tbl:equilaterals}\vspace{-20pt}
\end{table}

} %
\fi

\vspace{-15pt}
\bibliographystyle{amsplain}
\bibliography{tam,experimental_refs}

\ifabstract
\newpage
\appendix

\begin{center}
	\Huge\bfseries
	Appendix
\end{center}

\magicappendix
\fi

\end{document}